\documentclass[aoas]{imsart}

\RequirePackage{amsthm,amsmath,amsfonts,amssymb}
\RequirePackage[authoryear]{natbib}
\RequirePackage[colorlinks,citecolor=blue,urlcolor=blue]{hyperref}
\RequirePackage{graphicx}

\usepackage{algorithm}
\usepackage{algorithmic}
\usepackage{cleveref}
\usepackage{caption}
\usepackage{subcaption}
\usepackage{booktabs}
\usepackage{pifont}

\startlocaldefs

\newtheorem{theorem}{Theorem}[section]

\newtheorem{proposition}[theorem]{Proposition}
\theoremstyle{remark}
\newtheorem{definition}[theorem]{Definition}

\newcommand{\R}{\mathbb{R}}
\newcommand{\paren}[1]{\left({#1}\right)}
\newcommand{\set}[1]{\left\{{#1}\right\}}
\newcommand{\abs}[1]{\left|{#1}\right|}
\newcommand{\defeq}{\overset{\text{def.}}{=}}
\newcommand{\cmark}{\ding{51}}
\newcommand{\xmark}{\ding{55}}
\endlocaldefs

\begin{document}

\begin{frontmatter}
\title{On Numerical Considerations for Riemannian Manifold Hamiltonian Monte Carlo}
\runtitle{Numerical Considerations for Riemannian Hamiltonian Monte Carlo}

\begin{aug}
\author[A]{\fnms{James A.} \snm{Brofos}\ead[label=e1]{james.brofos@yale.edu}},
\and
\author[A]{\fnms{Roy R.} \snm{Lederman}}
\address[A]{Yale University, \printead{e1}}
\end{aug}

\begin{abstract}
Riemannian manifold Hamiltonian Monte Carlo (RMHMC) is a sampling algorithm that seeks to adapt proposals to the local geometry of the posterior distribution. The specific form of the Hamiltonian used in RMHMC necessitates {\it implicitly-defined} numerical integrators in order to sustain reversibility and volume-preservation, two properties that are necessary to establish detailed balance of RMHMC. In practice, these implicit equations are solved to a non-zero convergence tolerance via fixed-point iteration. However, the effect of these convergence thresholds on the ergodicity and computational efficiency properties of RMHMC are not well understood. The purpose of this research is to elucidate these relationships through numerous case studies. Our analysis reveals circumstances wherein the RMHMC algorithm is sensitive, and insensitive, to these convergence tolerances.
Our empirical analysis examines several aspects of the computation: (i) we examine the ergodicity of the RMHMC Markov chain by employing statistical methods for comparing probability measures based on collections of samples; (ii) we investigate the degree to which detailed balance is violated by measuring errors in reversibility and volume-preservation; (iii) we assess the efficiency of the RMHMC Markov chain in terms of time-normalized ESS. In each of these cases, we investigate the sensitivity of these metrics to the convergence threshold and further contextualize our results in terms of comparison against Euclidean HMC.
We propose a method by which one may select the convergence tolerance within a Bayesian inference application using techniques of stochastic approximation and we examine Newton's method, an alternative to fixed point iterations, which can eliminate much of the sensitivity of RMHMC to the convergence threshold.
\end{abstract}


\end{frontmatter}
\tableofcontents

\section{Introduction}

Bayesian inference provides a theoretical basis for reasoning under uncertainty in statistical applications. In particular, let $q\in\R^m$ be a parameter vector of interest and let $\mathcal{L}(q)$ denote the log-density of $q$; in many applications of interest $\mathcal{L}(q)$ will be, up to a additive constant, the sum of the log-likelihood of data given $q$ and a log-density representing a prior belief over $q$. A central task in Bayesian inference is the estimation of expectations of functions of the random variable $q$. However, since $m$ may be large, computing expectations via an $m$-dimensional integral may be intractable. Instead, in order to compute expectations over the distribution of $q$, we seek to generate samples from the distribution of $q$ and form a Monte Carlo expectation; this idea has led to the development of Markov chain Monte Carlo (MCMC), which require only that $\mathcal{L}(q)$ be known up to an additive constant. See \citet{brooks2011handbook} for an overview of MCMC sampling methods. In MCMC, one establishes a Markov chain whose stationary distribution has density proportional to $\exp(\mathcal{L}(q))$; by running the Markov chain for a ``large'' number of steps, one obtains a sequence of auto-correlated samples whose marginal distributions have converged -- one hopes -- to the distribution of interest. We refer the interested reader to \citet{meyn1993markov} for a discussion of the ergodicity properties of Markov chain.

Hamiltonian Monte Carlo (HMC) \citep{Duane1987216,betancourt2017conceptual,neal-hmc} is a MCMC sampling procedure that seeks to efficiently explore the typical set of a smooth probability density by integrating Hamilton's equations of motion. A variant of HMC that is the focus of this work in Riemannian manifold Hamiltonian Monte Carlo (RMHMC), which seeks to exploit the local geometry of the posterior distribution in order to align proposals along directions of the posterior that exhibit the greatest local variation. The {\it de facto} standard numerical integrator for RMHMC is the generalized leapfrog method \citep{softabs,rmhmc,cobb2019introducing,pmlr-v139-brofos21a}, which involves implicitly defined steps. These steps are typically resolved to a finite precision via fixed point iteration, thereby necessitating the practitioner to choose a suitable convergence tolerance for their posterior. This begs the questions:
\begin{enumerate}
    \item To what extent is detailed balance violated in practice with variable convergence thresholds?
    \item How much do non-zero convergence thresholds effect the ergodicity of the RMHMC Markov chains?
    \item How does the choice of convergence threshold impact the computational efficiency and time-normalized performance of the sampler?
\end{enumerate}
The role of these convergence tolerances in the generalized leapfrog method, and its effect on the ergodicity of the Markov chain, does not appear to have been thoroughly evaluated in the literature. Therefore, the purpose of this work is to elucidate these relationships, paying special attention to the manner in which these convergence tolerances impact the detailed balance condition in RMHMC, for which specialized integrators such as the generalized leapfrog is the principle motivation. In our experimental evaluation, we find that there tends to be a convergence threshold after which point no further improvement in the ergodicity of the Markov chain can be obtained, even though detailed balance may be more precisely enforced by decreasing the convergence threshold. Since sample ergodicity, rather than detailed balance, is what matters in Bayesian analysis, computational benefits -- such as faster sampling -- emerge by careful selection of this convergence parameter. We propose one mechanism, based on Ruppert averaging, by which one may choose a convergence tolerance for a given posterior. We also examine the use of Newton's method as an alternative to fixed point iterations with special attention to the reversibility, volume-preservation, and ergodicity of the resulting method.

The outline of the balance of this work is as follows. In \cref{sec:preliminaries} we review the elements of HMC and RMHMC, the derivation of detailed balance, the generalized leapfrog algorithm, and the transition kernels employed in RMHMC. Then, in \cref{sec:related-work}, we review related work on RMHMC, with special attention toward those works that have investigated the computational efficiency of the technique. \Cref{sec:analytical-apparatus} states the research question we investigate -- to understand the importance of these thresholds in the RMHMC algorithm -- and the measures we employ for this purpose. We then turn our attention to experimentation in \cref{sec:experiments}, which consider non-identifiable models, hierarchical Bayesian logistic regression, hierarchical models with all posterior quantities sampled jointly, a hierachical stochastic volatility and log-Gaussian Cox-Poisson model, Bayesian inference in stochastic differential equations, and a multiscale Student-$t$ distribution.

\section{Preliminaries}\label{sec:preliminaries}

In this section we review the foundations of RMHMC. We begin by deducing that reversibility and volume preservation of the numerical integrator are sufficient to establish detailed balance of HMC. We then discuss a particular numerical integrator that is the focus of our study, the generalized leapfrog method; we review how this numerical integrator is reversible and volume preserving and how these properties are violated when solving implicit relations with fixed point iteration to a prescribed tolerance. We then reproduce the transition kernel of the RMHMC algorithm and carefully distinguish between transition kernels computed with variable convergence thresholds.

\subsection{Notation and Background} 

We denote by $\mathrm{Id}_m$ the $m\times m$ identity matrix and by $0_m$ the $m$-dimensional zero vector. Given a collection of numbers $\{a_1,\ldots, a_n\}$, we denote the median of the collection by $\mathrm{median}(\{a_1,\ldots,a_n\})$. Let $\Omega \subseteq \R^m$; a partition of $\Omega$ is a collection of subsets $\set{A_i}_{i=1}^\infty$ with $A_i\subseteq\Omega$ such that $A_i\cap A_j=\emptyset$ when $i\neq j$ and $\cup_{i=1}^\infty A_i = \Omega$. Given a subset $\Omega$ of $\R^m$, we denote by $\mathrm{Vol}(\Omega)$ its $m$-dimensional volume; i.e. $\mathrm{Vol}(\Omega) = \int_\Omega ~\mathrm{d}x$. Denote by $\mathbb{S}^{m-1}$ the $(m-1)$-dimensional sphere embedded in $\R^m$ in the natural way. We denote the $(m-1)$-dimensional surface element of $\mathbb{S}^{m-1}$ by $\mathrm{dVol}_{\mathbb{S}^{m-1}}$. Given a vector $z\in\R^m$ we adopt the notation $\Vert z\Vert_2 = \sqrt{\sum_{i=1}^m z^2}$ and $\Vert z\Vert_{\infty} = \max_{i=1,\ldots, m} \abs{z_i}$. Let $\Phi:\R^m\to\R^m$ be a smooth map; the divergence of $\Phi$ at $z\in\R^m$ is $\mathrm{div}(\Phi)(z) = \sum_{i=1}^m \frac{\partial}{\partial z_i} \Phi(z)$.

Let $\mathcal{X}$ be a set and let $\mathfrak{B}(\mathcal{X})$ be the Borel $\sigma$-algebra on $\mathcal{X}$. A probability measure is a map $\Pi : \mathfrak{B}(\mathcal{X}) \to [0, 1]$ with the properties that $\Pi(\emptyset)=0$ and $\Pi(\mathcal{X}) = 1$. Moreover, $\Pi$ satisfies the property of countable additivity: for any countable collection $\set{A_i}$ where $A_i\in \mathfrak{B}(\mathcal{X})$ and $A_i\cap A_j = \emptyset$ for $i\neq j$, we have $\Pi(\cup_{i} A_i) = \sum_{i=1} \Pi(A_i)$. The Dirac measure at $x\in\mathcal{X}$ is a probability measure $\delta_x$ with the property that $\delta_x(A) = \mathbf{1}\set{x\in A}$ where $A\in\mathfrak{B}(\mathcal{X})$. Given two probability measures $\Pi$ and $\Pi'$, the total variation distance between them is defined by $\Vert \Pi - \Pi'\Vert_{\mathrm{TV}} = \frac{1}{2} \sup_{A\in\mathfrak{B}(\mathcal{X})} \abs{\Pi(A) - \Pi'(A)}$.

A Markov chain is a sequence of random variables $(X_n)_{n=1}^\infty$ which are determined by the Markov transition kernel. Suppose $X_i=x_i$ for $i=1,\ldots, n$; then $X_{n+1} \vert X_{1}=x_1, \ldots, X_{n}=x_n$ has the same distribution as $X_{n+1}\vert X_n=x_n$ and in particular $X_{n+1}\vert X_n=x_n \sim K(x_n, \cdot)$ where $K : \mathcal{X}\times \mathfrak{B}(\mathcal{X}) \to [0, 1]$ is a mapping with the properties \citep{10.5555/1051451},
\begin{enumerate}
    \item For any $A\in\mathfrak{B}(\mathcal{X})$, $K(\cdot, A)$ is measurable.
    \item For any $x\in\mathcal{X}$, $K(x,\cdot)$ is a probability measure.
\end{enumerate}
A Markov chain is said to be stationary for the distribution $\Pi$ if $\underset{x\sim \Pi}{\mathbb{E}} K(x, A) = \Pi(A)$ for any $A\in\mathfrak{B}(\mathcal{X})$. A condition stronger than stationarity but which is often employed in the practical development of Markov chains for sampling is the detailed balance condition, which says that if $X_n\sim \Pi$ then $\mathrm{Pr}(X_n\in A ~\mathrm{and}~ X_{n+1}\in B) = \mathrm{Pr}(X_n\in B ~\mathrm{and}~ X_{n+1}\in A)$. Given $X_1=x_1$, we write $X_{n}\vert X_n=x_n \sim K^n(x_1, \cdot)$ so that $K^n$ denotes the $n$-step transition probability measure. A Markov chain is said to be ergodic for the distribution $\Pi$ if $\lim_{n\to\infty} \Vert K^n(x, \cdot) - \Pi\Vert_{\mathrm{TV}}$ for every $x\in\mathcal{X}$.

\subsection{Hamiltonian Mechanics}

Formally, given a smooth function $H : \R^m\times\R^m\to\R$, Hamilton's equations of motion are described by the initial value problem,
\begin{align}
  \label{eq:hamiltonian-position-evolution} \dot{q}_t &= \nabla_p H(q_t, p_t) \\
  \label{eq:hamiltonian-momentum-evolution} \dot{p}_t &= -\nabla_q H(q_t, p_t),
\end{align}
with initial condition $(q_0, p_0)\in \R^m\times\R^m$. Physically, the solutions $q_t$ and $p_t$ are referred to as the {\it position} and {\it momentum} variables, respectively, and together $(q_t, p_t)$ represents a point in {\it phase space}. Hamilton's equations of motion possess several special properties, which we now recall; see \citet{mechanics-and-symmetry} for a thorough treatment. First, energy is conserved over the course of the trajectory so that $\frac{\mathrm{d}}{\mathrm{d}t} H(q_t, p_t) = 0$. Second, volume in $\R^m\times\R^m$ is conserved under the evolution prescribed by Hamilton's equations of motion; i.e. $\mathrm{div}((\dot{q}_t, \dot{p}_t)) = 0$. Third, under the dual conditions that $-\nabla_p H(q, -p) = \nabla_p H(q, p)$ and $\nabla_q H(q, -p) = \nabla_q H(q, p)$, which are satisfied by Hamiltonians with quadratic potential energies, the solutions to Hamilton's equations of motion are symmetric under negation of the momentum $\mathbf{F} : p\mapsto -p$. In (Euclidean) HMC, one defines a Hamiltonian of the form,
\begin{align}
  \label{eq:euclidean-hamiltonian} H(q, p) = -\mathcal{L}(q) + \frac{1}{2} p^\top p.
\end{align}
A Hamiltonian of this form is called ``separable'' because it is the sum of two functions, each depending only on position or only on momentum. However, excepting the most elementary forms of $\mathcal{L}$, the equations of motion corresponding to this Hamiltonian are not available in closed-form. Instead, one resorts to numerical approximations of these equations of motion. The most popular numerical integrator for approximating solutions to Hamilton's equations of motion for a separable Hamiltonian is called the leapfrog integrator, pseudo-code for which is given is \cref{alg:leapfrog}. Unlike the analytical equations of motion, however, the leapfrog integrator does not conserve the Hamiltonian energy, necessitating a Metropolis-Hastings accept reject criterion in order to maintain detailed balance. Like the equations of motion that it seeks to approximate, however, the leapfrog integrator is volume preserving and reversible under negation of the momentum; these two properties are critical to the proof that the Metropolis-Hastings criterion used in HMC suffices to obtain detailed balance with respect to the target distribution. Pseudo-code implementing HMC is provided in \cref{alg:hamiltonian-monte-carlo}. In practice the momentum variable $p_n$ is randomized in order to facilitate moving between energy level sets of the distribution. For Hamiltonians of the form \cref{eq:euclidean-hamiltonian}, we sample $p_n\sim\mathrm{Normal}(0_m, \mathrm{Id}_m)$, which can be viewed as a Gibbs sampling step in between Metropolis-Hastings updates of $q_n$.

\begin{algorithm}[t!]
  \caption{The leapfrog integrator, which is a volume-preserving, reversible,
    second-order accurate numerical integrator for separable Hamiltonians. The
    leapfrog integrator is fully explicit, involving no fixed point iterations,
    giving it an important computational advantage over Riemannian methods.}
  \label{alg:leapfrog}
  \begin{algorithmic}
    \STATE {\bf Input:} An initial state $(q_n, p_n)\in \R^m\times\R^m$, a smooth, separable Hamiltonian $H:\R^m\times\R^m\to\R$, a step-size $\epsilon\in\R$.
    \STATE Compute an explicit half-step in the momentum variable,
    \begin{align}
      p_{n+1/2} = p_n - \frac{\epsilon}{2} \nabla_q H(q_n, p_{n}).
    \end{align}
    \STATE Compute an explicit full step in the position variable,
    \begin{align}
      q_{n+1} = q_n + \epsilon\nabla_p H(q_n, p_{n+1/2}).
    \end{align}
    \STATE Apply an explicit half-step in the momentum variable,
    \begin{align}
      p_{n+1} = p_{n+1/2} - \frac{\epsilon}{2} \nabla_q H(q_{n+1}, p_{n+1/2}).
    \end{align}
    \STATE {\bf Return:} $(q_{n+1}, p_{n+1})$, the next position in phase-space along the integrated trajectory.
  \end{algorithmic}
\end{algorithm}

\begin{algorithm}[t!]
  \caption{Pseudo-code implementing the Hamiltonian Monte Carlo (HMC) algorithm. Given an initial position, a random momentum is generated and a reversible, volume preserving integrator is applied to this location in phase space. A Metropolis-Hastings correction is then applied to the output of the integrator in order to ensure detailed balance. When the metric $\mathbf{G}$ is constant with respect to $q$, the resulting Hamiltonian is separable and the leapfrog integrator may be used; on the other hand, if $\mathbf{G}$ depends on position then the generalized leapfrog integrator should be used to ensure reversibility and volume preservation.}
  \label{alg:hamiltonian-monte-carlo}
  \begin{algorithmic}
    \STATE {\bf Input:} An initial state $q_n\in \R^m$ and momentum $p_n \in\R^m$, a smooth log-density
    $\mathcal{L}:\R^m\to\R$, a Riemannian metric $\mathbf{G}:\R^m\to\R^m\times\R^m$, a reversible and volume-preserving numerical integrator $\Phi_\epsilon^k : \R^m\times\R^m\to\R^m\times\R^m$ consisting of $k\in\mathbb{N}$ steps with step-size $\epsilon\in \R$.
    \STATE Define the function
    \begin{align}
      H(q, p) = -\mathcal{L}(q) + \frac{1}{2} \log \mathrm{det}(\mathbf{G}(q)) + \frac{1}{2} p^\top\mathbf{G}(q)^{-1}p.
    \end{align}
    \STATE Apply the numerical integrator to obtain the proposal state for the Markov chain $(\tilde{q}_{n+1}, \hat{p}_{n+1}) = \Phi_\epsilon^k(q_n, p_n)$.
    \STATE Apply the momentum flip operator to produce a self-inverse proposal $\tilde{p}_{n+1} = -\hat{p}_{n+1}$.
    \STATE Apply the Metropolis-Hastings accept-reject decision.
    \STATE Set $M\gets H(q_n, p_n) - H(\tilde{q}_{n+1}, \tilde{p}_{n+1})$ and sample $u\sim\mathrm{Uniform}(0, 1)$.
    \IF{$u < M$}
    \STATE Accept the proposal and set $(q_{n+1}, p_{n+1}) \gets (\tilde{q}_{n+1}, \tilde{p}_{n+1})$.
    \ELSE
    \STATE Reject the proposal and set $(q_{n+1}, p_{n+1}) \gets (q_n, p_n)$.
    \ENDIF
    \STATE {\bf Return:} The next state of the Markov chain $(q_{n+1}, p_{n+1})$.
  \end{algorithmic}
\end{algorithm}

One difficulty that arises in HMC is that if the density has multiple spatial scales, employing the leapfrog integrator with a Hamiltonian in the form of \cref{eq:euclidean-hamiltonian} will produce severe oscillations in the directions of smallest variation; see \cref{subfig:banana-trajectory-threshold-precondition} for an illustration. This can be alleviated by employing a preconditioner that continuously adapts the momentum to the local variation of the posterior \citep{rmhmc}. Riemannian manifold Hamiltonian Monte Carlo (RMHMC) is a MCMC sampling procedure that seeks to adapt proposals to the directions of greatest variation in the posterior locally. This is accomplished by introducing second-order geometric information into the Hamiltonian Monte Carlo transition operator; second-order information is typically represented by the sum of the Fisher information matrix of the log-likelihood and the negative Hessian of the log-prior, or by the Hessian of the log-density $\mathcal{L}(\theta)$. As in HMC, the proposal generated in RMHMC is obtained by numerically integrating Hamilton's equations of motion corresponding to the non-separable Hamiltonian,
\begin{align}
  \label{eq:riemannian-hamiltonian} H(q, p) = -\mathcal{L}(q) + \frac{1}{2}\log\mathrm{det}(\mathbf{G}(q)) + \frac{1}{2}p^\top \mathbf{G}(q)^{-1}p,
\end{align}
where $\mathcal{L} : \R^m\to\R$ is the posterior log-density. This form of Hamiltonian is motivated by geometric principles; the Hamiltonian $H(q, p) = \frac{1}{2} p^\top \mathbf{G}(q)^{-1}p$ corresponds to co-geodesic motion on $(\R^m, \mathbf{G})$ \citep{Calin2004GeometricMO}. Thus, physically, $\mathcal{L}(q) + \frac{1}{2}\log\mathrm{det}(\mathbf{G}(q))$ represents a potential energy function that causes motion to deviate from the co-geodesics of the manifold. As conceived by \citet{rmhmc}, when the density $(q, p)$ is proportional to $\exp(-H(q, p))$, the term $\frac{1}{2}\log\mathrm{det}(\mathbf{G}(q))$ is included in the Hamiltonian so that we obtain the conditional distribution $p\vert q\sim \mathrm{Normal}(0_m, \mathbf{G}(q))$. The marginal distribution of $q$ has density $\pi(q) \propto \exp(\mathcal{L}(q))$. The quadratic form in $p$ causes the form of the Hamiltonian to be ``non-separable;'' this means that the Hamiltonian is not expressible as the sum of two functions, one a function of the position alone and the other a function of the momentum alone.  As a result, this Hamiltonian cannot be integrated by the (non-generalized) leapfrog method in a way that preserves reversibility and volume preservation. The form of \cref{eq:riemannian-hamiltonian} produces the following equations of motion:
\begin{align}
  \dot{q}_t &= \mathbf{G}(q_t)^{-1} p_t \\
  \dot{p}_t &= \nabla_q \mathcal{L}(q_t) - \frac{1}{2} \mathrm{trace}\paren{\mathbf{G}(q_t)^{-1}\nabla \mathbf{G}(q_t)} + \frac{1}{2} p_t^\top \mathbf{G}(q_t)^{-1} \nabla \mathbf{G}(q_t) \mathbf{G}(q_t)^{-1} p_t,
\end{align}
where $\mathrm{trace}(\mathbf{G}(q)\nabla \mathbf{G}(q))$ is an $m$-dimensional vector whose $k$-th element is,
\begin{align}
    \mathrm{trace}(\mathbf{G}(q)^{-1}\nabla \mathbf{G}(q))_{k} = \mathrm{trace}(\mathbf{G}(q)^{-1} \frac{\partial}{\partial q_k} \mathbf{G}(q)).
\end{align}
Similarly, the $k$-th element of $p^\top \mathbf{G}(q)^{-1} \nabla \mathbf{G}(q) \mathbf{G}(q)^{-1} p$ is,
\begin{align}
    (p^\top \mathbf{G}(q)^{-1} \nabla \mathbf{G}(q) \mathbf{G}(q)^{-1} p)_k = p^\top \mathbf{G}(q)^{-1} \frac{\partial}{\partial q_k} \mathbf{G}(q) \mathbf{G}(q)^{-1} p.
\end{align}

\subsection{Establishing Detailed Balance in Hamiltonian Monte Carlo}

\begin{definition}
  A map $\Phi:\R^{2m}\to\R^{2m}$ is called self-inverse if $\Phi\circ\Phi = \mathrm{Id}$.
\end{definition}
Self-inverse maps are also called involutions.
\begin{definition}
  A smooth map $\Phi:\R^{2m}\to\R^{2m}$ is called volume preserving if $\mathrm{div}(\nabla \Phi)(z) = 0$ for all $z\in\R^{2m}$.
\end{definition}

\begin{theorem}
  Let $\Phi$ be a self-inverse and volume-preserving map from $\R^{2m}$ to itself. Consider a probability distribution whose density is proportional to $\exp(-H(q, p))$ where $H:\R^m\times\R^m\to\R$ is a smooth function. Consider the Markov chain transition kernel defined as,
  \begin{align}
    \begin{split}
      K((q, p), B) &= \min\set{1, \exp(H(q, p) - H(\Phi(q, p)))} \cdot\mathbf{1}\set{\Phi(q, p)\in B} \\
      &\qquad + (1 - \min\set{1, \exp(H(q, p) - H(\Phi(q, p)))})\cdot \mathbf{1}\set{(q, p)\in B},
    \end{split}
  \end{align}
  where $B$ is a Borel subset of $\R^{2m}$. Then the stationary distribution of $K$ is the distribution with density $\pi(q, p) \propto \exp(-H(q, p))$.
\end{theorem}
\begin{proof}
  The following proof is similar to that in \citet{pmlr-v139-brofos21a}. It suffices to show that the transition kernel satisfies detailed balance with respect to the density $\pi$. Let $A$ and $B$ be Borel subsets of $\R^{2m}$. Let $z$ be a random variable drawn from the distribution with density $\pi(z)$ and suppose that $z'\vert z \sim K(z, \cdot)$. Then,
  \begin{align}
  \begin{split}
      \mathrm{Pr}\left[z\in A~\mathrm{and}~ z'\in B\right] &= \int_A \pi(z) \min\set{1, \frac{\pi(\Phi(z))}{\pi(z)}} \mathbf{1}\set{\Phi(z)\in B} ~\mathrm{d}z \\&\qquad + ~\int_A \pi(z) \paren{1 - \min\set{1, \frac{\pi(\Phi(z))}{\pi(z)}}}\mathrm{1}\set{z\in B}~\mathrm{d}z
  \end{split} \\
  \begin{split}
      &= \int_{A\cap \Phi(B)} \pi(z) \min\set{1, \frac{\pi(\Phi(z))}{\pi(z)}} ~\mathrm{d}z \\&\qquad + ~\int_B \pi(z) \paren{1 - \min\set{1, \frac{\pi(\Phi(z))}{\pi(z)}}}\mathrm{1}\set{z\in A}~\mathrm{d}z
  \end{split} \\
  \begin{split}
      \label{eq:detailed-balance-change-of-variables} &= \int_{\Phi(A)\cap B} \pi(\Phi(z')) \min\set{1, \frac{\pi(z')}{\pi(\Phi(z'))}} ~\mathrm{d}z' \\&\qquad + ~\int_B \pi(z) \paren{1 - \min\set{1, \frac{\pi(\Phi(z))}{\pi(z)}}}\mathrm{1}\set{z\in A}~\mathrm{d}z
  \end{split} \\
  \begin{split}
      &= \int_{\Phi(A)\cap B} \pi(z') \min\set{1, \frac{\pi(\Phi(z'))}{\pi(z')}} ~\mathrm{d}z' \\&\qquad + ~\int_B \pi(z) \paren{1 - \min\set{1, \frac{\pi(\Phi(z))}{\pi(z)}}}\mathrm{1}\set{z\in A}~\mathrm{d}z
  \end{split} \\
  \begin{split}
      &= \int_{B} \mathbf{1}\set{z'\in \Phi(A)} \pi(z') \min\set{1, \frac{\pi(\Phi(z'))}{\pi(z')}} ~\mathrm{d}z' \\&\qquad + ~\int_B \pi(z) \paren{1 - \min\set{1, \frac{\pi(\Phi(z))}{\pi(z)}}}\mathrm{1}\set{z\in A}~\mathrm{d}z
  \end{split} \\
  &= \mathrm{Pr}\left[z\in B~\mathrm{and}~z'\in A\right]
  \end{align}
  The fact that distribution whose density is $\pi$ is the stationary distribution of such a Markov chain then follows immediately from the choice $A = \R^{2m}$; mathematically,
  \begin{align}
      \mathrm{Pr}\left[z\in\R^{2m}~\mathrm{and}~z'\in B\right] = \mathrm{Pr}\left[z'\in B\right] = \mathrm{Pr}\left[z\in B\right].
  \end{align}
  The fact that $\Phi$ is a self-inverse function with unit Jacobian determinant is used in the change-of-variables in \cref{eq:detailed-balance-change-of-variables}. Thus, we see how these two properties play a critical role in establishing detailed balance of the HMC algorithm.
\end{proof}

\subsection{The Generalized Leapfrog Integrator}\label{subsec:generalized-leapfrog-integrator}

\begin{algorithm}[t!]
  \caption{The generalized leapfrog integrator for producing a symmetric,
    volume-preserving, second-order accurate to Hamilton's equations of motion
    for an arbitrary smooth Hamiltonian. This procedure implements a single step
    of the generalized leapfrog integrator, which may then be repeatedly
    composed in order to integrate over longer trajectories.}
  \label{alg:moral-leapfrog}
  \begin{algorithmic}
    \STATE {\bf Input:} An initial state $(q_n, p_n)\in \R^m\times\R^m$, a smooth Hamiltonian $H:\R^m\times\R^m\to\R$, a step-size $\epsilon\in\R$.
    \STATE Compute a half-step in the momentum variable,
    \begin{align}
      \label{eq:moral-leapfrog-momentum-implicit} p_{n+1/2} = p_n - \frac{\epsilon}{2} \nabla_q H(q_n, p_{n+1/2}).
    \end{align}
    \STATE Compute a full step in the position variable,
    \begin{align}
      \label{eq:moral-leapfrog-position} q_{n+1} = q_n + \frac{\epsilon}{2} \paren{\nabla_p H(q_n, p_{n+1/2}) + \nabla_p H(q_{n+1}, p_{n+1/2})}.
    \end{align}
    \STATE Apply an explicit half-step update to the momentum variable,
    \begin{align}
      \label{eq:moral-leapfrog-momentum-explicit} p_{n+1} = p_{n+1/2} - \frac{\epsilon}{2} \nabla_q H(q_{n+1}, p_{n+1/2}).
    \end{align}
    \STATE {\bf Return:} $(q_{n+1}, p_{n+1})$, the next position in phase-space along the integrated trajectory.
  \end{algorithmic}
\end{algorithm}

\begin{algorithm}[t!]
  \caption{In practice the generalized leapfrog integrator is implemented with
    fixed point iterations, as advocated by \citet{}, to resolve the
    implicitly-defined relations. As a result, the integrator violates
    reversibility and volume preservation above the level of numerical
    precision.}
  \label{alg:threshold-leapfrog}
  \begin{algorithmic}
    \STATE {\bf Input:} An initial state $(q_n, p_n)\in \R^m\times\R^m$, a
    smooth Hamiltonian $H:\R^m\times\R^m\to\R$, a step-size $\epsilon\in\R$, a
    convergence threshold $\delta > 0$, a maximal number of allowable fixed
    point iterations $M\in \mathbb{N}$.
    \STATE Compute a half-step in the momentum variable:
    \STATE Set $p_{n+1/2} \gets p_n$, $l_p \gets 0$, $\Delta_p \gets \infty$.
    \WHILE{$\Delta_p > \delta$ and $l_p < M$}
    \STATE
    \begin{align}
      \label{eq:generalized-leapfrog-momentum-fixed-point-i} p_{n+1/2}' &\gets p_n - \frac{\epsilon}{2} \nabla_q H(q_n, p_{n+1/2}) \\
      \label{eq:generalized-leapfrog-momentum-fixed-point-ii} \Delta_p &\gets \Vert p_{n+1/2}' - p_{n+1/2}\Vert_\infty
    \end{align}
    \STATE Assign $p_{n+1/2} \gets p_{n+1/2}'$ and set $l_p \gets l_p + 1$.
    \ENDWHILE
    \STATE Compute a full step in the position variable:
    \STATE Set $q_{n+1} \gets q_n$, $l_q\gets 0$, $\Delta_q \gets \infty$.
    \WHILE{$\Delta_q > \delta$ and $l_q < M$}
    \STATE
    \begin{align}
      \label{eq:generalized-leapfrog-position-fixed-point-i} q_{n+1}' &\gets q_n + \frac{\epsilon}{2} \paren{\nabla_p H(q_n, p_{n+1/2}) + \nabla_p H(q_{n+1}, p_{n+1/2})} \\
      \label{eq:generalized-leapfrog-position-fixed-point-ii} \Delta_q &\gets \Vert q_{n+1}' - q_{n+1}\Vert_\infty
    \end{align}
    \STATE Assign $q_{n+1} \gets q_{n+1}'$ and set $l_q \gets l_q + 1$.
    \ENDWHILE
    \STATE Compute an explicit half-step update to the momentum variable:
    \begin{align}
      p_{n+1} \gets p_{n+1/2} - \frac{\epsilon}{2} \nabla_q H(q_{n+1}, p_{n + 1/2}).
    \end{align}
    \STATE {\bf Return:} $(q_{n+1}, p_{n+1})$, the next position in phase-space along the integrated trajectory.
  \end{algorithmic}
\end{algorithm}

The generalized leapfrog integrator \citep{leimkuhler_reich_2005,10.2307/43692571,PhysRev.159.98} is a numerical method that provides a reversible, volume-preserving, and second-order accurate approximate solution to Hamilton's equations of motion \cref{eq:hamiltonian-position-evolution,eq:hamiltonian-momentum-evolution} for an arbitrary smooth Hamiltonian $H$. By combining the generalized leapfrog integrator with the negation of the momentum variable, one obtains a self-inverse and volume-preserving function, as described in \cref{app:reversibility-volume-preservation-generalized-leapfrog}.

The form of the Hamiltonian in \cref{eq:riemannian-hamiltonian} introduces computational complexities into the need for reversible and volume-preserving numerical integration. The most common numerical integrator for non-separable Hamiltonians is the generalized leapfrog method, which maintains reversibility and volume-preservation, but at the cost of introducing implicitly-defined equations for the updates to the position and momentum variables. In practice, these implicit updates are solved to a prescribed convergence tolerance via fixed point iteration, as advocated by \citet{hairer-geometric}. However, in \cref{subsec:alternatives-newtons-method} we discuss using Newton's method to resolve the implicit update to the momentum variable in the generalized leapfrog integrator. Over the years, some guidance has been provided as to the implementation of the fixed-point solution; in \citet{rmhmc}, the authors suggest five or six fixed point iterations whereas \citet{softabs} pursued fixed-point solutions to within a given convergence tolerance. We will use the second of these alternatives since it clearly defines the accuracy of the solution that we seek to obtain. Moreover, unlike a prescribed number of iterations, a convergence tolerance can avoid unnecessary additional iterations once the threshold has been achieved.

A crucial distinction in our analysis lies in the difference between the exact generalized leapfrog integrator (\cref{alg:moral-leapfrog}), which is reversible and volume preserving, and the practical implementation of the generalized leapfrog integrator using fixed point iterations to resolve the implicit updates (\cref{alg:threshold-leapfrog}), which is neither. (We note that the use of the term ``exact'' to describe \cref{alg:moral-leapfrog} does not mean that the output of the integrator are exact solutions to Hamilton's equations of motion, but merely that the steps of the integrator are defined exactly by solutions to implicit equations.) The degree to which an implementation of the generalized leapfrog algorithm exhibits reversibility and volume preservation is controlled by the level of precision one demands from the solution to the fixed point equation: the greater the precision exacted by the threshold, the better the fidelity of the method to true reversibility and volume preservation. On the other hand, smaller thresholds require more fixed point iterations, thereby reducing the computational expediency of the method. The proof that the exact implementation of the generalized leapfrog method preserves volume invokes the symmetry of partial derivatives; namely, that $\nabla_q \nabla_p H(q, p) = \nabla_p\nabla_q H(q, p)$. In most cases, this condition will hold; however, we mention it now in preparation for \cref{subsec:experiment-fitzhugh-nagumo} wherein we will examine two circumstance -- one benign and the other more sinister -- where this assumption is violated.

We illustrate the difference between the exact implementation of the generalized leapfrog integrator and a practical implementation by giving pseudo-code for each procedure in \cref{alg:moral-leapfrog} and \cref{alg:threshold-leapfrog}, respectively. We emphasize this distinction by writing $\Phi_\epsilon^k :\R^m\times\R^m\to\R^m\times\R^m$ as the expression for $k$-steps of the exact generalized leapfrog integrator with step-size $\epsilon$ and $\Phi_\epsilon^k((\cdot, \cdot); \delta)$ for an implementation of the generalized leapfrog method depending on the convergence threshold $\delta > 0$. For Hamiltonians of the form of \cref{eq:riemannian-hamiltonian}, the composition $\mathbf{F}\circ \Phi_\epsilon^k$ is a self-inverse, volume-preserving map; the map $\mathbf{F}\circ \Phi_\epsilon^k((\cdot,\cdot), \delta)$ is neither self-inverse nor volume-preserving, but the error in these quantities is controllable by the threshold $\delta$. As a further implementation detail, we note that fixed point iterations may diverge; in this case, in order to avert a scenario wherein the fixed point iterations continue indefinitely, it is common to impose an upper bound on the allowable number of fixed point iterations.

We briefly remark on the special case when $\frac{\partial}{\partial q_i} \mathbf{G}(q) = \mathbf{0}$ for $i=1,\ldots, m$ in Hamiltonians of the form in \cref{eq:riemannian-hamiltonian}. In this case, the implicit update to the momentum becomes,
\begin{align}
    p_{n+1/2} &= p_n - \frac{\epsilon}{2} \left[-\nabla \mathcal{L}(q_n) + \frac{1}{2} \mathrm{trace}(\mathbf{G}^{-1}(q_n) \nabla \mathbf{G}(q_n) + \frac{1}{2} p_{n+1/2}^\top \mathbf{G}^{-1}(q_n) \nabla \mathbf{G}(q_n) \mathbf{G}^{-1}(q_n) p_{n+1/2} \right] \\
    &= p_n + \frac{\epsilon}{2} \nabla \mathcal{L}(q_n)
\end{align}
so that the first update to momentum becomes explicit. Similarly, the implicit update to position is,
\begin{align}
    q_{n+1} &= q_n + \frac{\epsilon}{2}\left[ \mathbf{G}^{-1}(q_n)p_{n+1/2} + \mathbf{G}^{-1}(q_{n+1}) p_{n+1/2}\right] \\
    &= q_n + \frac{\epsilon}{2} \left[ \mathbf{G}^{-1}p_{n+1/2} + \mathbf{G}^{-1} p_{n+1/2} \right] \\
    &= q_n + \epsilon \mathbf{G}^{-1}p_{n+1/2},
\end{align}
so that the update to position also becomes explicit in this case. In fact, these simplifications show that the generalized leapfrog method reduces to the leapfrog method in \cref{alg:leapfrog} when the metric is constant. This leads us to conclude that when the metric is constant, the convergence tolerance becomes irrelevant. Intuitively, in cases wherein the metric is ``close to constant'' the convergence tolerance will not have a large effect on the ergodicity of the algorithm.

\subsubsection{Computational Complexity of Implicit Updates}

The generalized leapfrog integrator involves two fixed point iterations. The first of these, given in \cref{eq:generalized-leapfrog-momentum-fixed-point-i,eq:generalized-leapfrog-momentum-fixed-point-ii}, provides a half-step update to the momentum variable. A principle advantage of the generalized leapfrog method is that it facilitates caching of reusable quantities from iteration to iteration. In the context of the first momentum update, the computation of the gradient of the log-posterior, the Riemannian metric and its inverse, and the Jacobian of the Riemannian metric can be cached from the explicit update to the momentum in \cref{eq:moral-leapfrog-momentum-explicit}; this averts recomputing these quantities within each fixed-point iteration. However, at iteration $n$, if we define $v_{n+1/2} = \mathbf{G}(q_n)^{-1}p_{n+1/2}$ then the quantity $v_{n+1/2}^\top \frac{\partial \mathbf{G}(q_n)}{\partial q_i} v_{n+1/2}$ incurs a computational cost $\mathcal{O}(m^2)$. Since this computation is replicated for $i=1,\ldots,m$, the total computational complexity to update the momentum is $\mathcal{O}(m^3)$. 

The second fixed-point iteration provides, given in \cref{eq:generalized-leapfrog-position-fixed-point-i,eq:generalized-leapfrog-position-fixed-point-ii} a full-step update to the position variable. At iteration $n$, $v_{n+1/2} = \mathbf{G}(q_n)^{-1} p_{n+1/2}$ may be cached but $\mathbf{G}(q_{n+1})^{-1} p_{n+1/2}$ requires the the Riemannian metric, and its inverse, be recomputed at each iteration. Because matrix inversion scales as $\mathcal{O}(m^3)$, this position update shares the same overall computational complexity with the fixed point update to the momentum. Although by this analysis both implicitly defined updates scale as $\mathcal{O}(m^3)$, the update to position differs fundamentally in the sense that it must compute new quantities from the posterior at each fixed point iteration; namely, the Riemannian metric must be recomputed for every new candidate solution to the fixed point equation to update the position variable. As a practical matter, the computational effort required to compute the Riemannian metric may be more substantial than that required to invert it, such as in the case of \cref{subsec:experiment-fitzhugh-nagumo} wherein the Riemannian metric takes values in $\R^{3\times 3}$ but whose computation requires the solution to an eight-dimensional initial value problem at two-hundred predetermined locations in time. Therefore, in these cases we expect the implicit update to the position to be the more costly of the two.

The final step of the generalized leapfrog integrator requires us to compute the gradient of the log-posterior, the Riemannian metric and its inverse, and the Jacobian of the Riemannian metric in order to give an explicit half-step update to the momentum variable. We measure the total computational complexity of the generalized leapfrog integrator by calculating the number of fixed point iterations used in updating the momentum and position variables. The number of fixed point iterations depends on the convergence threshold, thereby revealing that the computational efficiency of the RMHMC algorithm will depend on the threshold. Concretely, the quantities $l_q$ and $l_p$ defined in \cref{alg:threshold-leapfrog} represent the number of fixed point iterations required by the integrator to solve the implicit relations to within a prescribed tolerance.

\subsection{The Implicit Midpoint Integrator}

\begin{algorithm}[t!]
  \caption{The implicit midpoint integrator for producing a symmetric,
    volume-preserving, second-order accurate to Hamilton's equations of motion
    for an arbitrary smooth Hamiltonian. This procedure implements a single step
    of the implicit midpoint integrator, which may then be repeatedly
    composed in order to integrate over longer trajectories.}
  \label{alg:moral-implicit-midpoint}
  \begin{algorithmic}
    \STATE {\bf Input:} An initial state $(q_n, p_n)\in \R^m\times\R^m$, a smooth Hamiltonian $H:\R^m\times\R^m\to\R$, a step-size $\epsilon\in\R$.
    \STATE Solve the implicit equation,
    \begin{align}
      \label{eq:moral-implicit-midpoint} \begin{pmatrix} q_{n+1} \\ p_{n+1} \end{pmatrix} &= \begin{pmatrix} q_n \\ p_n \end{pmatrix} + \epsilon \begin{pmatrix} \nabla_p H(\bar{q}_n, \bar{p}_n) \\ -\nabla_q H(\bar{q}_n, \bar{p}_n) \end{pmatrix} \\
      \bar{q}_n &= \frac{q_{n+1} + q_n}{2} \\
      \bar{p}_n &= \frac{p_{n+1} + p_n}{2}
    \end{align}
    \STATE {\bf Return:} $(q_{n+1}, p_{n+1})$, the next position in phase-space along the integrated trajectory.
  \end{algorithmic}
\end{algorithm}

The focus of our experimental evaluation is on the generalized leapfrog integrator and the sensitivity of its performance to the choice of convergence threshold. However, in \cref{subsec:experiment-banana-shaped}, we illustrate an alternative numerical method, the implicit midpoint integrator, which averts certain pathologies associated to the generalized leapfrog method in that example. Like the generalized leapfrog method, the implicit midpoint integrator is a symmetric, volume-preserving, second-order integrator. We give the exact implementation of the implicit midpoint method in \cref{alg:moral-implicit-midpoint} and note that the implicit relation in \cref{eq:moral-implicit-midpoint} is resolved by fixed point iteration in practice. Further experimental evaluation of the implicit midpoint integrator, including an evaluation of the relative conservation of volume and reversibility error exhibited by the implicit midpoint method and the generalized leapfrog, may be found in \citet{pmlr-v139-brofos21a}.

\subsection{The Transition Kernel of Riemannian Manifold Hamiltonian Monte Carlo}

Denoting the step-size of the numerical integrator by $\epsilon$ and the number
of integration steps by $k\in\mathbb{N}$, the transition kernel, expressing the
probability of transitioning to a Borel subset $A\subset\R^m\times\R^m$ given the
current state $(q, p)$ in phase-space is given by,
\begin{align}
  \begin{split}
    K((q, p), A; \epsilon, k) &= \min\set{1, \exp(H(q, p) - H(\mathbf{F}\circ\Phi_\epsilon^k(q, p))} \cdot \mathbf{1}\{\mathbf{F}\circ\Phi_{\epsilon}^k(q, p)\in A\} \\
    &\qquad~+ \paren{1 - \min\set{1, \exp(H(q, p) - H(\mathbf{F}\circ\Phi_{\epsilon}^k(q, p))}} \cdot \mathbf{1}\{(q, p)\in A\},
  \end{split}
\end{align}
Recognizing
the transition kernel as a simple mixture distribution of two Dirac measures, we
may sample from $K((q, p), \cdot; \epsilon, k)$ by defining the function,
\begin{align}
  \Psi_\epsilon^k((q, p), u) = \begin{cases}
    \mathbf{F}\circ\Phi_{\epsilon}^k(q, p) &~\mathrm{if}~ u <  \min\set{1, \exp(H(q, p) - H(\mathbf{F}\circ\Phi_{\epsilon}^k(q, p))} \\
    (q, p) &~\mathrm{otherwise}. 
  \end{cases}
\end{align}
Therefore $\Psi_\epsilon^k((q, p), u) \sim K((q, p), \cdot; \epsilon, k)$ when
$u\sim\mathrm{Uniform}(0, 1)$. Pseudo-code demonstrating Riemannian HMC using this transition kernel, and randomization of the momentum using Gibbs sampling, is shown in \cref{alg:hamiltonian-monte-carlo}. Given an initial $(q_0, p_0)$, we may sample $(q_1, p_1) \sim K((q_0, p_0),\cdot; \epsilon, k)$ and, in general, generate $(q_{n+1}, p_{n+1}) \sim K((q_n, p_n), \cdot; \epsilon, k)$. This produces a Markov chain whose marginal distribution in the $q$-variable, at stationarity, has density $\pi(q)\propto \exp(\mathcal{L}(q))$.

As described in \cref{subsec:generalized-leapfrog-integrator}, an actual implementation of RMHMC necessitates the introducting of fixed point iterations in order to resolve the implicitly-defined updates. Therefore, in practice, we sample from the transition kernel by replacing the exact leapfrog method by its thresholded counterpart:
\begin{align}
  \Psi_\epsilon^k((q, p), u; \delta) = \begin{cases}
    \mathbf{F}\circ\Phi_{\epsilon}^k(q, p;\delta) &~\mathrm{if}~ u <  \min\set{1, \exp(H(q, p) - H(\mathbf{F}\circ\Phi_{\epsilon}^k(q, p; \delta))} \\
    (q, p) &~\mathrm{otherwise}. 
  \end{cases}
\end{align}
When $u\sim\mathrm{Uniform}(0, 1)$, we can regard $\Psi_\epsilon^k((q, p), u; \delta)$ as an approximate sample from $K((q, p), \cdot; \epsilon, k)$; the use of a non-zero convergence tolerance prevents the sample from being exact to numerical precision. Instead, we will adopt the notation $K((q, p), \cdot; \epsilon, k, \delta)$ to denote the distribution of $\Psi_\epsilon^k((q, p), u; \delta)$ when $u\sim\mathrm{Uniform}(0, 1)$.

\subsection{Stochastic Approximation}

We will propose a mechanism to adapt the convergence threshold during sampling as a component of a burn-in phase. Consider a function $B : \R\to\R$ and suppose we seek to identify a $\delta\in\R$ such that $B(\delta) = \kappa$ for some $\kappa \in \R$. In the case where $B$ is directly computable, this could be accomplished via a line search or, under suitable smoothness assumptions, via Newton's method. Suppose that $B$ is not directly computable but that there exists a surrogate {\it random variable} $L(\delta)$ such that $\mathbf{E} L(\delta) = B(\delta)$. \Citet{10.1214/aoms/1177729586} proposed a sequential method by which to identify $\delta$ such that $B(\delta) = \kappa$ which depends only on a sequence $L(\delta_n)$. Specifically, they proposed the sequence,
\begin{align}
    \label{eq:primal-sequence} \delta_{n+1} = \delta_n - \gamma_n (L(\delta_n) - \kappa),
\end{align}
where $(\gamma_n)_{n=1}^\infty$ is a sequence of step-sizes. \Citet{10.1214/aoms/1177729586} demonstrated that $\mathbb{E} (\delta_n-\delta)^2 \to 0$ under the following conditions:
\begin{enumerate}
    \item $\sum_{i=1}^\infty \gamma_n = \infty$ and $\sum_{i=1}^\infty \gamma_n^2 < \infty$.
    \item The function $B$ is monotonically increasing.
    \item There exists $\delta\in\R$ such that $B(\delta) = \kappa$.
    \item The function $B$ is differentiable in a neighborhood of $\delta$ and $\frac{\mathrm{d}}{\mathrm{d}\delta} B(\delta) > 0$.
\end{enumerate}
\Citet{ruppert-dual-averaging} proposed to study step-size sequences of the form $\gamma_n = D n^{-\omega}$ for $\omega\in (1/2, 1)$ and $D\in\R_+$. With the sequence $\set{\delta_n}_{n=1}^\infty$ constructed as in \cref{eq:primal-sequence}, we then construct the average sequence,
\begin{align}
    \label{eq:dual-sequence} \bar{\delta}_{n+1} = \frac{n}{n+1} \bar{\delta}_n + \frac{1}{n+1} \delta_n.
\end{align}
The sequence $\set{\bar{\delta}_n}_{n=1}^\infty$ is, in fact, asymptotically efficient as an estimator of $\delta$ \citep{ruppert-dual-averaging}. Methods similar to this averaging procedure have been considered previously in application to MCMC; perhaps most notable among these is the dual averaging method of \citet{JMLR:v15:hoffman14a}, which seeks to tune the step-size of HMC in order to achieve a desired acceptance probability.

\subsection{Convergence of Fixed Point Iterations and Newton's Method}

Fixed point iteration can be used to solve an implicit equation of the form $f(z) = z$ where $z\in \R^m$ and $f:\R^m\to\R^m$; equations of this form appear in the generalized leapfrog (see \cref{eq:generalized-leapfrog-first-momentum,eq:generalized-leapfrog-position}) and implicit midpoint numerical integrators (\cref{eq:moral-implicit-midpoint}). By rearranging terms we arrive at the equivalent problem: find $z$ such that $g(z) = f(z) - z = 0$, which relates the solution to fixed point equations to root-finding methods. In examining these methods, we first recall the definition of order of convergence.
\begin{definition}
  Let $(z_0,z_1,z_2,\ldots)$ be an $\R^m$-valued sequence converging to some $z\in \R^m$. The sequence is said to convergence with order $k$ to the value $z$ if,
  \begin{align}
      \lim_{n\to\infty} \frac{\Vert z_{n+1} - z\Vert}{\Vert z_n - z\Vert^k} < L
  \end{align}
  for some $L\in \R$.
\end{definition}
The case $k=1$ is called linear convergence while the case $k=2$ is called quadratic convergence.
\begin{proposition}
Let $f : \R^m\to\R^m$ be a contraction map (i.e. $\Vert f(z) - f(z')\Vert \leq L \Vert z - z'\Vert$ for some $L < 1$). Then, from any initial position $z_0$, the sequence formed by $z_{n+1} = f(z_n)$ converges with order one to $z$.
\end{proposition}
\begin{proof}
Convergence of the iterates $(z_0,z_1,\ldots)$ is an immediate consequence of the Banach fixed point theorem. To see that convergence is order one, write
\begin{align}
    \frac{\Vert z_{n+1} - z\Vert}{\Vert z_n - z\Vert} &= \frac{\Vert f(z_{n}) - f(z)\Vert}{\Vert z_n - z\Vert} \\
    &\leq L \\
    &< L+\epsilon
\end{align}
by the assumption that $f$ is a contraction map and where $\epsilon \in (0, 1-L)$. Taking limits on both sides gives the result.
\end{proof}
Methods besides fixed point iteration have been suggested for root-finding. The most famous of these is Newton's method, which exhibits converge of order two. Given a function $g$ whose root is sought, Newton's method computes the following sequence:
\begin{align}
z_{n+1} = z_n - (\nabla g(z_n))^{-1} g(z_n).
\end{align}
We therefore see that Newton's method requires not only the evaluation of the function $g$ at the current iterate, but also its Jacobian, which must be inverted.

\section{Related Work}\label{sec:related-work}

The method of RMHMC was originally conceived in \citet{rmhmc}, though the fundamental idea of using adaptive geometry to inform the proposals can be found earlier in \citet{NIPS1999_d2cdf047}. This latter work used the (non-generalized) leapfrog integrator to compute proposals, which \citet{rmhmc} criticized as lacking reversibility and volume-preservation, and therefore failing to maintain detailed balance with respect to the target distribution. Several attempts at resolving the implicit updates used in applying the generalized leapfrog integrator. \Citet{Lan_2015} made progress in this direction by transforming from Hamiltonian to Lagrangian dynamics, which eliminates fixed point iterations but at a cost of requiring the determinant of an $m\times m$ matrix at every step of the integrator; moreover, the resulting integrator purports to have a larger global error rate than integration methods based on Hamiltonian mechanics. It is also not obvious how to incorporate general-purpose metrics such as SoftAbs \citep{softabs} into the Lagrangian framework. More recently, \citet{cobb2019introducing} proposed an explicit integrator for RMHMC by ``coupling'' two copies of the numerical trajectory; however, this method also suffers from the lack of a theory of reversibility and volume-preservation. The work of \citet{pmlr-v139-brofos21a} examined the use of the implicit midpoint integrator as an alternative to the generalized leapfrog method, though the implicit midpoint method lacks the computational efficiency properties of the generalized leapfrog method. We focus our attention on the generalized leapfrog integrator since it is the {\it de facto} standard integrator used in RMHMC and to enable a simplified analysis of the threshold effects under consideration. Yet another approach to improving the performance of RMHMC was explored by \citet{NIPS2014_a87ff679}, which proposes to utilize a Riemannian metric with a specific ``alternating block-wise'' structure that facilitates the use of the non-generalized leapfrog method. The importance of volume preservation and symmetry of the proposal operator is often stressed in the proofs of detailed balance for HMC (see {\it inter alia} \citet{neal-hmc}); \citet{Lelivre2019HybridMC} proposed a check for violations of reversibility in the case of embedded manifolds, wherein the numerical integrator involves solving for Lagrange multipliers.

\section{Analytical Apparatus}\label{sec:analytical-apparatus}

We now review the research question our work seeks to address and review the various metrics we will employ to measure quantities pursuant to the investigation of the research question.

\subsection{Research Question}

As we have seen, an implementation of the generalized leapfrog integrator requires the use of a non-zero convergence tolerance. How should one determine the convergence threshold? Under the principle that it is better to be slow and correct than fast and wrong, one possibility is to treat the numerics of the generalized leapfrog integrator with caution and adopt a minuscule convergence tolerance for all Bayesian inference problems. However, as discussed in \cref{subsec:measuring-computational-effort}, one step of RMHMC is $\mathcal{O}(m^3)$ whereas a single step of HMC is $\mathcal{O}(m^2)$ (and can be as fast as $\mathcal{O}(m)$ when using an identity mass matrix). Therefore, it seems reasonable to alleviate the computational burden by selecting a non-zero threshold that does not degrade the performance of the RMHMC algorithm. Thus, the research question we seek to answer is, ``To what extent are the ergodicity and the computational efficiency of the RMHMC sampler affected by these convergence tolerances?'' The question of ergodicity is important because the objective of a sampler is to produce a Markov chain that converges to the target density; if the choice of threshold produces detectable and substantial differences between the large-sample distribution of the chain and the target distribution, then ergodicity has been meaningfully violated. The point of computational efficiency is also important: if a larger threshold produces minute differences in the large-sample distribution of the chain, then one would like to characterize the computational savings associated to this larger threshold relative to a more stringent one. This will allow us to quantify and compare, among other properties, the effective sample size per unit of computation in the RMHMC algorithm with varying thresholds.


\subsection{Measures and Metrics}

Here we describe our procedures for measuring quantities of interest in assessing the dependence of the RMHMC algorithm on the threshold. We develop metrics for measuring reversibility and volume-preservation of the proposal operator, ergodicity, similarity between transition kernels, and the computational effort expended by the proposal operator. In measuring reversibility and volume-preservation, we modify the technique of \citet{pmlr-v139-brofos21a}. We emphasize that by measuring reversibility and volume-preservation, we are not measuring detailed balance. Instead, we are measuring two properties which, taken together, imply detailed balance and therefore stationarity of the RMHMC Markov chain. We carry out all computation in 64-bit precision.

\subsubsection{Measuring Reversibility}

As described previously, the easiest mechanism to ensure reversibility of the RMHMC proposal operator is by negating the momentum variable following integration. Therefore, to assess numerical reversibility, we define the momentum flip operator $\mathbf{F}(q, p) = (q, -p)$. Under reversibility, it follows that $\mathbf{F} \circ \Phi^k_{\epsilon} \circ \mathbf{F}\circ \Phi^k_{\epsilon} = \mathrm{Id}$ for any $k\in\mathbb{N}$ and $\epsilon\in \R$. To measure the degree of reversibility of the numerical integrator with fixed point convergence tolerance $\delta$, we set, $(q_f, p_f) = \mathbf{F}\circ\Phi^k_{\epsilon}(q, p; \delta)$ and $(q_r, p_r) = \mathbf{F}\circ \Phi^k_\epsilon(q_f, p_f; \delta)$ and compute the absolute reversibility error (ARE) by,
\begin{align}
  \mathrm{ARE} = \sqrt{\Vert q - q_r\Vert_2^2 + \Vert p - p_r\Vert_2^2}.
\end{align}
We additionally consider a normalized version of this absolute error by dividing by $\sqrt{\Vert q\Vert^2 + \Vert p\Vert^2}$: the relative reversibility error (RRE) is $\mathrm{RRE} = \mathrm{ARE} / \sqrt{\Vert q\Vert^2 + \Vert p\Vert^2}$. The results of the relative error analysis are shown in \cref{app:relative-reversibility-error}.

\subsubsection{Measuring Volume Preservation}

The (non-generalized) leapfrog integrator is always a volume-preserving transformation because its three steps consist only of translations of the position (resp. momentum) by quantities depending only on the momentum (resp. position). For volume-preservation property of the generalized leapfrog method is complicated by the implicit relations used to define the integrator. Nevertheless, viewing the concatenation of the momentum flip operator and the generalized leapfrog integrator as a map from $\R^{2m}$ to itself, we recall that a necessary and sufficient condition for a smooth map to be volume-preserving is that it has unit Jacobian determinant. writing $z=(q, p)$ and letting $e_i$ be the $i^\mathrm{th}$ standard basis vector of $\R^{2m}$ we approximate the Jacobian of the generalized leapfrog integrator by the central difference formula whose $i^\mathrm{th}$ column is given by,
\begin{align}
  \tilde{\mathbf{J}}_i = \frac{\Phi^k_\epsilon(z + \omega  e_i / 2; \delta) - \Phi^k_\epsilon(z - \omega e_i / 2; \delta)}{\omega},
\end{align}
where $\omega > 0$ is the finite difference perturbation size.
Forming the approximate Jacobian $\tilde{\mathbf{J}}$ in this way, we compute its determinant and compare its absolute difference against unity in order to obtain a measure of volume preservation: the volume preservation error (VPE) is
\begin{align}
  \mathrm{VPE} = \abs{\abs{\mathrm{det}(\tilde{\mathbf{J}})} - 1}.
\end{align}
The use of a finite difference approximation in computing the Jacobian will produce round-off and truncation errors that prevent perfect estimation of the true Jacobian. Therefore, we measure the true violation of volume preservation with error. In our experiments we search over values of $\omega$ from the set $\set{1\times10^{-8}, 1\times 10^{-7}, 1\times 10^{-6}, 1\times 10^{-5}, 1\times 10^{-4}, 1\times 10^{-3}}$ and report the volume preservation results for the value of $\omega$ that produces estimates of the Jacobian determinant that are closest to zero when using a convergence of $1\times 10^{-9}$. Further details on the sensitivity of these estimates are given in \cref{app:jacobian-determinant-perturbation}.

\subsubsection{Measuring Ergodicity}\label{subsubsec:measuring-ergodicity}


Ergodicity of a Markov chain refers to the property that, from any initial condition, the state of the Markov chain is, asymptotically, distributed according to the target distribution. In practice, of course, we cannot actually assess this asymptotic behavior because the chain must be stopped at some finite, but large, number of steps. However, in cases where it is possible to generate i.i.d. samples from the posterior, we can assess the similarity of the Markov chain samples and analytic samples drawn from the target distribution. Moreover, in these cases we can initialize the Markov chain in stationary distribution by drawing an initial state from the target distribution.
Therefore, in order to assess the ergodicity of RMHMC, we propose to compare the samples produced by RMHMC against a collection of i.i.d. samples drawn from the target distribution. In the case of the banana-shaped distribution and the Fitzhugh-Nagumo model, which are of dimension two and three, respectively, it is feasible to generate i.i.d. samples via rejection sampling. For Neal's funnel distribution and for the multi-scale Student-$t$ distribution, we can generate i.i.d. samples analytically. Formally, let $(q_1,\ldots,q_n)$ be a collection of $m$-dimensional samples produced by a Markov chain and let $(q'_1,\ldots,q'_{n'})$ be a collection of i.i.d. samples from an $m$-dimensional target distribution. 

\paragraph{Cram\'{e}r-Wold Methods}

To assess the ergodicity of the RMHMC samplers of varying thresholds, we draw on the Cram\'{e}r-Wold theorem \citep{Bill86} and the Kolmogorov-Smirnov statistic for inspiration.
\begin{theorem}[Cram\'{e}r-Wold]
A density function in $\R^m$ is determined by the its projection onto all the one-dimensional subspaces of $\R^m$.
\end{theorem}
\begin{definition}[Kolmogorov-Smirnov Statistic]
  Given two sets of data $(\omega_1,\ldots,\omega_n)$ and $(\omega'_1,\ldots,\omega'_m)$ in $\R$, denote their empirical cumulative distribution functions by $\hat{F}$ and $\hat{F}'$, respectively. The Kolmogorov-Smirnov statistic is
  \begin{align}
      \mathrm{KS}(\hat{F}, \hat{F}') = \sup_{x\in\R} \abs{\hat{F}(x) - \hat{F}'(x)}.
  \end{align}
\end{definition}
Therefore, we propose the measure the ergodicity of the sampler by comparing the Markov chain samples to the analytic samples across many random one-dimensional projections. For as many iterations as desired, compute the following:
\begin{enumerate}
    \item Sample $u\sim\mathrm{Uniform}(\mathbb{S}^{m-1})$.
    \item Compute the orthogonal projections onto the vector space spanned by $u$; that is, $\omega_i = u^\top q_i\in\R$ for $i=1,\ldots,n$ and $\omega_i' = u^\top q_i'\in\R$ for $i=1,\ldots, n'$.
    \item Compute the Kolmogorov-Smirnov statistic (the maximal absolute difference in the cumulative distribution functions) between $(\omega_1,\ldots,\omega_n)$ and $(\omega_1',\ldots,\omega'_{n'})$.
\end{enumerate}
A procedure similar to the one described above was previously advocated in \citet{two-sample-random-projection} for the purposes of crafting a two-sample test for equality of distributions. By constructing a histogram of these Kolmogorov-Smirnov statistics across numerous random directions, one obtains a quantitative measure of the closeness of the distribution of the Markov chain samples and the i.i.d. samples. We note that we do not attempt to compute a $p$-value associated to these Kolmogorov-Smirnov statistics due to the serial auto-correlation of the Markov chain samples, which would invalidate any independence assumption.

\paragraph{Maximum Mean Discrepancy}

We consider the method of maximum mean discrepancy developed in \citet{gretton-kernel}. Given two random variables $x\sim \pi$ and $y\sim \tilde{\pi}$, the maximum mean discrepancy is defined by,
\begin{align}
    \mathrm{MMD}[\mathcal{F},\pi,\tilde{\pi}] \defeq \sup_{\phi\in\mathcal{F}} \paren{\underset{x\sim \pi}{\mathbb{E}} \phi(x) - \underset{y\sim \tilde{\pi}}{\mathbb{E}}\phi(y)},
\end{align}
where $\mathcal{F}$ is a prescribed set of $\R$-valued functions. Given a reproducing kernel Hilbert space (RKHS) $\mathcal{H}$ with kernel $k$, \citet{gretton-kernel} showed that the squared maximum mean discrepancy enjoys the following characterization,
\begin{align}
    \mathrm{MMD}^2[\mathcal{F}, \pi, \tilde{\pi}] = \underset{x, x'\sim \pi}{\mathbb{E}} k(x, x') + \underset{y, y'\sim \tilde{\pi}}{\mathbb{E}} k(y, y') - 2\underset{x\sim \pi, y\sim \tilde{\pi}}{\mathbb{E}} k(x, y).
\end{align}
when $\mathcal{F}$ is the set of functions in the unit ball of the RKHS: $\mathcal{F} = \set{\phi \in \mathcal{H} : \Vert \phi\Vert_{\mathcal{H}} \leq 1}$. Under suitable conditions on the RKHS, it has been shown that $\pi=\tilde{\pi} \iff \mathrm{MMD}[\mathcal{F}, \pi, \tilde{\pi}] = 0$. We use the following unbiased estimator of the squared maximum mean discrepancy as a measure of similarity between the Monte Carlo and analytical samples,
\begin{align}
    \mathrm{MMD}_u^2[\mathcal{F}, \pi, \tilde{\pi}] = \frac{1}{n(n-1)} \sum_{i=1}^n\sum_{j\neq i}^n k(q_i, q_j) + \frac{1}{n'(n'-1)} \sum_{i=1}^{n'} \sum_{j\neq i}^{n'} k(q'_i, q'_j) - \frac{2}{nn'} \sum_{i=1}^{n}\sum_{j=1}^{n'} k(q_i, q'_j).
\end{align}
As a measure of ergodicity, when sampling has been effective we expect that $\mathrm{MMD}$ will be close to zero; since the estimator we employ is unbiased, this can produce a negative quantity. In our experiments, we report $\abs{\mathrm{MMD}_u^2[\mathcal{F}, \pi, \tilde{\pi}]}$ so as to avoid presenting negative values of what ought to be a non-negative quantity. We use the squared exponential kernel $k(q, q') = \exp(-\Vert q - q'\Vert_2^2 / h^2)$ where $h$ is the kernel bandwidth. We use the median distance heuristic among the i.i.d. samples to select the bandwidth $h$ so that $h = \mathrm{median}(\{ \Vert q'_i - q'_j\Vert \}_{i\neq j})$. The computational complexity of the MMD estimator is $\mathcal{O}(n^2 + (n')^2)$.

\paragraph{Sliced Wasserstein Distances}

A somewhat similar method to the Cram\'{e}r-Wold procedure described above can be formulated based on Wasserstein distances rather than Kolmogorov-Smirnov statistics along one-dimensional projections. Given two probability measures $\Pi$ and $\tilde{\Pi}$ on $\R$ with cumulative distribution functions $F_\Pi$ and $F_{\tilde{\Pi}}$, respectively, the 1-Wasserstein distance between $\Pi$ and $\tilde{\Pi}$ is \citep{e19020047},
\begin{align}
    W_1(\Pi, \tilde{\Pi}) = \int_0^1 \abs{F_\Pi^{-1}(t) - F_{\tilde{\Pi}}^{-1}(t)} ~\mathrm{d}t.
\end{align}
Using this one-dimensional characterization of the 1-Wasserstein distance, a distance -- called the sliced Wasserstein distance -- on probability measures on $\R^m$. Let $\Pi$ now be a probability measure on $\R^m$ may be constructed. The the sliced Wasserstein distance is defined by
\begin{align}
    \label{eq:sliced-wasserstein} SW_1(\Pi, \tilde{\Pi}) = \int_{\mathbb{S}^{m-1}} W_1(\mathcal{R}_\theta(\Pi), \mathcal{R}_\theta(\tilde{\Pi})) ~\mathrm{dVol}_{\mathbb{S}^{m-1}}(\theta),
\end{align}
where $\mathcal{R}_{\theta}(\Pi)$ is the probability measure of the random variable $x^\top \theta$ when $x\sim \Pi$. \citet{nadjahi2021fast} computes a Monte Carlo approximation of \cref{eq:sliced-wasserstein} by sampling unit vectors uniformly over $\mathbb{S}^{m-1}$ and treating $(\omega_1,\ldots,\omega_n)$ and $(\omega'_1,\ldots,\omega'_{n'})$ as the locations of Dirac measures by which to define discrete probability measures. (The quantities $\omega_i$ and $\omega_i'$ were defined in step two in our discussion of Cram\'{e}r-Wold Methods.) In the special case that $n=n'$, the one-dimensional Wasserstein distance assumes the simple form,
\begin{align}
    W_1(\Pi,\tilde{\Pi}) = \frac{1}{n} \sum_{i=1}^n \abs{\omega_{(i)} - \omega'_{(i)}},
\end{align}
where $(\omega_{(1)}, \ldots,\omega_{(n)})$ and $(\omega'_{(1)}, \ldots,\omega'_{(n)})$ are, respectively, $(\omega_1,\ldots,\omega_{n})$ and $(\omega'_1,\ldots,\omega'_{n})$ sorted in ascending order.

\paragraph{Discretized Differences in Probability}

In \cref{subsec:experiment-banana-shaped}, the low-dimensionality of the posterior distribution permits us to employ the method of \citet{10.36045/bbms/1170347810} which considers the $L_1$ distance between discretized probability densities. Given a partition $\mathcal{P}$ of $\R^m$, we compute the statistic,
\begin{align}
    DL_1(\set{q_i}_{i=1}^n, \set{q'_i}_{i=1}^{n'}) = \sum_{r\in\mathcal{P}} \mathrm{Vol}(r)\cdot \abs{\frac{1}{n} \sum_{i=1}^n \mathbf{1}\set{q_i\in r} - \frac{1}{n'} \sum_{i=1}^{n'} \mathbf{1}\set{q'_i\in r}}.
\end{align}
This can be interpreted as an approximation of the $L_1$ distance between probability densities, where the quality of the approximation depends on the number of samples $n$ and $n'$ and the fineness of the partition. In \cref{subsec:experiment-banana-shaped}, we partition $[-30, 10]\times[-10, 10]\subset \R^2$ (which contains virtually all of the probability mass of the posterior in that example) into $2,500$ equally sized rectangles and compute the $DL_1$ statistic.

\paragraph{Methods Based on Multiple Chains}

Let $q_0$ be a fixed initial position variable. Consider $r$ independent (RM)HMC Markov chains starting from this shared initial position. Let $q_{ij}$ denote the position variable at step $j$ of the $i$-th Markov chain. On the $j$-th step, how close is the distribution of the set $\set{q_{1j},\ldots, q_{rj}}$ to the target distribution? Note that, given the fixed initial position $q_0$, $q_{ij}$ is independent of $q_{kj}$. Therefore, the question as posed is clearly distinct from the question, ``How close is the distribution $\set{q_{i1},\ldots, q_{in}}$ to the target distribution?'' Answering this question can give an indication of the convergence {\it speed} of the HMC and RMHMC procedures and, in the latter case, its sensitivity to the convergence threshold.

Computing the previously described metrics would be prohibitively expensive to compute for every step of the Markov chain. However, in \cref{subsec:experiment-student-t,subsec:experiment-neal-funnel} there are singular dimensions of the posterior that are particularly challenging to sample due to the multiscale structure of the target distribution, but which nonetheless have known marginal distributions. By exclusively considering these single dimensions, we may assess convergence with respect to the most challenging dimension of the posterior. Let $l$ denote the index of this challenging dimension of the posetior so that $q_{ijl}$ is the $l$-th dimension of the state at the $j$-th step in the $i$-th chain. In our experiments, we compute the {\it single sample} Kolmogorov-Smirnov statistic comparing the distribution of $\set{q_{1jl},\ldots, q_{rjl}}$ against the known marginal. We set $r=1,000$ and consider $j=1,\ldots, 10,000$.

\subsubsection{Measuring Sample Independence}

Ergodicity measures how close the iterates of the Markov chain are to the target distribution of interest. However, the samples generated by Markov chains exhibit serial auto-correlation and therefore not independent. The degree of dependence between samples with effect the determine the precision of the Monte Carlo approximation of posterior expectations. A standard measure in the MCMC literature is the {\it effective} number of independent samples that a set of Markov chain samples represents. The effective sample size is the equivalent number of independent samples that would produce an estimator with the same variance as the auto-correlated samples produced by the Markov chain. Formally, following \citet{gelmanbda04}, let $(q_1,\ldots, q_n)$ be a sequence of univariate Markov chain samples and let $\rho_t$ be the auto-correlation of $(q_1,\ldots, q_n)$ with lag $t$. The effective sample size (ESS) is the quantity,
\begin{align}
    \mathrm{ESS} = \frac{n}{1 + 2\sum_{t=1}^\infty \rho_t}.
\end{align}
In practice, $\rho_t$ is not known and must be estimated from the sequence $(q_1,\ldots, q_n)$ itself. We utilize the procedure of \citet{arviz_2019} to compute the ESS in our experiments. As a practical matter, one is concerned not only with the effective sample size in absolute terms, but also with the effective sample size {\it per unit of computation}. To measure this quantity, we divide the ESS by the running time (in seconds) of the algorithm. In distributions with multiple dimensions, we may consider the mean ESS, which is the average ESS among each dimension of the Markov chain. Similarly, the minimum ESS is the minimum ESS among each dimension of the Markov chain.

\subsubsection{Measuring Transition Kernel Similarity}\label{subsubsec:transition-kernel-distance}

Given two RMHMC transition kernels $K((q, p), \cdot; \epsilon, k, \delta)$ and $K((q, p), \cdot; \epsilon, k, \delta')$ with the same step-size $\epsilon$ and number of integration steps $k$, but with differing fixed point convergence thresholds $\delta$ and $\delta'$, how can we measure their similarity? For Bayesian inference tasks, we are primarily concerned with their similarity in the $q$-dimensions, since the $p$-dimensions are auxiliary variables that serve only to facilitate the construction of a phase-space. Therefore, let $\mathrm{Proj}_q : \R^m\times\R^m\to\R^m$ denote the projection to the $q$-dimensions alone. We propose to measure the similarity of $K((q, p), \cdot; \epsilon, k, \delta)$ and $K((q, p), \cdot; \epsilon, k, \delta')$ as,
\begin{align}
    \underset{q\sim (\cdot)}{\mathbb{E}} ~\underset{p\sim \mathrm{Normal}(0_m, \mathbb{G}(q))}{\mathbb{E}} \underset{u\sim \mathrm{Uniform}(0, 1)}{\mathbb{E}} \Vert \mathrm{Proj}_q(\Psi_\epsilon^k((q, p), u; \delta)) - \mathrm{Proj}_q(\Psi_{\epsilon}^k((q, p), u; \delta'))\Vert_2,
\end{align}
where we average over a suitable distribution over the position and momentum variables. This is the expected difference in the samples generated by the transition kernels $K_{\delta}$ and $K_{\delta'}$ when ensuring that both transitions are computed using the same integration step-size, the same number of integration steps, the same current position in phase-space, and the same uniform random number used in applying the Metropolis-Hastings accept-reject criterion. The distribution over $q$, the random position variable, is arbitrary and one therefore has this degree of freedom when measuring the similarity of transition kernels. In our experiments, we will use either i.i.d. samples from the target distribution, or samples from another Markov chain to approximate the expectation over the $q$-variables. Note that when both transition kernels reject the proposal computed by the RMHMC integrator, the expected difference is zero; in our visualizations of this metric, we show a distribution of differences in the cases where at least one of the two transition kernels did not reject its proposal and the probability that both transition kernels reject, which we call the ``rejection agreement.'' In our experiments, we compare transition kernels with thresholds $\delta \in\set{10^{-1}, 10^{-2}, 10^{-3}, 10^{-4}, 10^{-5}, 10^{-6}, 1\times 10^{-7}, 10^{-8}, 10^{-9}}$ against the transition kernel with $\delta' = 1\times 10^{-10}$.

\subsubsection{Measuring Computational Effort}\label{subsec:measuring-computational-effort}

As a measure of computational effort, we report $l_p$ and $l_q$ as described in \cref{alg:threshold-leapfrog}. These count, respectively, the number of fixed point iterations required to resolve the implicit updates to the momentum and position.

\subsection{Stochastic Approximation for Threshold Identification}

We propose to adapt the threshold to achieve a prescribed average number of decimal digits of similarity with a numerical integrator using a strict convergence tolerance (such as $1\times 10^{-10}$). Let $\delta$ and $\delta'$ be two convergence thresholds and consider the quantity
\begin{align}
    G_{\epsilon, k}(q, p, \delta, \delta') = \begin{cases} 
    \log_{10}\paren{\Vert \Phi_\epsilon^k(q, p, \delta) - \Phi_{\epsilon}^k(q, p, \delta')\Vert_2} &  ~\mathrm{if}~ \delta > \delta' \\
    -16 & ~\mathrm{otherwise}.
    \end{cases}
\end{align}
This is the {\it negative} number of decimal digits of similarity between the numerical integrators with step-size $\epsilon$ and number of integration steps $k$, but with differing convergence tolerances $\delta$ and $\delta'$. This measure treats thresholds that are less than or equal to the baseline as equivalent; for thresholds greater than the baseline, we measure the number of decimal digits of similarity between their respective integrators. 
We may wish to choose $\delta$ to produce a given number of decimal digits of similarity between these two integrators; we denote the desired number of decimal digits of similarity by $\kappa$ and define the difference between the the observed and desired number of decimal digits of similarity by,
\begin{align}
    L(\delta) \equiv L_{\epsilon, k}(q, p, \delta, \delta') = G_{\epsilon, k}(q, p; \delta, \delta') + \kappa
\end{align}
We would like to find $\delta$ such that $B(\delta) = 0$ where
\begin{align}
    B(\delta) = \underset{q\sim \pi}{\mathbb{E}} ~\underset{p\sim \mathrm{Normal}(0_m, \mathbb{G}(q))}{\mathbb{E}} L(\delta).
\end{align}
We seek to produce a sequence of convergence thresholds $(\bar{\delta}_1,\bar{\delta}_2,\ldots)$ such that $B(\bar{\delta}_n) \to 0$. Given initial data $\bar{\delta}_1 = \delta_1$, we set
\begin{align}
    \log \delta_{n+1} &=  \log \delta_n - \gamma_n (L(\delta_n) - \kappa) \\
    \log \bar{\delta}_{n+1} &= \frac{n}{n+1} \log \bar{\delta}_n + \frac{1}{n+1} \log \delta_n,
\end{align}
where $\gamma_n = n^{-\omega}$. We also define $\bar{L}_n = n^{-1} \sum_{i=1}^n L(\delta_n)$ as a approximation to the expectation of $L_n$. In our experiments, we set $\delta'=1\times10^{-10}$ and $\omega = 3/4$, which produced reasonable behavior. Convergence may be faster for variations of these hyperparameters. We consider a maximum value of 1,000 for $n$. The number of decimal digits of similarity between transition integrators is only a proxy for the violation of reversibility and volume preservation. However, this is a measure that is simple to compute and deploy in practice; in contrast, a measure to directly match a prescribed average violation of volume preservation would be significantly more computationally expensive due to the use of finite differences to approximate the Jacobian. Moreover, if one can hypothesize a minimum scale of the posterior distribution, then scale can be used to guide how many decimal digits of similarity one should require on average from the RMHMC numerical integrator.

In our experimental evaluation, we make an effort to check that the assumptions of the averaging procedure are satisfied. To accomplish this, we compute a Monte Carlo approximation of $B(\delta)$ for one-hundred logarithmically-spaced values of $\delta$ between $1\times10^{-10}$ and $1\times10^{-1}$. In practice, these Monte Carlo approximations of $B(\delta)$ appear to be monotonically increasing, to have a value of $\delta$ satisfying $B(\delta)=0$, and {\it look} smooth. In computing the Monte Carlo approximation of $B(\delta)$ we use i.i.d. samples from the target distribution as the distribution over $q$; however, when reporting sequences $\set{\bar{L}_n}$ and $\set{\bar{\delta}_n}$, in order to mirror the actual practice of the technique, we instead draw samples using an RMHMC Markov chain. Therefore, the value of $\bar{\delta}_n$ after 1,000 iterations may not precisely match the apparent solution of the equation $B(\delta)=0$, though typically the two values are close.

\section{Experiments}\label{sec:experiments}

\begin{table}[t]
\begin{tabular}{@{}l|lllll@{}}
\toprule
Posterior                                                           & \# Dimensions & \# Samples & Metric                            & Hierarchical          & \begin{tabular}[c]{@{}l@{}}Exact\\ Gradients\end{tabular} \\ \midrule
Banana                                                              & 2             & 1,000,000  & Fisher Information + Hessian      & \xmark & \cmark                                     \\
\begin{tabular}[c]{@{}l@{}}Logistic\\ Regression\end{tabular}       & 14 + 1        & 100,000    & Fisher Information + Hessian      & \cmark & \cmark                                     \\
Neal's Funnel                                                       & 11            & 1,000,000  & SoftAbs Metric                    & \xmark & \cmark                                     \\
\begin{tabular}[c]{@{}l@{}}Stochastic\\ Volatility\end{tabular}     & 1000 + 3      & 100,000    & Fisher Information + Hessian      & \cmark & \cmark                                     \\
\begin{tabular}[c]{@{}l@{}}Log-Gaussian\\ Cox-Poisson\end{tabular}  & 1024 + 2      & 5,000      & Fisher Information + Hessian      & \cmark & \cmark                                     \\
Fitzhugh-Nagumo                                                     & 3             & 100,000    & Fisher Information + Hessian      & \xmark & \xmark                                     \\
\begin{tabular}[c]{@{}l@{}}Multivariate \\ Student-$t$\end{tabular} & 20            & 1,000,000  & Positive Definite Part of Hessian & \xmark & \cmark                                     \\ \bottomrule
\end{tabular}
\caption{Summary of the posterior distributions we examine in our experimental results. For the heirarchical models, we employ an alternating Metropolis-within-Gibbs-like strategy; we indicate the dimensionality of each alternating component by separating them with a plus-sign. The Neal Funnel distribution has a hierarchical structure but it is sampled in a non-hierarchical manner. As discussed in detail later, the Fitzhugh-Nagumo model computes approximate gradients by solving initial value problems.}
\label{tab:posteriors}
\end{table}

Here we present our empirical analysis of the role of thresholds in RMHMC. Throughout our experiments, we carefully control the seed of the pseudo-random number generator used in sampling the random momentum and in applying the Metropolis-Hastings accept-reject criterion. As a result of this experimental design, we may assess the causal effect of adjusting the convergence threshold. We summarize some critical aspects of our experimentation in \cref{tab:posteriors}. Code implementing these experiments may be found at \url{https://github.com/JamesBrofos/Thresholds-in-Hamiltonian-Monte-Carlo}.

\subsection{Banana-Shaped Distribution}\label{subsec:experiment-banana-shaped}

\begin{figure}[t!]
  \centering
  \begin{subfigure}[t]{0.32\textwidth}
    \includegraphics[width=\textwidth]{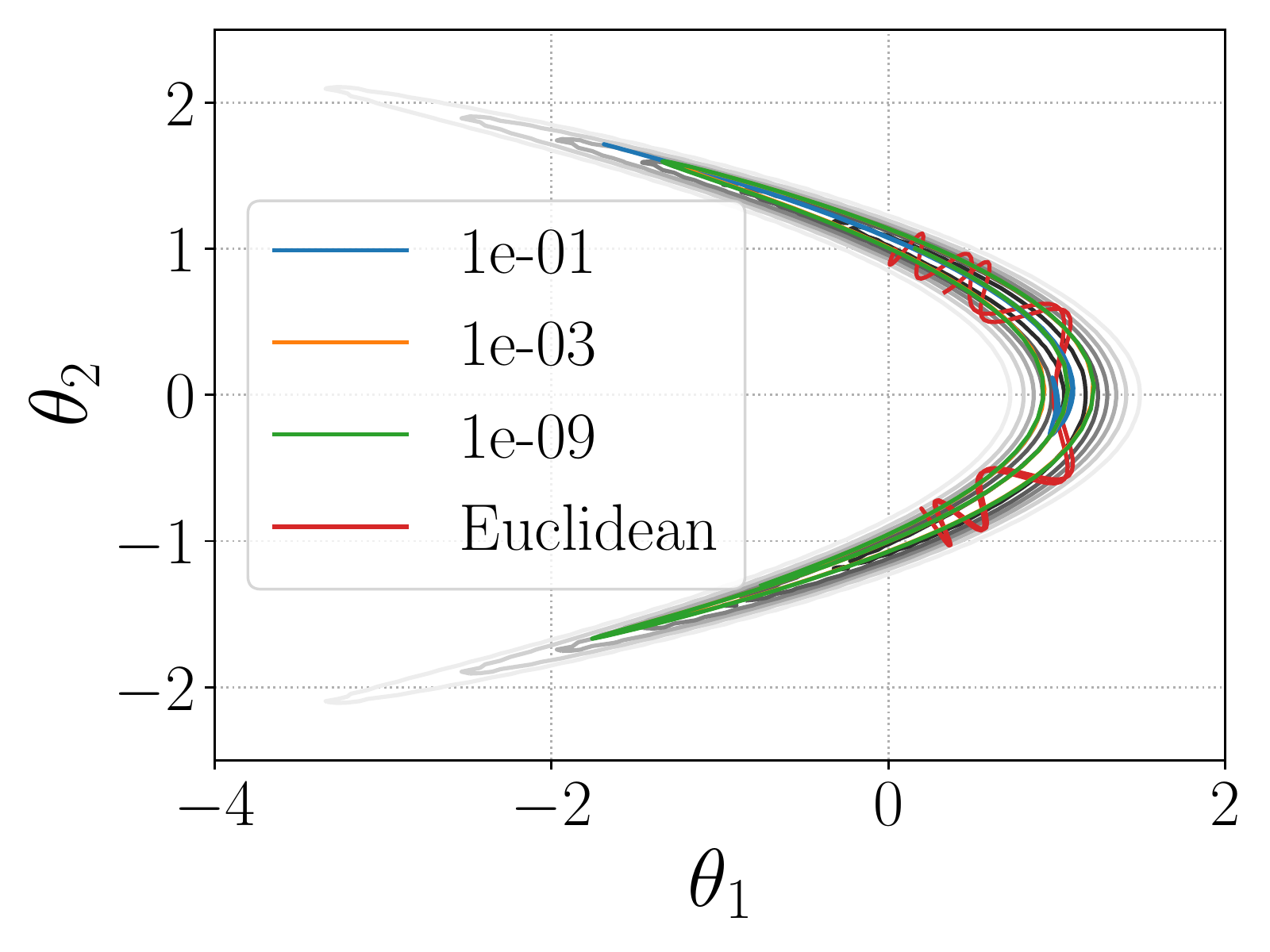}
    \caption{Generalized Leapfrog Trajectories}
    \label{subfig:banana-trajectory-threshold-precondition}
  \end{subfigure}
  ~
  \begin{subfigure}[t]{0.32\textwidth}
    \includegraphics[width=\textwidth]{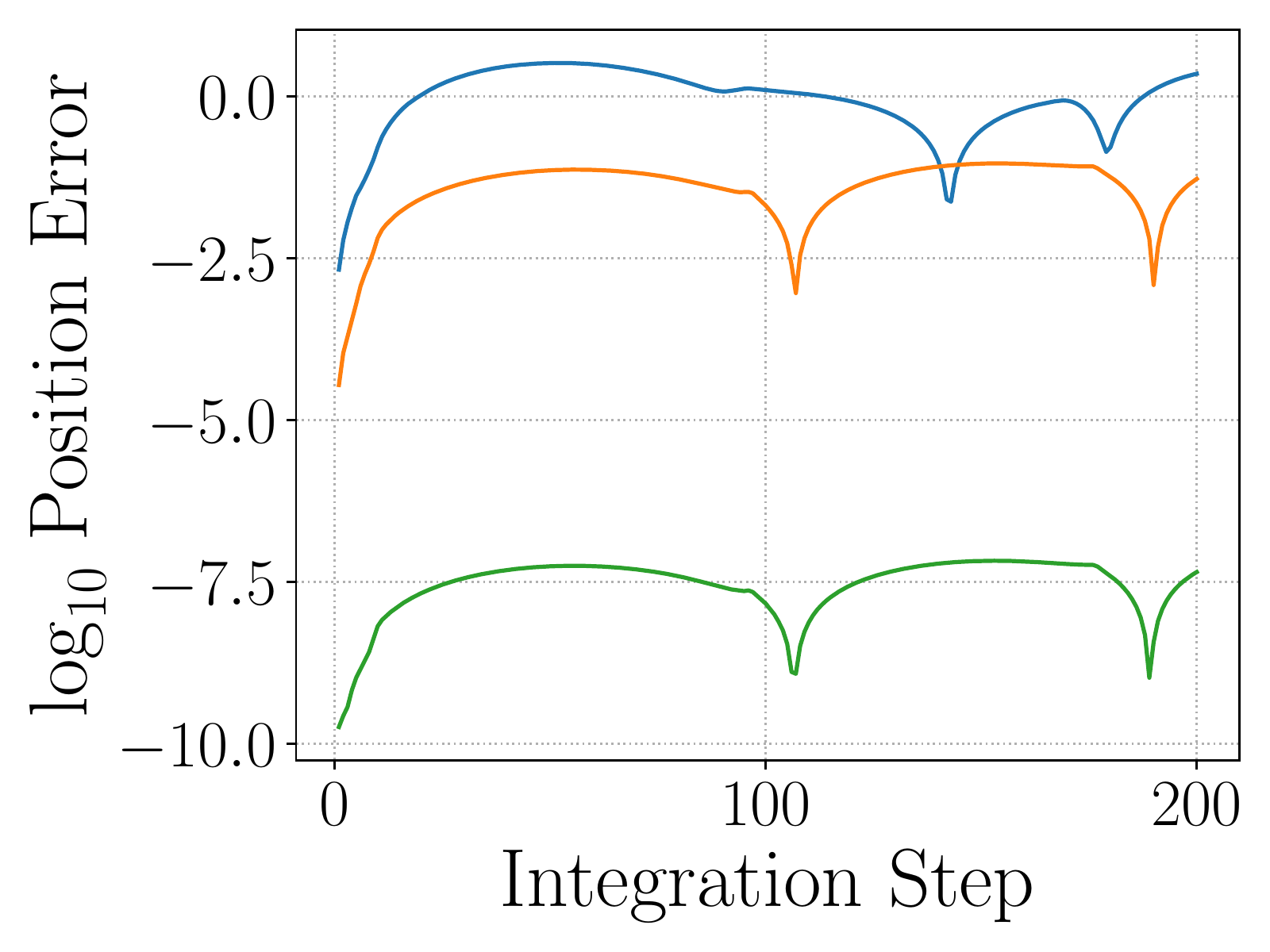}
    \caption{Error in the Position Variable}
    \label{subfig:banana-trajectory-position-deviation}
  \end{subfigure}
  ~
  \begin{subfigure}[t]{0.32\textwidth}
    \includegraphics[width=\textwidth]{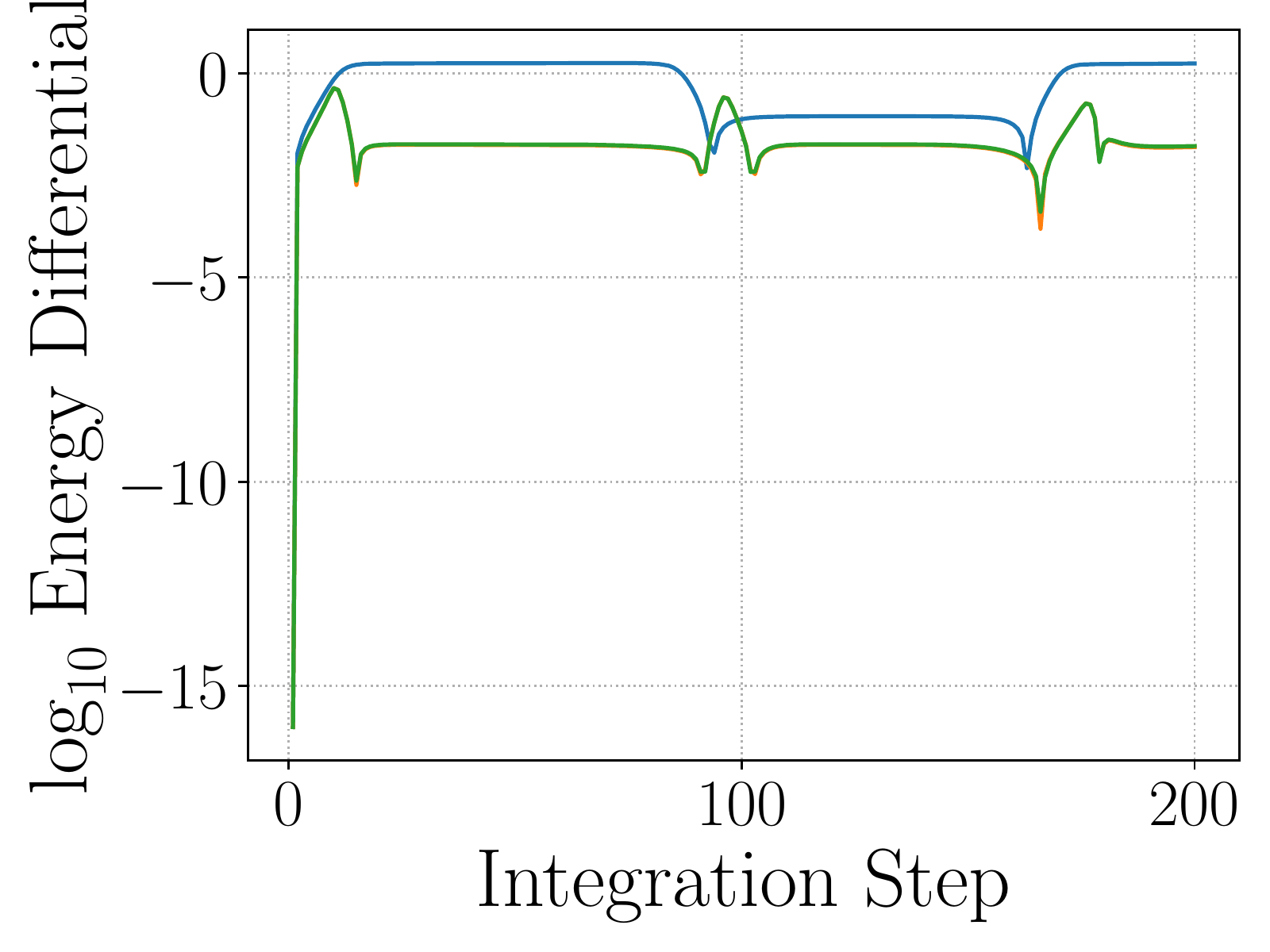}
    \caption{Conservation of Hamiltonian Energy}
    \label{subfig:banana-trajectory-energy-deviation}
  \end{subfigure}
  \caption{In \cref{subfig:banana-trajectory-threshold-precondition} we
    visualize the trajectories produced by the generalized leapfrog integrator
    with variable thresholds. A trajectory of the leapfrog integrator using a
    separable Hamiltonian is shown for comparison, which also reveals the
    preconditioning effect. Although the trajectories in \cref{subfig:banana-trajectory-threshold-precondition} are qualitatively similar, in
    \cref{subfig:banana-trajectory-position-deviation} we measure the
    number of decimal digits conserved in the position variable for differing
    thresholds relative to the generalized leapfrog integrator with threshold
    $1\times 10^{-10}$; this reveals the extent to which the trajectories are differentiated by the choice of threshold. The approximate conservation of the Hamiltonian is
    violated is visualized in \cref{subfig:banana-trajectory-energy-deviation}.}
  \label{fig:banana-trajectory}
\end{figure}

\begin{figure}[t!]
  \centering
  \begin{subfigure}[t]{0.32\textwidth}
    \includegraphics[width=\textwidth]{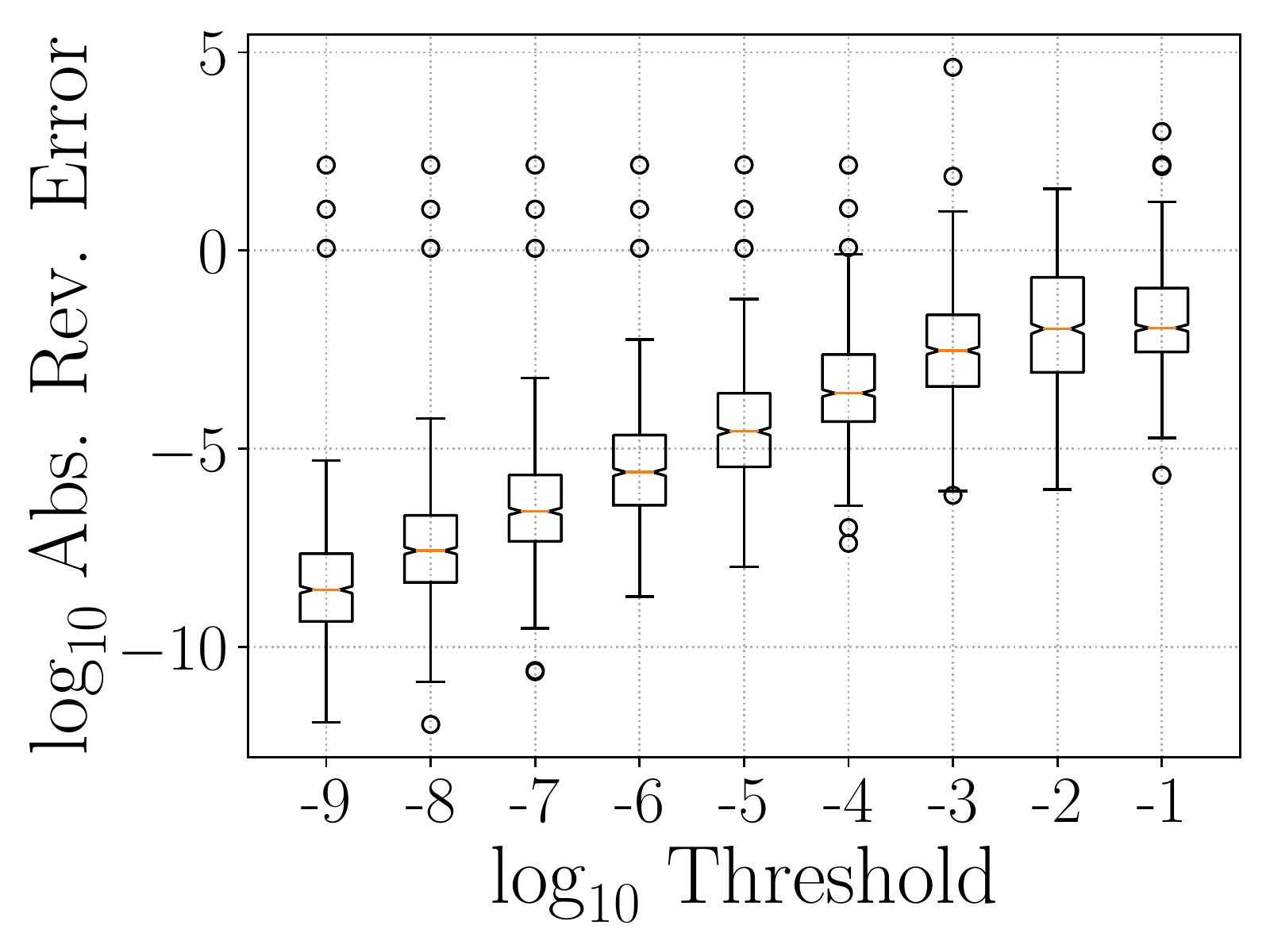}
    \caption{Error in Reversibility}
    \label{subfig:banana-reversibility}
  \end{subfigure}
  ~
  \begin{subfigure}[t]{0.32\textwidth}
    \includegraphics[width=\textwidth]{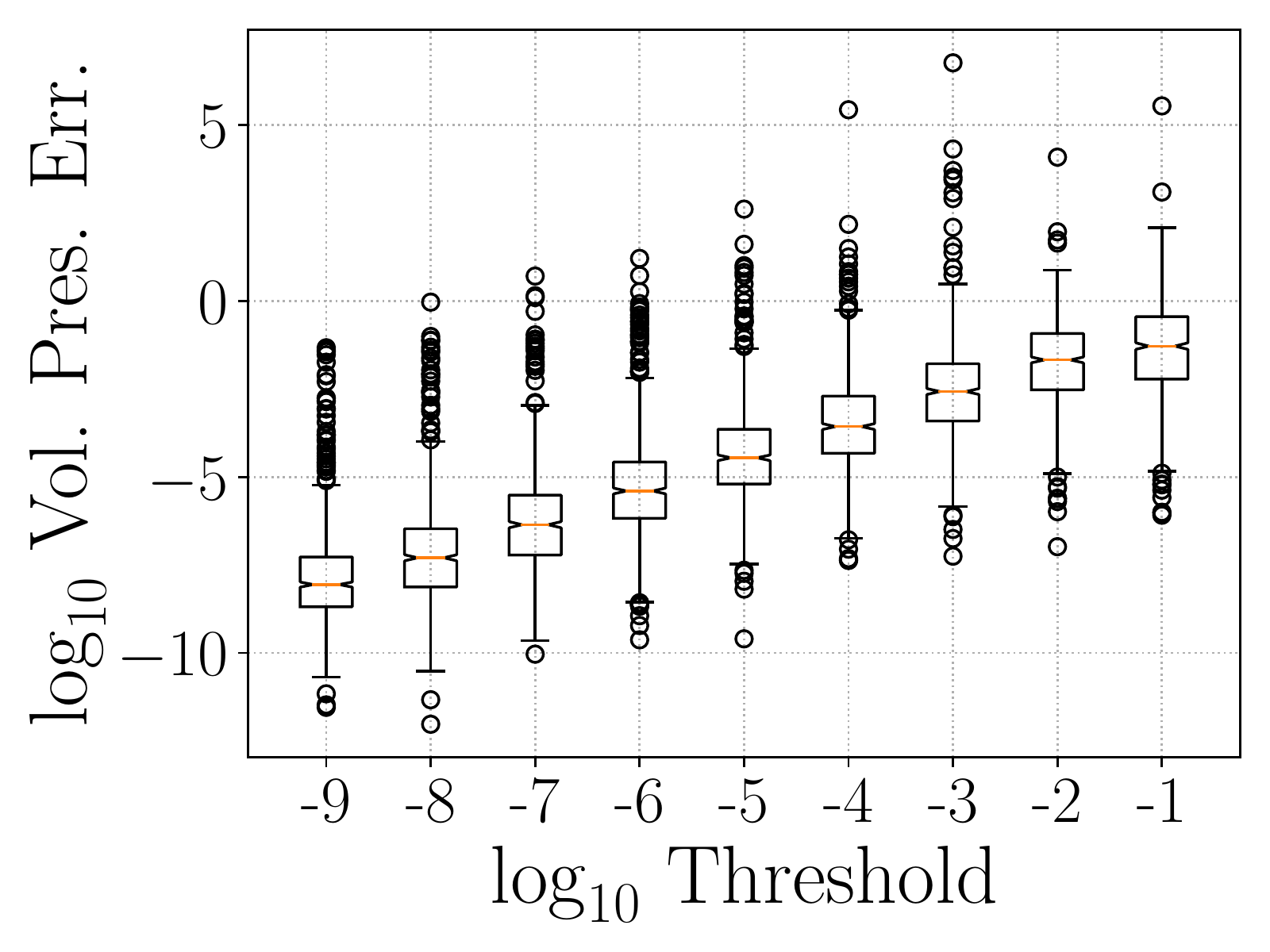}
    \caption{Error in Volume-Preservation}
    \label{subfig:banana-jacobian-determinant}
  \end{subfigure}
  ~
  \begin{subfigure}[t]{0.32\textwidth}
    \includegraphics[width=\textwidth]{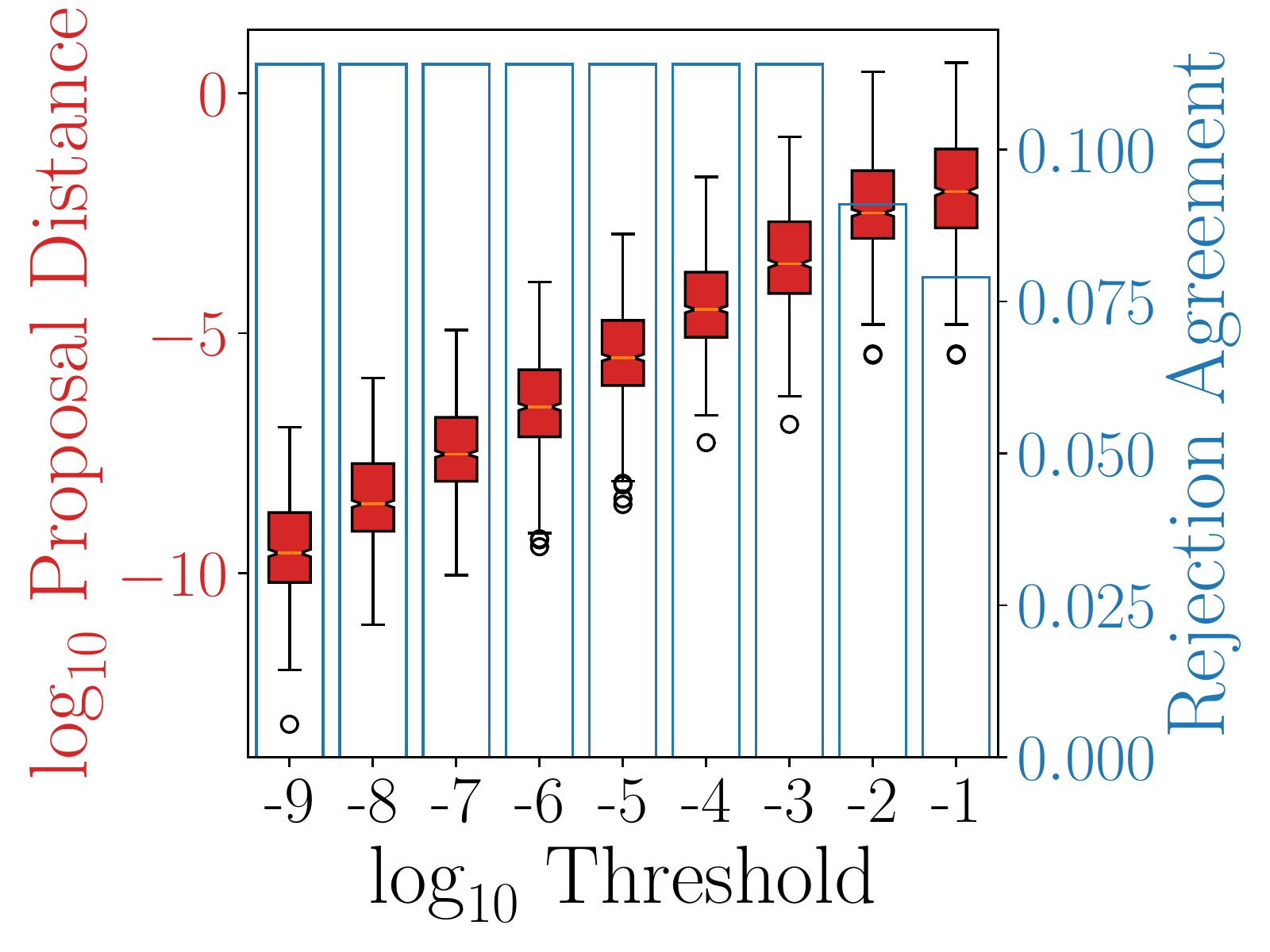}
    \caption{Difference in Transition Kernel}
    \label{subfig:banana-transition-difference}
  \end{subfigure}
  \caption{Visualization of the error in reversibility (see
    \cref{subfig:banana-reversibility}), error in
    volume-preservation (see
    \cref{subfig:banana-jacobian-determinant}), and the number of
    decimal digits of similarity in transition kernels (see
    \cref{subfig:banana-transition-difference}) for variable
    thresholds in the banana posterior distribution.}
\end{figure}

\begin{figure}[t!]
  \centering
  \begin{subfigure}[t]{0.49\textwidth}
    \includegraphics[width=\textwidth]{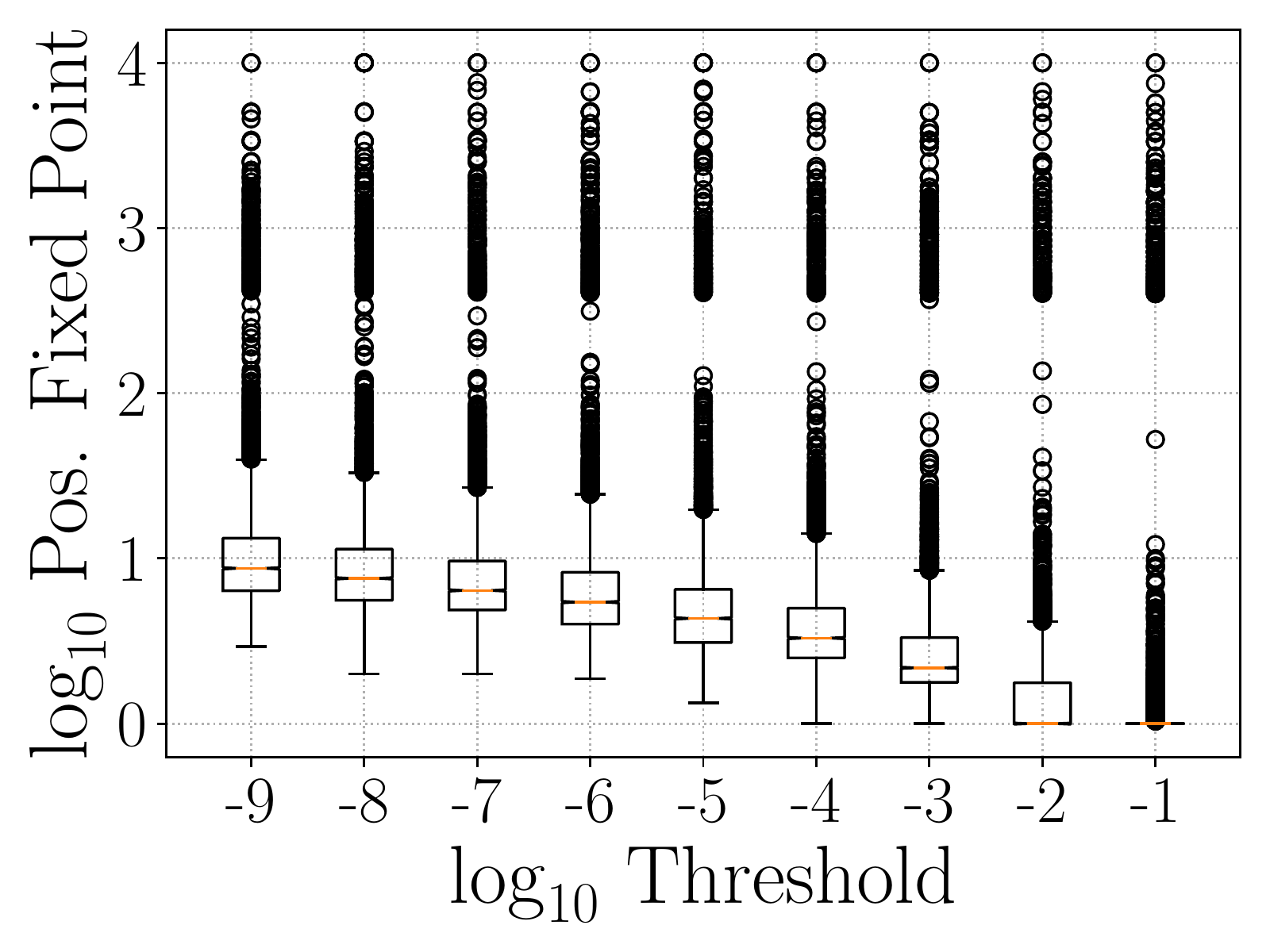}
    \caption{Number of fixed point iterations for position variable.}
    \label{subfig:banana-fixed-point-position}
  \end{subfigure}
  ~
  \begin{subfigure}[t]{0.49\textwidth}
    \includegraphics[width=\textwidth]{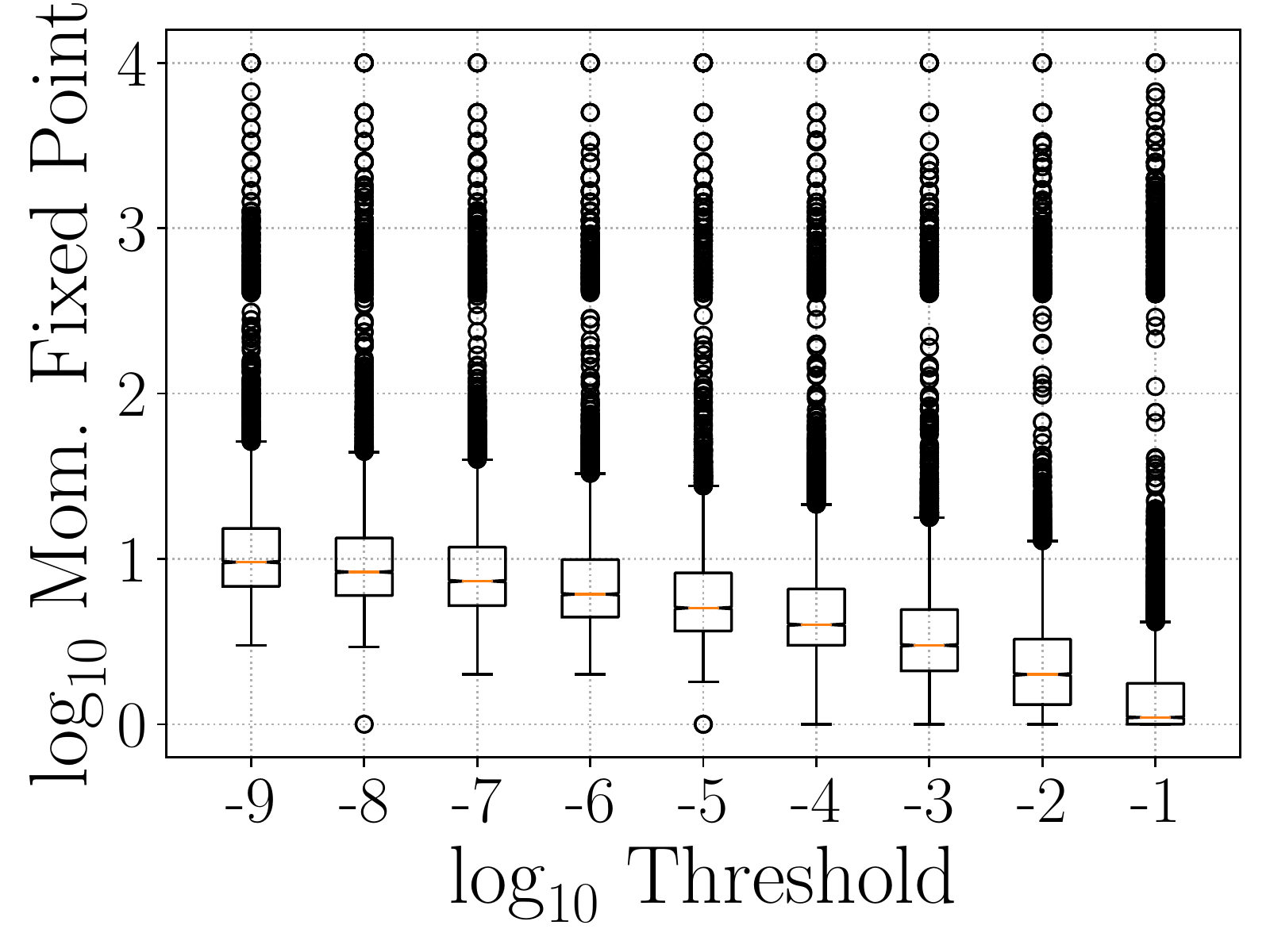}
    \caption{Number of fixed point iterations for momentum variable.}
    \label{subfig:banana-fixed-point-momentum}
  \end{subfigure}
  
  \begin{subfigure}[t]{\textwidth}
    \includegraphics[width=\textwidth]{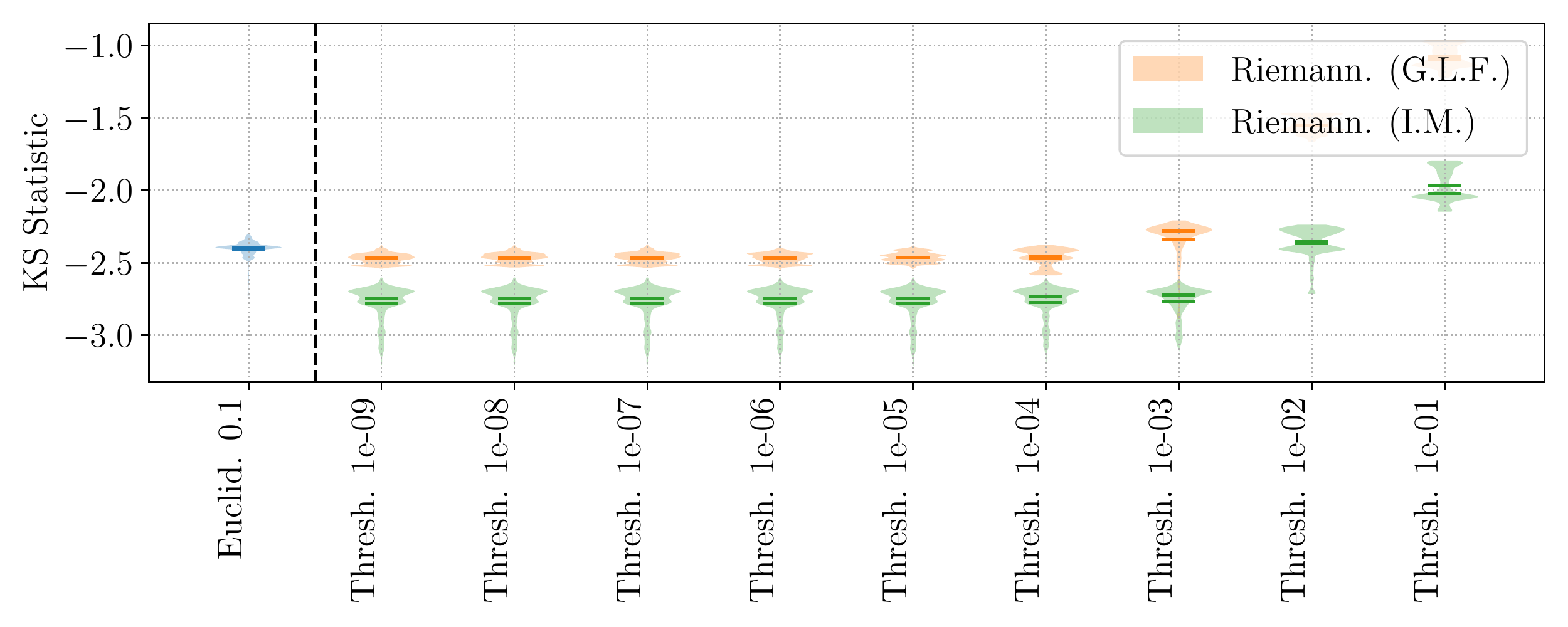}
    \caption{Ergodicity of RMHMC and HMC}
    \label{subfig:banana-ergodicity}
  \end{subfigure}
  
  \caption{Visualization of the computational effort required to sample with RMHMC from the banana posterior. We show the number of fixed point iterations required to compute the two implicit steps of the generalized leapfrog integrator as well as the distribution of Kolmogorov-Smirnov statistics over variable thresholds. The maximum Kolmogorov-Smirnov statistic of Euclidean HMC with a step-size of $0.1$ is shown as a dashed blue line.}
\end{figure}
\begin{figure}[t!]
  \begin{subfigure}[t]{0.3\textwidth}
    \centering
    \includegraphics[width=\textwidth]{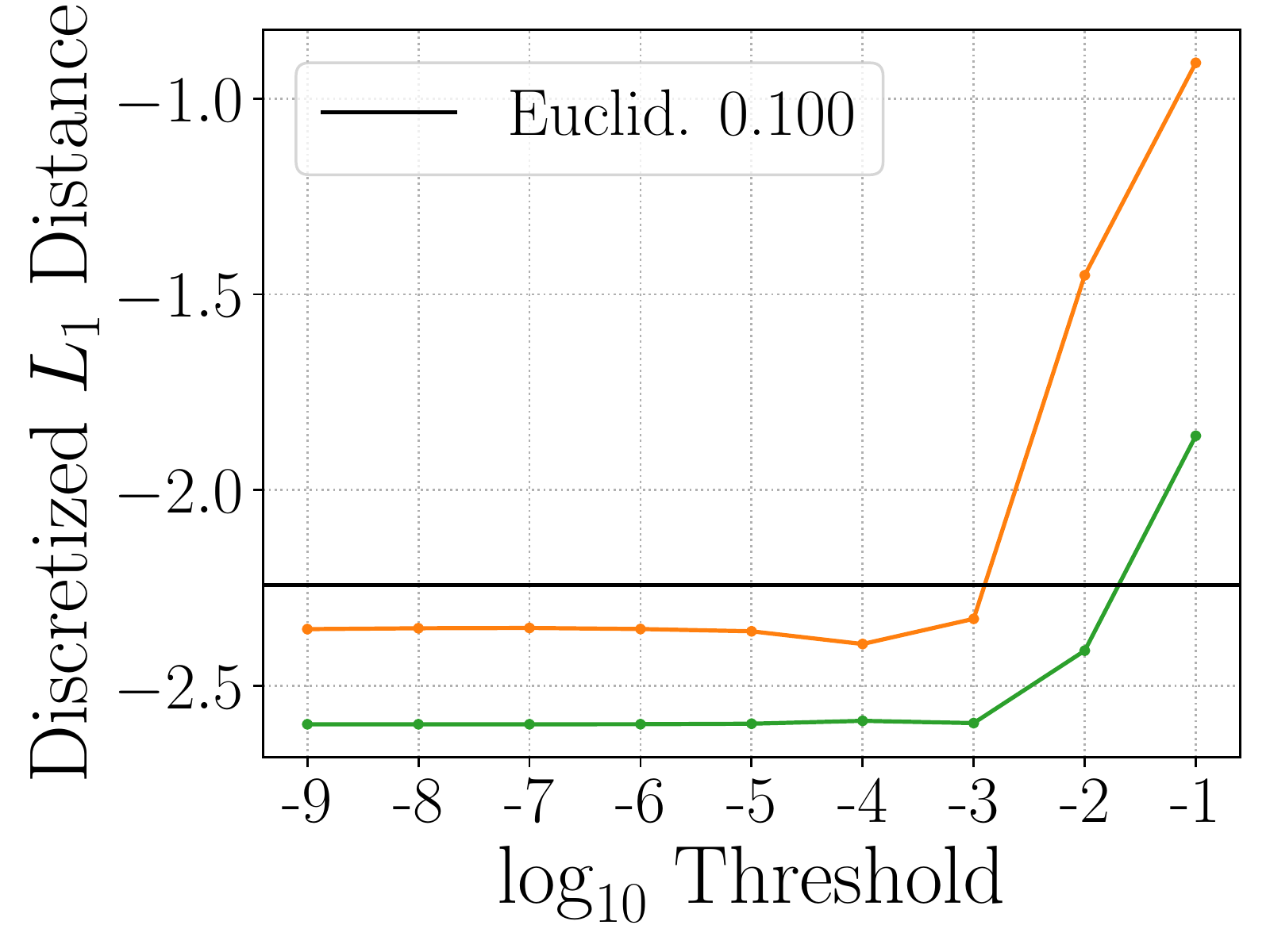}
    \caption{Biau Metric}
  \end{subfigure}
  ~
  \begin{subfigure}[t]{0.3\textwidth}
    \centering
    \includegraphics[width=\textwidth]{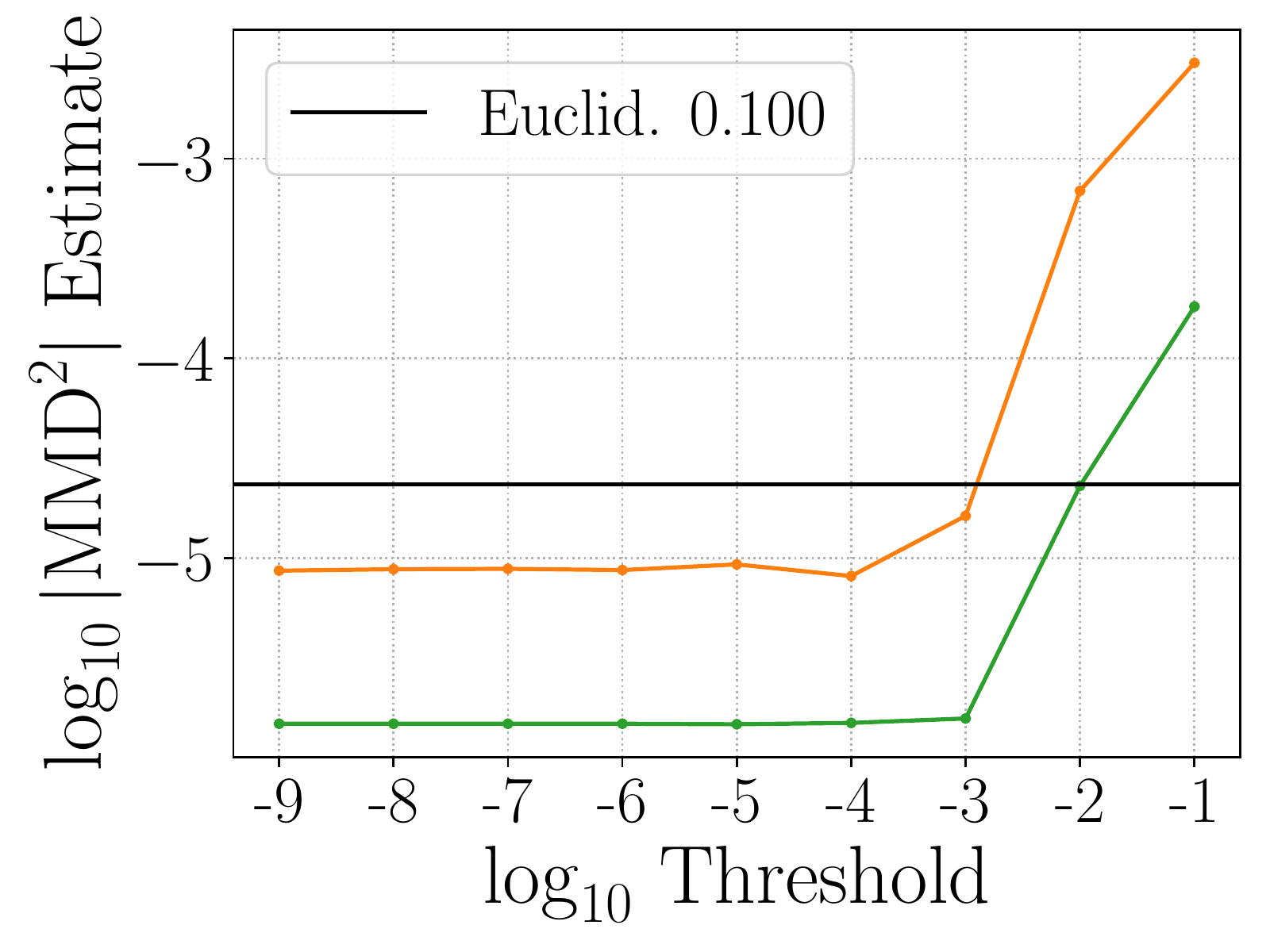}
    \caption{$\mathrm{MMD}_u^2$}
  \end{subfigure}
  ~
  \begin{subfigure}[t]{0.3\textwidth}
    \centering
    \includegraphics[width=\textwidth]{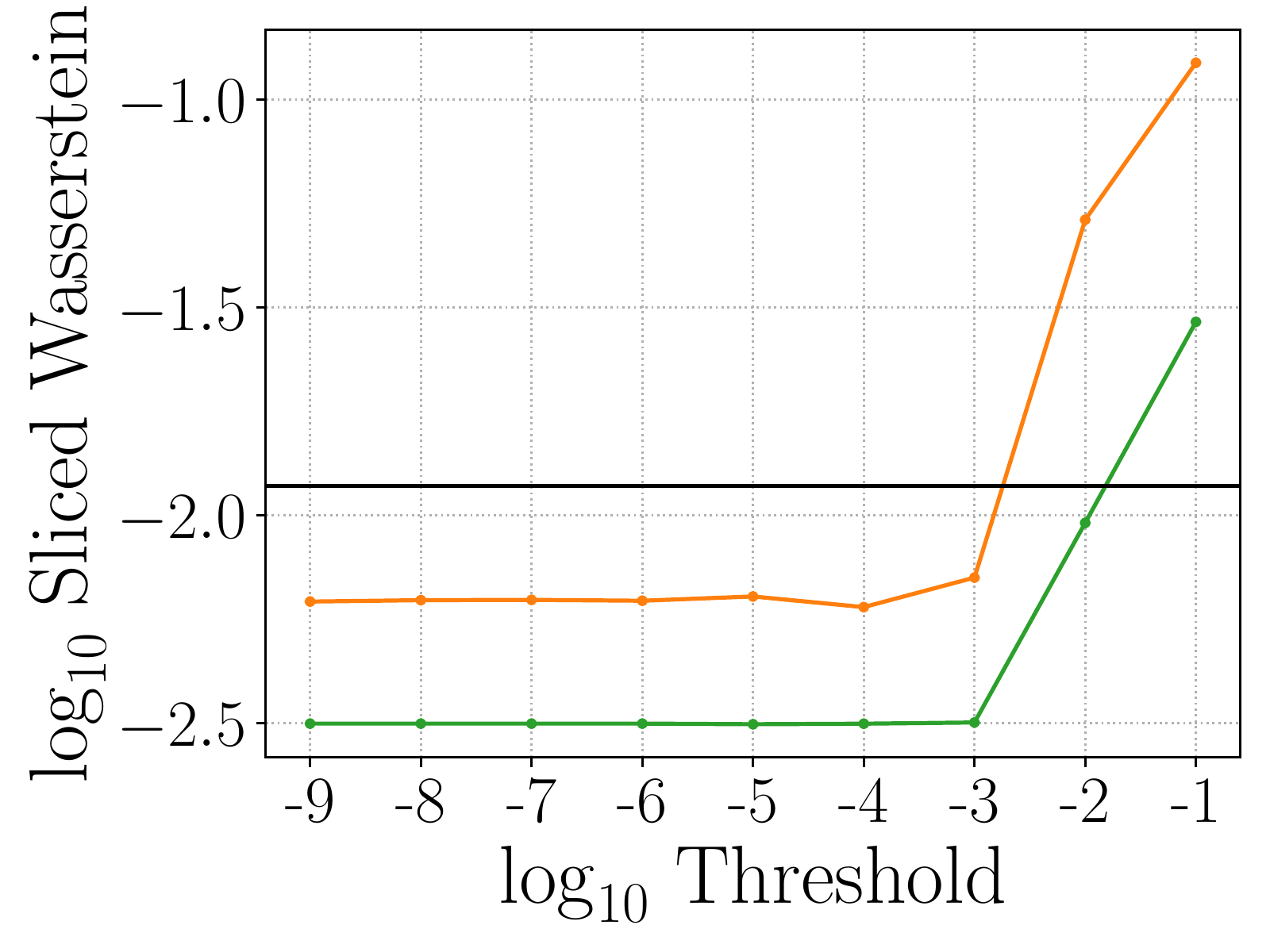}
    \caption{Sliced Wasserstein}
  \end{subfigure}
  \caption{Additional measures of ergodicity in the banana-shaped distribution. For thresholds smaller than $1\times10^{-2}$, the Riemannian methods produce a modest improvement in the ergodicity of the samples produced by the Markov chain. In computing the MMD metric, we employ a bandwidth of 1.727.}
  \label{fig:banana-extra-ergodicity}
\end{figure}

\begin{figure}[t!]
  \centering
  \includegraphics[width=\textwidth]{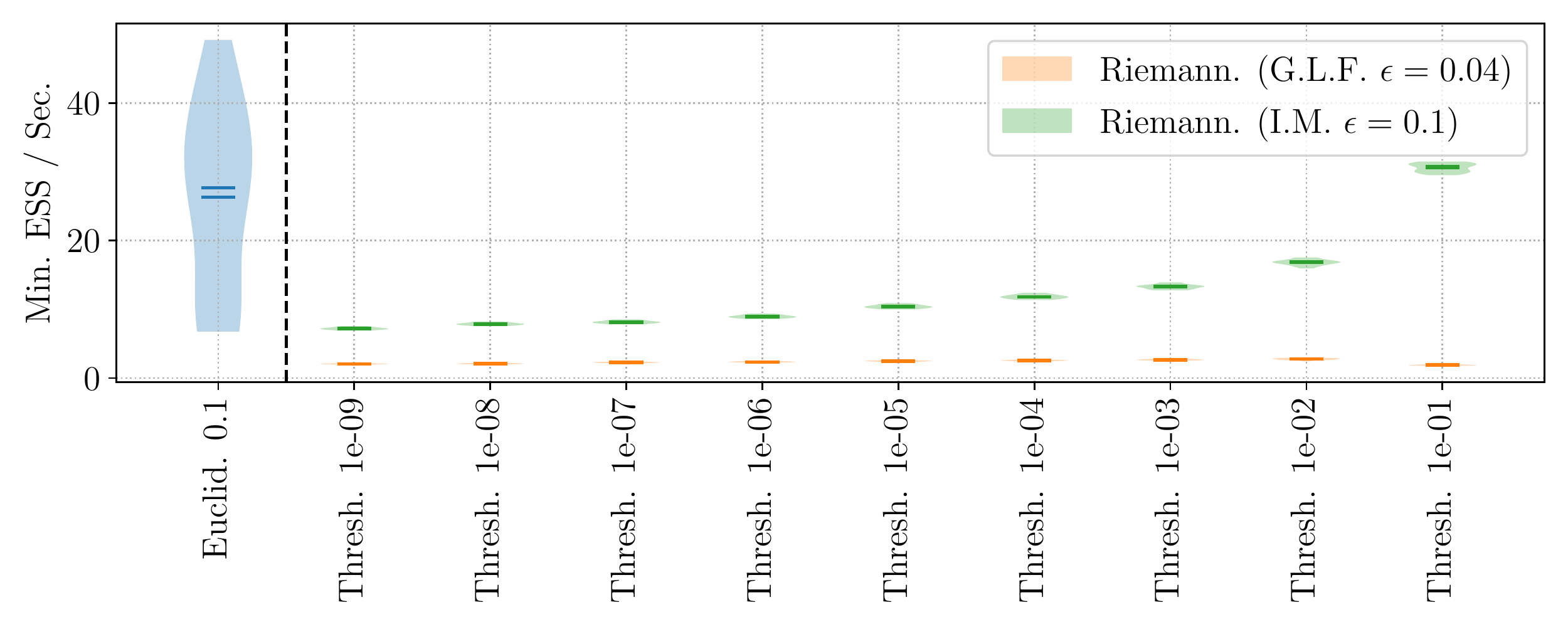}
  \caption{Visualization of the effective sample size per second of RMHMC compared to HMC on the banana-shaped distribution. Notice in particular the differentiation of the RMHMC with the generalized leapfrog integrator compared to the implicit midpoint method.}
\end{figure}
\begin{figure}[t!]
  \begin{subfigure}[t]{0.32\textwidth}
    \centering
    \includegraphics[width=\textwidth]{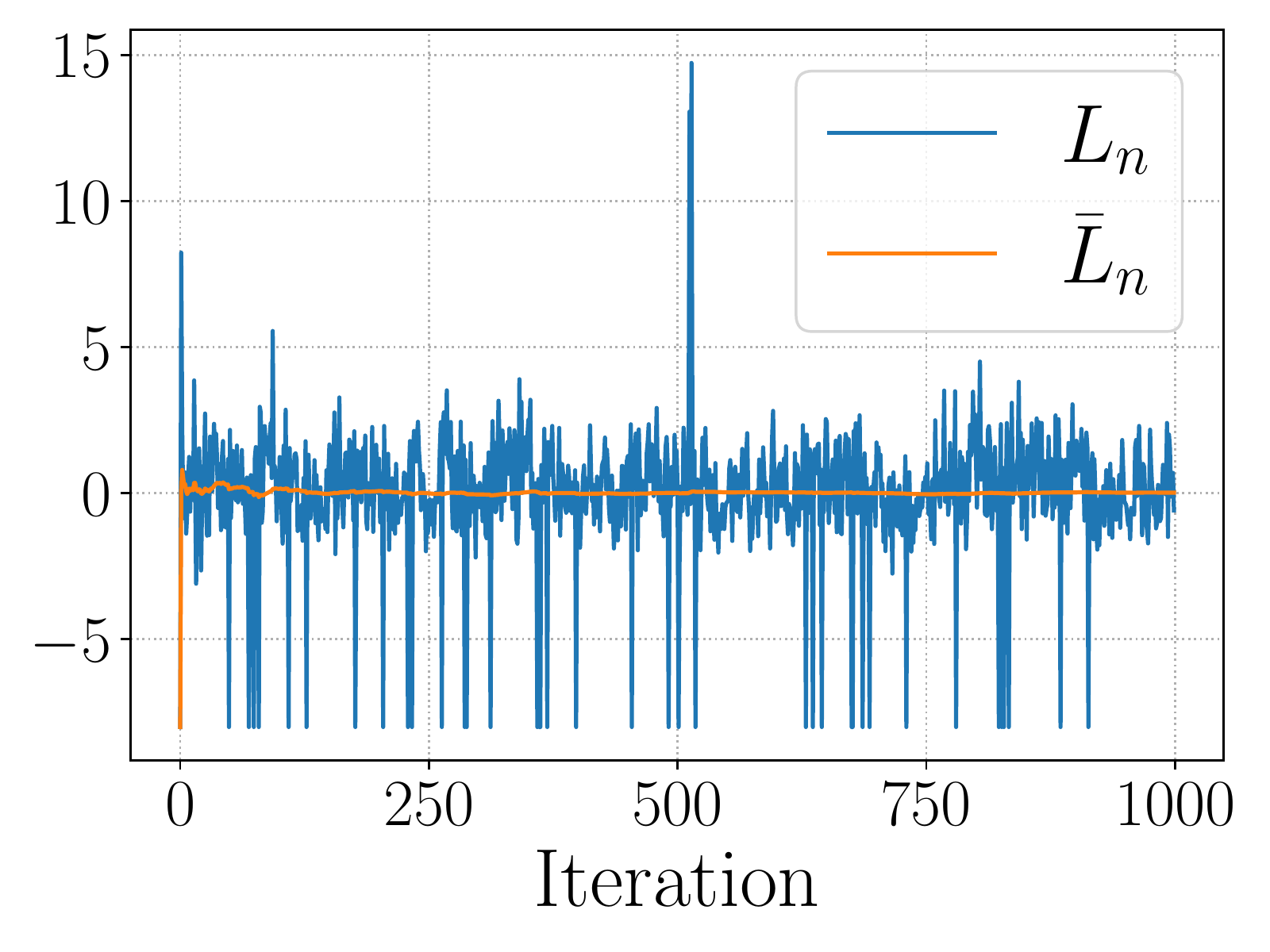}
    \caption{$L_n$ Sequence}
  \end{subfigure}
  ~
  \begin{subfigure}[t]{0.32\textwidth}
    \centering
    \includegraphics[width=\textwidth]{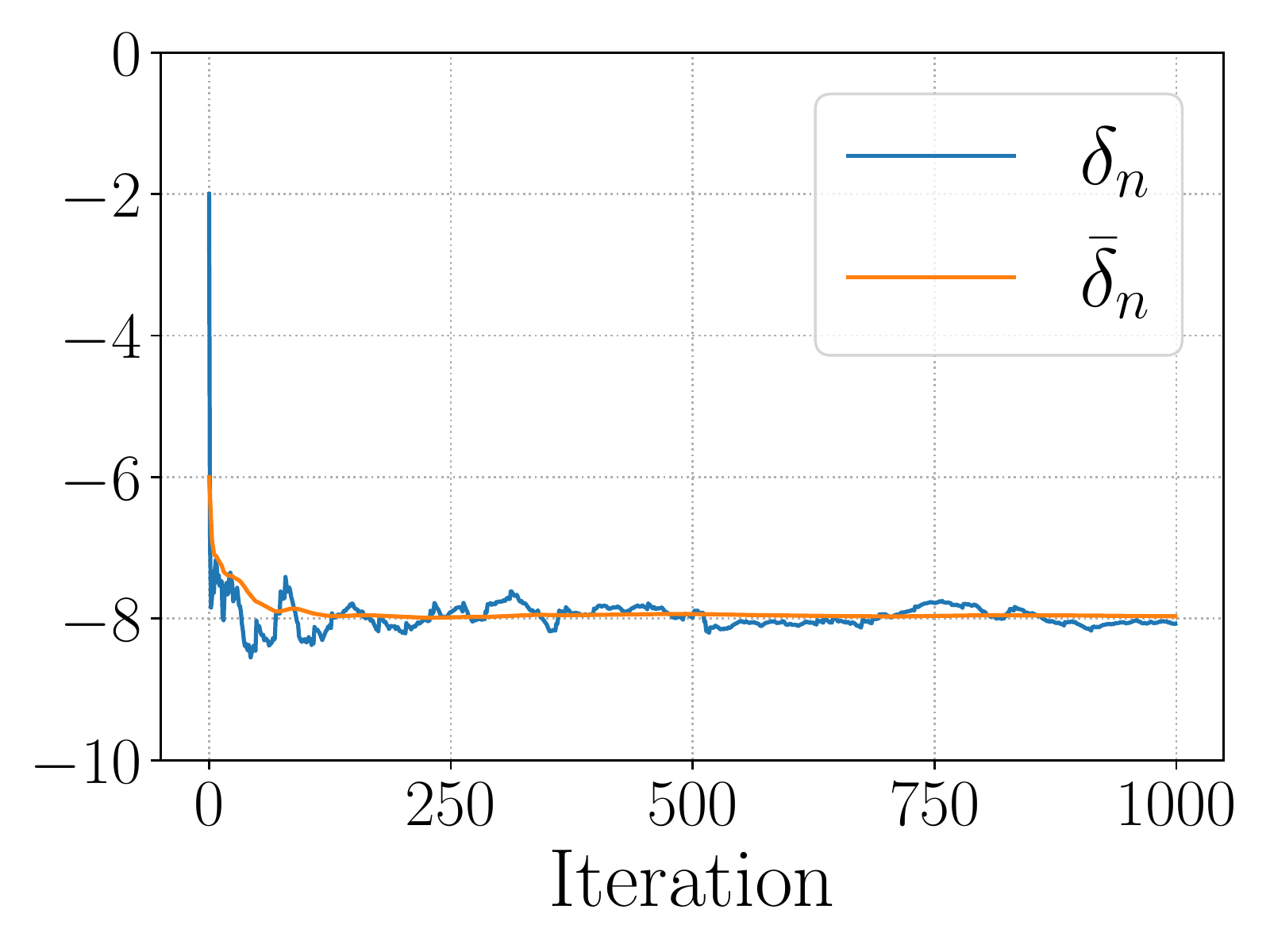}
    \caption{$\delta_n$ Sequence}
  \end{subfigure}
  ~
  \begin{subfigure}[t]{0.32\textwidth}
    \centering
    \includegraphics[width=\textwidth]{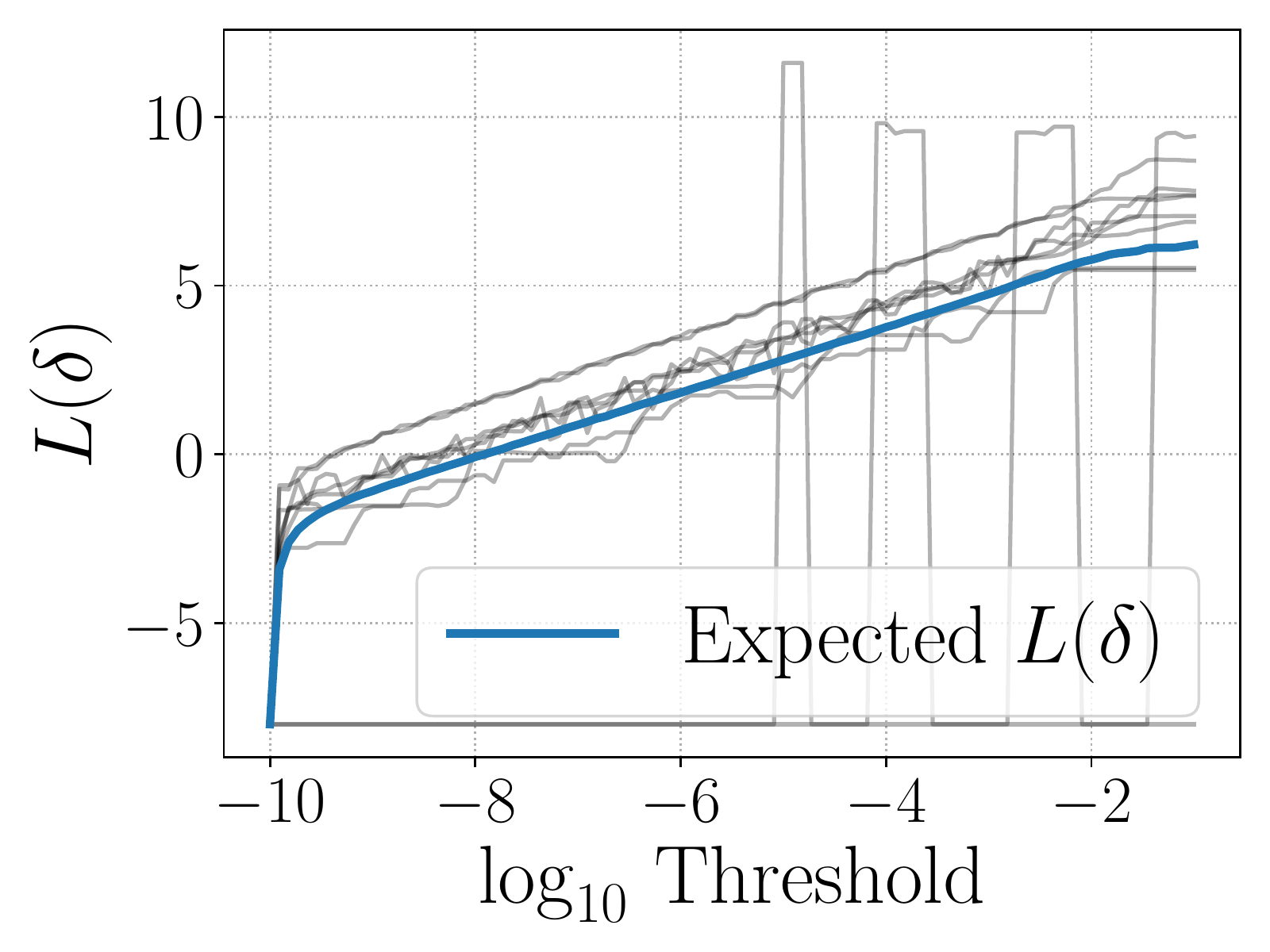}
    \caption{Monte Carlo $B(\delta)$}
    \label{subfig:banana-dual-averaging-function}
  \end{subfigure}
  \caption{The use of Ruppert averaging in the banana-shaped distribution to adaptively set the convergence threshold to achieve eight decimal digits of similarity compared to a transition kernel with a threshold of $1\times 10^{-10}$. We show a Monte Carlo approximation to $B(\delta)$, which appears smooth, monotonically increasing. The value of $\delta$ satisfying $B(\delta)=0$ is approximately $\delta = 1\times 10^{-8}$, which is approximately the value of produced by Ruppert averaging.}
  \label{fig:banana-dual-averaging}
\end{figure}

Consider the following generative model:
\begin{align}
  \theta_1,\theta_2 &\overset{\mathrm{i.i.d.}}{\sim} \mathrm{Normal}(0, \sigma_\theta^2) \\
  y_i \vert\theta_1,\theta_2 &\overset{\mathrm{i.i.d.}}{\sim} \mathrm{Normal}(\theta_1+\theta_2^2, \sigma_y^2) ~~\mathrm{for}~i=1,\ldots, 100.
\end{align}
As can be readily observed, the likelihood function of $(\theta_1,\theta_2)$ given $(y_1,\ldots,y_{100})$ is non-identifiable since for any real number $c \geq \theta_1$, there are two values of $\theta_2$ for which $c = \theta_1+\theta_2^2$; namely, $\theta_2 = \pm \sqrt{c - \theta_1}$. The purpose of of the banana-shaped distribution, suggested by Bornn and Cornebise in their rejoinder to \citet{rmhmc}, is therefore to give a representative example of a non-identifiable likelihood and the eponymous banana-shaped posterior that it produces. In our experiments we set $\sigma_y=2$ and $\sigma_\theta = 2$. When sampling observations $(y_1,\ldots, y_{100})$ we set $\theta_1 = 1/2$ and $\theta_2 = \sqrt{1 - 1/4}$. The sum of the Fisher information and the negative Hessian of the log-prior associated to the banana-shaped distribution is given by,
\begin{align}
  \mathbf{G}(\theta_1,\theta_2) = \begin{pmatrix}
    \frac{1}{\sigma_\theta^2}  + \frac{n}{\sigma_y^2} & \frac{2n\theta_2}{\sigma_y^2} \\
    \frac{2n\theta_2}{\sigma_y^2} & \frac{1}{\sigma_\theta^2} + \frac{4n\theta_2^2}{\sigma_y^2}
  \end{pmatrix}.
\end{align}

Due to its small dimensionality, we use the banana-shaped distribution as an opportunity to visualize the effect of varying threshold on the trajectories computed by the generalized leapfrog integrator. The principle advantage of the RMHMC algorithm lies in its ability to produce proposals that are adapted to directions of greatest variation locally. We would therefore like to assess if this property is preserved even for varying values of the threshold. We visualize a sample trajectory in \cref{subfig:banana-trajectory-threshold-precondition} where we observe that trajectories are qualitatively similar regardless of the threshold and extend along dimensions of greatest local variation; for both the Euclidean and the Riemannian trajectories, we consider integrating for $500$ steps using a step-size of 0.01. The behavior of the RMHMC proposal stands in contrast to HMC, which produces significant oscillations in directions of relatively little variation. This suggests that the preconditioning effect of RMHMC algorithm is not terribly sensitive to the threshold. We can obtain a more quantitative comparison of the trajectories by measuring the difference in trajectories at a given step for varying values of the thresold, which are shown in \cref{subfig:banana-trajectory-position-deviation}; as expected, the smaller thresholds exhibit greater fidelity toward the baseline but accumulating errors over the course of numerical integration produces trajectories that slowly diverge. We also evaluate conservation of the Hamiltonian over the course of the trajectory in \cref{subfig:banana-trajectory-threshold-precondition}. We observe that the largest threshold produces the largest variations in the Hamiltonian energy. The fact that the exact generalized leapfrog integrator is second-order accurate depends on its reversibility, since first-order accurate integrators that are reversible are automatically second-order accurate \citep{hairer-geometric}. As reversibility becomes increasingly violated with the larger thresholds, one expects the quality of the computed trajectory to diminish. Therefore, the fact the energy conservation becomes increasingly bad with larger thresholds is not entirely unexpected.

We now turn our discussion to inference in this statistical model. We consider Euclidean HMC with a step-size of $0.1$ and $20$ integration steps and with a step-size of 0.003 and $333$ integration steps. For RMHMC with the generalized leapfrog integrator, we consider a step-size of $0.04$ and $20$ integration steps. As described by Bornn and Cornebise, RMHMC requires a small integration step-size for the banana-shaped distribution because of divergences in the fixed point iterations for the momentum variable in \cref{eq:generalized-leapfrog-momentum-fixed-point-i,eq:generalized-leapfrog-momentum-fixed-point-ii}. The banana-shaped distribution a non-trivial amount of probability mass in long, thin ``tails'' that extend in the $\theta_1$-variable and which are symmetric for negations of the $\theta_2$-variable, whose exploration is inhibited by a small step-size. We also implement RMHMC using the implicit midpoint integrator with a step-size of $0.1$ and $20$ integration steps, which does not suffer the same divergences as the generalized leapfrog integrator due to its superior stability. We also observe the the implicit midpoint integrator has an acceptance probability of ninety-five percent, even greater than the generalized leapfrog integrator with a step-size less than half that used by the implicit midpoint method. The implicit midpoint method produces RMHMC with the best ergodicity properties, dominating both Eulcidean HMC and RMHMC with the implicit midpoint integrator when 1,000,000 samples are drawn. In terms of ergodicity, a threshold of $1\times 10^{-3}$ appears to be sufficient for inference in the banana-shaped distribution. However, the computational expediency of Euclidean HMC causes it to enjoy the best ESS per second, which must be balanced against the superior ergodicity of the the Riemannian variants. Notably, however, a threshold of $1\times 10^{-3}$ is by no means optimal in terms of producing the best reversibility and volume preservation among the implicitly defined integrators, indicating that computational benefits may be obtained via principled selection of this convergence parameter.

We also consider the use of the Ruppert averaging procedure in the banana-shaped posterior in order to identify a threshold that has, on average, eight ($\kappa=8$) decimal digits of similarity with a numerical integrator whose convergence threshold is $1\times 10^{-10}$. Results are shown in \cref{fig:banana-dual-averaging}. We see that convergence of $\bar{L}_n$ to zero is rapid; the sequence $(\bar{\delta}_1,\bar{\delta}_2,\ldots)$ converges by, approximately, iteration six-hundred. In \cref{subfig:banana-dual-averaging-function} we show a Monte Carlo approximation to $B(\delta)$ for the banana-shaped distribution and, in addition, ten random samples of $L(\delta)$ for randomly selected values of the position and momentum. One observes that individual samples are not monotonically increasing, nor do they appear to be particularly smooth. However, the average function, shown in blue, does appear monotonically increasing and smooth. These observations will be replicated in our other experiments.

\subsection{Bayesian Logistic Regression}

\begin{figure}[t!]
  \centering
  \begin{subfigure}[t]{\textwidth}
  \includegraphics[width=\textwidth]{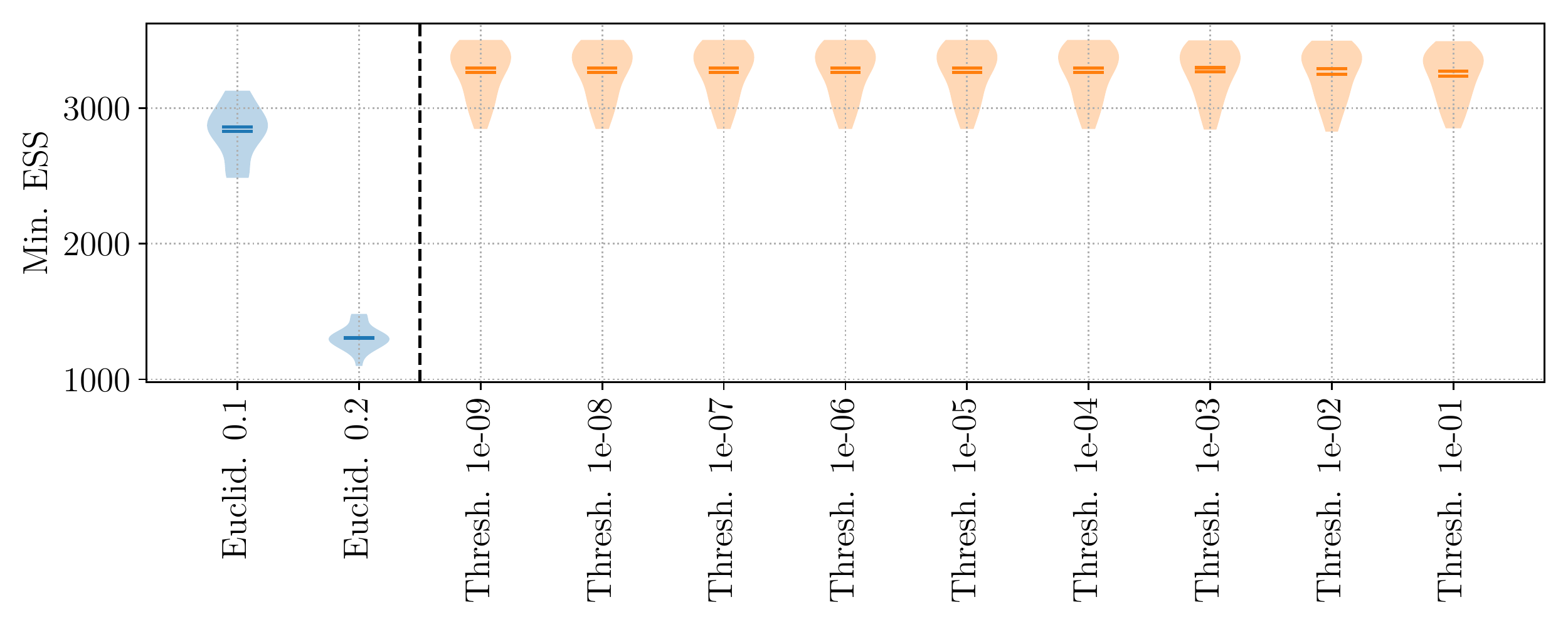}
  \label{subfig:logistic-regression-ess}
  \end{subfigure}
  
  \begin{subfigure}[t]{\textwidth}
  \includegraphics[width=\textwidth]{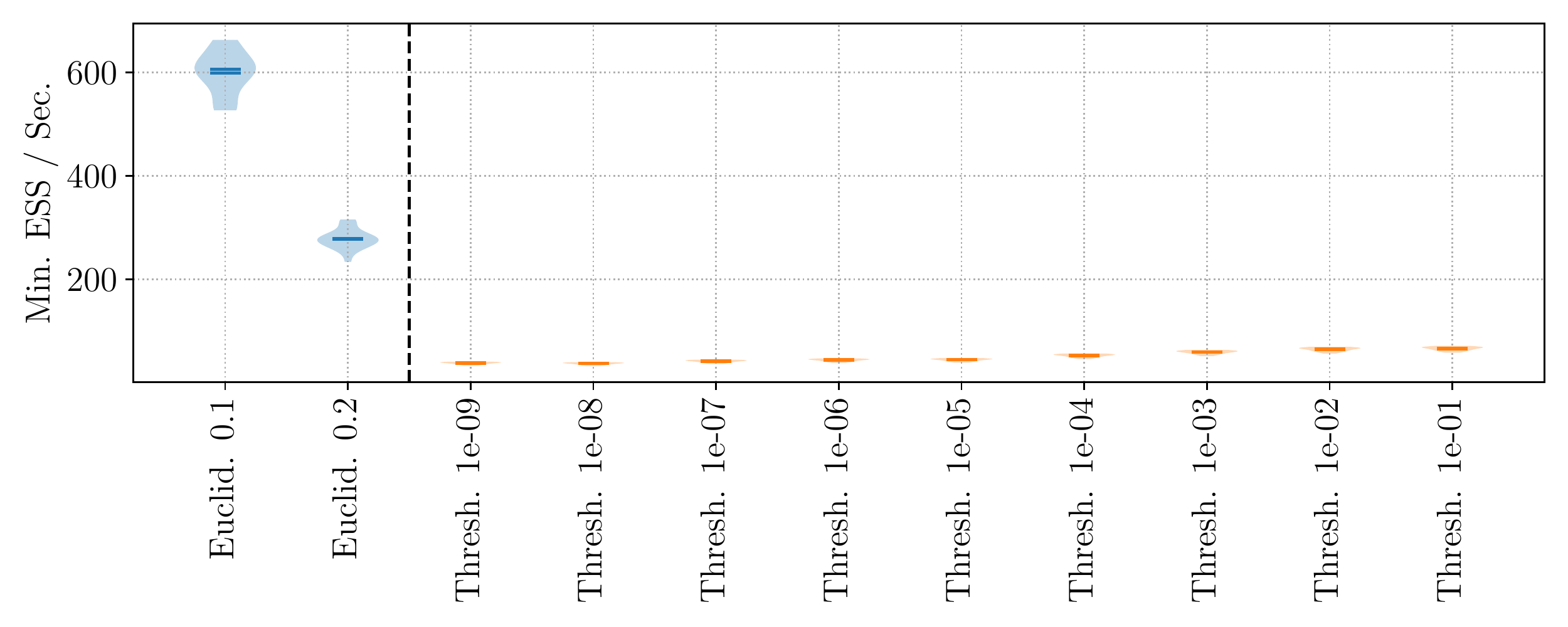}
  \label{subfig:logistic-regression-ess-per-second}
  \end{subfigure}
  \caption{Visualization of the effective sample sizes (ESS) of the variables in the logistic regression posterior. Distributions over ESS are computed by splitting a Markov chain of length 100,000 into twenty contiguous sequences of length 5,000. In \cref{subfig:logistic-regression-ess} we show the effective sample size and the time-normalized ESS per second is given in \cref{subfig:logistic-regression-ess-per-second}.}
  \label{fig:logistic-regression-ess}
\end{figure}

\begin{figure}[t!]
  \centering
  \begin{subfigure}[t]{0.32\textwidth}
    \includegraphics[width=\textwidth]{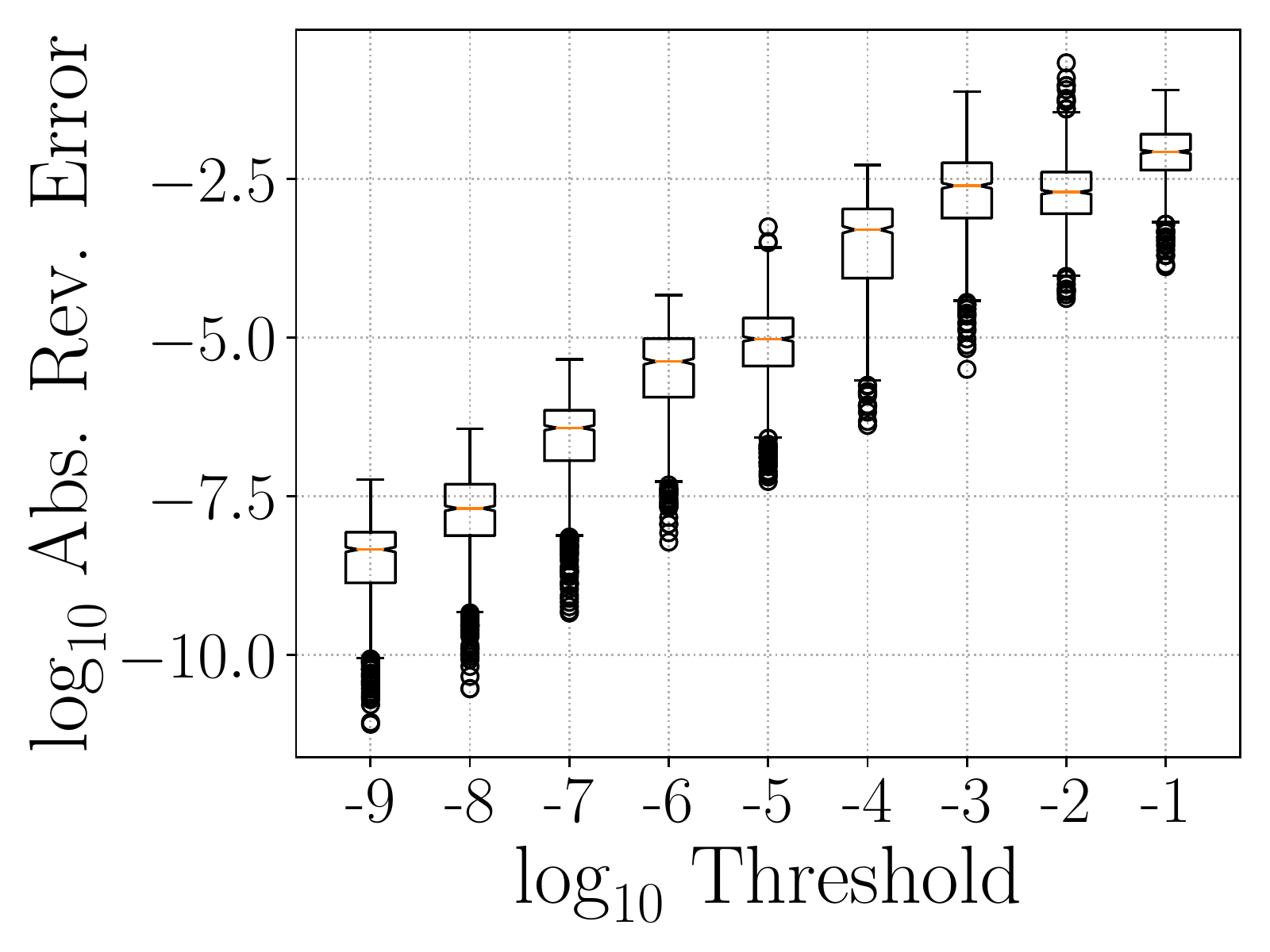}
    \caption{Error in Reversibility}
    \label{subfig:logistic-regression-reversibility}
  \end{subfigure}
  ~
  \begin{subfigure}[t]{0.32\textwidth}
    \includegraphics[width=\textwidth]{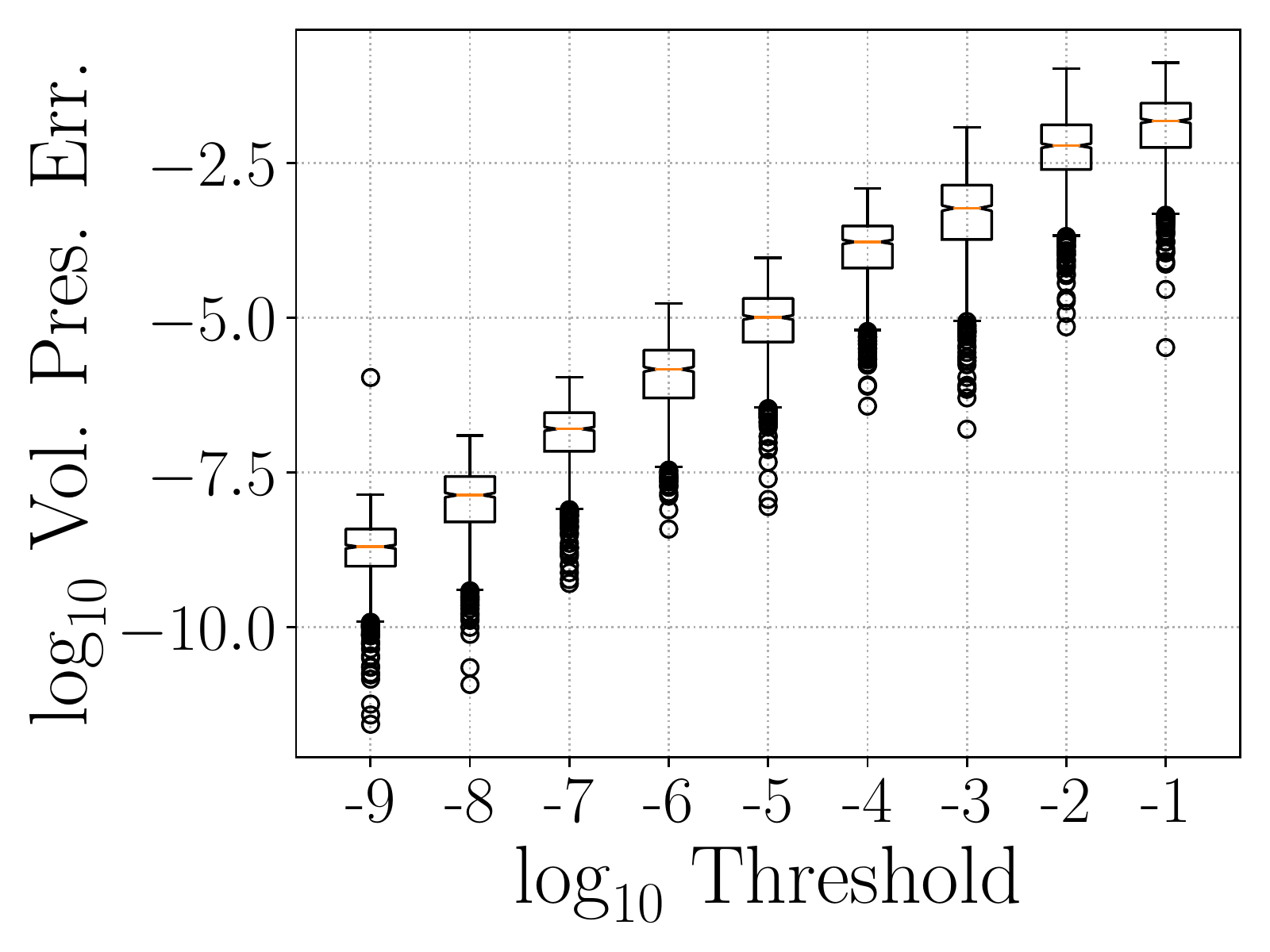}
    \caption{Error in Volume-Preservation}
    \label{subfig:logistic-regression-jacobian-determinant}
  \end{subfigure}
  ~
  \begin{subfigure}[t]{0.32\textwidth}
    \includegraphics[width=\textwidth]{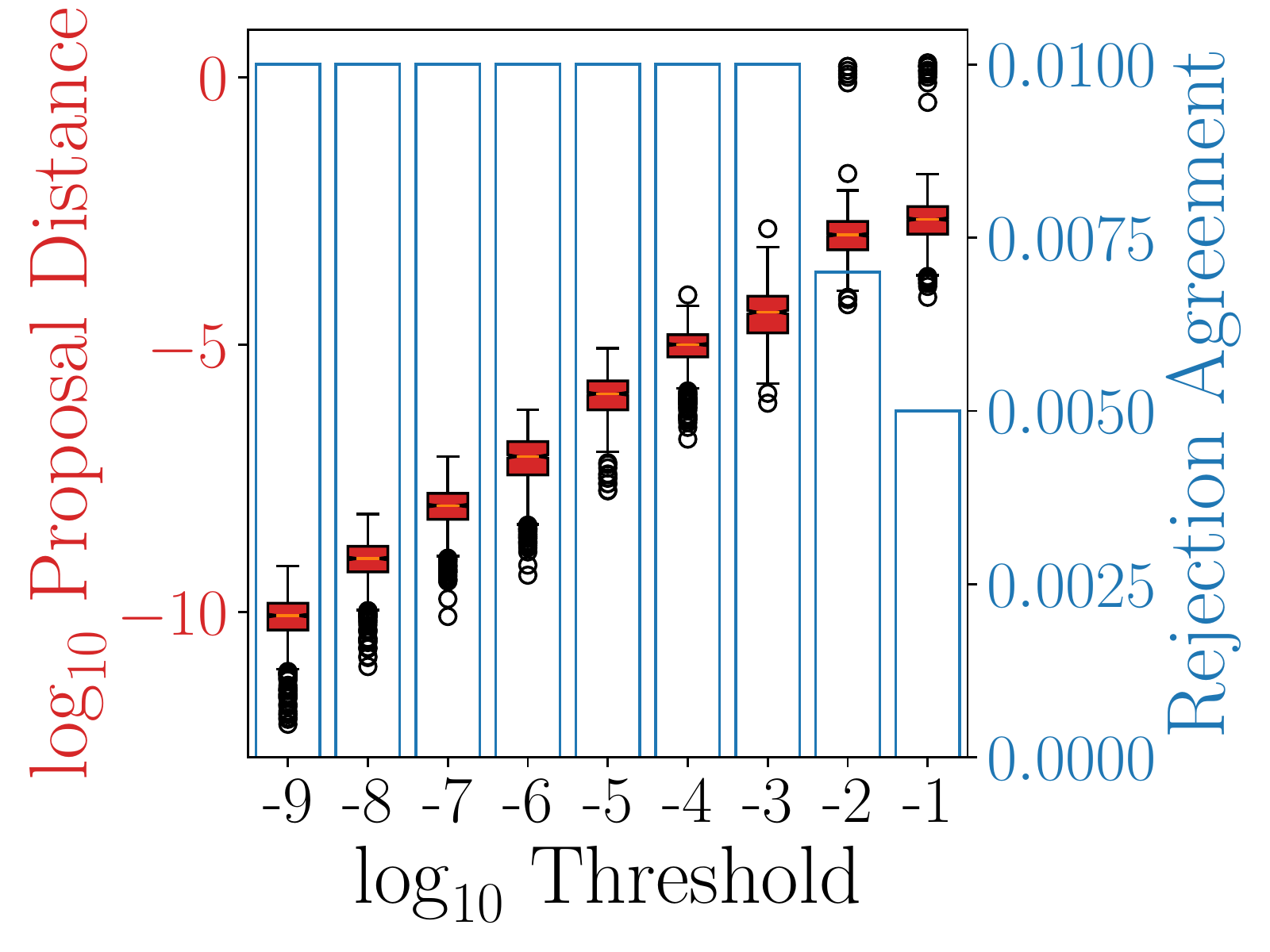}
    \caption{Difference in Transition Kernel}
    \label{subfig:logistic-regression-transition-difference}
  \end{subfigure}
  \caption{Visualization of the error in reversibility (see \cref{subfig:logistic-regression-reversibility}), error in volume-preservation (see \cref{subfig:logistic-regression-jacobian-determinant}), and the number of decimal digits of similarity in transition kernels (see \cref{subfig:logistic-regression-transition-difference}) for variable thresholds in the logistic posterior distribution.}
  \label{fig:logistic-regression-metrics}
\end{figure}

We consider a hierarchical Bayesian logistic regression defined by the following generative model
\begin{align}
  \alpha &\sim \mathrm{Gamma}(k, \theta) \\
  \beta_i \vert\alpha &\overset{\mathrm{i.i.d.}}{\sim} \mathrm{Normal}(0, \alpha^{-1}) ~~\mathrm{for}~i=1,\ldots, p \\
  y_i \vert \beta, \mathbf{x}_i &\overset{\mathrm{i.i.d.}}{\sim} \mathrm{Bernoulli}(\sigma(\mathbf{x}_i^\top\beta)) ~~\mathrm{for}~i=1,\ldots, n,
\end{align}
where $\sigma :\R\to (0, 1)$ is the sigmoid function. In our experiments we analyze the {\tt heart} dataset from \citet{rmhmc}. This dataset consists of $p=14$ regression coefficients and $n=270$ observations; including the latent prior precision, this produces a 15-dimensional posterior distribution. In our experiments, we set $k=1$ and $\theta=2$. We employ a Metropolis-within-Gibbs sampling strategy whereby we alternate between sampling $\beta\vert ((\mathbf{x}_1,y_1),\ldots,(\mathbf{x}_n, y_n)), \alpha$ using RMHMC and sampling $\alpha\vert\beta, ((\mathbf{x}_1,y_1),\ldots,(\mathbf{x}_n, y_n))$ analytically since,
\begin{align}
    \alpha\vert\beta,
((\mathbf{x}_1,y_1),\ldots,(\mathbf{x}_n, y_n)) \sim\mathrm{Gamma}\paren{k+\frac{p}{2}, \frac{1}{\frac{\beta^\top\beta}{2} + \frac{1}{\theta}}}.
\end{align}
The Fisher
information metric for sampling the former of these distributions depends on the value of $\beta$, thereby necessitating the use of the generalized leapfrog integrator. In particular, the sum of the Fisher information and the negative Hessian of the log-prior can be shown to be,
\begin{align}
  \mathbf{G}(\beta) = \mathbf{X}^\top \Lambda\mathbf{X} + \alpha ~\mathrm{Id}
\end{align}
where $\Lambda$ is a diagonal matrix with entries $\Lambda_{i, i} =
\sigma(\mathbf{x}_i^\top \beta)(1-\sigma(\mathbf{x}_i^\top\beta))$ and
$\mathbf{X}\in\R^{n\times p}$ is the row-wise concatenation of
$\mathbf{x}_1,\ldots,\mathbf{x}_n$.

We compare the ESS of the RMHMC algorithm in \cref{fig:logistic-regression-ess}. We compute the ESS by taking a single chain of length $100,000$ and splitting it into twenty contiguous arrays of length $5,000$. Within each contiguous sample, we compute the ESS and report the minimum ESS among the linear coefficients and the ESS of the precision variable. In this experiment, the minimum ESS is effectively constant as a function of the threshold, relative to the ESS per second exhibited by HMC. Indeed, this example provided a circumstance wherein no threshold employed in RMHMC was able to produce a time normalized minimum ESS which was competitive with Euclidean HMC. It is worth noting, however, that RMHMC tends to outperform HMC when time is {\it not} accounted for, as demonstrated in the top panel of \cref{fig:logistic-regression-ess}. We employ RMHMC with a step-size of $0.2$ and twenty integration steps. For reference, we also report these ESS statistics for HMC with a step-sizes of 0.1 and 0.2 and twenty integration steps.

We visualize the violation of reversibility and volume preservation over varying thresholds in \cref{fig:logistic-regression-metrics}. This analysis reveals that, in the worst case, the proposal operator enjoys approximately two decimal digits of reversibility and volume preservation. Moreover, for a threshold of $1\times 10^{-3}$, the transition kernels are in agreement to the third decimal place.

\subsection{Neal's Funnel Distribution}\label{subsec:experiment-neal-funnel}

\begin{figure}[t!]
  \centering
  \begin{subfigure}[t]{0.32\textwidth}
    \includegraphics[width=\textwidth]{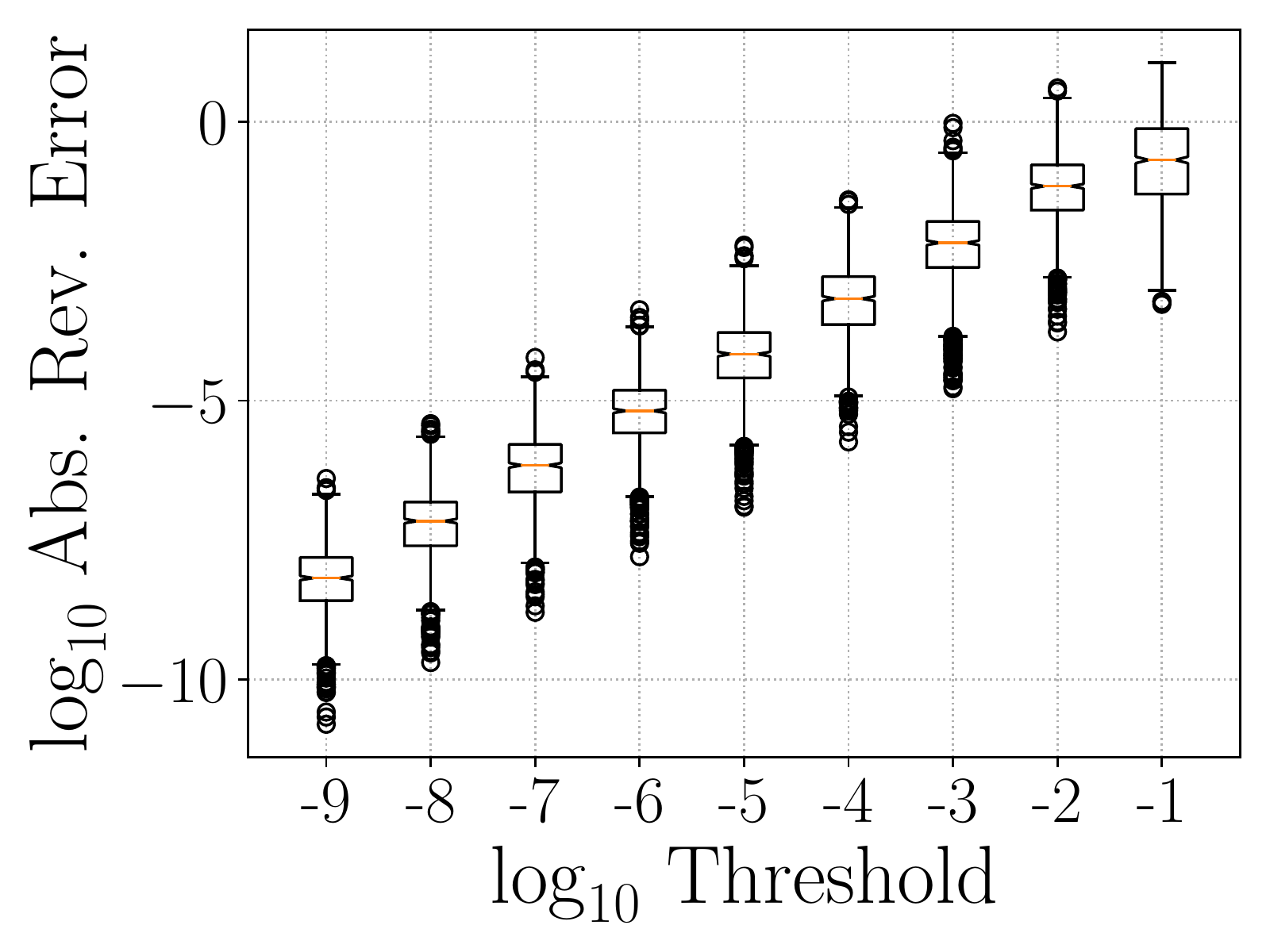}
    \caption{Error in Reversibility}
    \label{subfig:neal-funnel-reversibility}
  \end{subfigure}
  ~
  \begin{subfigure}[t]{0.32\textwidth}
    \includegraphics[width=\textwidth]{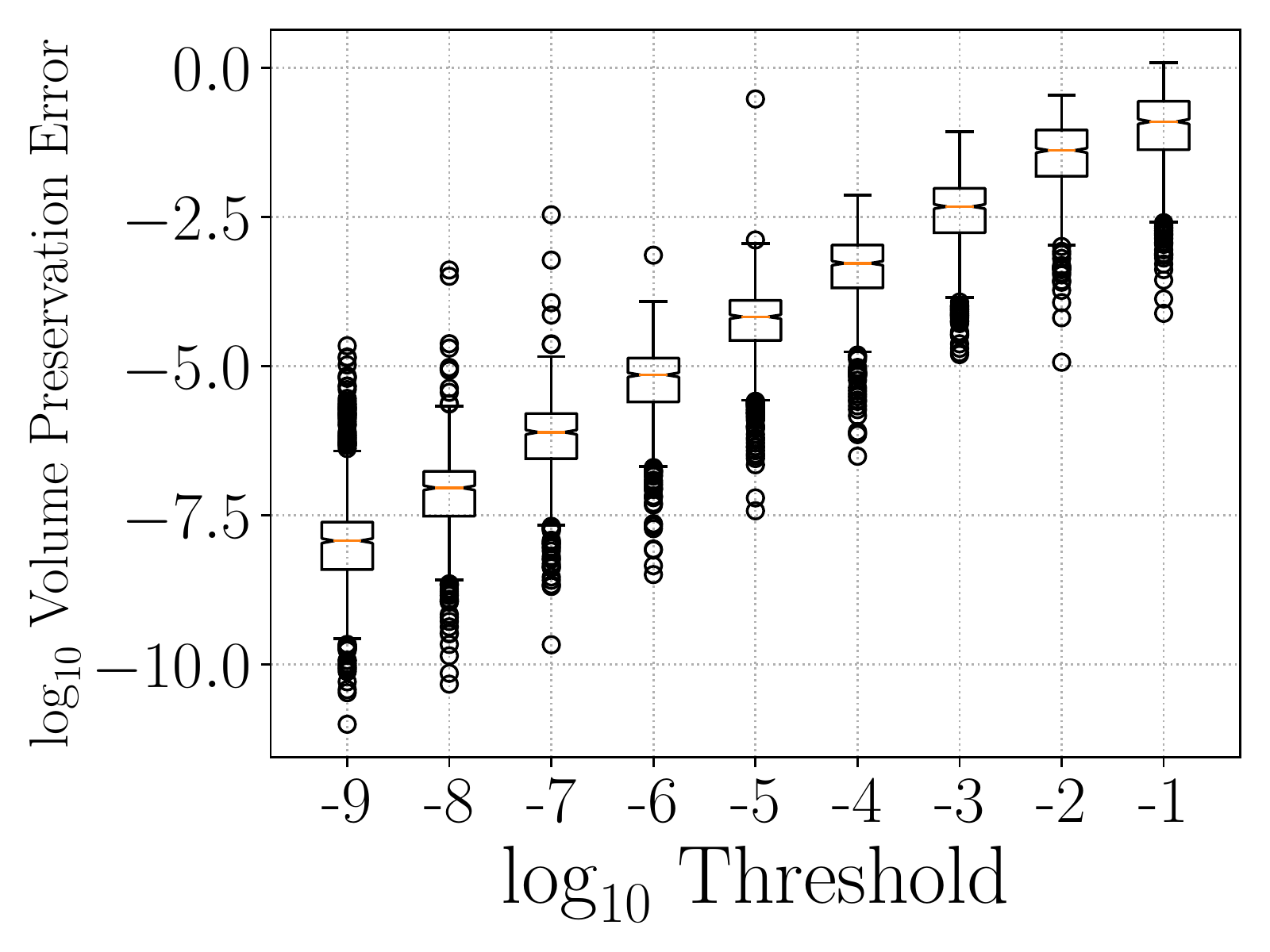}
    \caption{Error in Volume-Preservation}
    \label{subfig:neal-funnel-jacobian-determinant}
  \end{subfigure}
  ~
  \begin{subfigure}[t]{0.32\textwidth}
    \includegraphics[width=\textwidth]{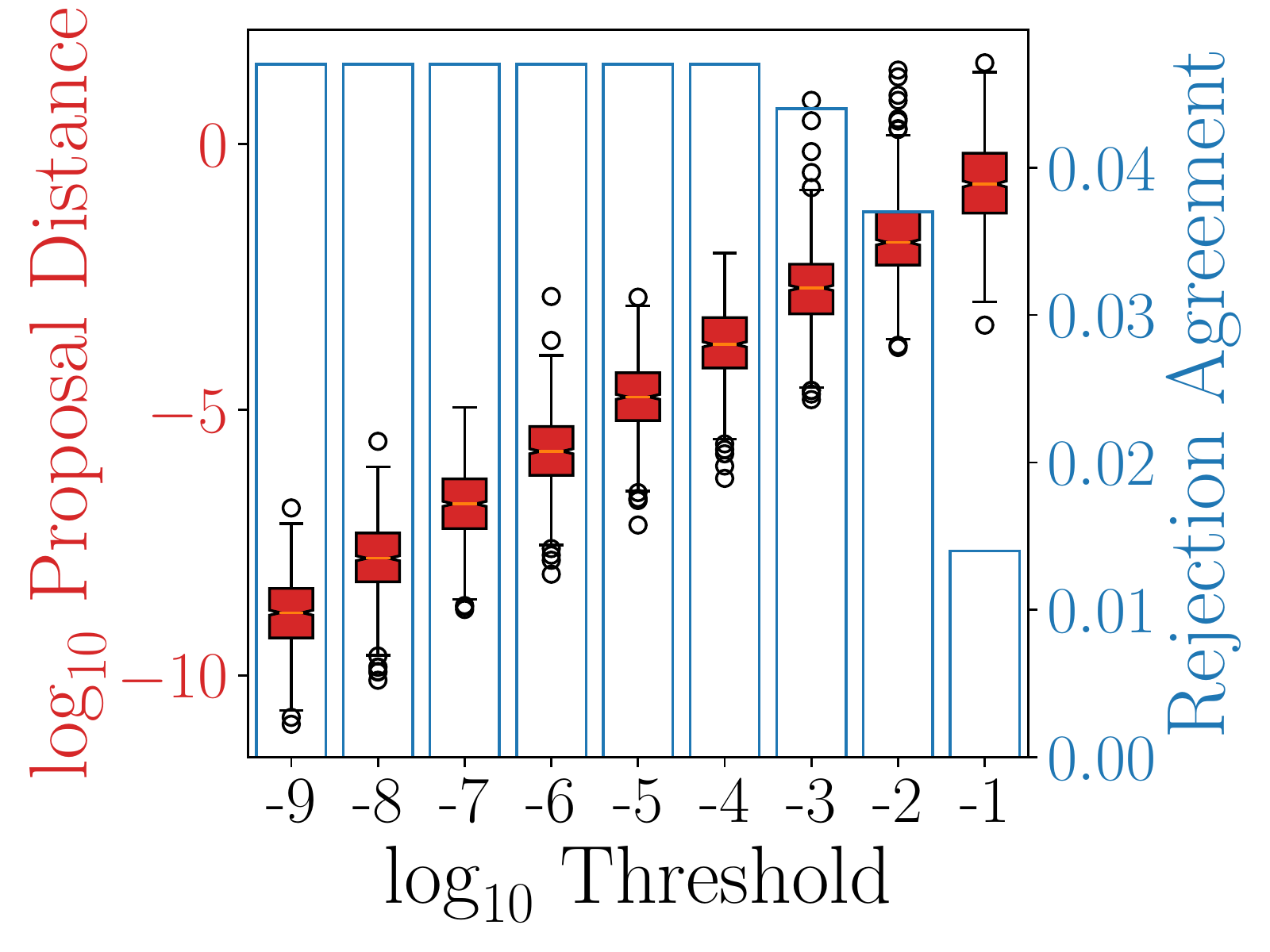}
    \caption{Difference in Transition Kernel}
    \label{subfig:neal-funnel-transition-difference}
  \end{subfigure}
  \caption{Visualization of the error in reversibility (see \cref{subfig:neal-funnel-reversibility}), error in volume-preservation (see \cref{subfig:neal-funnel-jacobian-determinant}), and the number of decimal digits of similarity in transition kernels (see \cref{subfig:neal-funnel-transition-difference}) for variable thresholds in Neal's funnel distribution with $10 + 1$ dimensions.}
\end{figure}

\begin{figure}[t!]
  \centering
  \begin{subfigure}[t]{0.49\textwidth}
    \includegraphics[width=\textwidth]{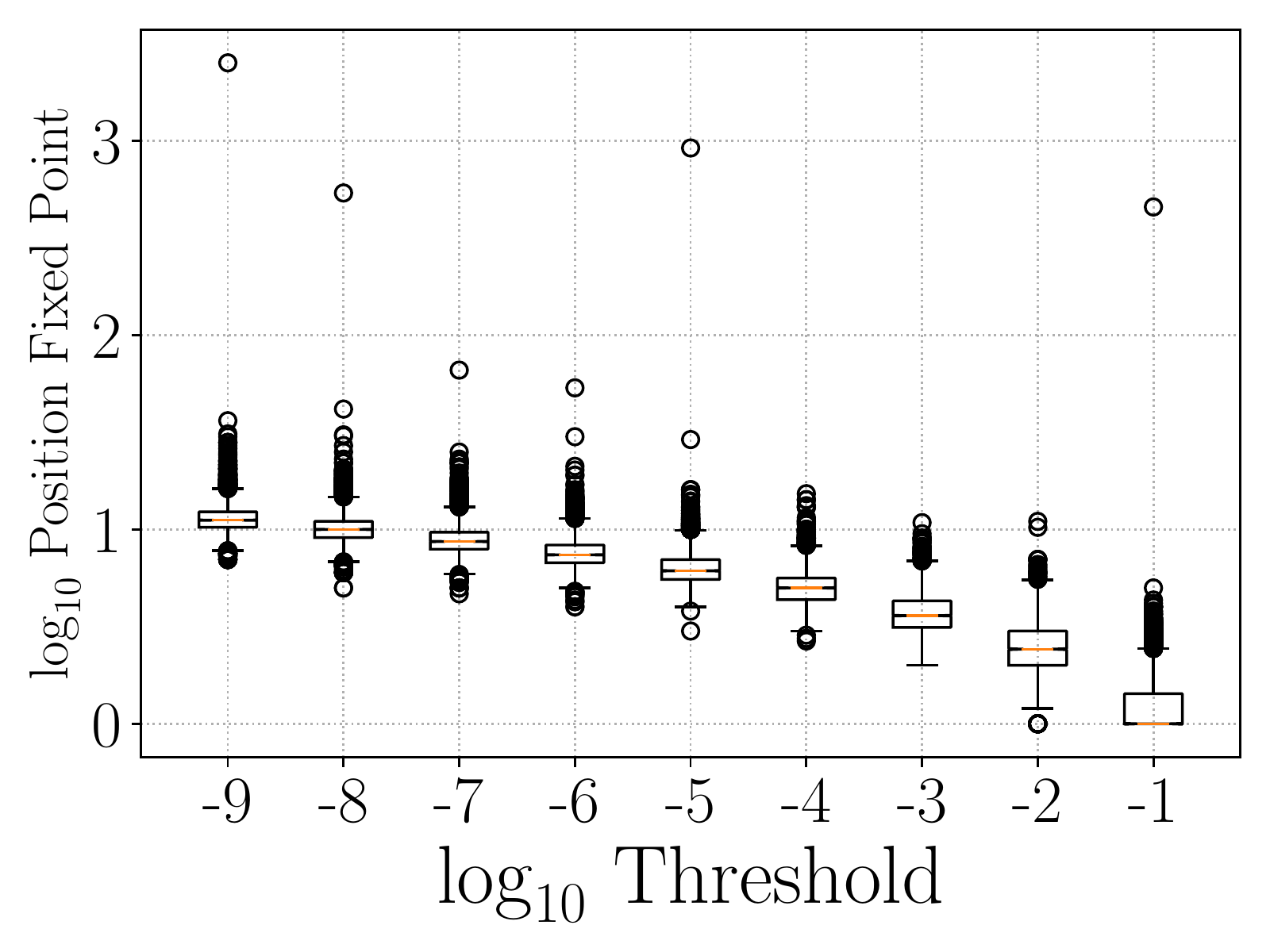}
    \caption{Number of fixed point iterations for position variable.}
    \label{subfig:neal-funnel-fixed-point-position}
  \end{subfigure}
  ~
  \begin{subfigure}[t]{0.49\textwidth}
    \includegraphics[width=\textwidth]{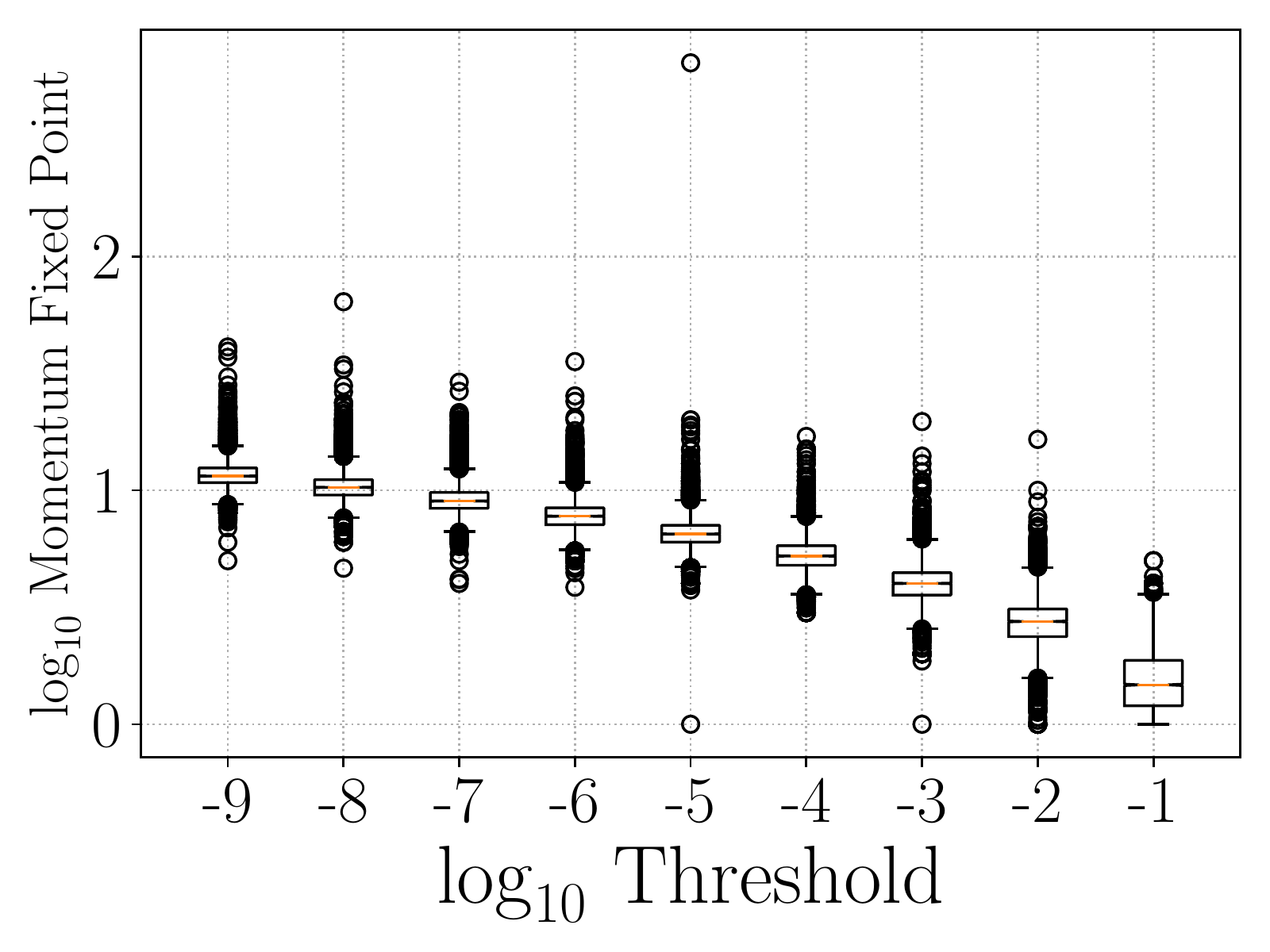}
    \caption{Number of fixed point iterations for momentum variable.}
    \label{subfig:neal-funnel-fixed-point-momentum}
  \end{subfigure}
  \caption{Visualization of the computational effort required to sample with RMHMC from Neal's funnel distribution. We show the number of fixed point iterations required to compute the two implicit steps of the generalized leapfrog integrato.}
\end{figure}

\begin{figure}[t!]
  \centering
  \begin{subfigure}[t]{\textwidth}
  \includegraphics[width=\textwidth]{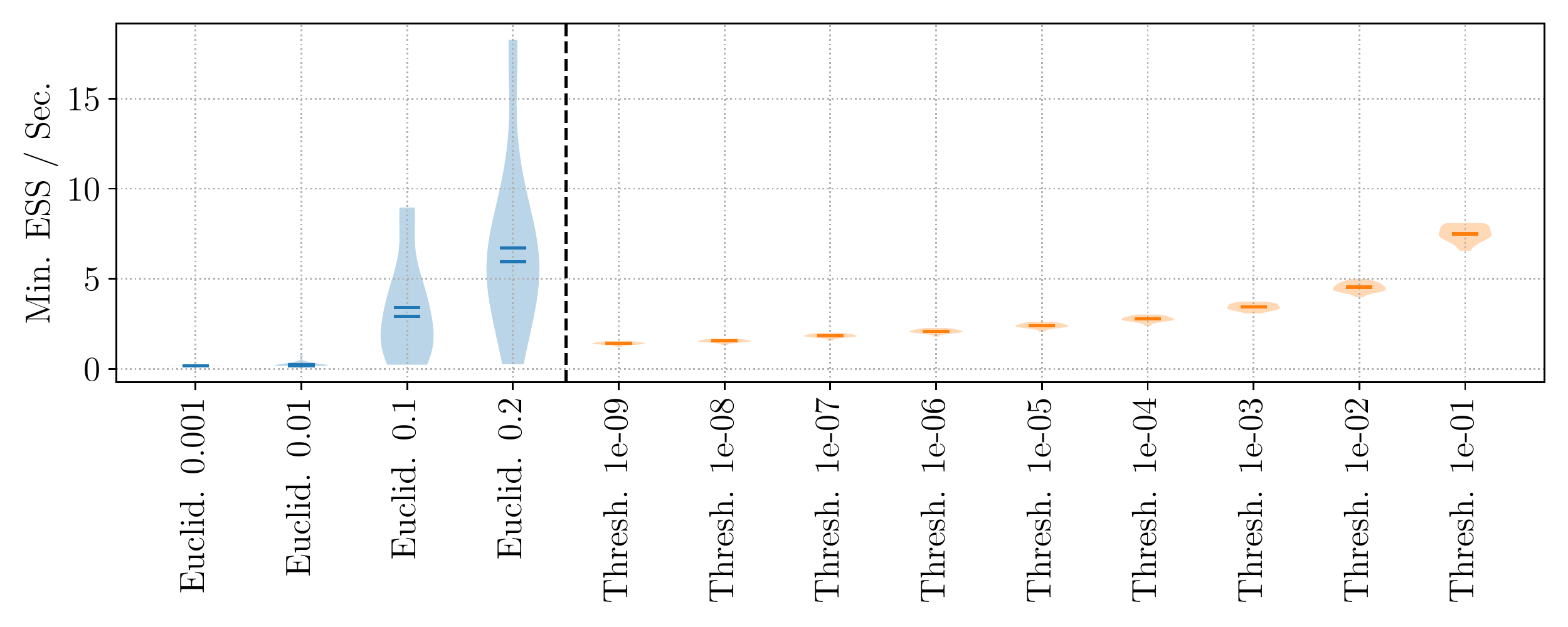}
  \caption{Minimum ESS per second in Neal's funnel distribution.}
  \label{subfig:neal-funnel-ess-per-second}
  \end{subfigure}
  
  \begin{subfigure}[t]{\textwidth}
  \includegraphics[width=\textwidth]{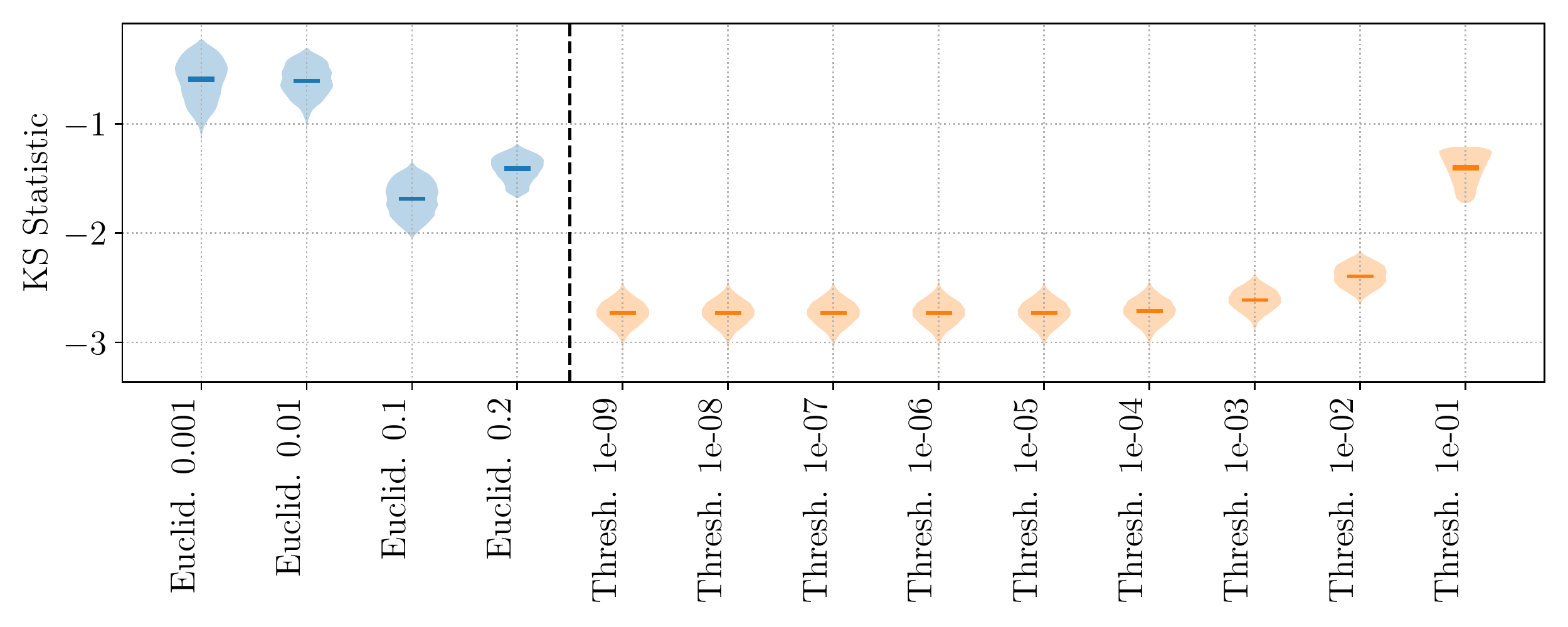}
  \caption{Sample ergodicity in Neal's funnel distribution.}
  \label{subfig:neal-funnel-ergodicity}
  \end{subfigure}
  
  \caption{Visualization of the sample quality of the variables in Neal's funnel distribution as measured by the ESS per second (\cref{subfig:neal-funnel-ess-per-second}) and the distribution of Kolmogorov-Smirnov (KS) statistics over random one-dimensional subspaces (\cref{subfig:neal-funnel-ergodicity}). Distributions over ESS are computed by splitting a Markov chain of length 1,000,000 into twenty contiguous sequences of length 50,000.}
  \label{fig:neal-funnel-ess}
\end{figure}

\begin{figure}[t!]
  \begin{subfigure}[t]{0.32\textwidth}
    \centering
    \includegraphics[width=\textwidth]{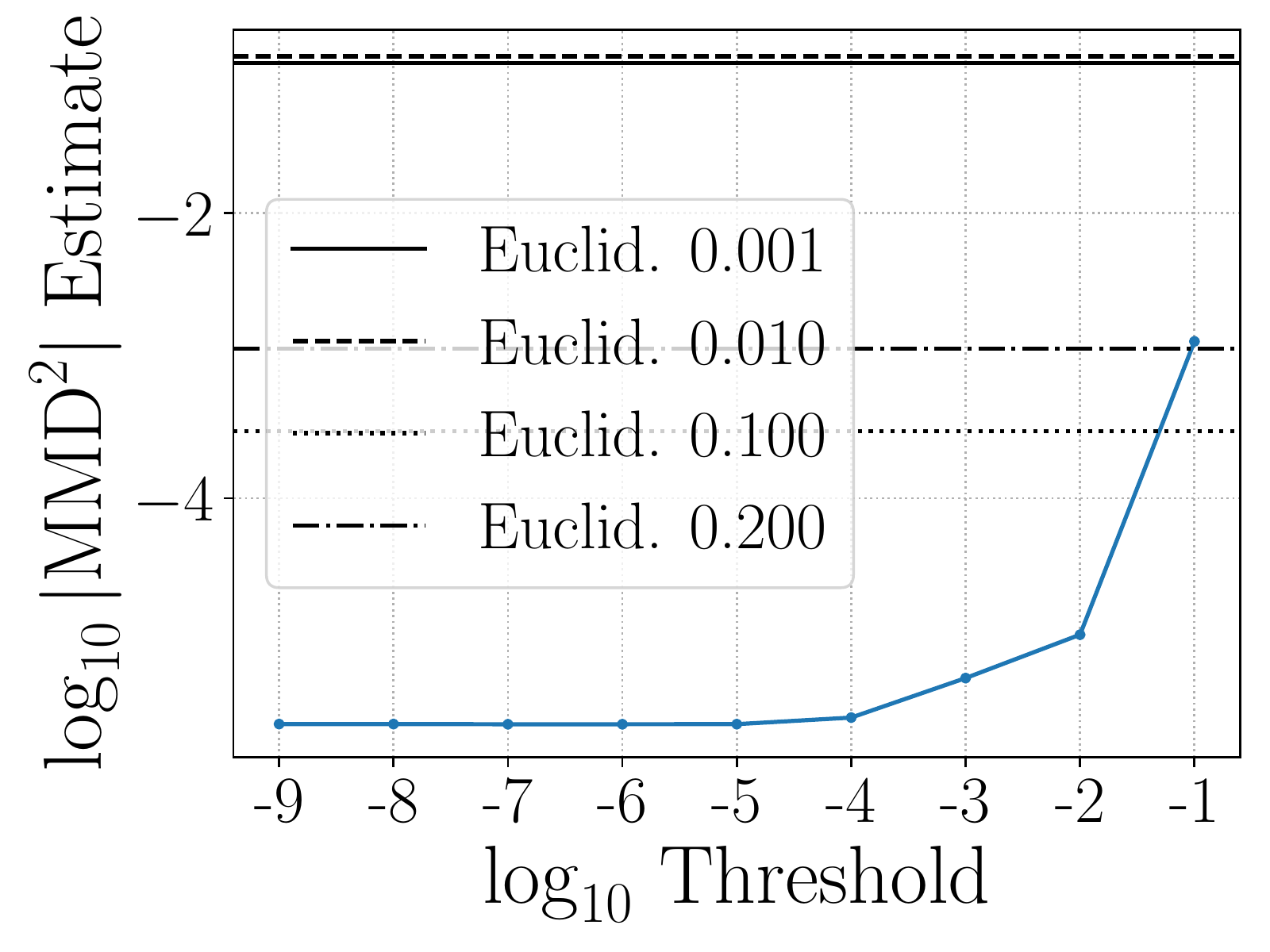}
    \caption{$\mathrm{MMD}_u^2$}
  \end{subfigure}
  ~
  \begin{subfigure}[t]{0.32\textwidth}
    \centering
    \includegraphics[width=\textwidth]{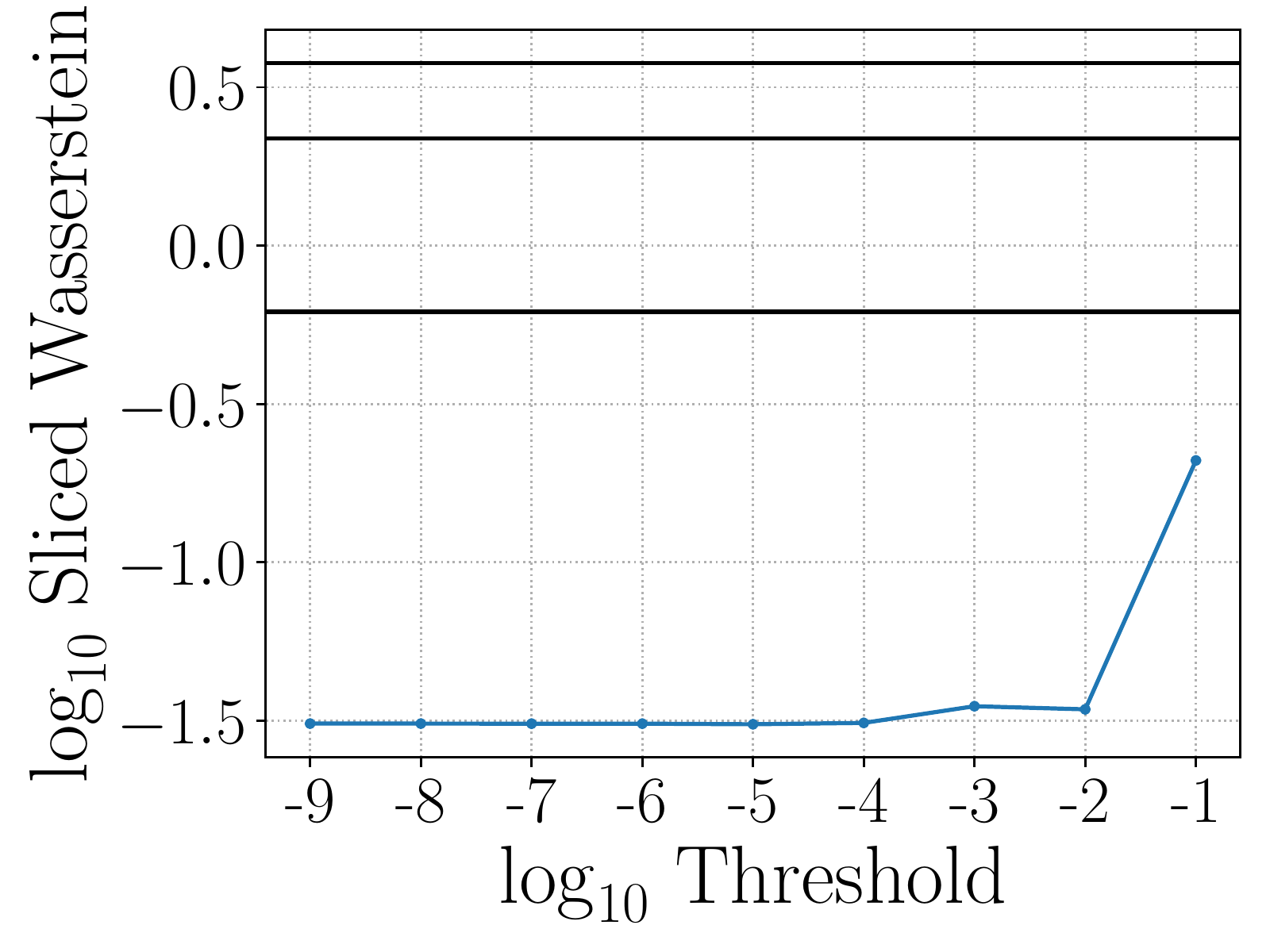}
    \caption{Sliced Wasserstein}
  \end{subfigure}
  ~
  \begin{subfigure}[t]{0.32\textwidth}
    \centering
    \includegraphics[width=\textwidth]{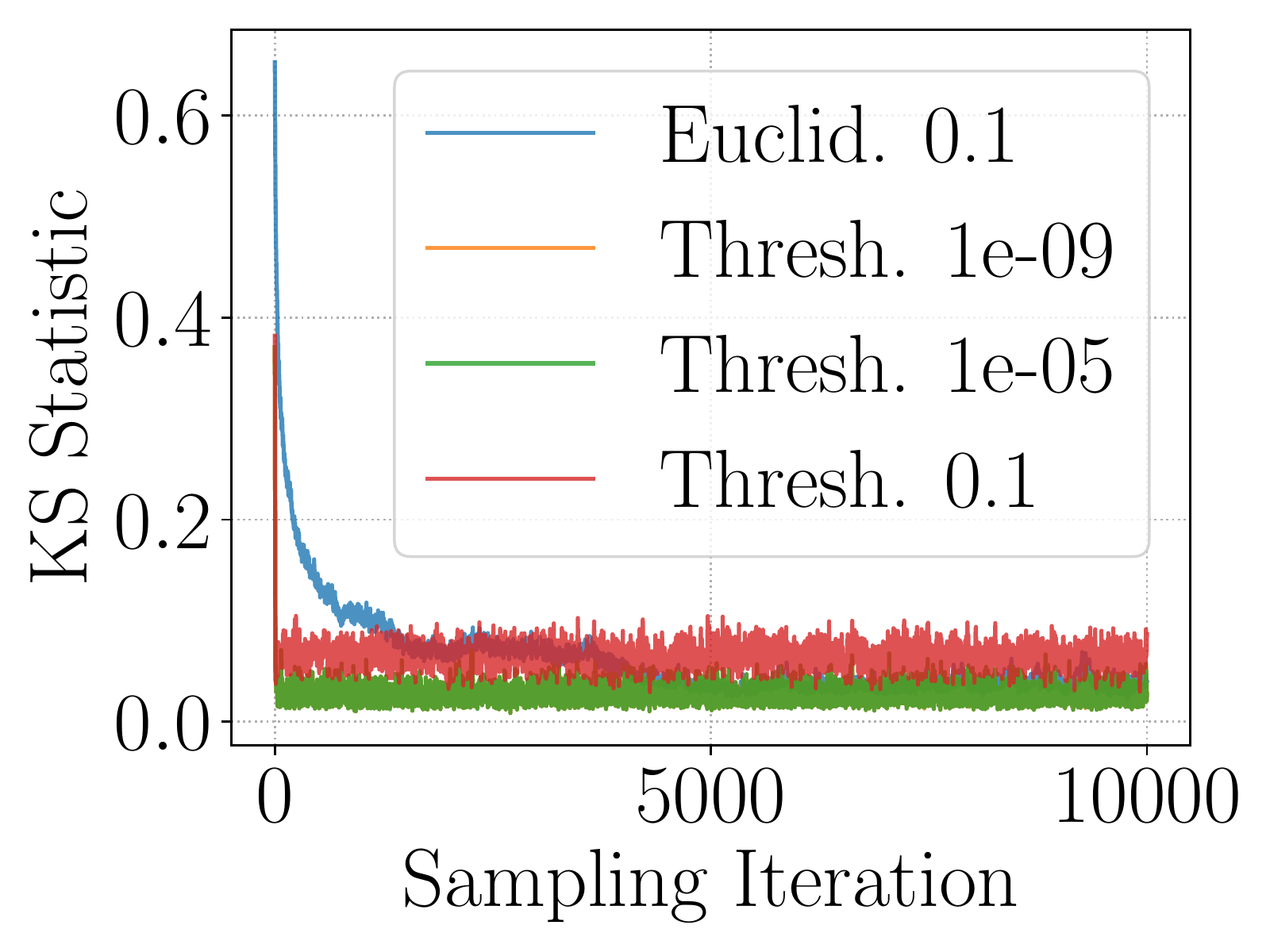}
    \caption{Ergodic Convergence}
  \end{subfigure}
  \caption{Additional measures of ergodicity in Neal's funnel distribution. According to the Kolmogorov-Smirnov and $\mathrm{MMD}_u^2$ statistics, the Riemannian methods with a threshold of $1\times10^{-1}$ are competitive with the best-performing Euclidean methods. However, by the sliced Wasserstein metric, all Riemannian methods have better ergodicity compared to the Euclidean variants. In any event, all metrics agree that ergodicity performance is essentially constant for thresholds less than $1\times 10^{-2}$. We show convergence of independent Markov chains in the marginal distribution of $v$ which is $\mathrm{Normal}(0, 3^2)$. We see that convergence of Riemannian methods is faster than for HMC; however, when employing the largest threshold, one can observe the bias in the stationary distribution. In computing the MMD statistic, we use a kernel bandwidth of 8.4.}
  \label{fig:neal-funnel-extra-ergodicity}
\end{figure}

\begin{figure}[t!]
  \begin{subfigure}[t]{0.32\textwidth}
    \centering
    \includegraphics[width=\textwidth]{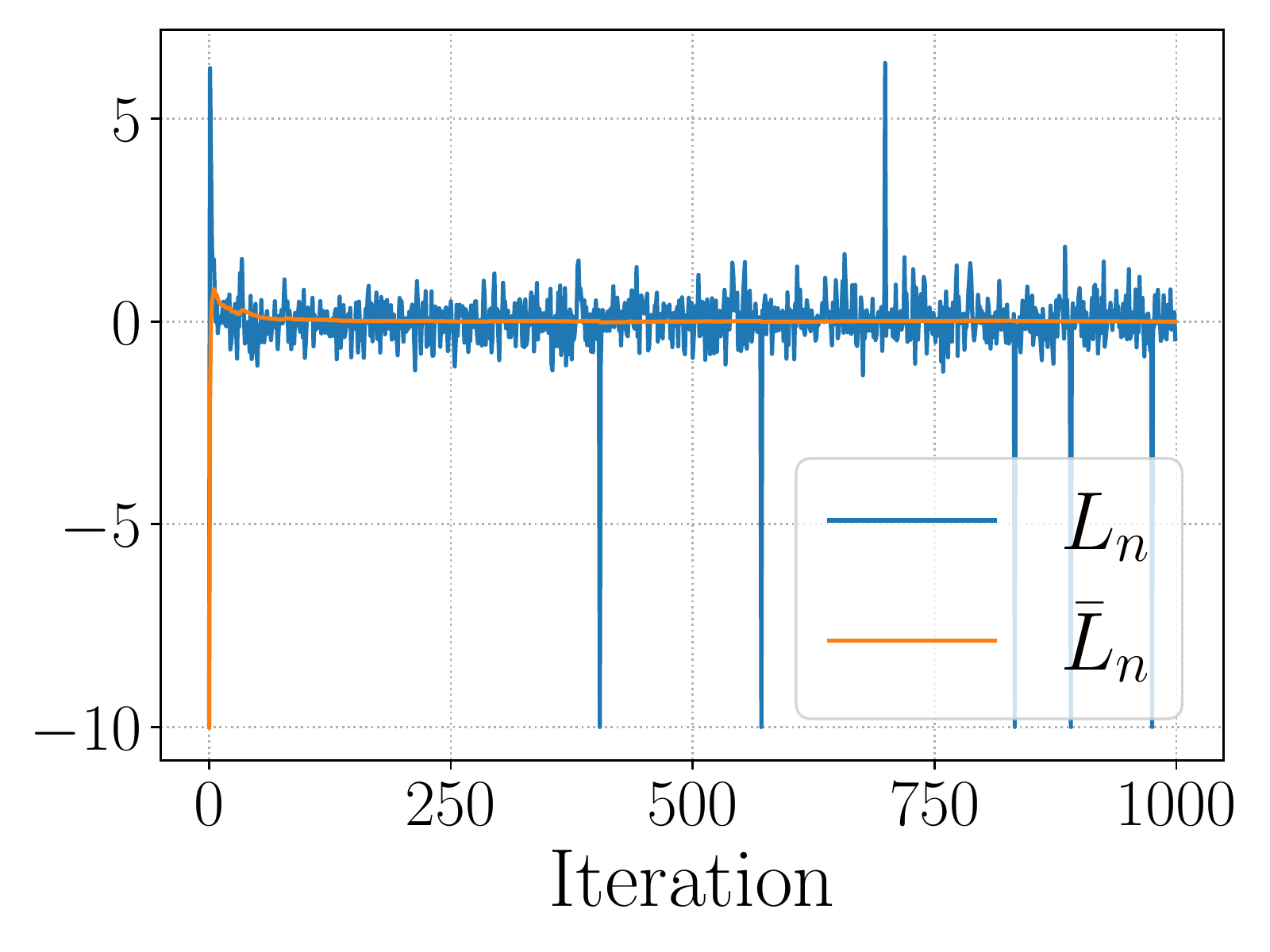}
    \caption{$L_n$ Sequence}
  \end{subfigure}
  ~
  \begin{subfigure}[t]{0.32\textwidth}
    \centering
    \includegraphics[width=\textwidth]{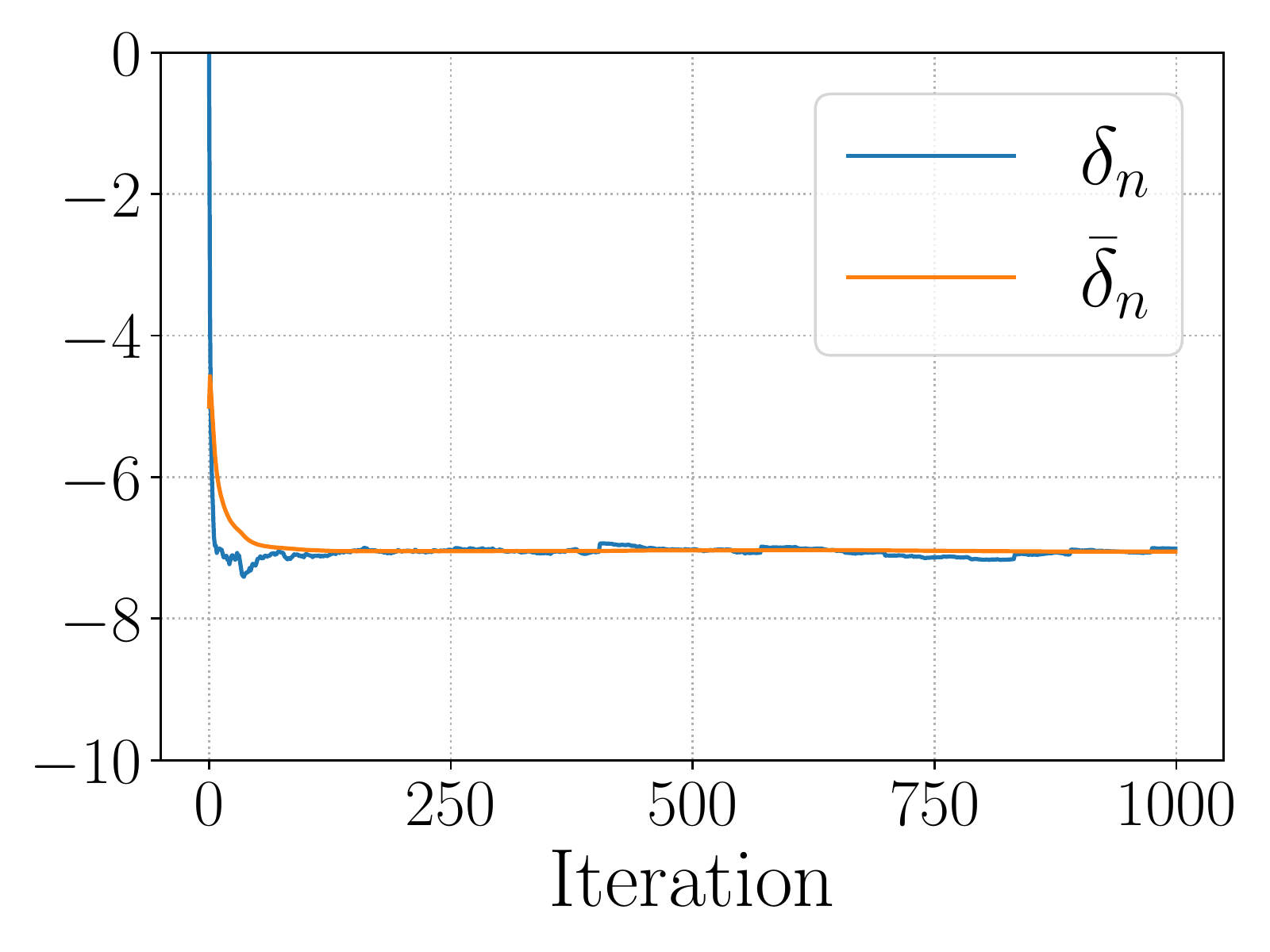}
    \caption{$\delta_n$ Sequence}
  \end{subfigure}
  ~
  \begin{subfigure}[t]{0.32\textwidth}
    \centering
    \includegraphics[width=\textwidth]{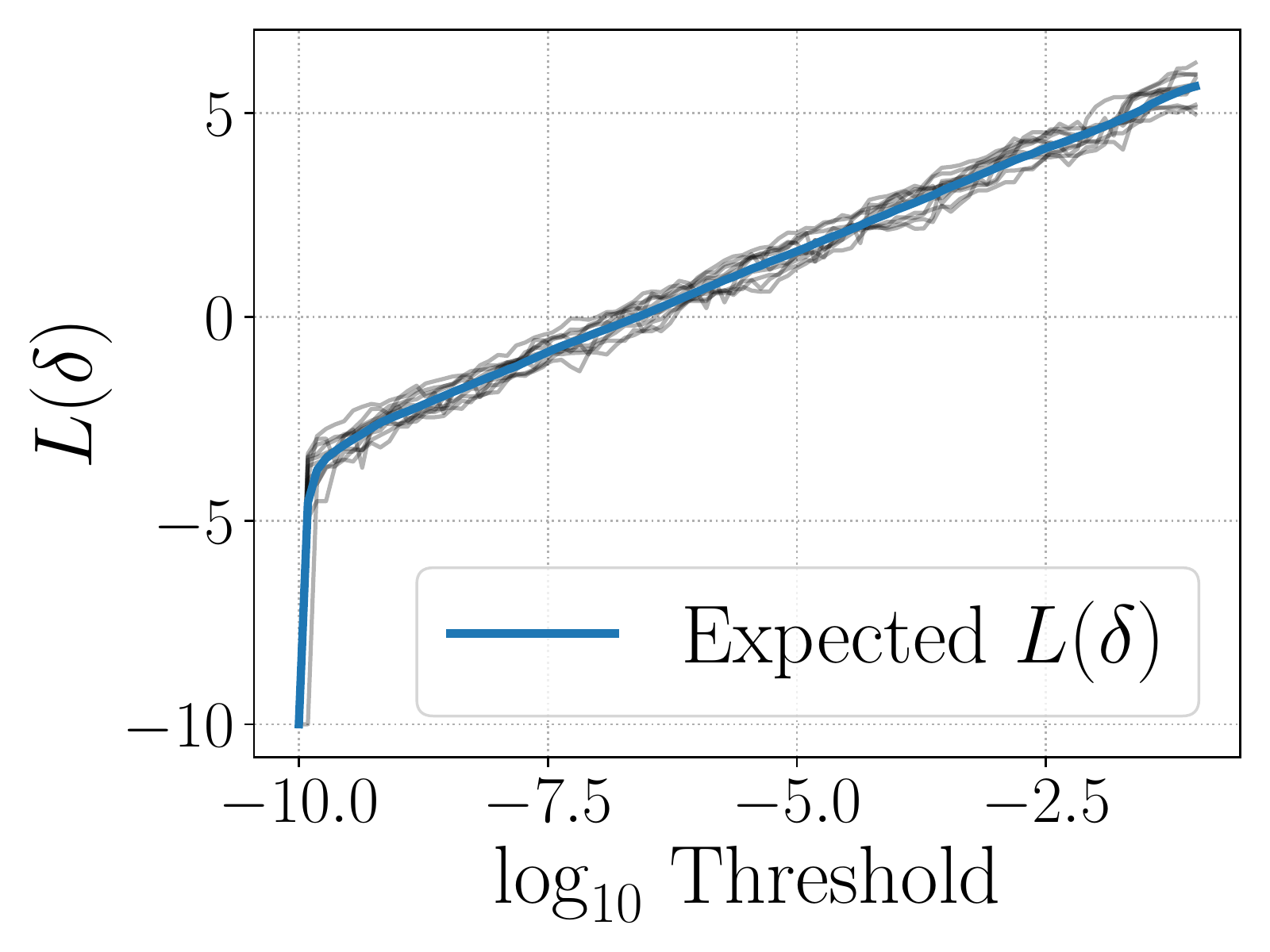}
    \caption{Monte Carlo $B(\delta)$}
  \end{subfigure}
  \caption{The use of Ruppert averaging in Neal's funnel distribution to adaptively set the convergence threshold to achieve six decimal digits of similarity compared to a transition kernel with a threshold of $1\times 10^{-10}$. We show a Monte Carlo approximation to $B(\delta)$, which appears smooth, monotonically increasing. The value of $\delta$ satisfying $B(\delta)=0$ is approximately $\delta = 1\times 10^{-6.5}$, which is somewhat greater than the value of $1\times 10^{-7}$ produced by Ruppert averaging.}
  \label{fig:neal-funnel-dual-averaging}
\end{figure}

Neal's funnel distribution \citep{10.1214/aos/1056562461} is a hierarchical distribution constructed as follows,
\begin{align}
  v &\sim \mathrm{Normal}(0, 9) \\
  x_i \vert v &\sim\mathrm{Normal}(0, e^{-v}) ~~\mathrm{for}~i=1,\ldots, 10.
\end{align}
One sees by inspection that this distribution is trivial to sample analytically. However, the purpose of Neal's funnel distribution is to provide an example of a distribution which HMC struggles to sample. Indeed, for large values of $v$, the conditional distribution $x_i\vert v$ becomes increasingly concentrated near zero, producing the eponymous funnel shape. Without preconditioning, HMC is unable to penetrate this narrow funnel. Neal's funnel distribution is also challenging because it represents a distribution in which no global preconditioning is apparent. Therefore, in applying RMHMC to this task, we follow \citet{softabs} and employ the SoftAbs Riemannian metric. The SoftAbs metric is constructed as follows. Let $\mathbf{H}(x, v)$ denote the Hessian of the joint density of $(x, v)$. Let $\mathbf{H}(x, v) = \mathbf{Q} \Lambda\mathbf{Q}^\top$ be an eigen-decomposition of $\mathbf{H}$, where $\mathbf{Q}\equiv\mathbf{Q}(x, v)$ is the matrix of eigen-vectors and $\Lambda\equiv\Lambda(x, v)$ is a diagonal matrix of eigen-values. The Hessian of Neal's funnel distribution is not positive definite and it is therefore inadmissible as a Riemannian metric. The SoftAbs metric is constructed from the Hessian by smoothly transforming the eigen-values of the Hessian according to,
\begin{align}
  \tilde{\lambda}_{i}(x, v) &= \alpha \mathrm{coth}(\alpha\lambda_i(x, v)) \\
  \tilde{\Lambda}(x, v) &= \mathrm{diag}(\tilde{\lambda}_1(x, v),\ldots,\tilde{\lambda}_{11}(x, v)),
\end{align}
where $\alpha > 0$ is a tunable parameter controlling the smoothness of the SoftAbs transformation. Indeed, the transformation $\lambda\mapsto \alpha\mathrm{coth}(\alpha\lambda)$ is a smooth approximation to the absolute value function. The SoftAbs metric is then $\mathbf{G}(x, v) = \mathbf{Q}(x, v)\tilde{\Lambda}(x, v)\mathbf{Q}(x, v)^\top$. In our experiments we set $\alpha = 10^4$.

In our experiments we consider RMHMC with varying thresholds and with an integration step-size set to $\epsilon = 0.2$ and a maximum number of integration steps equal to twenty-five. We also compare RMHMC against HMC with eight integration steps and an integration step-size $\epsilon\in \set{0.001, 0.01, 0.1, 0.2}$; the parameters of RMHMC and HMC were chosen based off the discussion in \citet{softabs}. When assessing ergodicity in Neal's funnel distribution, results are reported in \cref{subfig:neal-funnel-ergodicity}; one observes that the weakest threshold produces a chain whose similarity to the target distribution is approximately the same as HMC with step-size $0.1$ or $0.2$, yielding around 1.5 digits of similarity in the Kolmogorov-Smirnov statistics along a randomly chosen subspace. When the threshold is decreased to $1\times10^{-2}$, around 2.5 digits of similarity are obtained for a randomly chosen one-dimensional subspace. All of the thresholds smaller than $1\times 10^{-2}$ produce nearly indistinguishable measures of ergodicity as measured by the Kolmogorov-Smirnov statistic along a random subspace. Although employing RMHMC with a threshold of $1\times10^{-3}$ typically exhibits only two digits of reversibility and volume preservation, the transition kernels exhibit a similarity of around 2.5 digits. This level of performance can be achieved with three or four fixed point iterations on average for each implicit update, compared with the eleven or twelve required by a stronger convergence tolerance of $1\times 10^{-9}$, which offers negligible benefits in terms of ergodicity.

We apply the Ruppert averaging procedure in Neal's funnel distribution in order to identify a threshold that produces, on average, six ($\kappa=6$) decimal digits of similarity with a numerical integrator whose convergence threshold is $1\times 10^{-10}$. We show the results of this procedure in \cref{fig:neal-funnel-dual-averaging}. The sequence of $\bar{L}_n$ stabilizes at zero by iteration 100. The sequence of $\bar{\delta}_n$ has converged by iteration 500.

Neal's funnel distribution offers one of the most convincing examples of the benefit of a Riemannian approach to MCMC. We illustrate this phenomenon in \cref{fig:neal-funnel-ess}, which compares HMC with variable step-sizes against RMHMC with variable thresholds. For both the variables $(x_1,\ldots,x_{10})$ and the hierarchical variance $v$, the ESS of the Riemannian methods are orders of magnitude larger than the MCMC procedures without preconditioning.

\subsection{Stochastic Volatility Model}\label{subsec:experiment-stochastic-volatility-model}

\begin{figure}[t!]
  \centering
  \includegraphics[width=\textwidth]{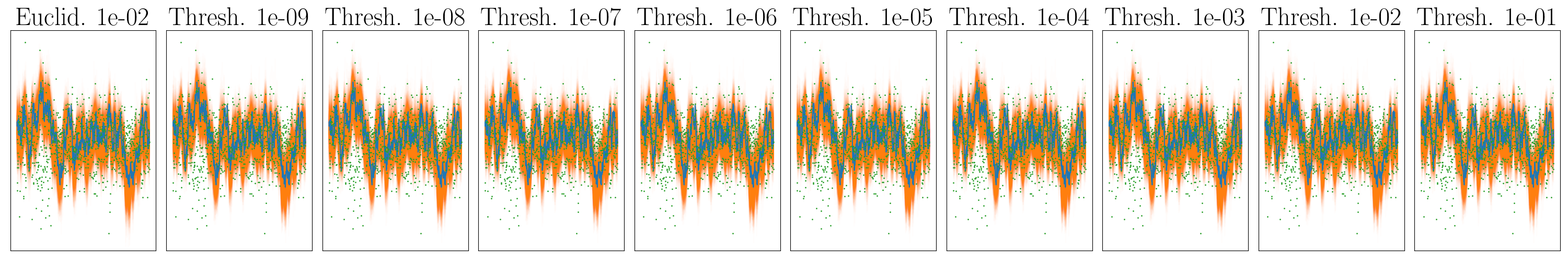}
  \caption{Visualization of the posterior distribution over the stochastic
    volatilities wherein the process hyperparameters $(\sigma^2,\phi,\beta)$ are
    sampled by RMHMC with variable threshold. Visually, the posterior
    distributions are indistinguishable.}
  \label{fig:stochastic-volatility-stochastic-volatility}
\end{figure}

\begin{figure}[t!]
  \centering
  \begin{subfigure}[t]{0.32\textwidth}
    \includegraphics[width=\textwidth]{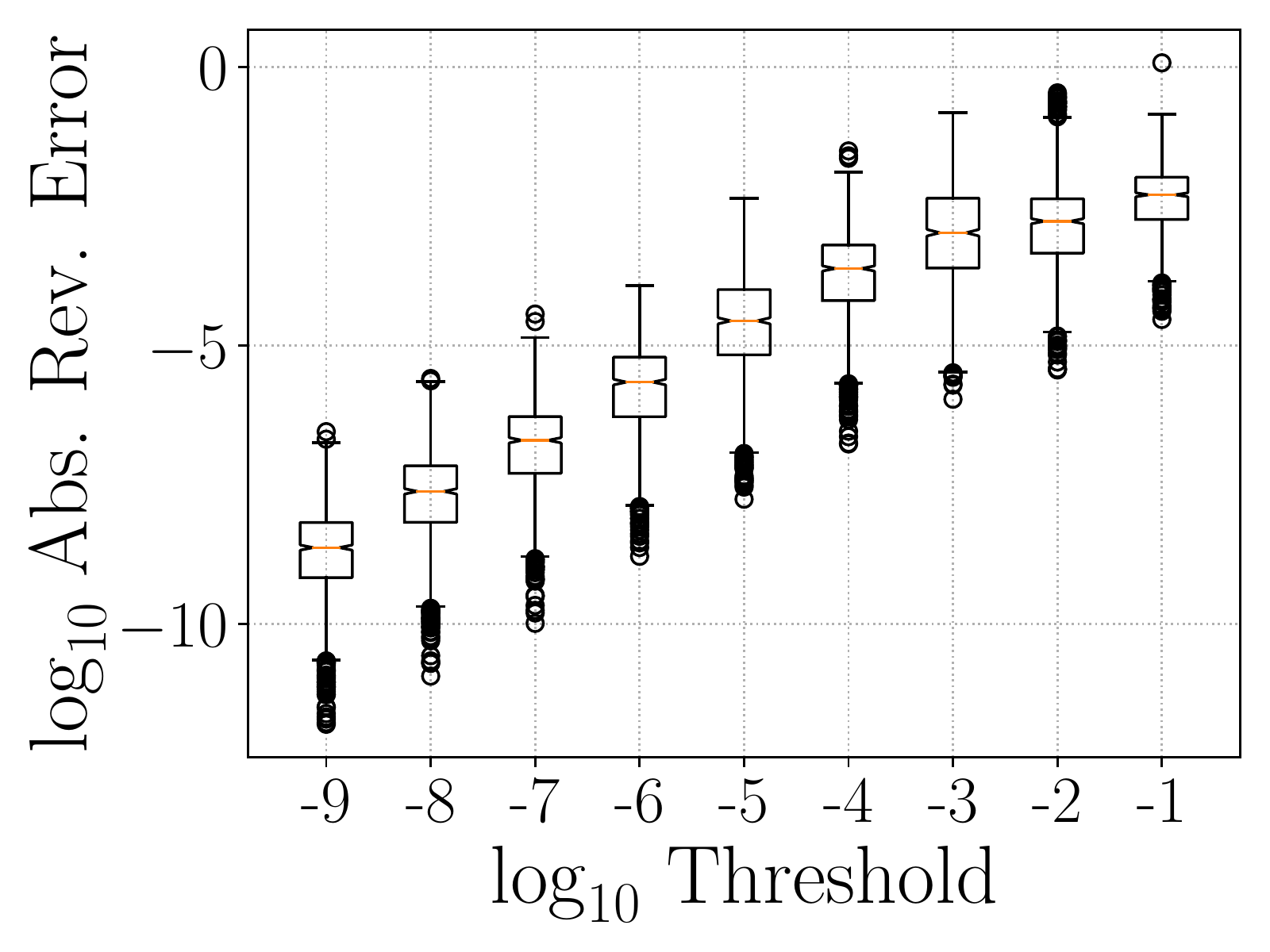}
    \caption{$\log_{10}$ reversibility error}
    \label{subfig:stochastic-volatility-reversibility}
  \end{subfigure}
  ~
  \begin{subfigure}[t]{0.32\textwidth}
    \includegraphics[width=\textwidth]{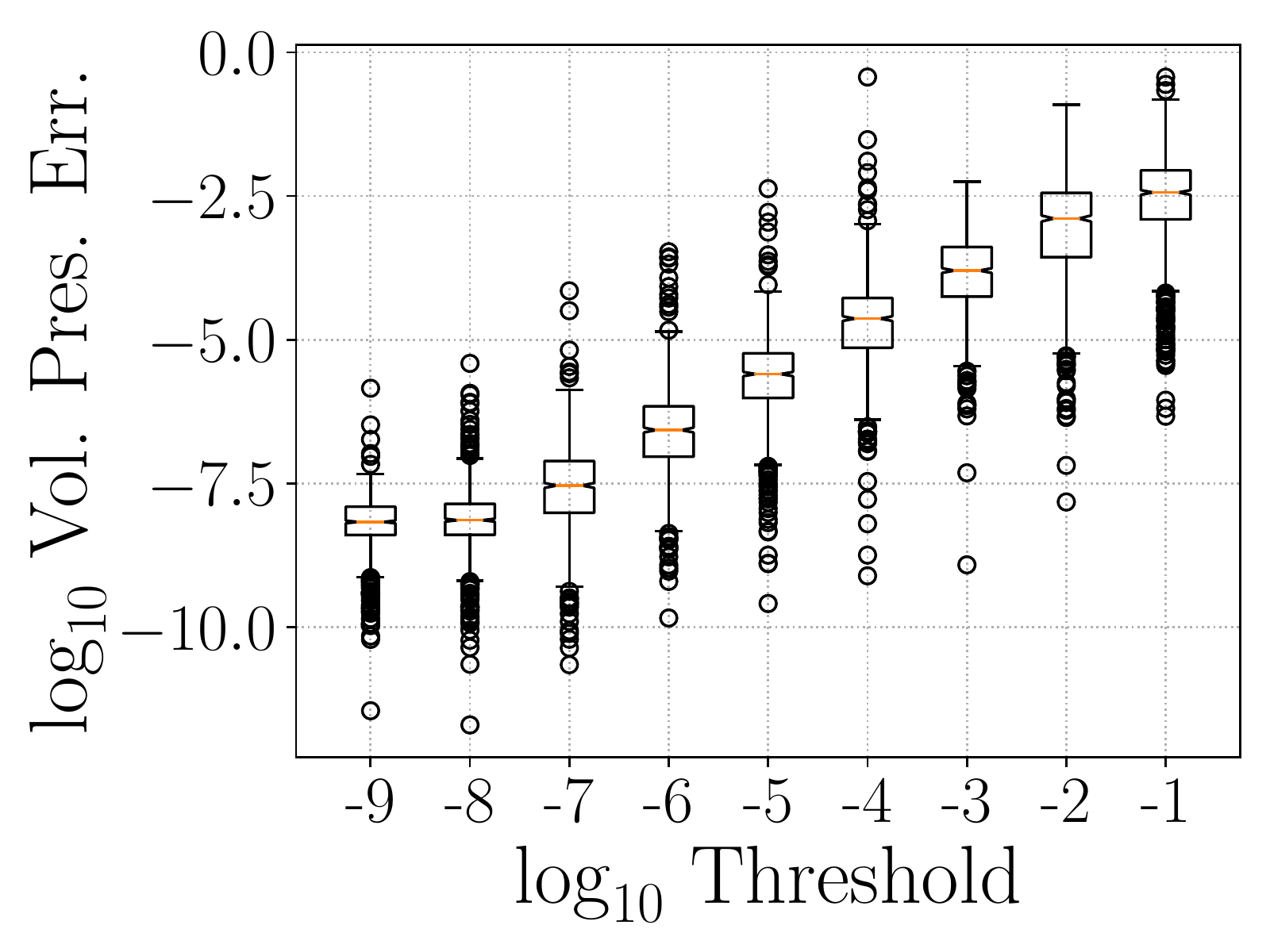}
    \caption{$\log_{10}$ volume-preservation error}
    \label{subfig:stochastic-volatility-jacobian-determinant}
  \end{subfigure}
  ~
  \begin{subfigure}[t]{0.32\textwidth}
    \includegraphics[width=\textwidth]{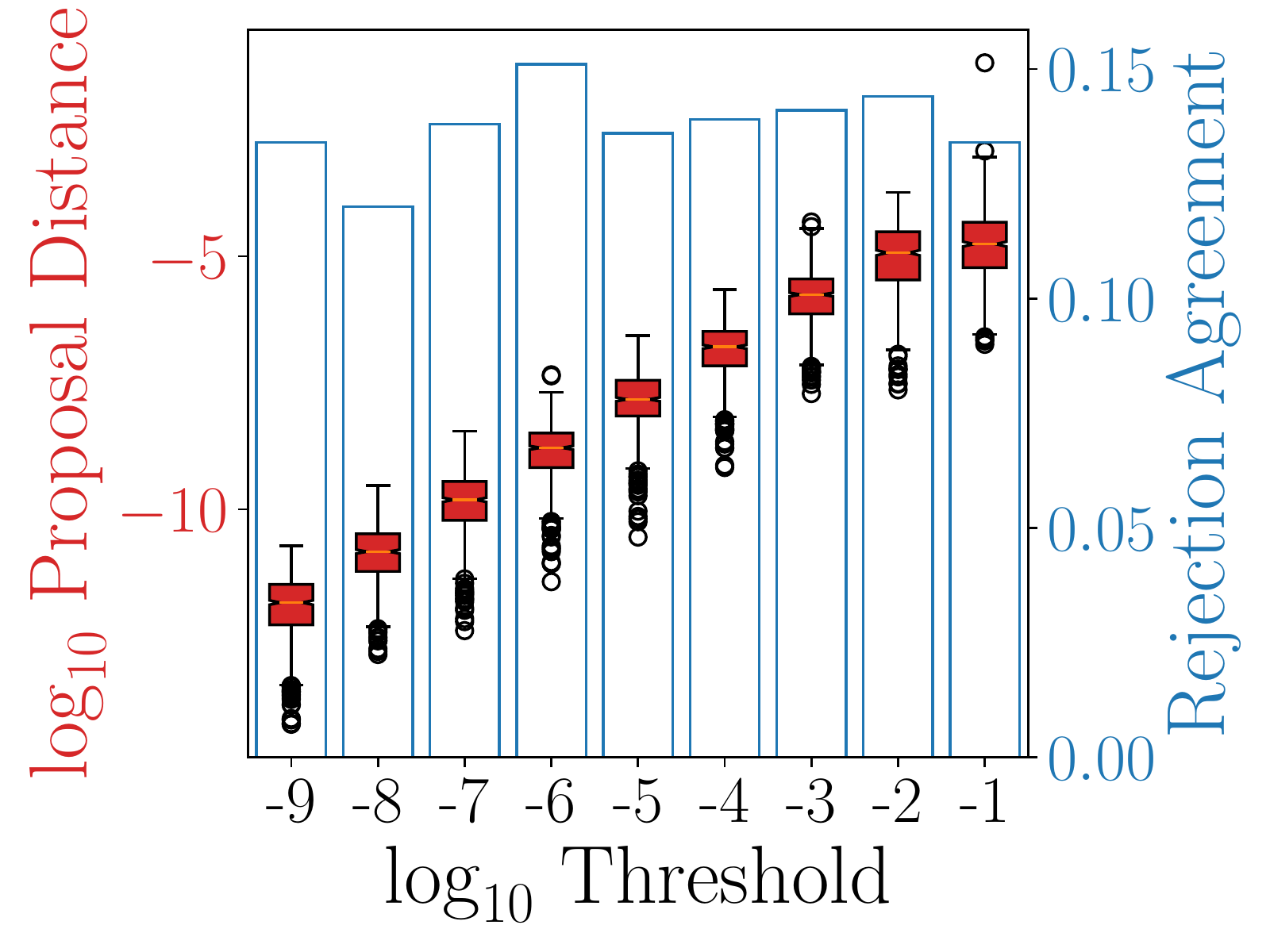}
    \caption{$\log_{10}$ transition kernel difference}
    \label{subfig:stochastic-volatility-transition-difference}
  \end{subfigure}
  \caption{Visualization of the error in reversibility (see \cref{subfig:stochastic-volatility-reversibility}), error in volume-preservation (see \cref{subfig:stochastic-volatility-jacobian-determinant}), and the number of decimal digits of similarity in transition kernels (see \cref{subfig:stochastic-volatility-transition-difference}) for variable thresholds in the stochastic volatility posterior.}
\end{figure}

\begin{figure}[t!]
  \centering
  \begin{subfigure}[t]{0.49\textwidth}
    \includegraphics[width=\textwidth]{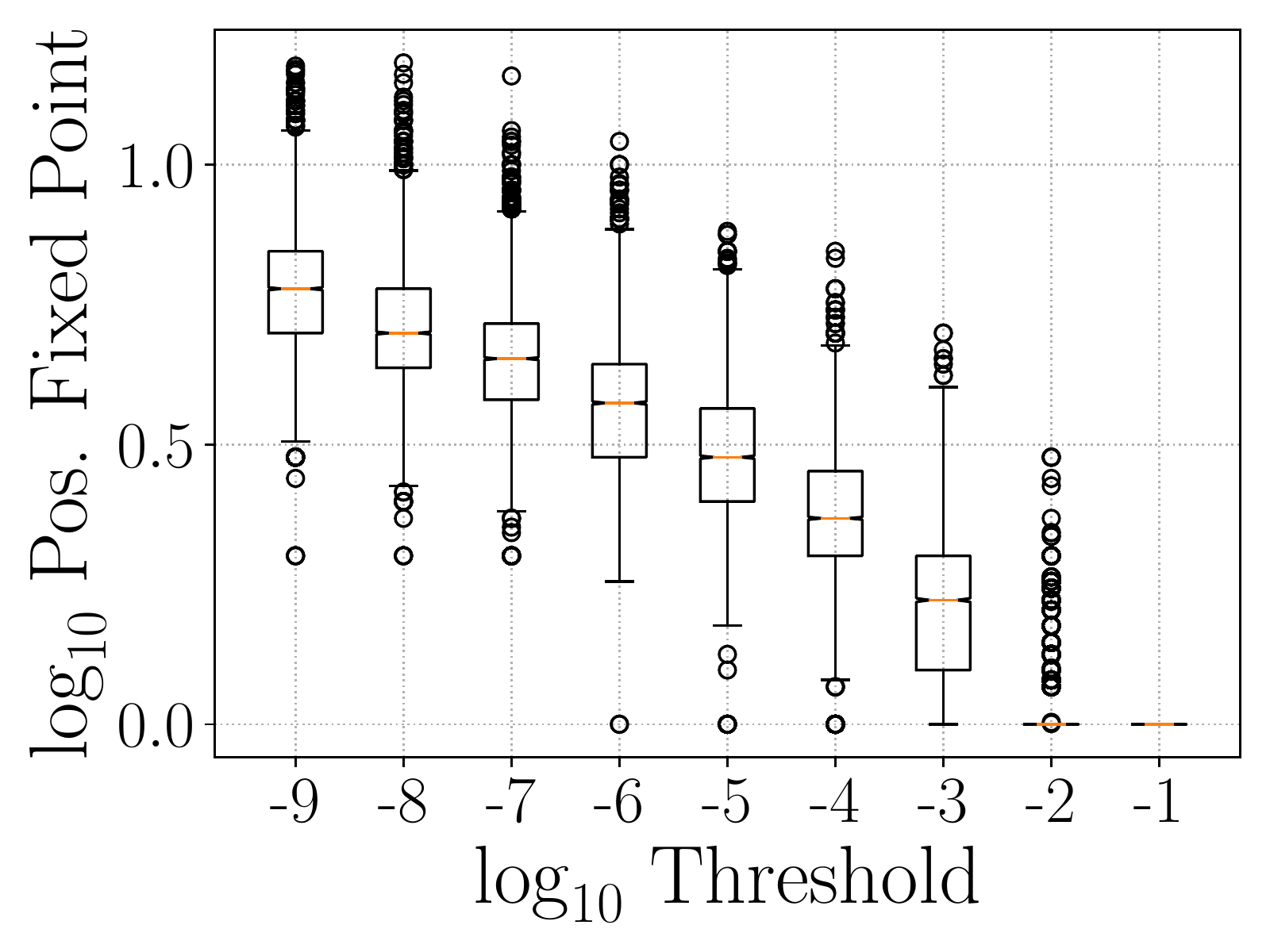}
    \caption{Position fixed point iterations}
    \label{subfig:stochastic-volatility-fixed-point-position}
  \end{subfigure}
  ~
  \begin{subfigure}[t]{0.49\textwidth}
    \includegraphics[width=\textwidth]{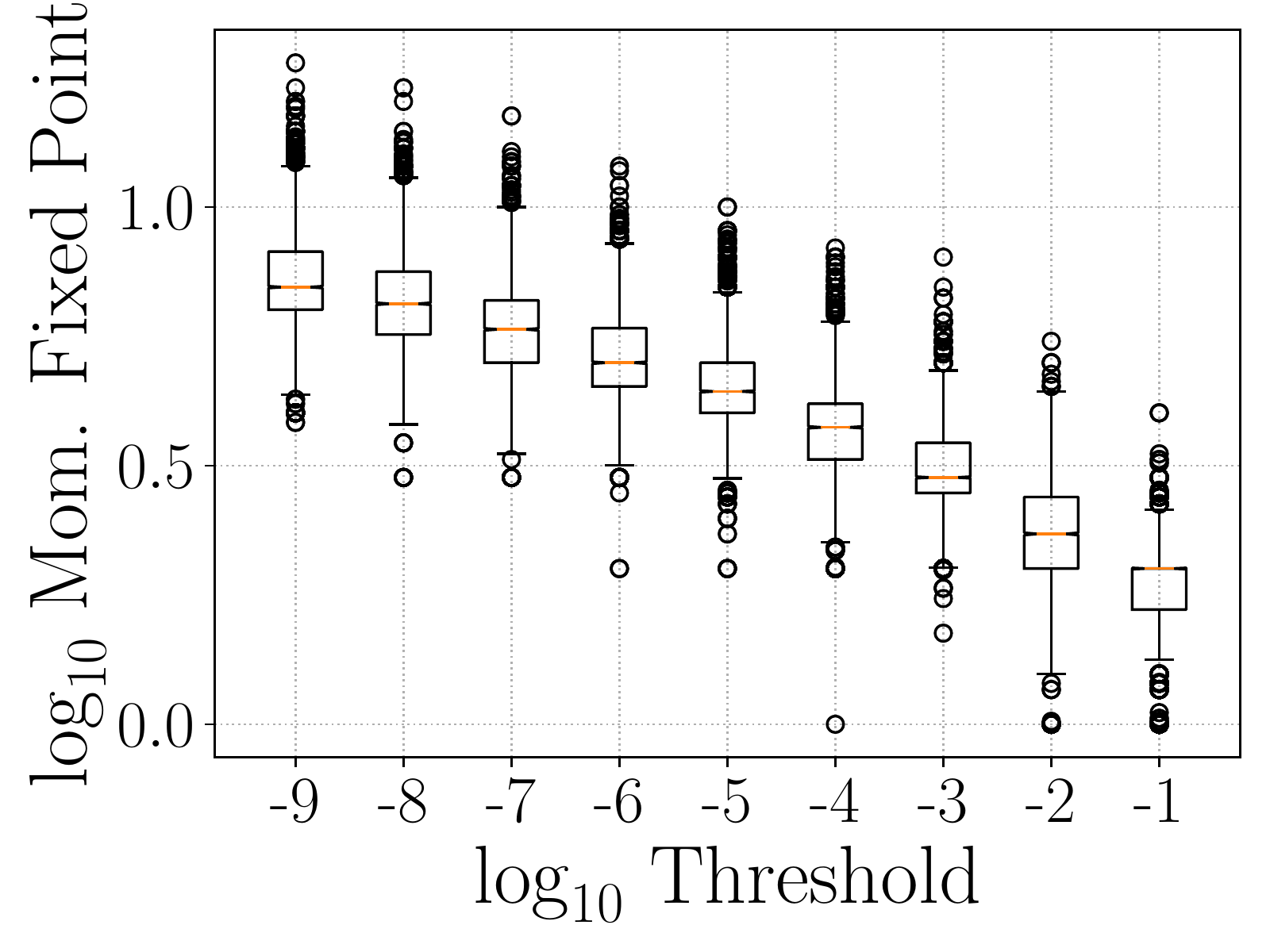}
    \caption{Momentum fixed point iterations}
    \label{subfig:stochastic-volatility-fixed-point-momentum}
  \end{subfigure}

  \caption{Visualization of the number of fixed point iterations required to compute the implicit updates to position and momentum required by the generalized leapfrog integrator in the stochastic volatility model.}
\end{figure}

\begin{figure}[t!]
  \centering
  \begin{subfigure}[t]{0.32\textwidth}
    \includegraphics[width=\textwidth]{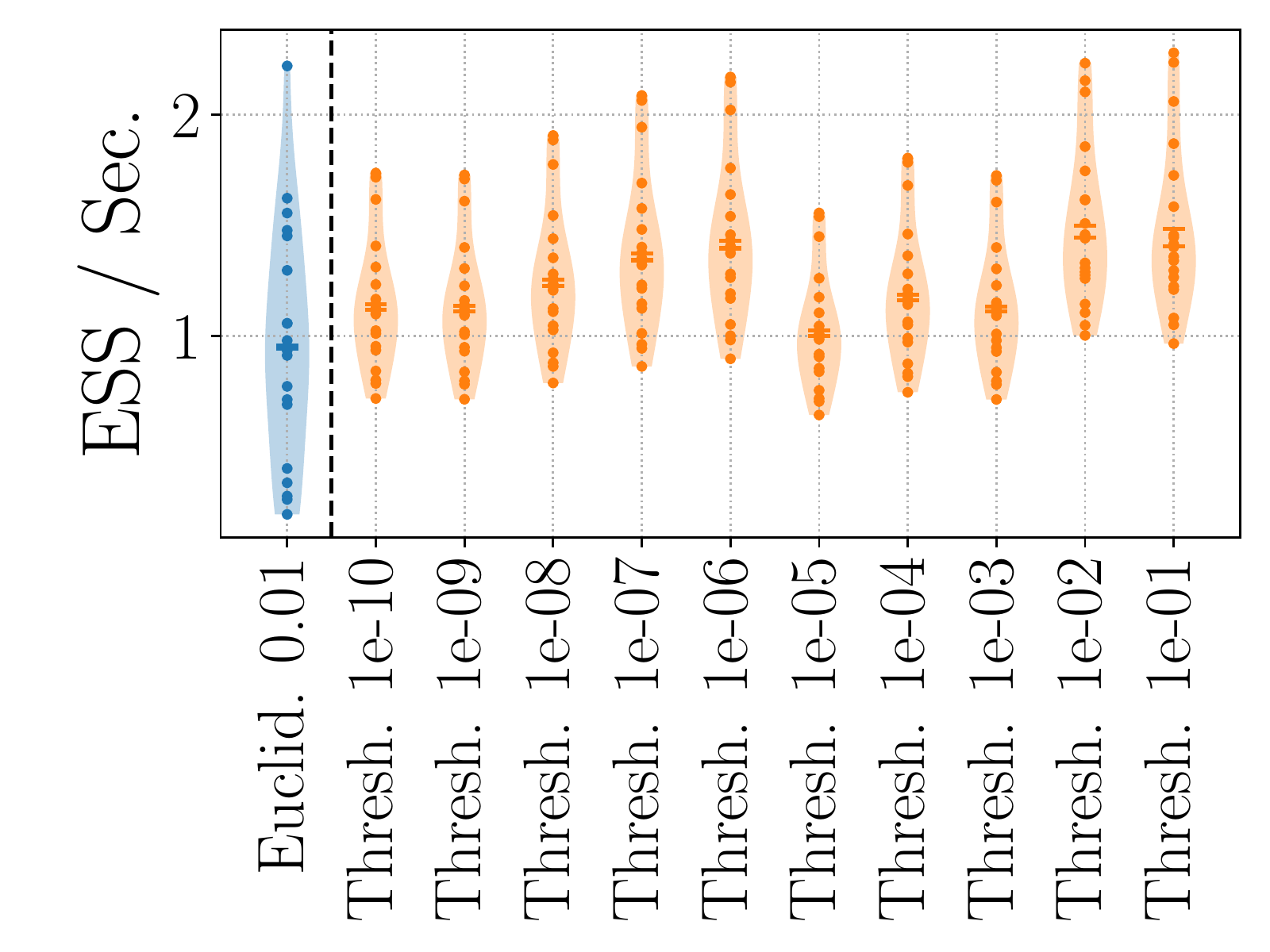}
    \caption{$\phi$}
  \end{subfigure}
  ~
  \begin{subfigure}[t]{0.32\textwidth}
    \includegraphics[width=\textwidth]{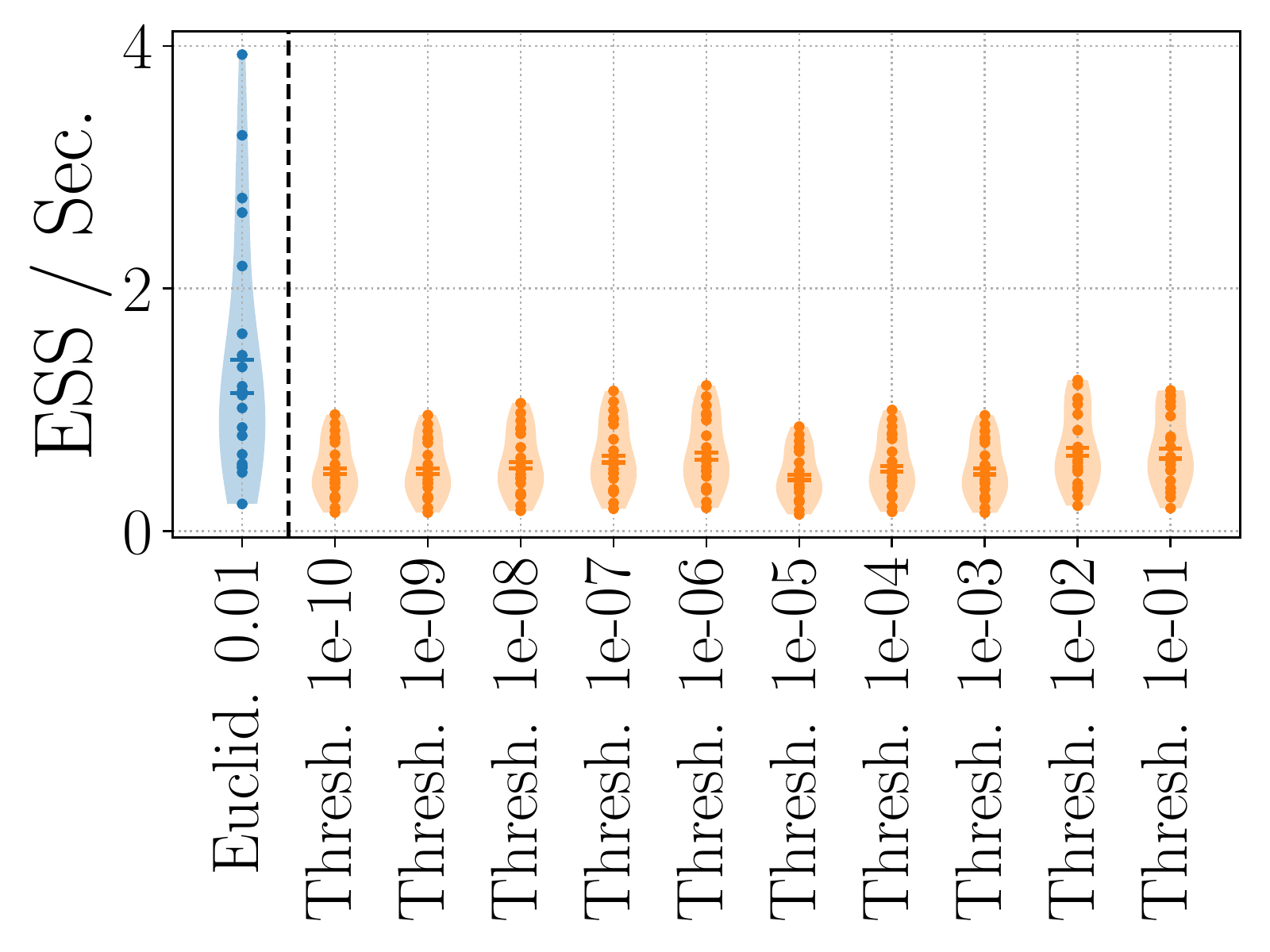}
    \caption{$\beta$}
  \end{subfigure}
  ~
  \begin{subfigure}[t]{0.32\textwidth}
    \includegraphics[width=\textwidth]{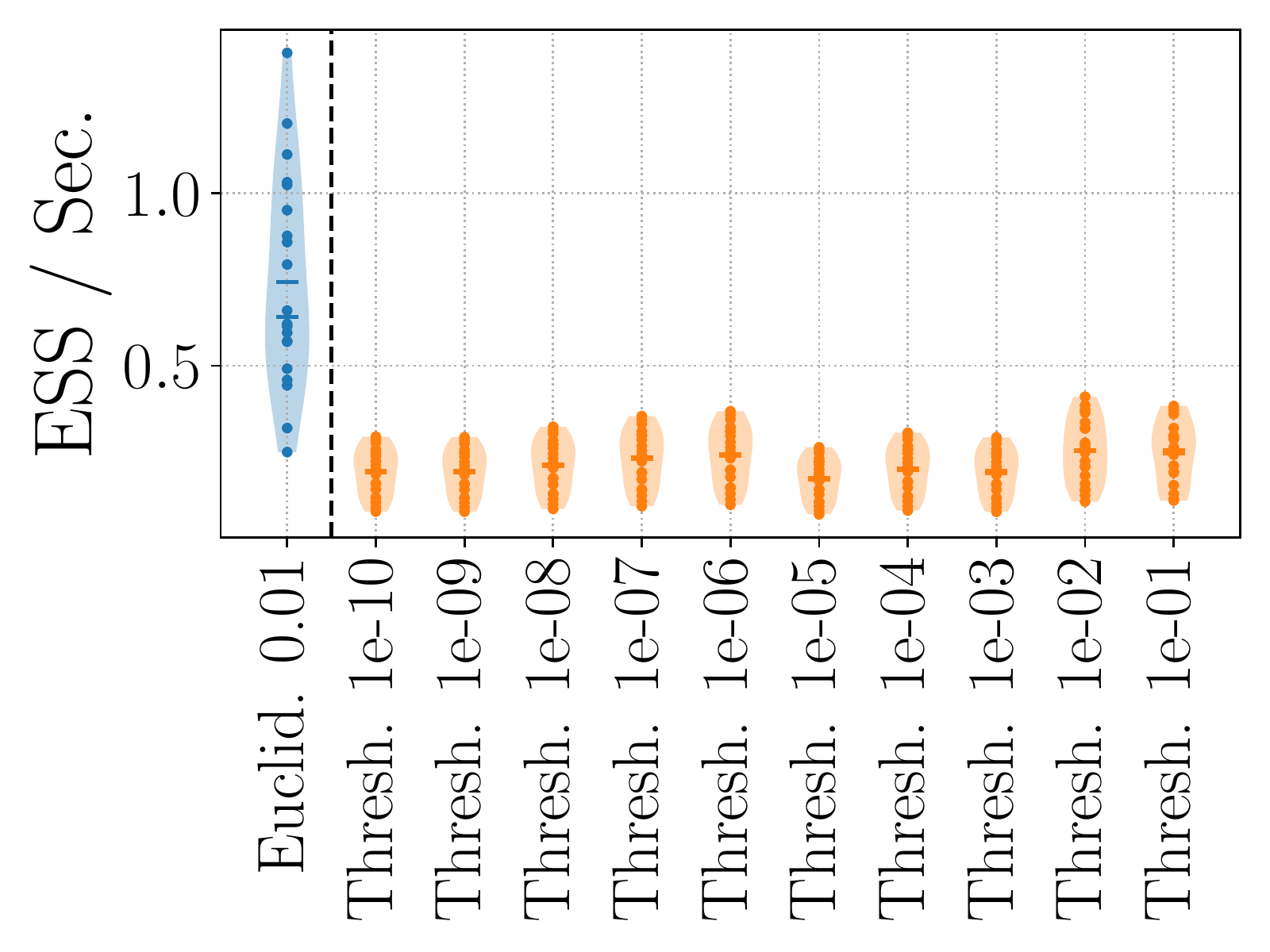}
    \caption{$\sigma^2$}
  \end{subfigure}
  \caption{Visualization of the ESS per second for the three latent parameters of the stochastic volatility model.}
  \label{fig:stochastic-volatility-ess-per-second}
\end{figure}

Stochastic volatility models are random processes which are characterized by the randomness inherent in their variance (or ``volatility''). We consider the following generative model,
\begin{align}
  x_1 \vert \phi, \sigma^2 &\sim \mathrm{Normal}(0, \sigma^2 / (1-\phi^2)) \\
  x_{t+1} \vert x_{t},\phi,\sigma^2 &\sim \mathrm{Normal}(\phi x_t, \sigma^2) ~~\mathrm{for}~t=1,\ldots, T-1 \\
  y_t \vert \beta, x_t &\sim \mathrm{Normal}(0, \beta^2 \exp(x_t)) ~~\mathrm{for}~t=1,\ldots, T
\end{align}
with priors $\sigma^2 \sim\mathrm{Inv-}\chi^2(10, 1/20)$, $(\phi+1)/2\sim\mathrm{Beta}(20, 3/2)$ (so that $\phi\in (-1, +1)$), and the prior density over $\beta$ is proportional to $1/\beta$. The Bayesian inference task is to infer the posterior distribution of $(x_1,\ldots, x_T, \sigma^2, \phi, \beta)$ given observations $(y_1,\ldots,y_T)$. In our experiments we set $T=1,000$, producing a posterior of dimensionality $1,003$. Following \citet{rmhmc}, we employ a Metropolis-within-Gibbs-like strategy by alternating between sampling the distribution $(x_1,\ldots,x_T)\vert \sigma^2,\phi,\beta, (y_1,\ldots, y_T)$ and sampling the distribution $(\phi, \beta,\sigma^2)\vert (x_1,\ldots, x_T), (y_1,\ldots, y_T)$. In the former case, the Fisher information metric is constant with respect to $(x_1,\ldots,x_T)$, thereby allowing us to use the usual leapfrog integrator to produce samples; moreover, the metric has a special tri-diagonal structure, fascilitating the use of specialized numerical linear algebra routines. However, the Fisher information metric of the distribution $(\phi, \beta,\sigma^2)\vert (x_1,\ldots, x_T), (y_1,\ldots, y_T)$ depends on $(\phi, \beta, \sigma^2)$, thereby producing a non-separable Hamiltonian and necessitating the use of the generalized leapfrog integrator. In order to respect the constraints $\sigma^2 > 0$ and $\phi \in (-1, +1)$, we define auxiliary variables $(\gamma, \alpha)$ and employ the smooth, invertible transformations $\sigma^2 = \exp(\gamma)$ and $\phi=\mathrm{tanh}(\alpha)$. The Fisher information metric for the transformed variables is,
\begin{align}
  \mathbf{G}(\gamma,\alpha,\beta) = \begin{pmatrix}
    2T & 2\phi & 0 \\
    2\phi & 2\phi^2 + (T-1)(1-\phi^2) & 0 \\
    0 & 0 & 2T / \beta^2
  \end{pmatrix}.
\end{align}
In sampling the posterior of $(\sigma^2,\phi)$, we use six integration steps with a step-size of 0.5. We seek to draw 100,000 observations from the posterior.

We visualize the posterior over the stochastic volatilities for variable thresholds in \cref{fig:stochastic-volatility-stochastic-volatility}. Visually, the posterior distributions are indistinguishable; this conclusion is reinforced by the close similarity of the posterior marginals of over $(\phi,\beta,\sigma^2)$, wherein only the largest threshold $1\times 10^{-1}$ shows any dissimilarity, which is nonetheless minor. This similarity is quantified in our assessment of the similarity of the Markov chain transition kernel, wherein we see that the threshold $1\times 10^{-1}$ enjoys nearly five decimal digits of similarity relative to the transition kernel with threshold $1\times 10^{-10}$. In \cref{subfig:stochastic-volatility-fixed-point-position,subfig:stochastic-volatility-fixed-point-momentum} we visualize the number of fixed point iterations required by the generalized leapfrog method by convergence tolerance. We observe that for a threshold of $1\times 10^{-1}$, only one or two fixed point iterations are required to resolve the implicit updates of the position and momentum variables, respectively, which compares favorably to the six or seven required by the generalized leapfrog method with a convergence tolerance of $1\times 10^{-9}$. In \cref{fig:stochastic-volatility-ess-per-second} we evaluate the effective sample size per second of the Riemannian and non-Riemannian HMC variants. On in the case of the variable $\phi$ does RMHMC offer any benefits, while sampling efficiency is degraded by using RMHMC on the variables $\beta$ and $\sigma^2$. This occurs due to the relative computational burden of RMHMC, which cannot always be compensated for, in terms of time-normalized metrics, by the geometric advantages.

Is the similarity of the posterior over $(\phi,\beta,\sigma^2)$ a result of a step-size that is sufficiently small so as to enable near-perfect simulation of Hamilton's equations of motion? One piece of evidence against this hypothesis is that the acceptance rate of the Markov chain is approximately eighty-seven percent; thereby showing the the numerical trajectory does not conserve the Hamiltonian, as near-perfect simulation of the underlying equations of motion must.

\subsection{Log-Gaussian Cox-Poisson Model}

\begin{figure}[t!]
  \centering
  \includegraphics[width=\textwidth]{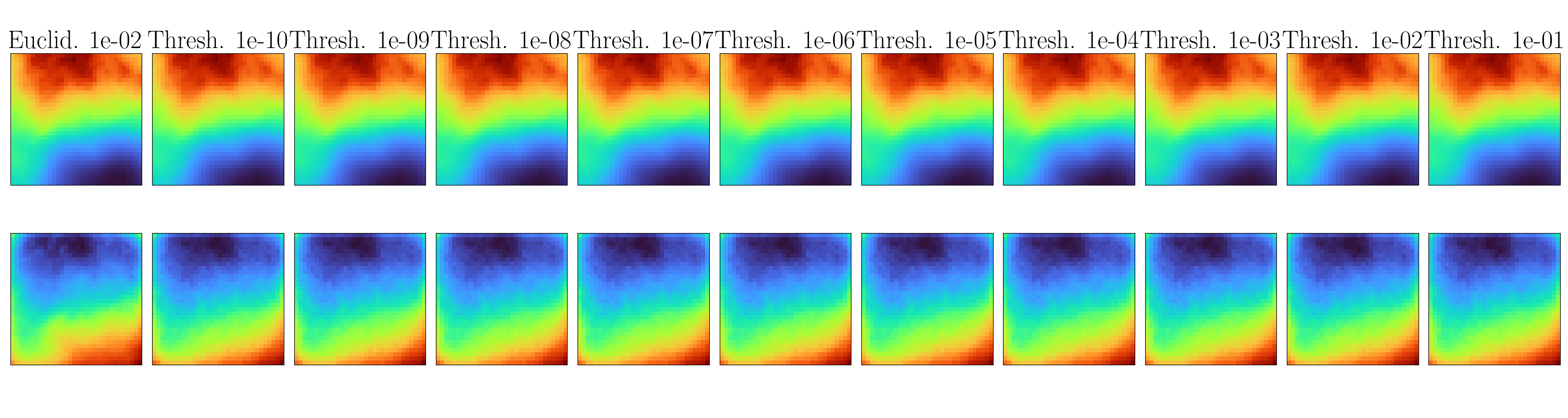}
  \caption{Visualization of the posterior mean (in the first row) and posterior standard deviation (in the second row) of the spatial latent process wherein the kernel parameters $(\sigma^2,\beta)$ are sampled by RMHMC with variable thresholds. Within the Riemannian algorithms with varying thresholds, these quantities are visually indistinguishable. However, one does observe differences in the posterior standard deviation when comparing RMHMC to Euclidean HMC.}
  \label{fig:cox-poisson-cox-poisson}
\end{figure}

\begin{figure}[t!]
  \centering
  \begin{subfigure}[t]{0.49\textwidth}
    \includegraphics[width=\textwidth]{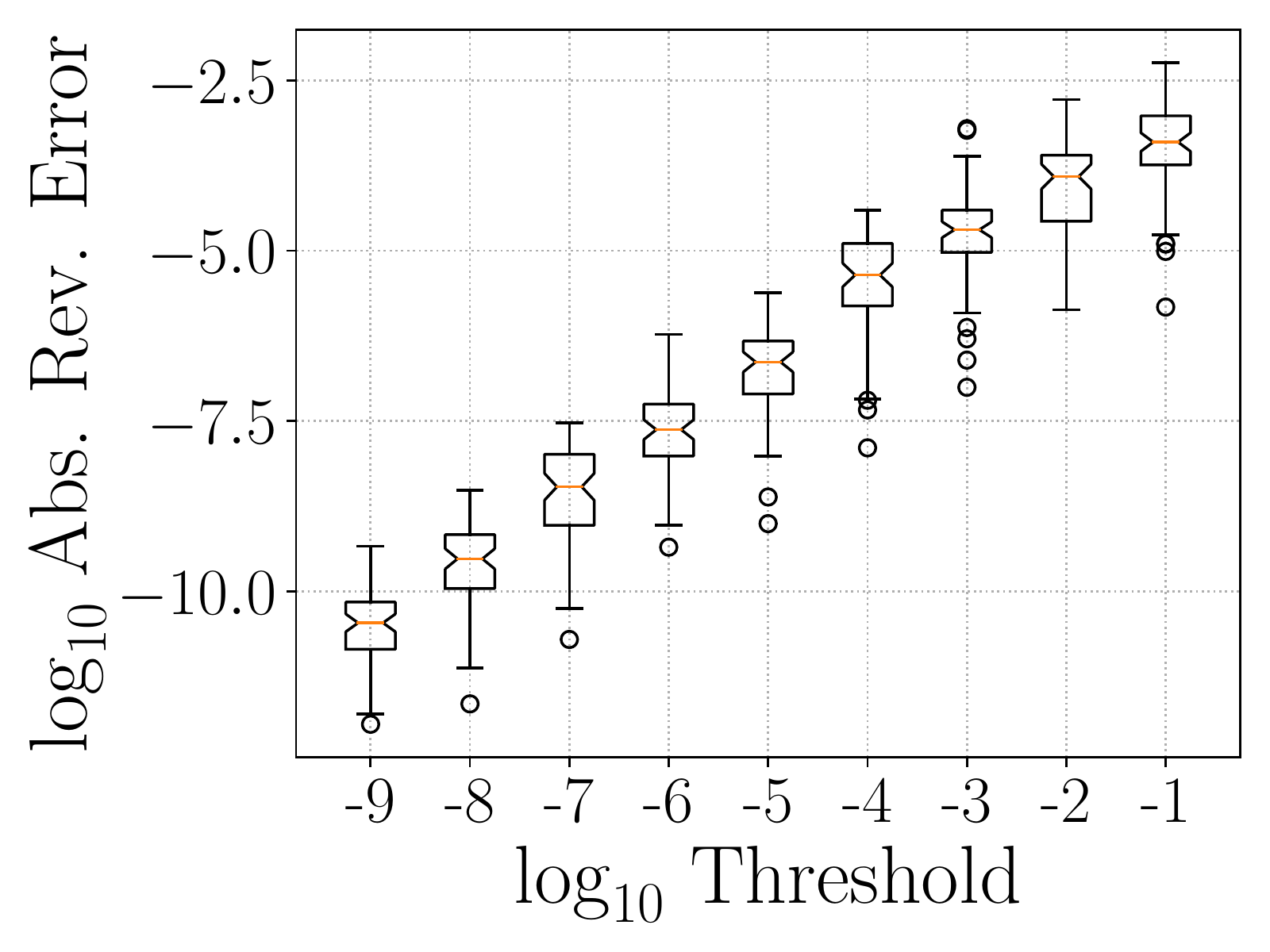}
    \caption{Error in Reversibility}
    \label{subfig:cox-poisson-reversibility}
  \end{subfigure}
  ~
  \begin{subfigure}[t]{0.49\textwidth}
    \includegraphics[width=\textwidth]{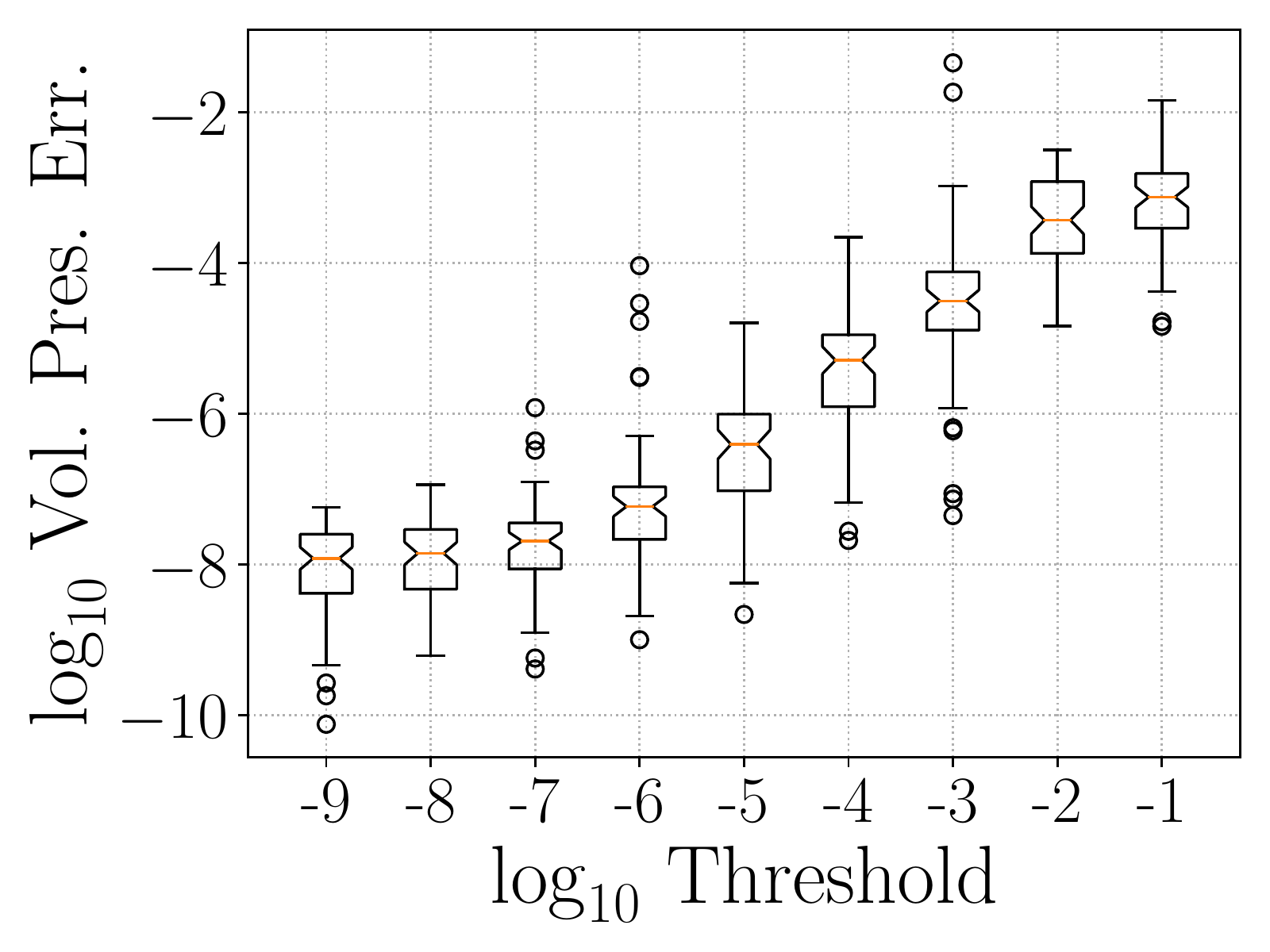}
    \caption{Error in Volume-Preservation}
    \label{subfig:cox-poisson-jacobian-determinant}
  \end{subfigure}
  
  \begin{subfigure}[t]{0.49\textwidth}
    \includegraphics[width=\textwidth]{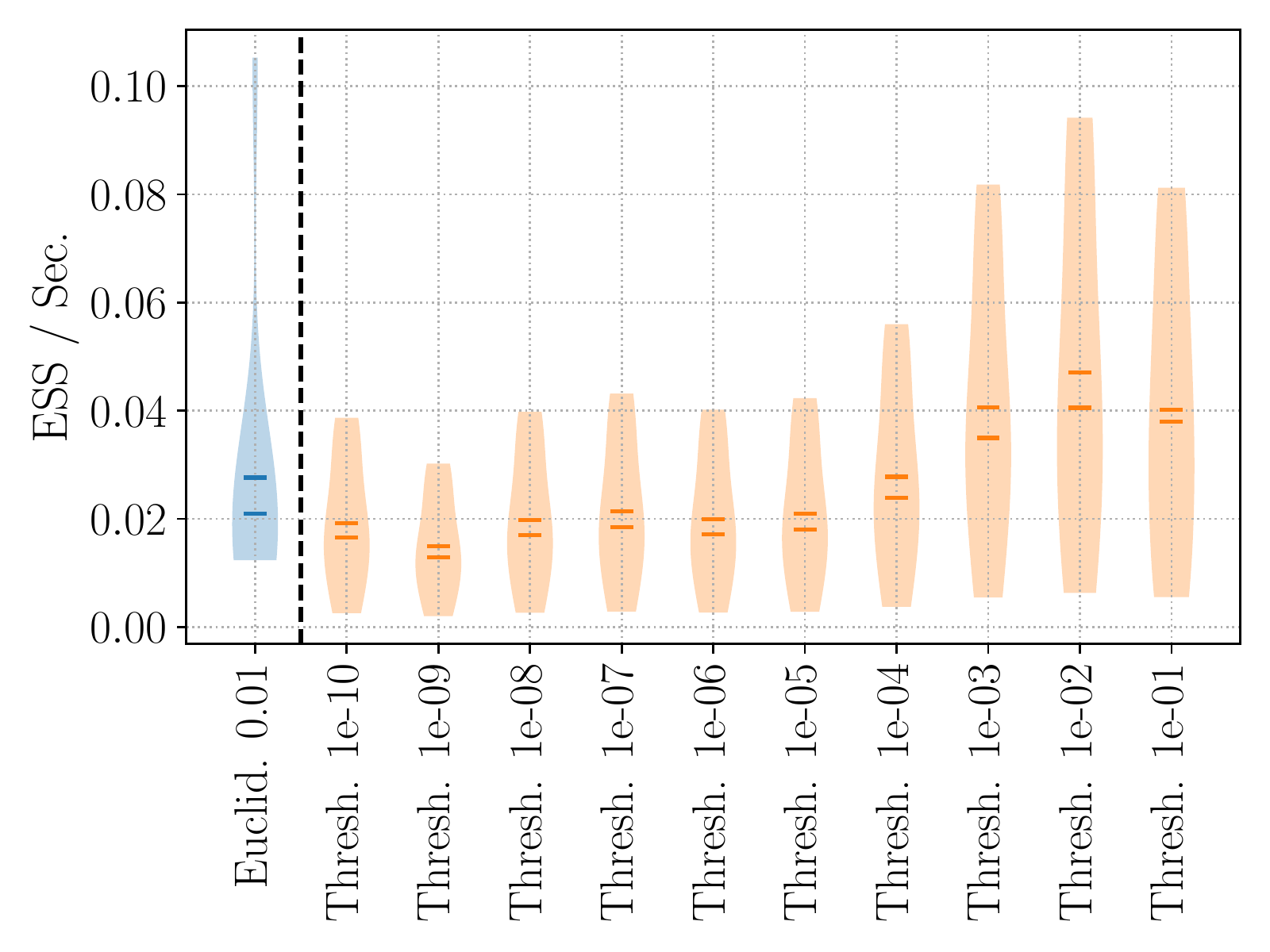}
    \caption{$\beta$}
    \label{subfig:cox-poisson-beta-ess-per-second}
  \end{subfigure}
  ~
  \begin{subfigure}[t]{0.49\textwidth}
    \includegraphics[width=\textwidth]{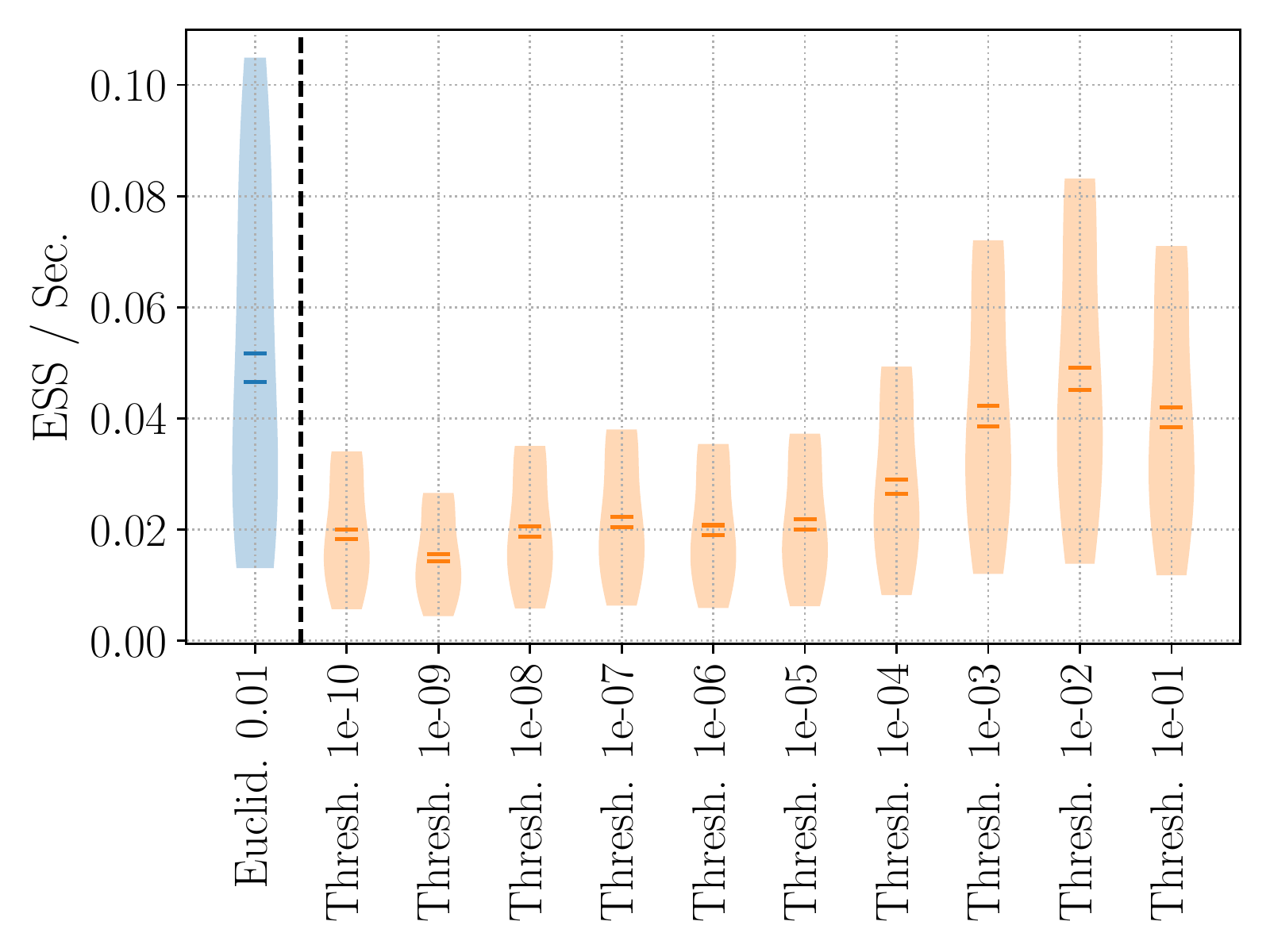}
    \caption{$\sigma$}
    \label{subfig:cox-poisson-sigma-ess-per-second}
  \end{subfigure}

  \caption{Visualization of the error in reversibility (see
    \cref{subfig:cox-poisson-reversibility}) and error in
    volume-preservation (see
    \cref{subfig:cox-poisson-jacobian-determinant}) in the log-Gaussian Cox-Poisson posterior. We also illustrate the ESS per second for sampling the hyperparameters of the Cox-Poisson model and a comparison against a Euclidean HMC method.}
\end{figure}

\begin{figure}[t!]
  \centering
  \begin{subfigure}[t]{0.32\textwidth}
    \includegraphics[width=\textwidth]{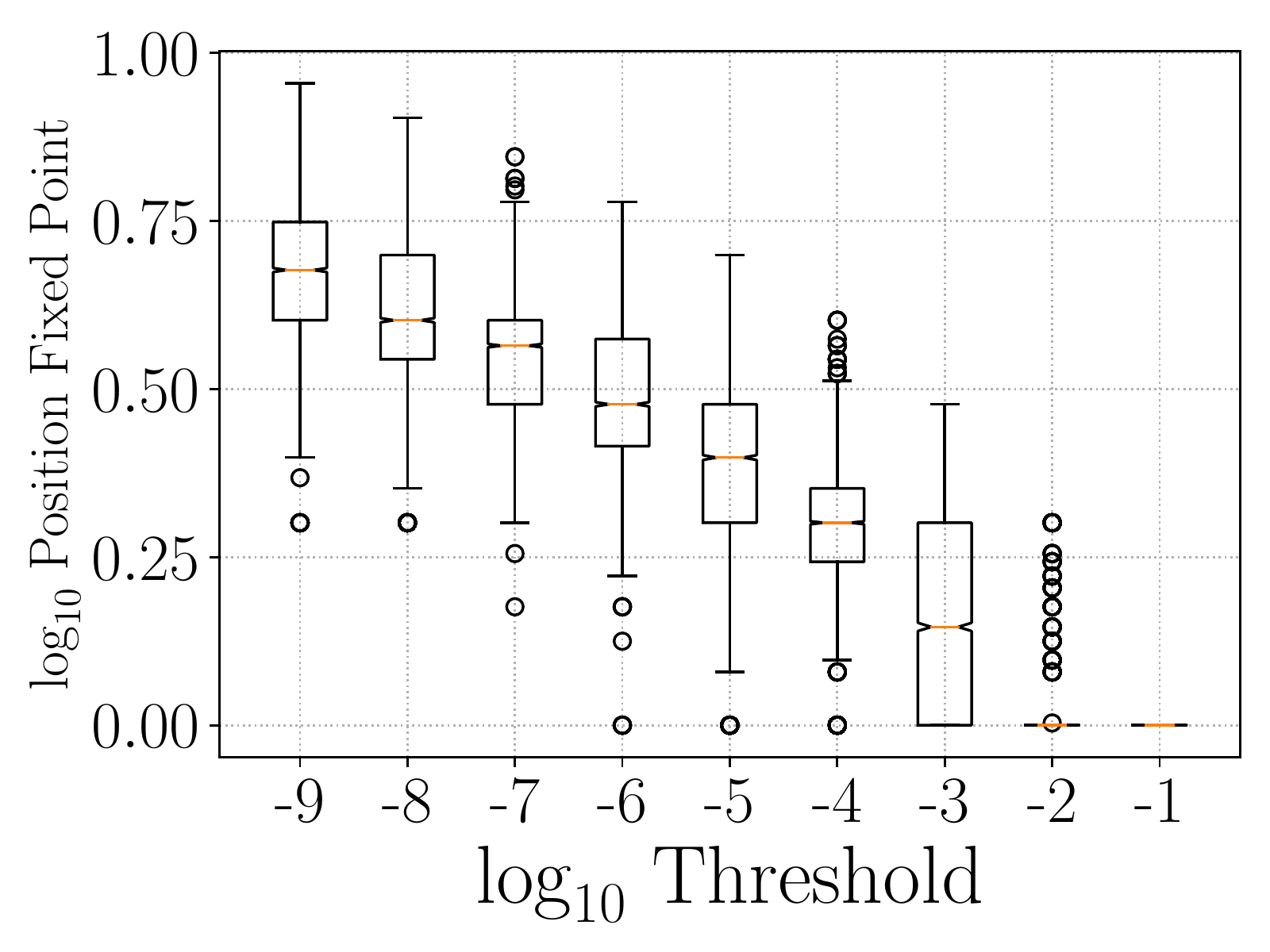}
    \caption{Position fixed point iterations}
    \label{subfig:cox-poisson-fixed-point-position}
  \end{subfigure}
  ~
  \begin{subfigure}[t]{0.32\textwidth}
    \includegraphics[width=\textwidth]{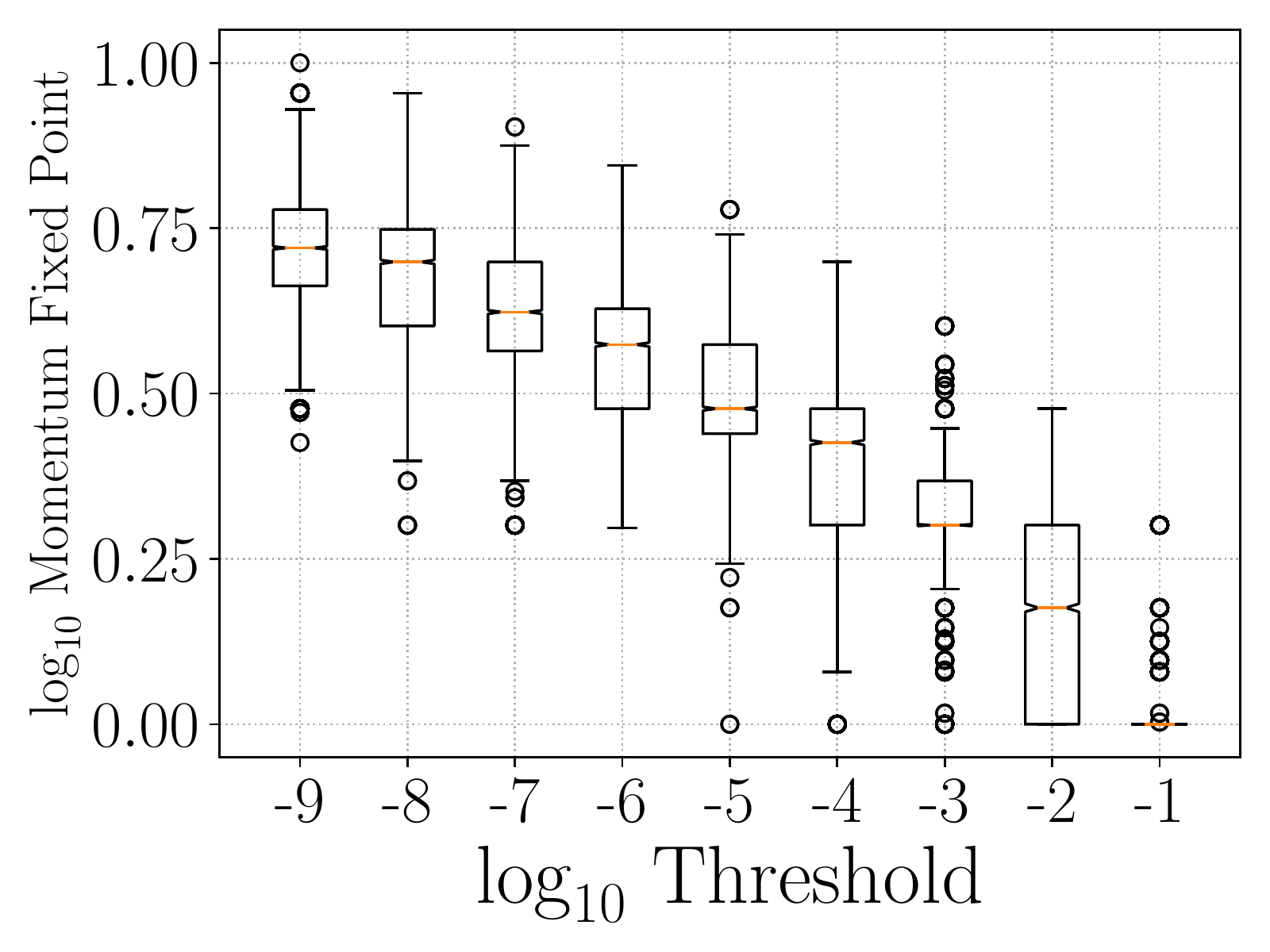}
    \caption{Momentum fixed point iterations}
    \label{subfig:cox-poisson-fixed-point-momentum}
  \end{subfigure}
  ~
  \begin{subfigure}[t]{0.32\textwidth}
    \includegraphics[width=\textwidth]{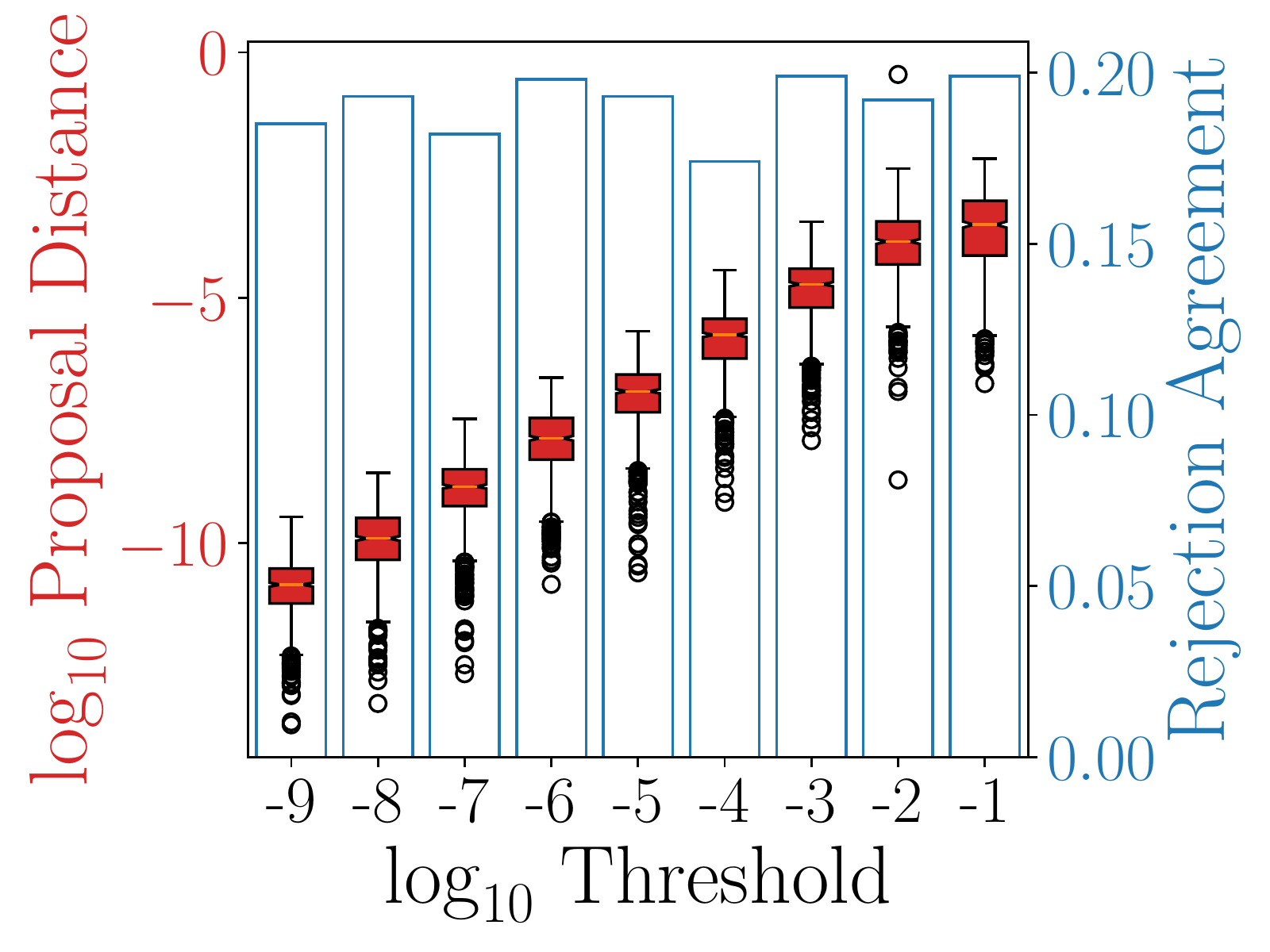}
    \caption{Difference in Transition Kernel}
    \label{subfig:cox-poisson-transition-difference}
  \end{subfigure}
  \caption{Visualization of the number of fixed point iterations required to compute the implicit updates to position and momentum required by the generalized leapfrog integrator in the log-Gaussian Cox-Poisson model. The number of decimal digits of similarity in transition kernels are shown in  \cref{subfig:cox-poisson-transition-difference} for variable thresholds.}
\end{figure}

The log-Gaussian Cox-Poisson model allows us to model count data within a spatial grid. In particular, consider a $N\times N$ grid on the unit square. Within the $(i, j)$ sub-region, the counts of some quantity of interest are denoted $y_{i,j}$. We model these count observations as following a Poisson distribution whose rate is spatially correlated according to a Gaussian process. Formally, we consider the following generative procedure:
\begin{align}
  \beta &\sim \mathrm{Normal}(2, 1/2) \\
  \sigma^2 &\sim\mathrm{Normal}(2, 1/2) \\
  \Sigma_{(i, j), (i', j')} \vert \sigma^2, \beta &= \sigma^2 \exp\paren{-\frac{\sqrt{(i-i')^2 + (j-j')^2}}{N \beta}} \\
  \mathrm{vec}(\mathbf{x}) \vert \Sigma &\sim \mathrm{MultivariateNormal}(\mu\mathbf{1}, \Sigma) \\
  y_{i, j}\vert x_{i, j} &\sim \mathrm{Poisson}(\exp(x_{i, j}) / N^2),
\end{align}
where $\mathrm{vec} : \R^{N\times N} \to\R^{N^2}$ by concatenating columns into a single vector.
The Bayesian inference task is to infer both the Gaussian process posterior $\mathbf{x}$ and the posterior distribution of $(\beta,\sigma^2)$ given observations of the Poisson process. Following \citet{rmhmc}, we employ a Metropolis-within-Gibbs-like strategy and alternate between sampling the posterior of $\mathbf{x}\vert (y_{i, j})_{i,j=1}^N, \sigma^2,\beta$ and the posterior of $(\sigma^2,\beta) \vert \mathbf{x}$. In generating data from this model, we set $\sigma^2 = 1.91$, $\beta = 1/33$, and $\mu=\log 126 - 1.91 / 2$. In our experiments we set $N=32$ so that the total dimensionality of the posterior is $32^2 + 2=1026$. Similarly to the situation in \cref{subsec:experiment-stochastic-volatility-model}, the $\mathbf{x}\vert (y_{i, j})_{i,j=1}^N, \sigma^2,\beta$ has a constant Fisher information metric; however, the conditional distribution of the hyperparameters $(\sigma^2,\beta) \vert \mathbf{x}$ depends on $\sigma^2$ and $\beta$, thereby necessitating the use of the generalized leapfrog integrator. Once again, to respect the constraints $\sigma^2 > 0$ and $\beta>0$, we employ the transformation $\sigma^2=\exp(\phi_1)$ and $\beta=\exp(\phi_2)$. Since the conditional distribution of $\mathrm{vec}(\mathbf{x})$ given $\phi_1$ and $\phi_2$ is simply a multivariate Gaussian, the Fisher information is
\begin{align}
  \mathbf{G}(\phi_1,\phi_2) &= \frac{1}{2}\begin{pmatrix}
    N^2 & \Omega(\phi_1,\phi_2) \\
    \Omega(\phi_1,\phi_2) & \Lambda(\phi_1,\phi_2)
  \end{pmatrix} \\
  \Omega(\phi_1,\phi_2) &= \mathrm{trace}\paren{\Sigma^{-1}(\phi_1,\phi_2) \frac{\partial}{\partial \phi_2} \Sigma(\phi_1,\phi_2)} \\
  \Lambda(\phi_1,\phi_2) &= \mathrm{trace}\paren{\Sigma^{-1}(\phi_1,\phi_2) \frac{\partial}{\partial \phi_2} \Sigma(\phi_1,\phi_2)\Sigma^{-1}(\phi_1,\phi_2) \frac{\partial}{\partial \phi_2} \Sigma(\phi_1,\phi_2)}.
\end{align}
In sampling from the posterior distribution of $(\sigma^2,\beta)$ we use the generalized leapfrog integrator with six integration steps and a step-size of 0.5. We seek to draw 5,000 samples after an initial burn-in period of 1,000 samples.

We visualize the posterior mean and standard deviation of the log-Gaussian Cox-Poisson process in \cref{fig:cox-poisson-cox-poisson}. As in the case of \cref{subsec:experiment-stochastic-volatility-model}, we observe few detectable differences within the first two moments of the posterior. We also visualize the marginal posteriors of the parameters $(\sigma^2,\beta)$, which shows substantial overlap, regardless of the threshold used in generating the samples. When assessing the degree to which reversibility and volume preservation are violated, we observe that the worst-case behavior of the generalized leapfrog integrator still exhibits error only in the third or fourth decimal digit. The similarity of transition kernels also reveals that the worst case difference among transitions still maintains approximately three digits of similarity with a transition kernel whose threshold is $1\times 10^{-10}$. When contextualized in terms of computational effort, one sees that a threshold of $1\times 10^{-1}$ requires only a single fixed point iteration on average for both the position and momentum variables, compared with three or four required by a threshold of $1\times 10^{-9}$. Since computing the Riemannian metric requires the inverse of $\Sigma$, which is a $N^2\times N^2$ matrix, one infers that the computational complexity of computing the Riemannian metric for the Cox-Poisson model scales as $\mathcal{O}(N^6)$; therefore, one desires few fixed point iterations, particularly in the implicit update to position, for which the Riemannian metric must be recomputed at each iteration.

We evaluate the ESS per second for the Log-Gaussian Cox-Poisson model in \cref{subfig:cox-poisson-sigma-ess-per-second,subfig:cox-poisson-beta-ess-per-second}. Of the two hyperparameters in the Cox-Poisson model, $\beta$ is the more challenging, having the smallest time-normalized ESS. Even with the most conservative convergence threshold, RMHMC outperforms the Euclidean algorithm; as a result, RMHMC has superior minimum ESS per second relative to Euclidean HMC.
As in \cref{subsec:experiment-stochastic-volatility-model}, one questions if the similarity of these posteriors is attributable to the choice of a very small step-size. However, the acceptance rate of the model parameters $(\sigma^2, \beta)$ is approximately eighty-five percent, thereby indicating that the step-size of the numerical integrator is not so small as to imply near-perfect conservation of the Hamiltonian energy.

\subsection{Fitzhugh-Nagumo Differential Equation Posterior}\label{subsec:experiment-fitzhugh-nagumo}

\begin{figure}[t!]
  \centering
  \includegraphics[width=\textwidth]{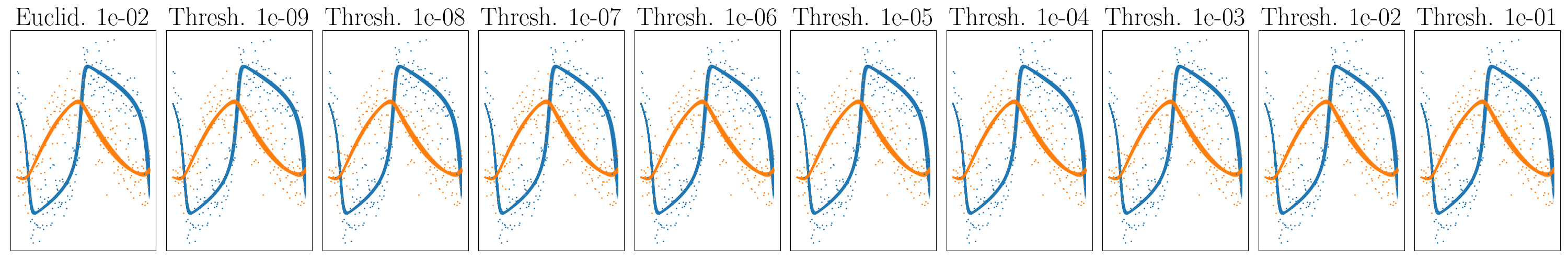}
  \caption{Visualization of the posterior mean of $v_t$ (in blue) and $r_t$ (in orange) in the Fitzhugh-Nagumo posterior distribution. The posterior means are visually indistinguishable irrespective of the convergence tolerance used in RMHMC.}
  \label{fig:fitzhugh-nagumo-fitzhugh-nagumo}
\end{figure}

\begin{figure}[t!]
  \centering
  \begin{subfigure}[t]{0.49\textwidth}
    \includegraphics[width=\textwidth]{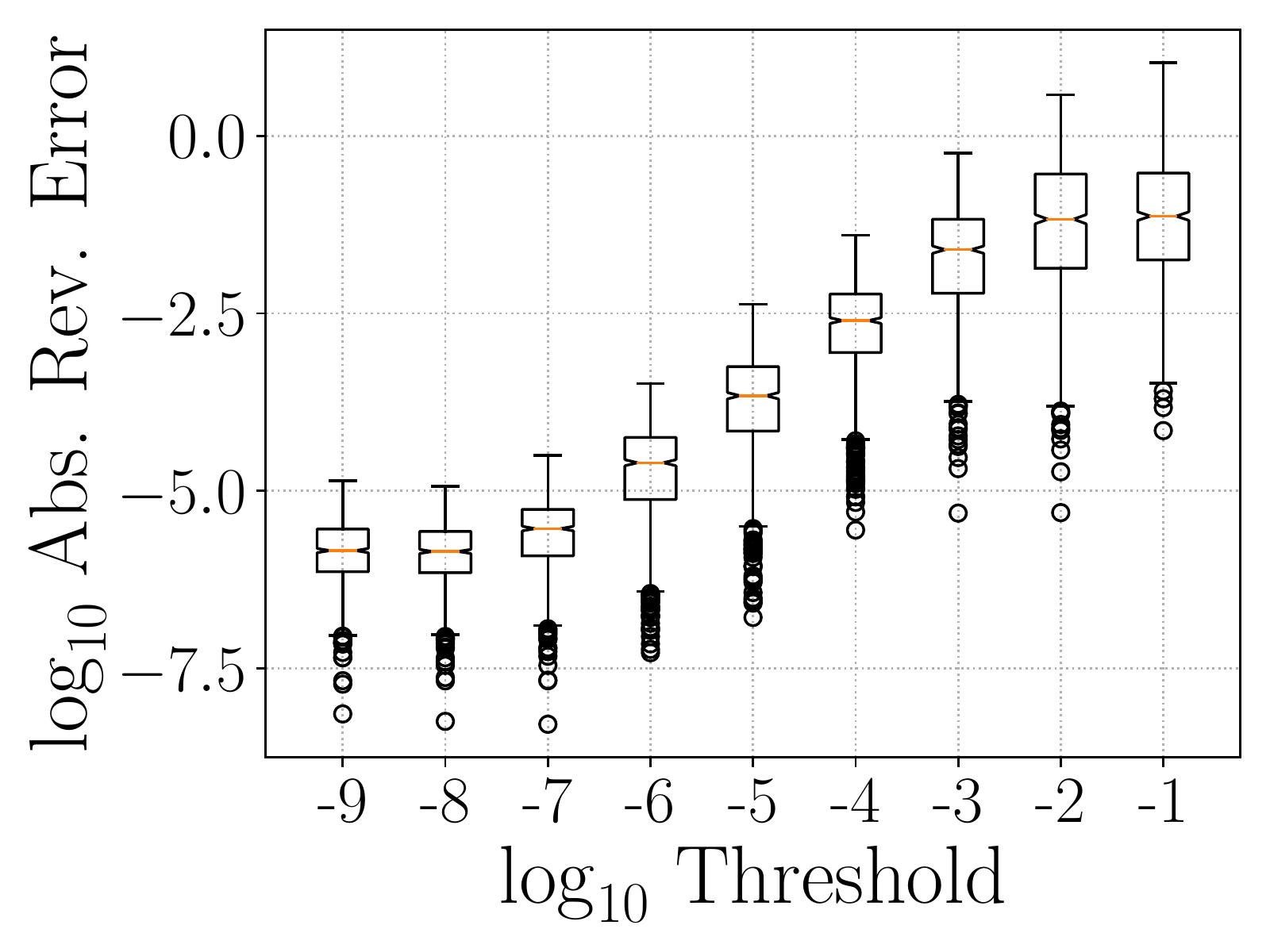}
    \caption{Error in Reversibility}
    \label{subfig:fitzhugh-nagumo-reversibility}
  \end{subfigure}
  ~
  \begin{subfigure}[t]{0.49\textwidth}
    \includegraphics[width=\textwidth]{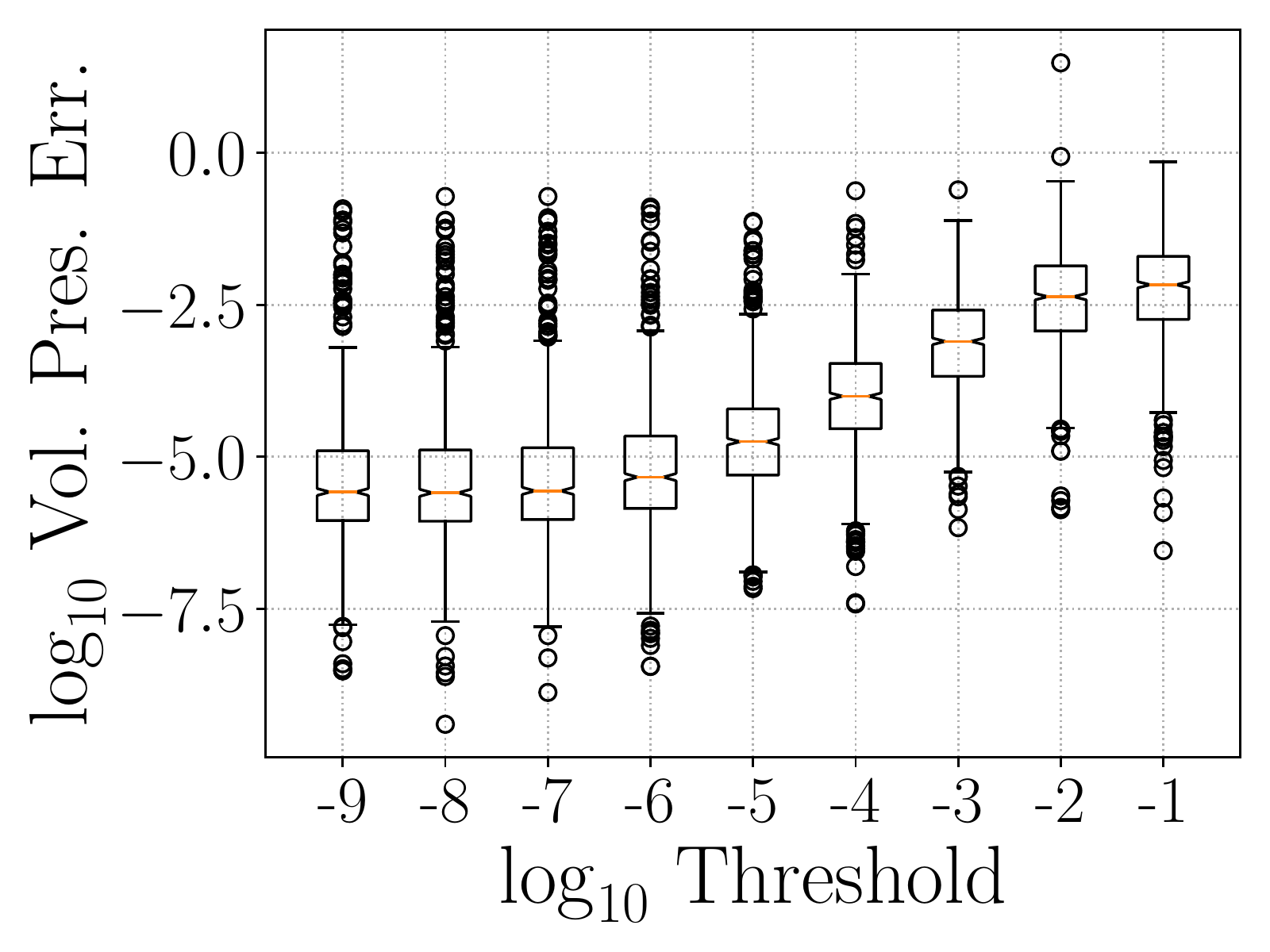}
    \caption{Error in Volume-Preservation}
    \label{subfig:fitzhugh-nagumo-jacobian-determinant}
  \end{subfigure}
  
  \begin{subfigure}[t]{\textwidth}
    \includegraphics[width=\textwidth]{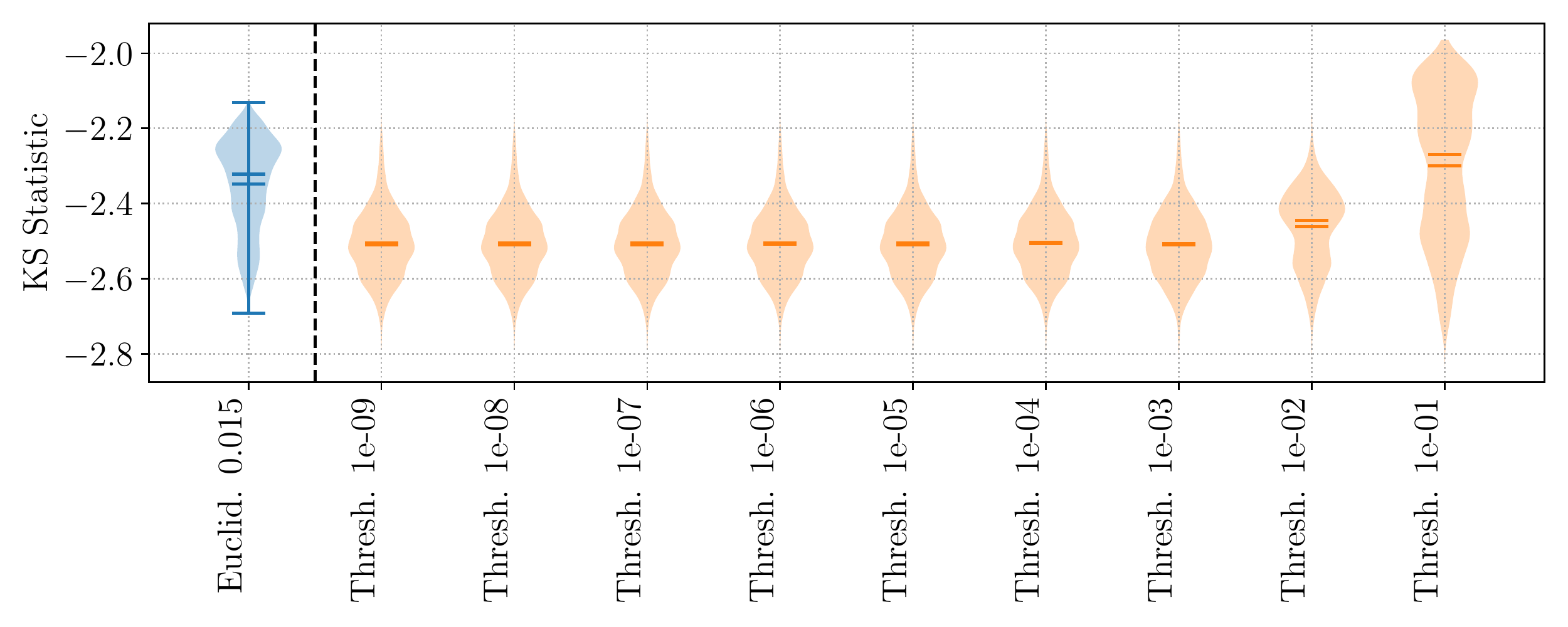}
    \caption{Ergodicity of RMHMC and HMC}
    \label{subfig:fitzhugh-nagumo-ergodicity}
  \end{subfigure}
  \caption{Visualization of the computational effort required to sample with RMHMC from the Fitzhugh-Nagumo posterior. We show the number of fixed point iterations required to compute the two implicit steps of the generalized leapfrog integrator as well as the distribution of Kolmogorov-Smirnov statistics over variable thresholds.}
\end{figure}

\begin{figure}[t!]
  \centering
  \begin{subfigure}[t]{0.32\textwidth}
    \includegraphics[width=\textwidth]{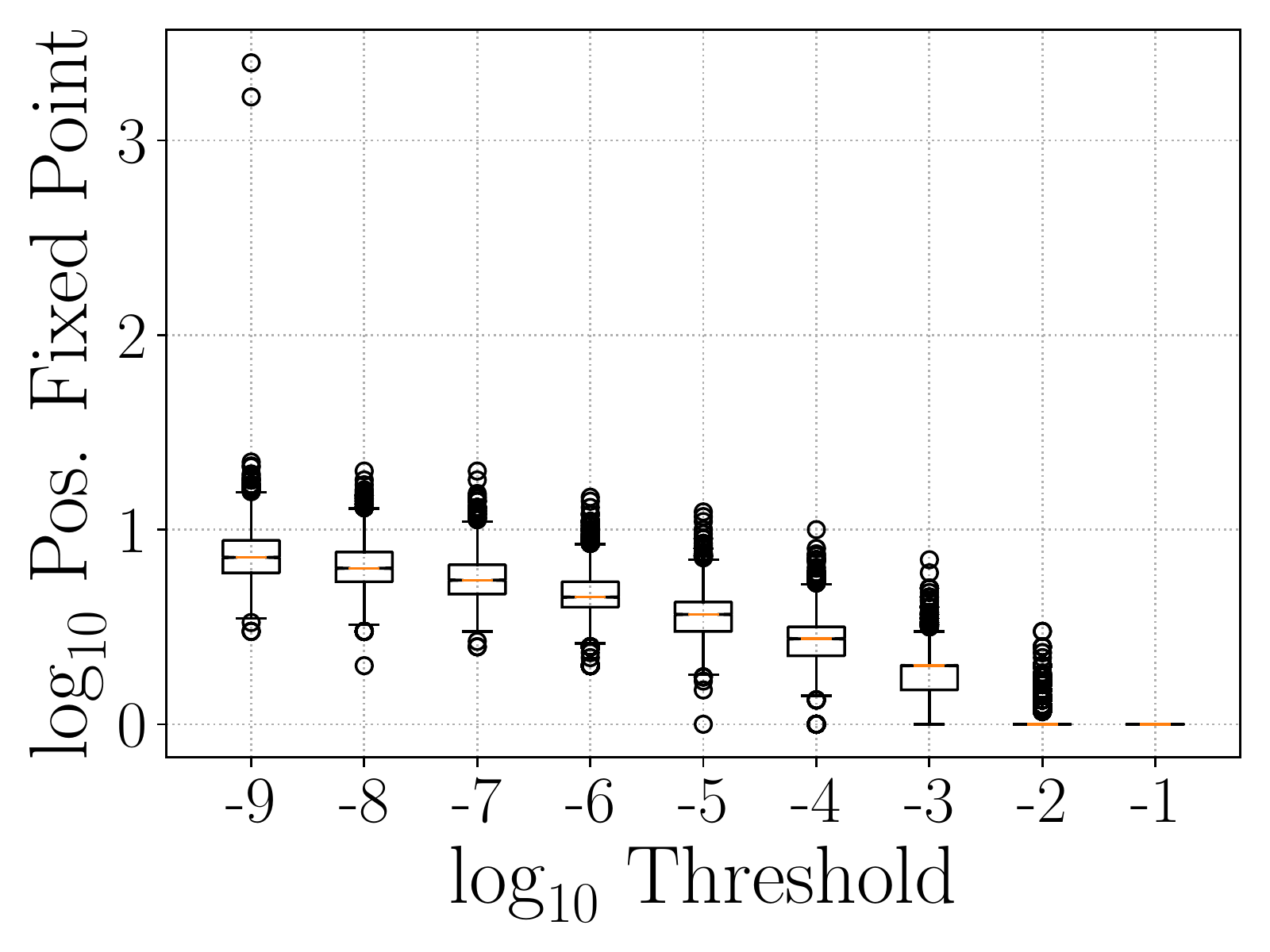}
    \caption{Number of fixed point iterations for position variable.}
    \label{subfig:fitzhugh-nagumo-fixed-point-position}
  \end{subfigure}
  ~
  \begin{subfigure}[t]{0.32\textwidth}
    \includegraphics[width=\textwidth]{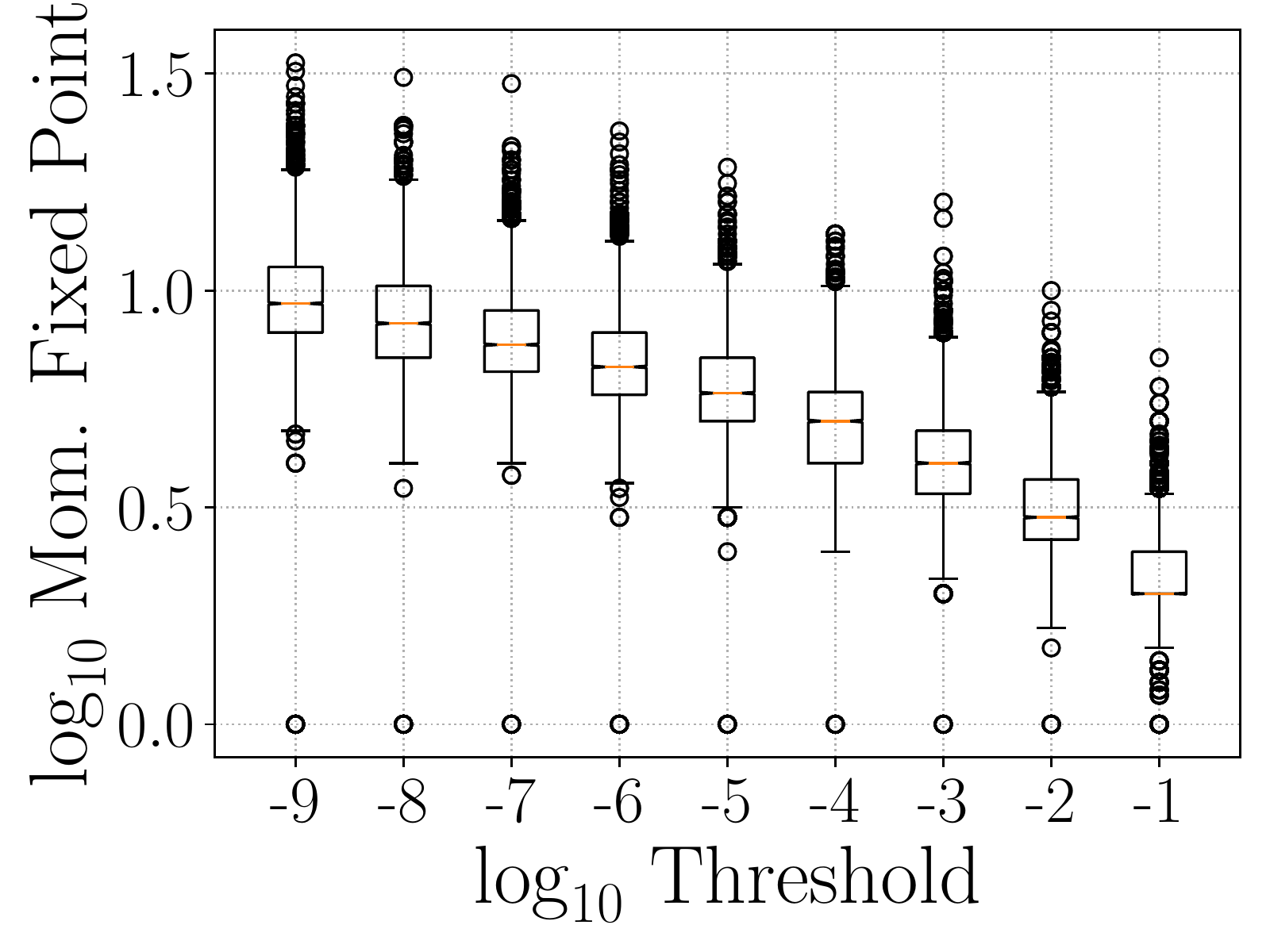}
    \caption{Number of fixed point iterations for momentum variable.}
    \label{subfig:fitzhugh-nagumo-fixed-point-momentum}
  \end{subfigure}
  ~
  \begin{subfigure}[t]{0.32\textwidth}
    \includegraphics[width=\textwidth]{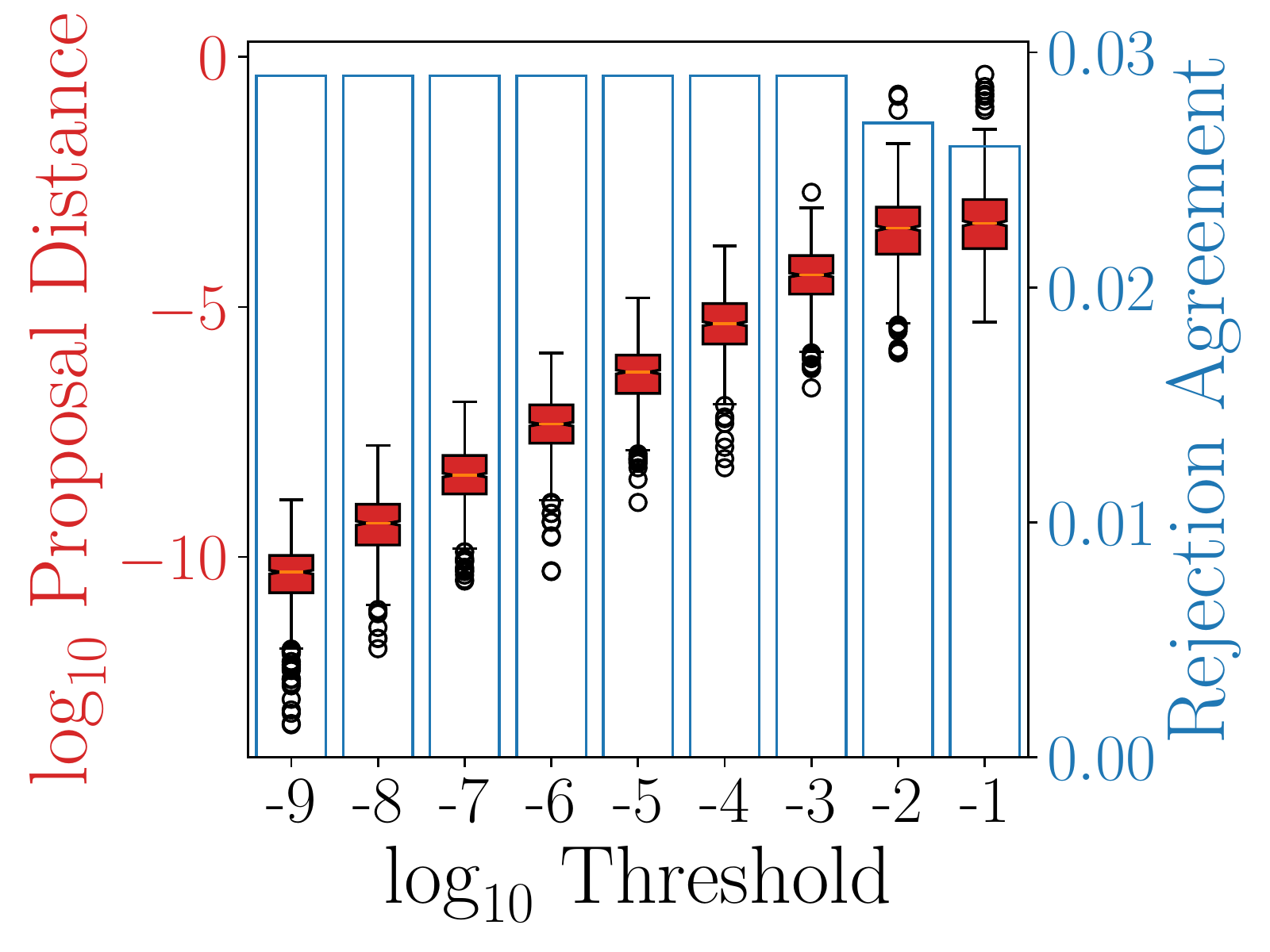}
    \caption{Difference in Transition Kernel}
    \label{subfig:fitzhugh-nagumo-transition-difference}
  \end{subfigure}
  \caption{Visualization of the error in reversibility (see
    \cref{subfig:fitzhugh-nagumo-reversibility}), error in
    volume-preservation (see
    \cref{subfig:fitzhugh-nagumo-jacobian-determinant}), and the number of
    decimal digits of similarity in transition kernels (see
    \cref{subfig:fitzhugh-nagumo-transition-difference}) for variable
    thresholds in the Fitzhugh-Nagumo posterior distribution.}
\end{figure}

\begin{figure}[t!]
  \centering
  \begin{subfigure}[t]{0.49\textwidth}
    \includegraphics[width=\textwidth]{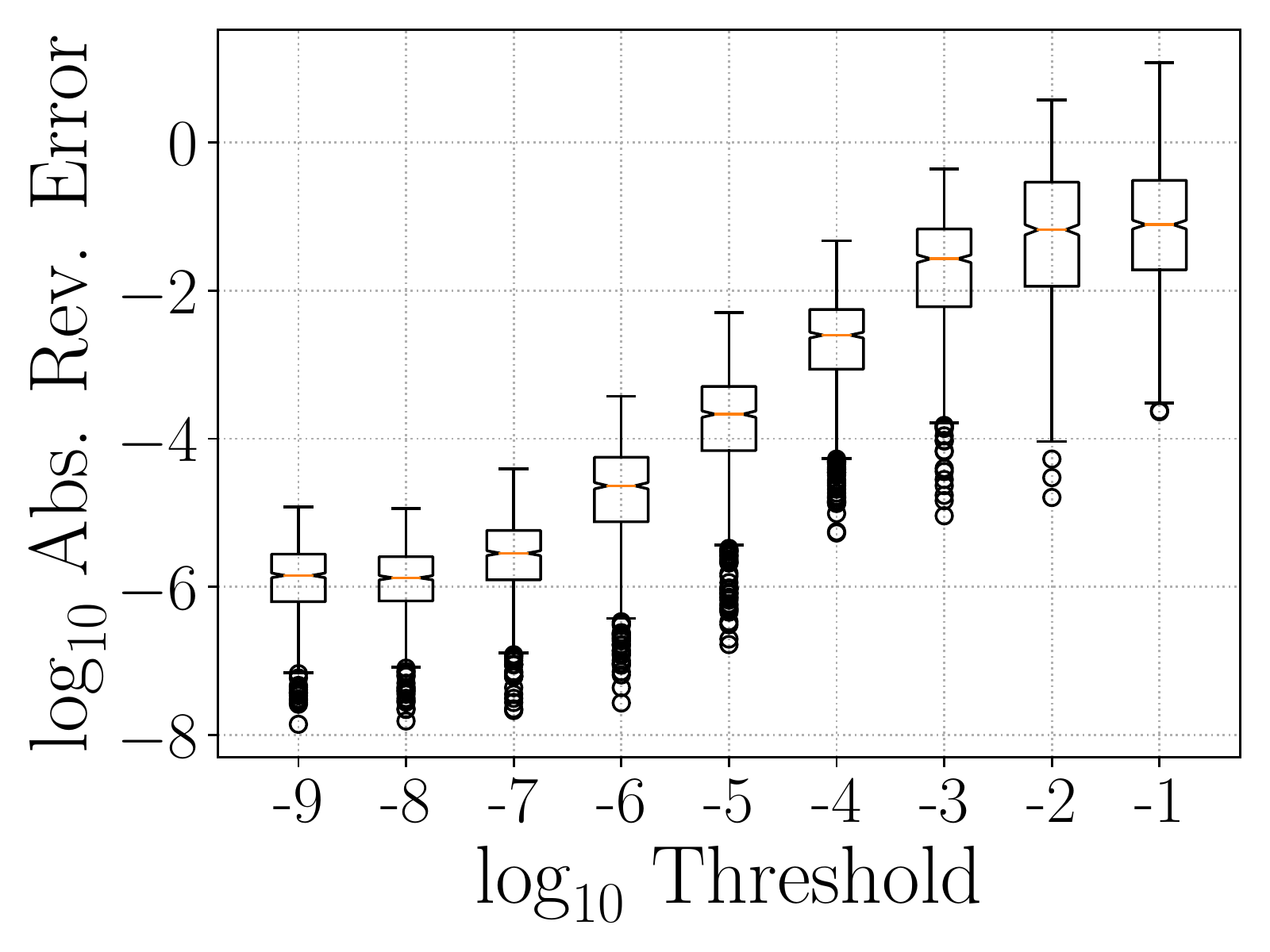}
    \caption{Error in Reversibility}
  \end{subfigure}
  ~
  \begin{subfigure}[t]{0.49\textwidth}
    \includegraphics[width=\textwidth]{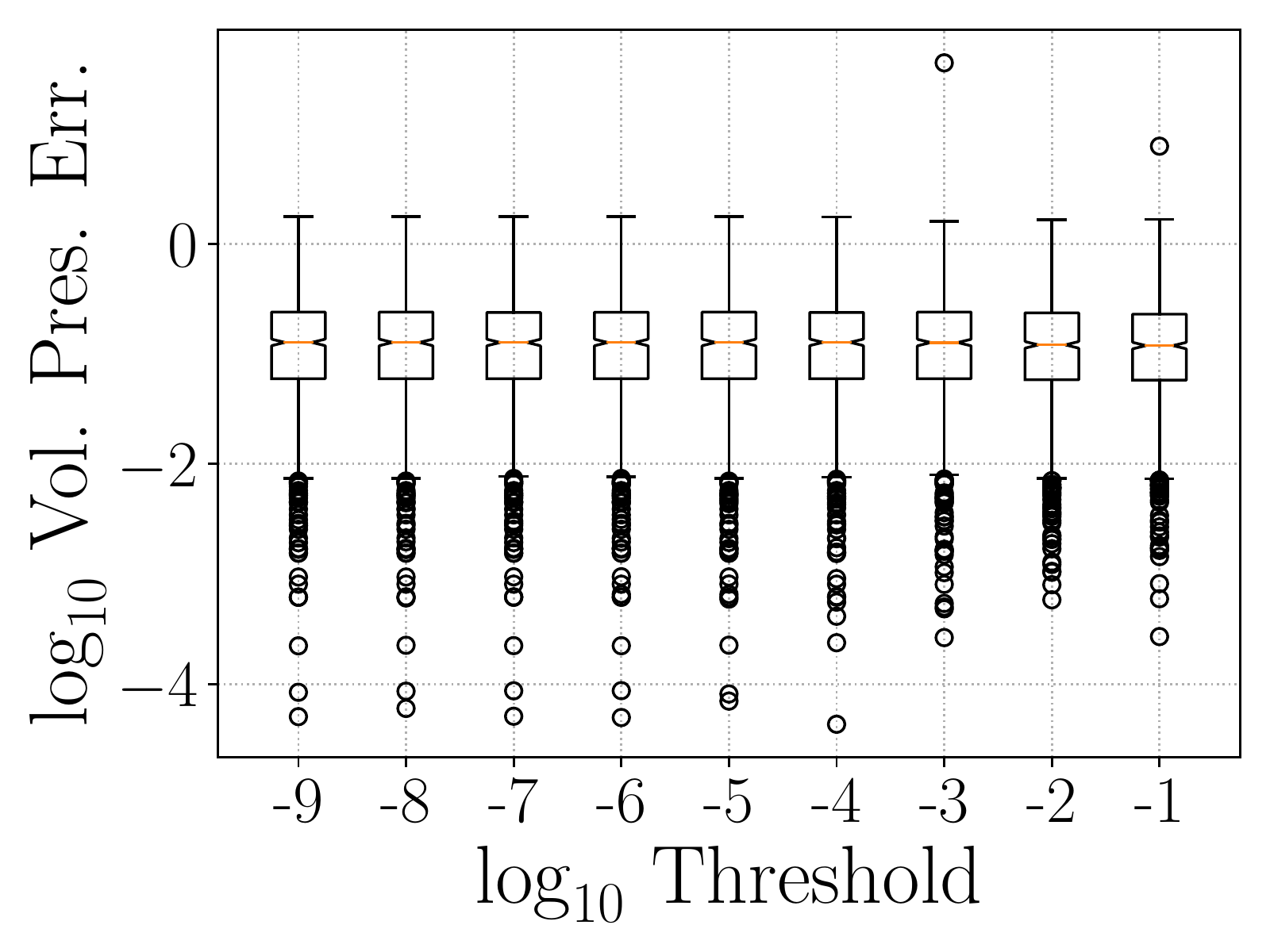}
    \caption{Error in Volume-Preservation}
  \end{subfigure}
  
  \begin{subfigure}[t]{\textwidth}
    \includegraphics[width=\textwidth]{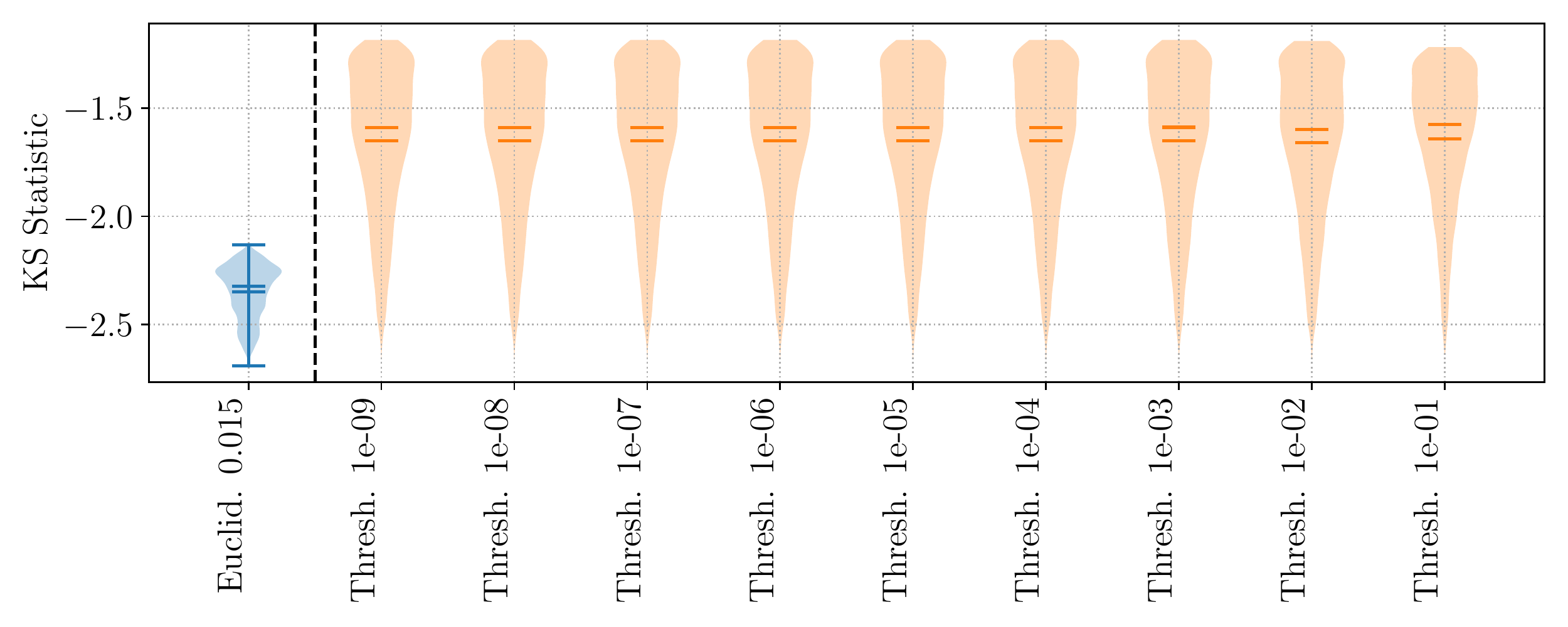}
    \caption{Ergodicity of RMHMC with incorrectly specified partial derivatives}
    \label{subfig:fitzhugh-nagumo-incorrect-ergodicity}
  \end{subfigure}
  \caption{Visualization of how a failure to respect the symmetry of partial derivatives yields severely degraded sample quality; indeed, the target distribution may be warped. Although reversibility can be reduced by a decreasing threshold, the property of volume-preservation cannot be recovered irrespective of threshold. As a consequence, the RMHMC algorithm is no longer ergodic for the target distribution, which is reflected in the degraded ergodicity measure.}
  \label{fig:fitzhugh-nagumo-incorrect-metrics}
\end{figure}

\begin{figure}[t!]
  \centering
  \begin{subfigure}[t]{0.32\textwidth}
    \includegraphics[width=\textwidth]{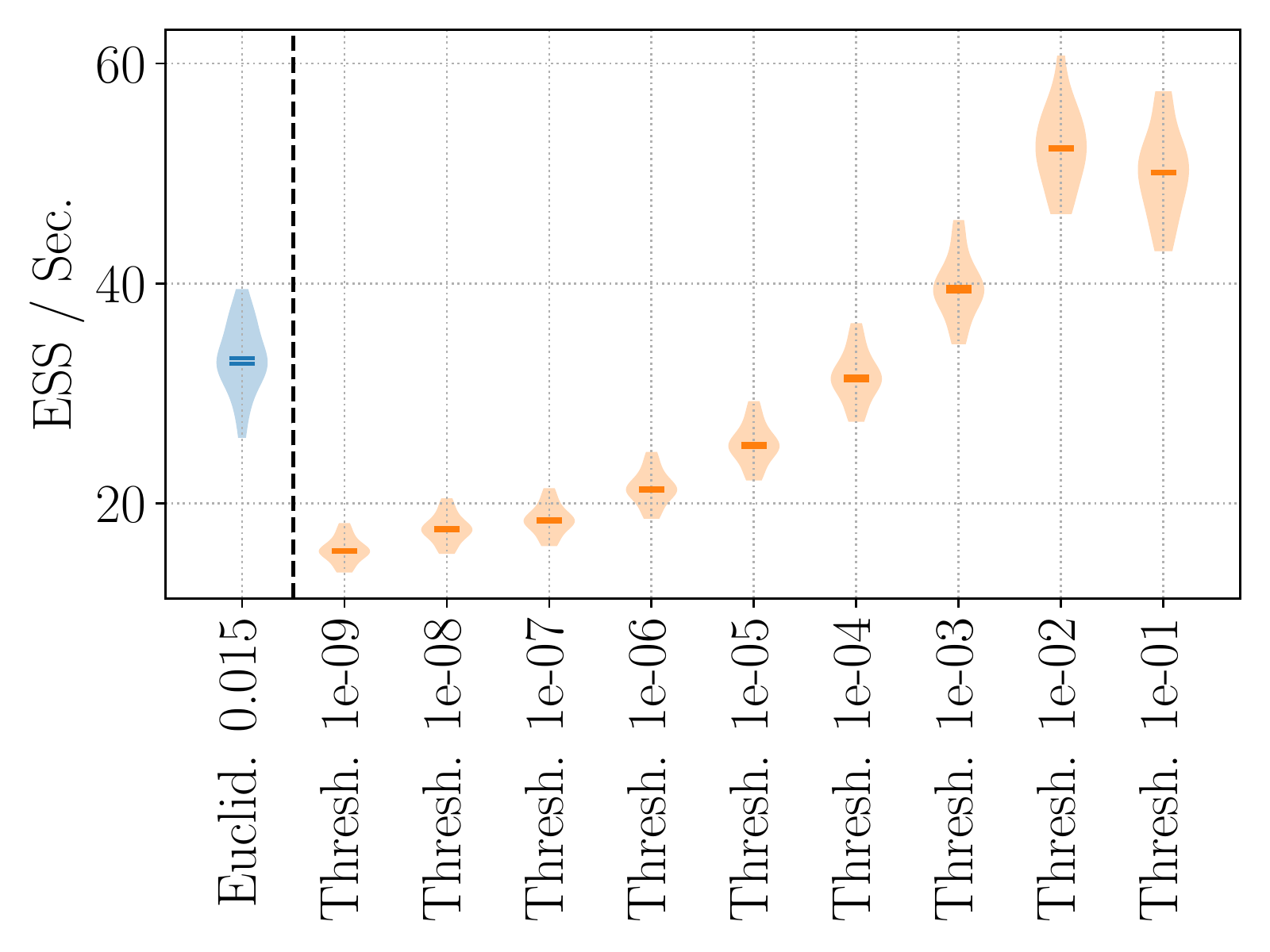}
    \caption{$a$}
  \end{subfigure}
  ~
  \begin{subfigure}[t]{0.32\textwidth}
    \includegraphics[width=\textwidth]{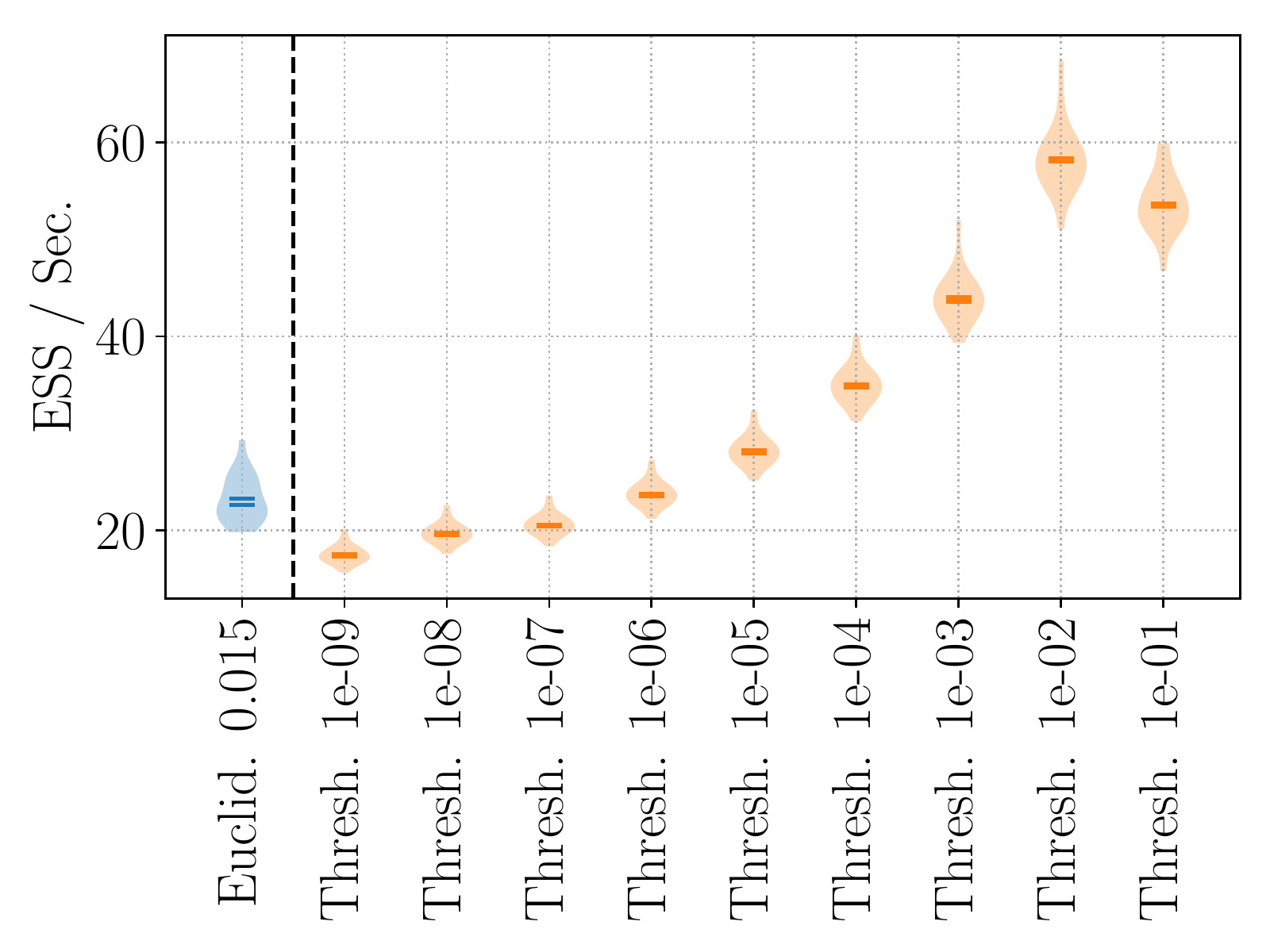}
    \caption{$b$}
  \end{subfigure}
  ~
  \begin{subfigure}[t]{0.32\textwidth}
    \includegraphics[width=\textwidth]{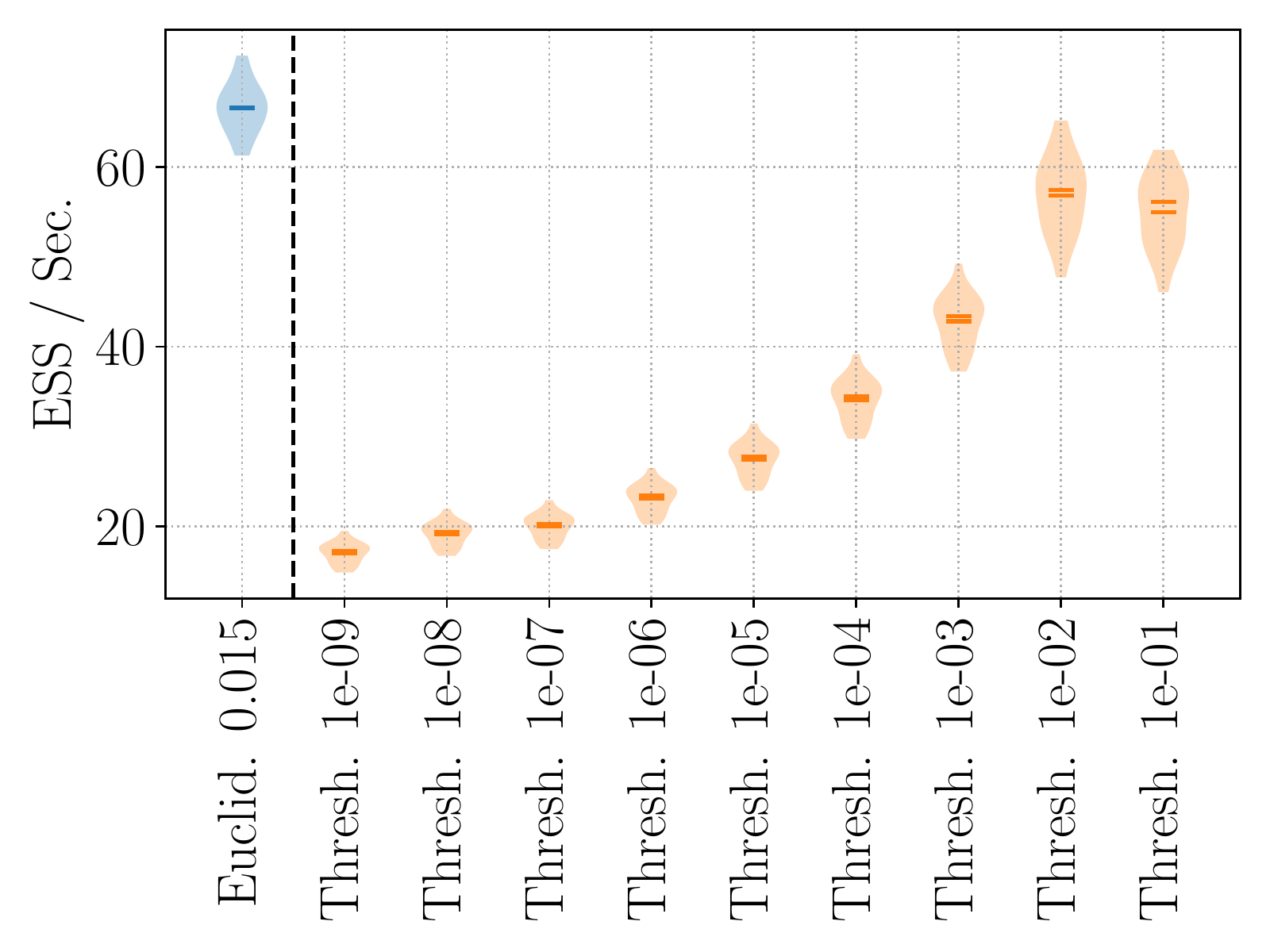}
    \caption{$c$}
  \end{subfigure}
  \caption{Visualization of the effective sample sizes (ESS) of the variables in the Fitzhugh-Nagumo posterior. Distributions over ESS are computed by splitting a Markov chain of length 100,000 into twenty contiguous sequences of length 5,000. We observe that the effective sample size for RMHMC is in excess of 5,000, indicating super-efficient sampling.}
  \label{fig:fitzhugh-nagumo-ess}
\end{figure}

\begin{figure}[t!]
  \begin{subfigure}[t]{0.49\textwidth}
    \centering
    \includegraphics[width=\textwidth]{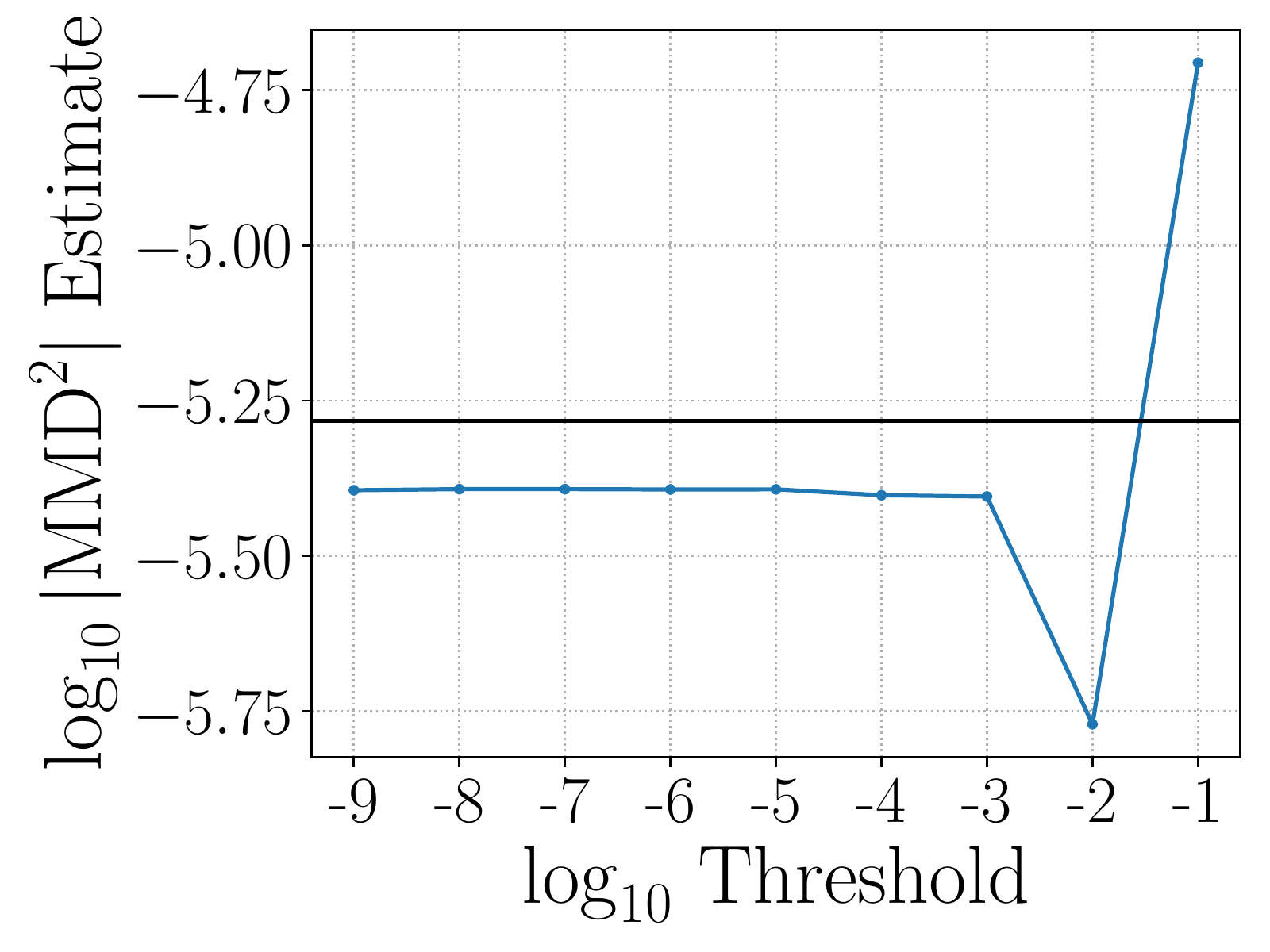}
    \caption{$\mathrm{MMD}_u^2$}
  \end{subfigure}
  ~
  \begin{subfigure}[t]{0.49\textwidth}
    \centering
    \includegraphics[width=\textwidth]{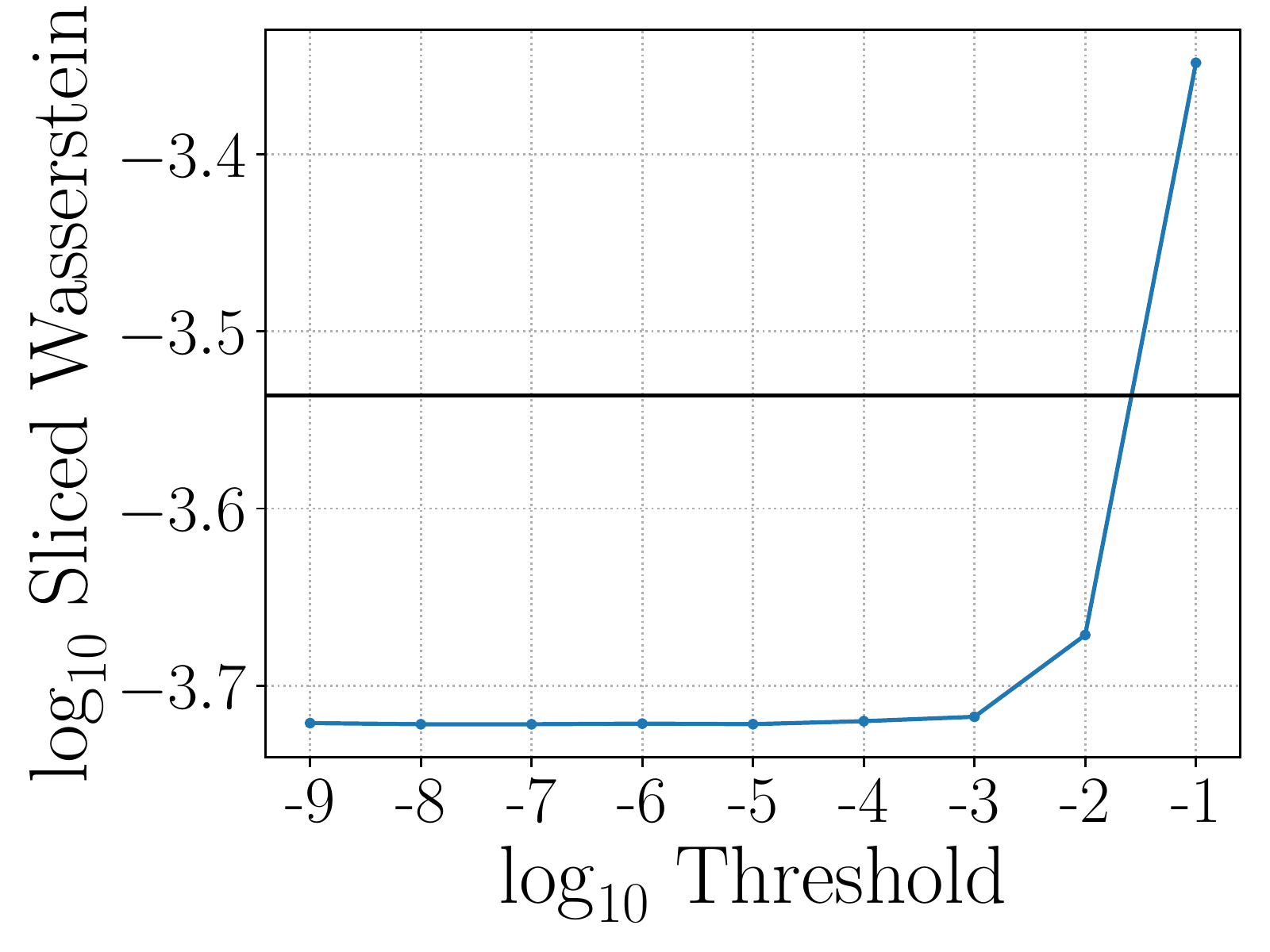}
    \caption{Sliced Wasserstein}
  \end{subfigure}
  \caption{Additional measures of ergodicity in the Fitzhugh-Nagumo differential equation model. In this example, the $\mathrm{MMD}_u^2$ statistic indicates that a threshold of $1\times 10^{-2}$ has the smallest value of the squared maximum mean discrepancy. This is in contrast to the Kolmogorov-Smirnov and sliced Wasserstein metrics, which indicate the ergodicity improves to a threshold of $1\times 10^{-3}$. In the MMD statistic we use a kernel bandwidth of 0.088.}
  \label{fig:fitzhugh-nagumo-extra-ergodicity}
\end{figure}

\begin{figure}[t!]
  \begin{subfigure}[t]{0.32\textwidth}
    \centering
    \includegraphics[width=\textwidth]{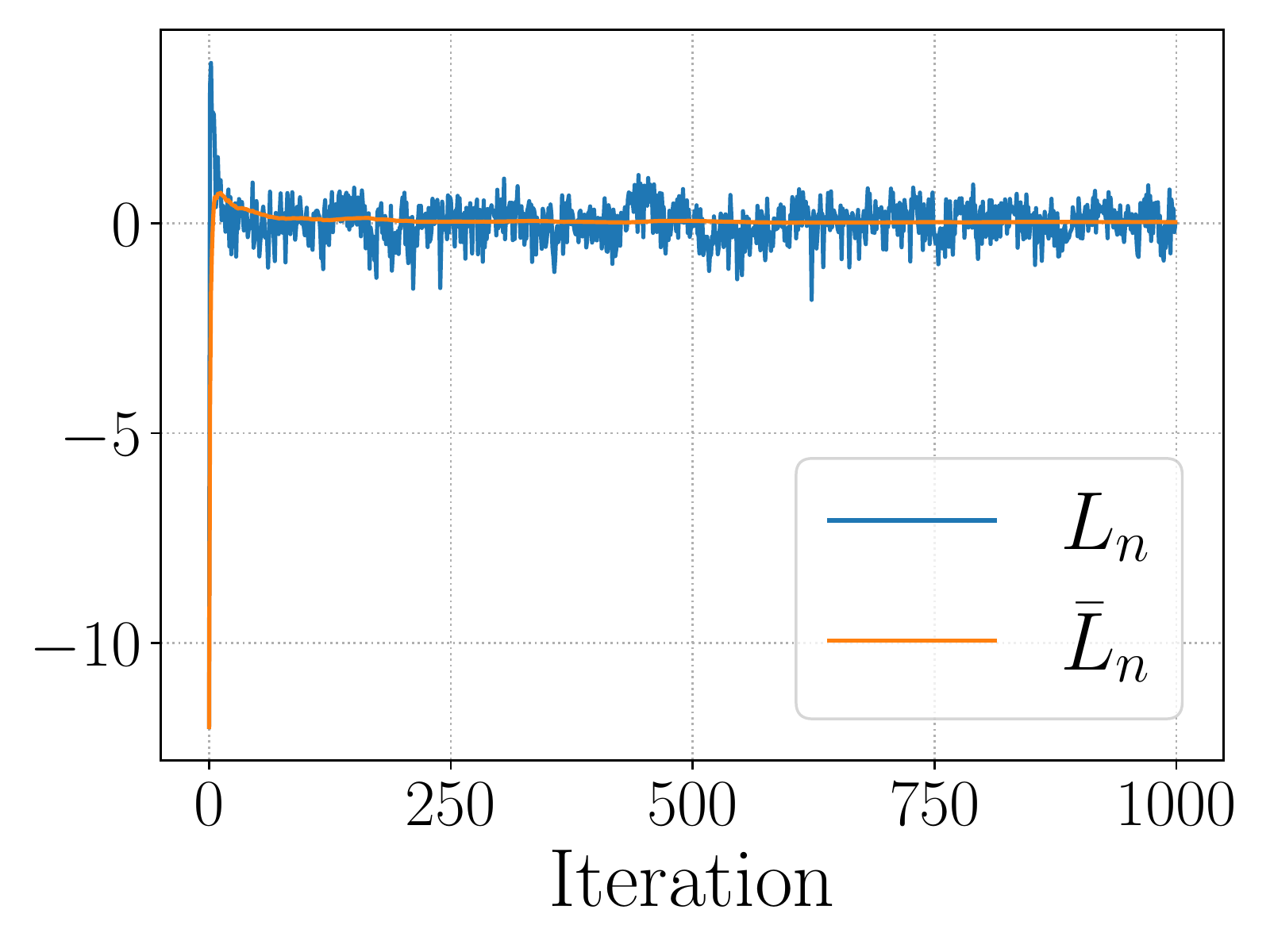}
    \caption{$L_n$ Sequence}
  \end{subfigure}
  ~
  \begin{subfigure}[t]{0.32\textwidth}
    \centering
    \includegraphics[width=\textwidth]{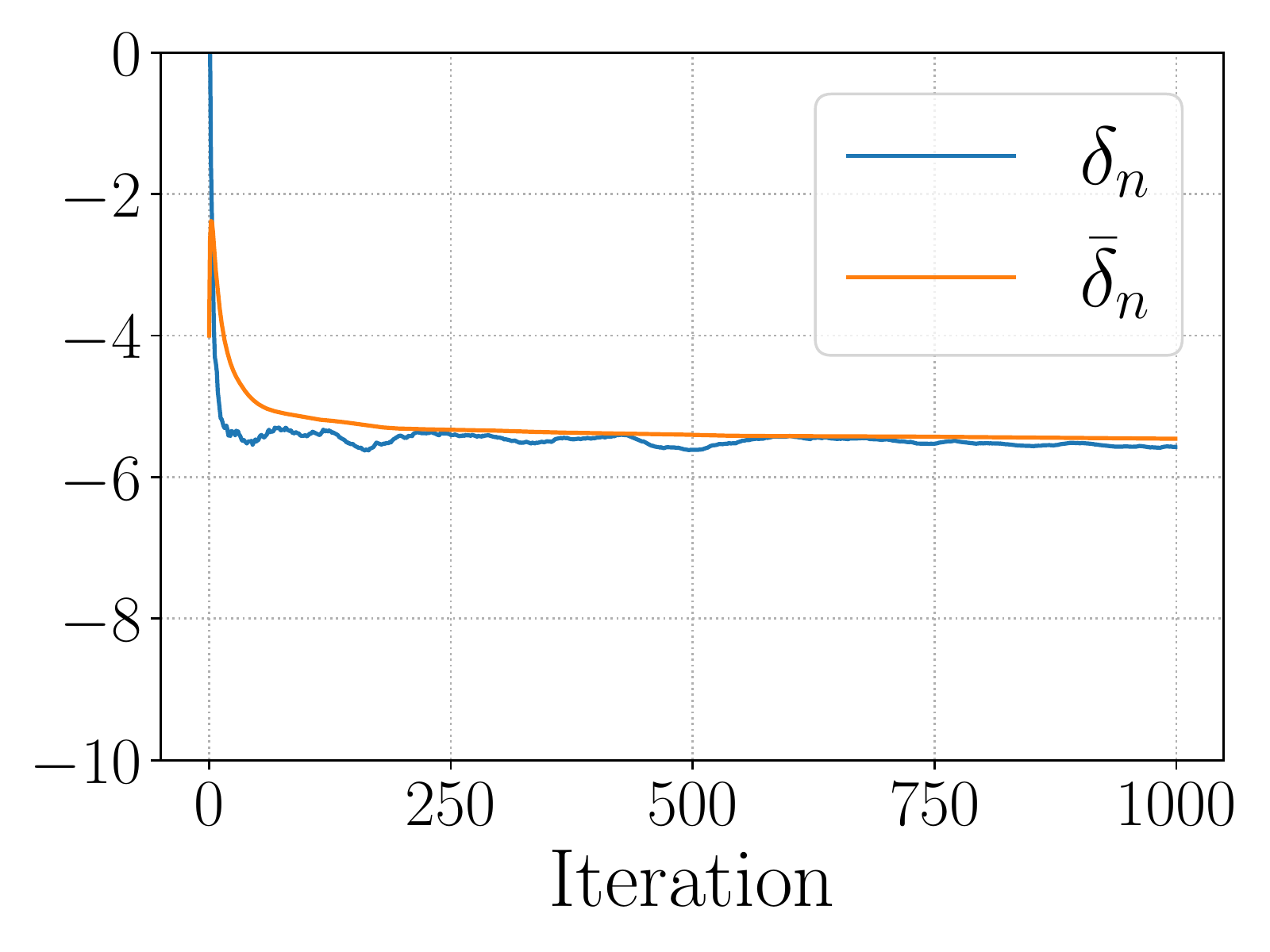}
    \caption{$\delta_n$ Sequence}
  \end{subfigure}
  ~
  \begin{subfigure}[t]{0.32\textwidth}
    \centering
    \includegraphics[width=\textwidth]{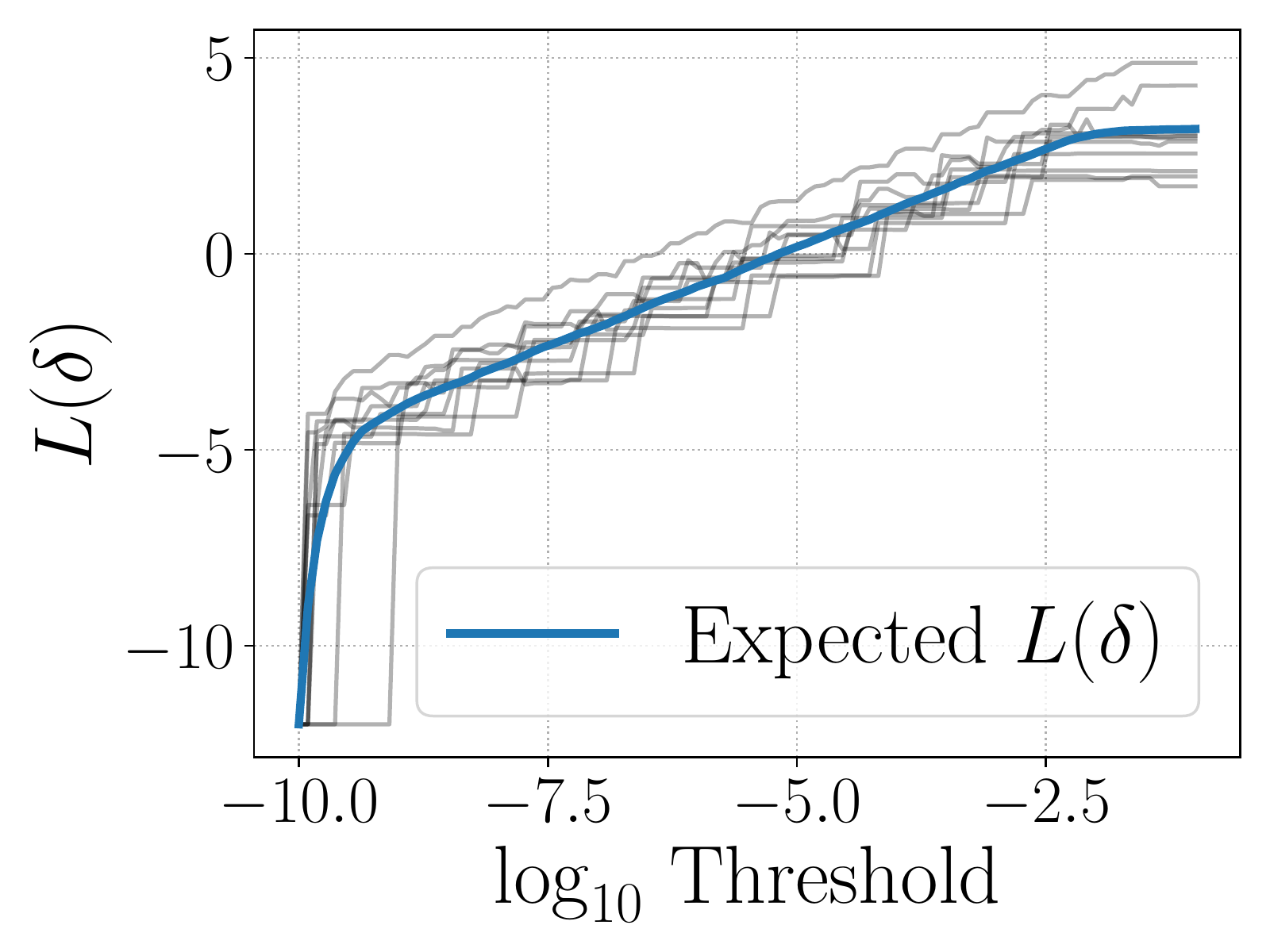}
    \caption{Monte Carlo $B(\delta)$}
  \end{subfigure}
  \caption{The use of Ruppert averaging in the Fitzhugh-Nagumo posterior distribution to adaptively set the convergence threshold to achieve four decimal digits of similarity compared to a transition kernel with a threshold of $1\times 10^{-10}$.  We show a Monte Carlo approximation to $B(\delta)$, which appears smooth, monotonically increasing. The value of $\delta$ satisfying $B(\delta)=0$ is approximately $\delta = 1\times 10^{-5.3}$, which is approximately the value of produced by Ruppert averaging.}
  \label{fig:fitzhugh-nagumo-dual-averaging}
\end{figure}

The Fitzhugh-Nagumo differential equation is a two-dimensional ordinary differential equation of the form,
\begin{align}
  \label{eq:fitzhugh-nagumo-v} \dot{v}_t &= c\paren{v_t - \frac{v_t^3}{3} + r_t} \\
  \label{eq:fitzhugh-nagumo-r} \dot{r}_t &= -\paren{\frac{v_t - a + b r_t}{c}},
\end{align}
where $(a, b, c)$ are parameters of the system. Given the initial condition $v_0 = -1$ and $r_0 = +1$, consider observing $\tilde{v}_{t_k} = v_{t_k} + \epsilon_{v, k}$ and $\tilde{r}_{t_k} = r_{t_k} + \epsilon_{r, k}$ where $\epsilon_{v, k}, \epsilon_{r, k}\overset{\mathrm{i.i.d.}}{\sim} \mathrm{Normal}(0, \sigma^2)$ for $k=1,\ldots, 200$ and where $(t_1,\ldots, t_{200})$ are 200 equally-spaced points in the interval $[0, 10)$. The Bayesian inference task at hand is to sample from the posterior distribution of $(a, b, c)\vert ((\tilde{v}_1, \tilde{r}_1), \ldots, (\tilde{v}_{200}, \tilde{r}_{200}))$ when $a$, $b$, and $c$ are equipped with independent standard normal priors.

Collectively denoting the parameters of the Fitzhugh-Nagumo differential equation model by $q=(a, b, c)$, it follows from the general expression for the Fisher information of a multivariate normal that the Riemannian metric formed by the sum of the Fisher information of the log-likelihood and the negative Hessian of the log-prior assumes the following form,
\begin{align}
  \label{eq:fitzhugh-nagumo-metric} \mathbf{G}_{ij}(q) = \frac{1}{\sigma^2}\sum_{k=1}^{200} \paren{\frac{\partial v_{t_k}}{\partial q_i}\frac{\partial v_{t_k}}{\partial q_j} + \frac{\partial r_{t_k}}{\partial q_i}\frac{\partial r_{t_k}}{\partial q_j}} + \mathbf{1}\set{i=j}
\end{align}
We must give some additional details about how this metric is computed in practice. One first observes that there is no immediate closed-form expression for the partial derivatives $\partial v_{t_k}/\partial q_i$ or $\partial r_{t_k}/\partial q_i$ based on the model description. However, we may leverage implicit differentiation of \cref{eq:fitzhugh-nagumo-v,eq:fitzhugh-nagumo-r} in order to deduce sensitivity equations \citep{Arriola2009} for these quantities. For instance,
\begin{align}
  \frac{\mathrm{d}}{\mathrm{d}t} \frac{\partial v_t}{\partial a} &= \frac{\partial}{\partial a} \frac{\mathrm{d} v_t}{\mathrm{d}t}  \\
  &= c\paren{\frac{\partial r_t}{\partial a} - (v_t^2 - 1)\frac{\partial v_t}{\partial a}}
\end{align}
In order to obtain initial conditions for these sensitivity dynamics, one observes that the initial condition of the Fitzhugh-Nagumo system $(v_0, r_0) = (-1, +1)$ does not depend on $a$, $b$, or $c$. The RMHMC algorithm requires the derivatives of the Riemannian metric; applying the chain rule to \cref{eq:fitzhugh-nagumo-metric} yields the following equations,
\begin{align}
  \begin{split}
    \frac{\partial}{\partial q_l} \mathbf{G}_{ij}(q) &= \frac{1}{\sigma^2} \sum_{k=1}^{200} \left( \frac{\partial^2 v_{t_k}}{\partial q_l\partial q_i}\frac{\partial v_{t_k}}{\partial q_j} + \frac{\partial v_{t_k}}{\partial q_i} \frac{\partial^2 v_{t_k}}{\partial q_l\partial q_j} \right. \\
    &\qquad \left. +~\frac{\partial^2 r_{t_k}}{\partial q_l\partial q_i}\frac{\partial r_{t_k}}{\partial q_j} + \frac{\partial r_{t_k}}{\partial q_i} \frac{\partial^2 r_{t_k}}{\partial q_l\partial q_j} \right)
  \end{split}
\end{align}
To compute these derivatives requires additional sensitivity equations: the second partial derivatives of the states. However, this presents no additional
conceptual difficulty since we may apply implicit differentiation a second time to obtain the required system of differential equations. For instance,
\begin{align}
  \frac{\mathrm{d}}{\mathrm{d}t} \frac{\partial^2 v_t}{\partial a^2} &= c\paren{\frac{\partial^2 r_t}{\partial a^2} - (v_t^2 - 1) \frac{\partial^2 v_t}{\partial a^2} - 2v_t \paren{\frac{\partial v_t}{\partial a}}^2} \\
  \frac{\mathrm{d}}{\mathrm{d}t} \frac{\partial^2 v_t}{\partial b\partial a} &= c\paren{\frac{\partial^2 r_t}{\partial a\partial b} - 2v_t\frac{\partial v_t}{\partial a} \frac{\partial v_t}{\partial b} - (v_t^2 - 1) \frac{\partial^2 v_t}{\partial a\partial b}} \\
  \label{eq:fitzhugh-nagumo-v-a-c} \frac{\mathrm{d}}{\mathrm{d}t} \frac{\partial^2 v_t}{\partial c\partial a} &= c\paren{\frac{\partial^2 r_t}{\partial a\partial c} - (v_t^2 - 1) \frac{\partial^2 v_t}{\partial a\partial c}} + \frac{\partial r_t}{\partial a} - 2c \frac{\partial v_t}{\partial c} \frac{\partial v_t}{\partial a} v_t + (v_t^2 - 1) \frac{\partial}{\partial a} v_t.
\end{align}

We see, therefore, that unlike the other distributions considered thus far, the values of the log-posterior, its gradient, and Riemannian metric of the Fitzhugh-Nagumo model are not available in closed-form. Instead, these quanities are approximated by numerically solving initial value problems involving sensitivities of the required quantities. In our implementation, we solve these initial value problems using SciPy's {\tt odeint} function with its default parameters. In total there are twenty initial value problems to be solved in the Fitzhugh-Nagumo posterior: the two equations for the Fitzhugh-Nagumo dynamics in \cref{eq:fitzhugh-nagumo-v,eq:fitzhugh-nagumo-r}, six equations for the first-order sensitivities of $v_t$ and $r_t$ with respect to the parameters, and twelve equations for the second order sensitivities. As an interesting consequence of this approximation, the computed ``derivatives'' of the Hamiltonian are no longer exact up to machine precision due to accumulating errors associated to the numerical solution of the initial value problems. As demonstrated in \cref{app:reversibility-volume-preservation-generalized-leapfrog}, the proof that the generalized leapfrog integrator conserves volume is predicated on the symmetry of the partial derivatives of the Hamiltonian, which may be violated when analytical expressions for derivatives are supplanted by numerical solutions to differential equations.

We can apply the Ruppert averaging method in order to find a threshold in the Fitzhugh-Nagumo differential equation posterior that produces a numerical integrator with four decimal digits of similarity. \Cref{fig:fitzhugh-nagumo-dual-averaging} shows convergence of the sequences $\bar{L}_n$ and $\bar{\delta}_n$. The sequence $\bar{\delta}_n$ by iteration 100 and the sequence $\bar{L}_n$ has converged by the five-hundredth iteration.

We follow \citet{rmhmc} and set a number of integration steps equal to six and use an integration step-size of $0.5$ in our experiments. For the case of HMC, we follow \citet{pmlr-v70-tripuraneni17a} and use ten integration steps and a step-size of $0.015$, which produces an acceptance rate of around ninety-percent. We see in \cref{subfig:fitzhugh-nagumo-transition-difference} that a threshold of $1\times 10^{-2}$ is sufficent to obtain around 2.5 or 3 decimal digits of similarity relative to a transition kernel with threshold $1\times 10^{-10}$. In terms of ergodicity, we observe that RMHMC with a threshold of $1\times 10^{-1}$ produces samples that are arguably of lesser quality than a HMC baseline; see \cref{subfig:fitzhugh-nagumo-ergodicity}. However, a threshold of $1\times 10^{-2}$ appears to produce an improvement over HMC and for thresholds less than $1\times 10^{-2}$ there is no evident ergodicity advantage. One notes that the the Riemannian metric and its gradients constitute an expensive metric to compute, even though the parameter space of the posterior is only three-dimensional. This is because the Riemannian metric and its gradients require computing solutions to initial value problems. Therefore, the computational burden of computing fixed point solutions is significant in the Fitzhugh-Nagumo posterior, particularly in the implicit update of the position variable, for which we must recompute the metric at each iteration. 

Violations of volume preservation may appear in other, more subtle ways, which would raise no immediate indication of error unless one knew to look for it. As a concrete example, we consider a transcription error in the specification of the forward sensitivity equations necessary to compute the derivatives of the Riemannian metric used in the Fitzhugh-Nagumo model. Specifically, we consider replacing the expression in \cref{eq:fitzhugh-nagumo-v-a-c} as follows:
\begin{align}
  \frac{\mathrm{d}}{\mathrm{d}t} \frac{\partial^2 v_t}{\partial c\partial a} &\overset{\mathrm{wrong}}{=} c\paren{\frac{\partial^2 r_t}{\partial a\partial c} - (v_t^2 - 1) \frac{\partial^2 v_t}{\partial a\partial c}} + \frac{\partial r_t}{\partial a} - 2c \frac{\partial v_t}{\partial c} \frac{\partial v_t}{\partial a} v_t + 1 - v_t^2 \frac{\partial v_t}{\partial a},
\end{align}
and we incorrectly set,
\begin{align}
    \frac{\mathrm{d}}{\mathrm{d}t} \frac{\partial^2 v_t}{\partial c\partial b} &\overset{\mathrm{wrong}}{=} -2cv_t \frac{\partial v_t}{\partial c}\frac{\partial v_t}{\partial a} + 1 - v_t^2 \frac{\partial v_t}{\partial b} - c\paren{v_t^2 - 1} \frac{\partial^2 v_t}{\partial c\partial b} + \frac{\partial r_t}{\partial b} + \frac{\partial r_t}{\partial c\partial b} \\
    \frac{\mathrm{d}}{\mathrm{d}t} \frac{\partial^2 v_t}{\partial^2 c} &\overset{\mathrm{wrong}}{=} -2cv_t \frac{\partial v_t}{\partial c}\frac{\partial v_t}{\partial c} + 2 - v_t^2 \frac{\partial v_t}{\partial c} - c\paren{v_t^2 - 1} \frac{\partial^2 v_t}{\partial^2 c} - v_t^2 \frac{\partial v_t}{\partial c} + 2\frac{\partial r_t}{\partial c} + c\frac{\partial^2 r_t}{\partial^2 c}
\end{align}
The result of this modification is that the symmetry of partial derivatives is violated; therefore, no convergence threshold can be used in the generalized leapfrog integrator to produce a volume preserving proposal. In \cref{app:stationary-distribution-non-volume-preserving} we give a theoretical treatment of what occurs when the Metropolis-Hastings correction is applied without properly accounting for the change in volume due to the proposal. We analyze the effect of this modification in \cref{fig:fitzhugh-nagumo-incorrect-metrics}. We see that although reversibility may be reduced via a diminished threshold, it is not possible to produce a volume-preserving proposal in the presence of mis-specified sensitivity equations that break the symmetry of partial derivatives. In terms of ergodicity,  the failure to account for substantial changes in volume has destroyed the stationarity property of the desired target distribution and the ergodicity metric we employ reveals degraded performance relative to a correct implementation of the Fitzhugh-Nagumo sensitivity equations. It is worth observing that despite the incorrectly specified derivatives, volume preservation is still preserved to around one decimal digit, which explains why samples produced by this incorrect procedure, while certainly degraded, are not absolutely awful, as shown in \cref{subfig:fitzhugh-nagumo-incorrect-ergodicity} (around 1.5 decimal digits of similarity with the target posterior as measured by Kolmogorov-Smirnov statistics along random one-dimensional sub-spaces).

When assessing ergodicity of the RMHMC algorithm in the Fitzhugh-Nagumo model, we use rejection sampling to generate $100,000$ independent samples as a point of comparison. To apply rejection sampling, we use a uniform distribution in a cube centered at the posterior mode and whose side lengths are ten times the marginal standard deviations computed from a Laplace approximation to the posterior at the mode. Similar to the conclusion of \citet{rmhmc}, we find the Euclidean HMC performs competitively with RMHMC, though RMHMC does exhibit a distribution of Kolmogorov-Smirnov statistics somewhat more tightly concentrated near zero. The RMHMC algorithm also yields samples that are less auto-correlated, producing a larger ESS, as shown in \cref{fig:fitzhugh-nagumo-ess}.

\subsection{Multi-Scale Phenomena in a Multivariate Student-t Distribution}\label{subsec:experiment-student-t}

\begin{figure}[t!]
  \centering
  \begin{subfigure}[t]{0.49\textwidth}
    \includegraphics[width=\textwidth]{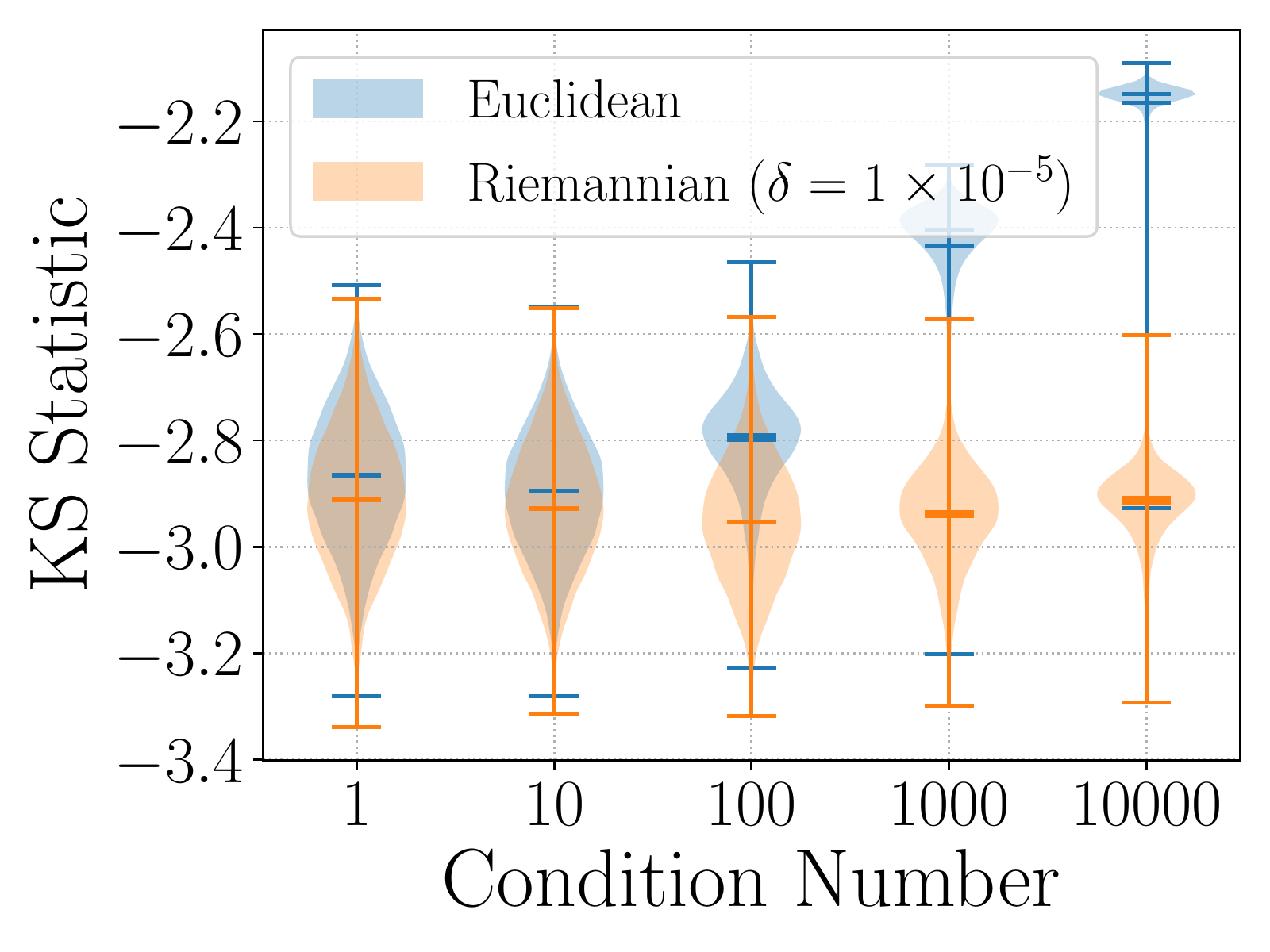}
    \caption{Error in Reversibility}
    \label{subfig:t-kolmogorov-smirnov-scale}
  \end{subfigure}
  ~
  \begin{subfigure}[t]{0.49\textwidth}
    \includegraphics[width=\textwidth]{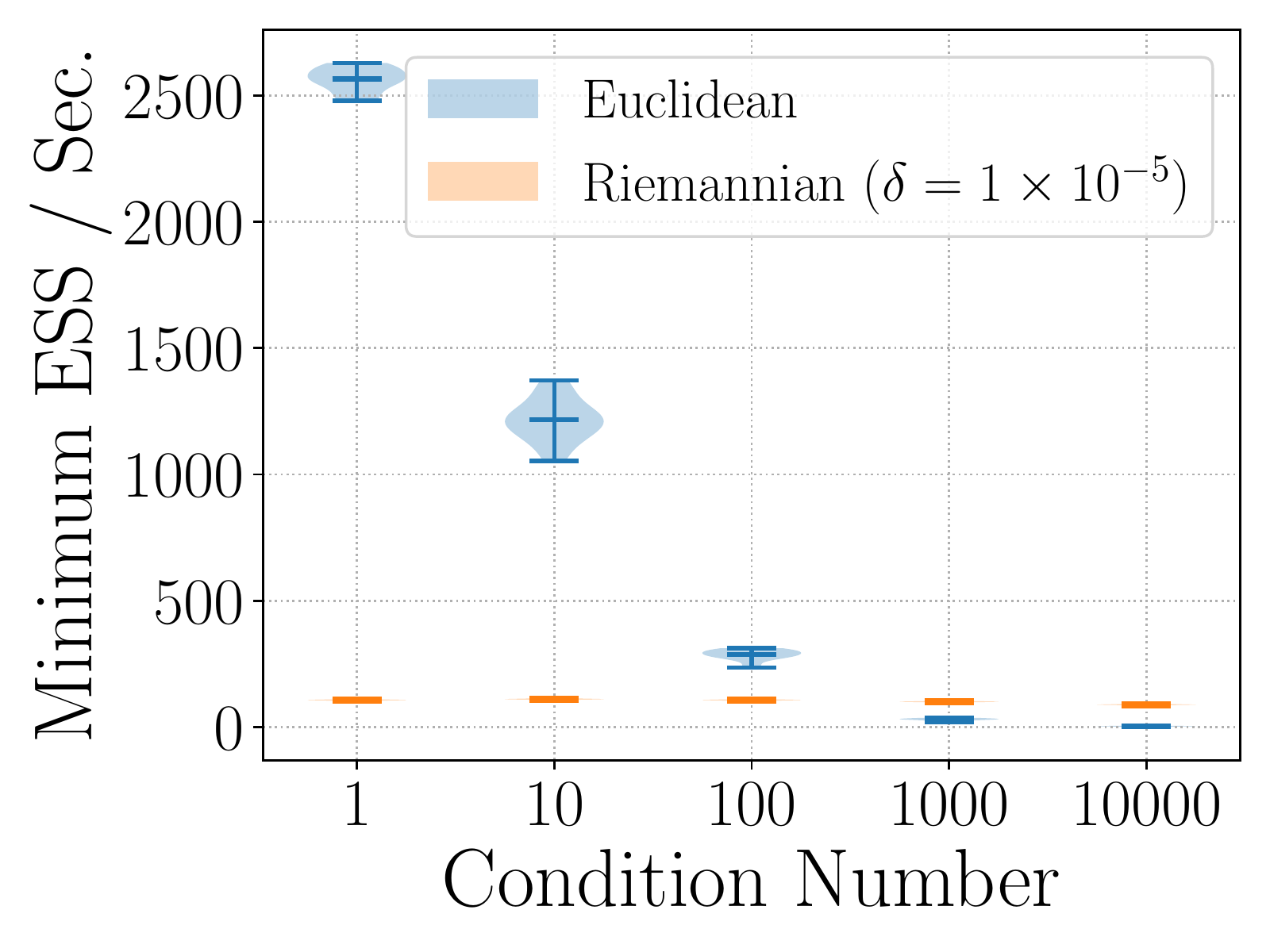}
    \caption{Error in Volume-Preservation}
    \label{subfig:t-ess-per-second-scale}
  \end{subfigure}

  \caption{Visualization of how multiscale properties of a target multivariate Student-$t$ distribution affect the ergodicity and computational efficiency of the Euclidean HMC and RMHMC algorithms. In the prescence of severe multiscale attributes of the posterior, Euclidean methods of HMC exhibit poor sampling performance by these metrics.}
  \label{fig:t-multiscale-phenomena}
\end{figure}

\begin{figure}[t!]
  \centering
  \begin{subfigure}[t]{0.32\textwidth}
    \includegraphics[width=\textwidth]{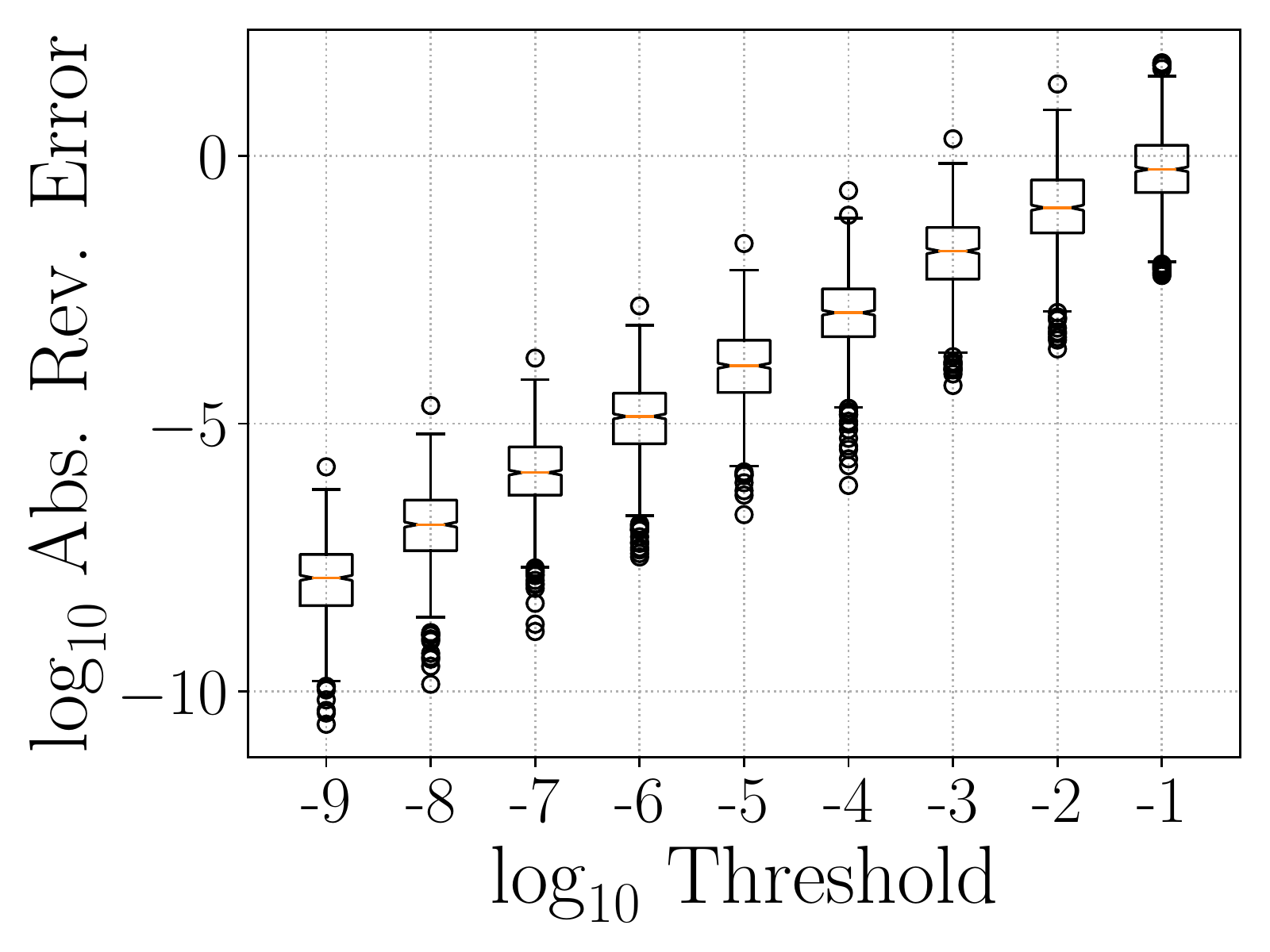}
    \caption{Error in Reversibility}
    \label{subfig:t-reversibility}
  \end{subfigure}
  ~
  \begin{subfigure}[t]{0.32\textwidth}
    \includegraphics[width=\textwidth]{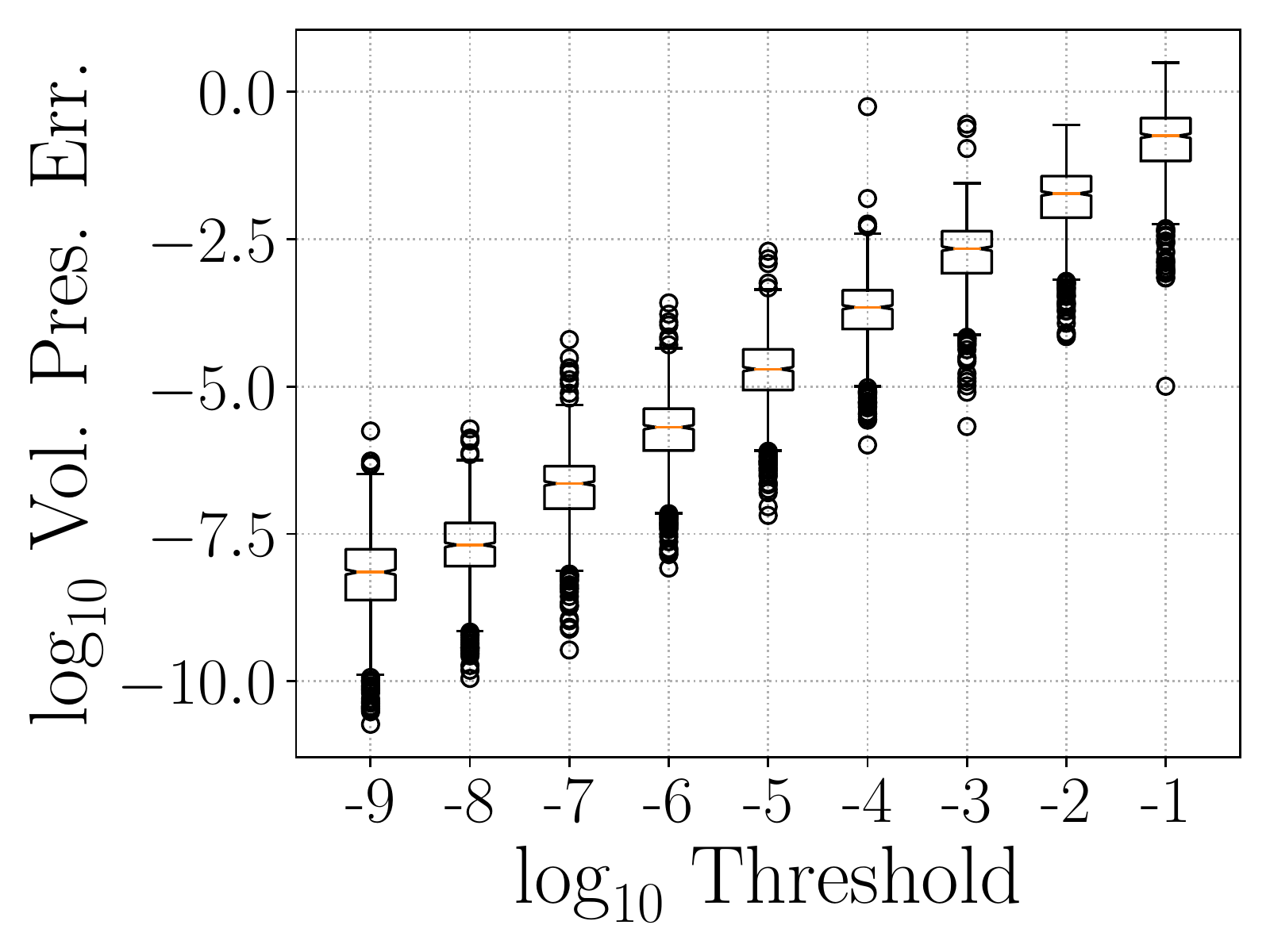}
    \caption{Error in Volume-Preservation}
    \label{subfig:t-jacobian-determinant}
  \end{subfigure}
  ~
  \begin{subfigure}[t]{0.32\textwidth}
    \includegraphics[width=\textwidth]{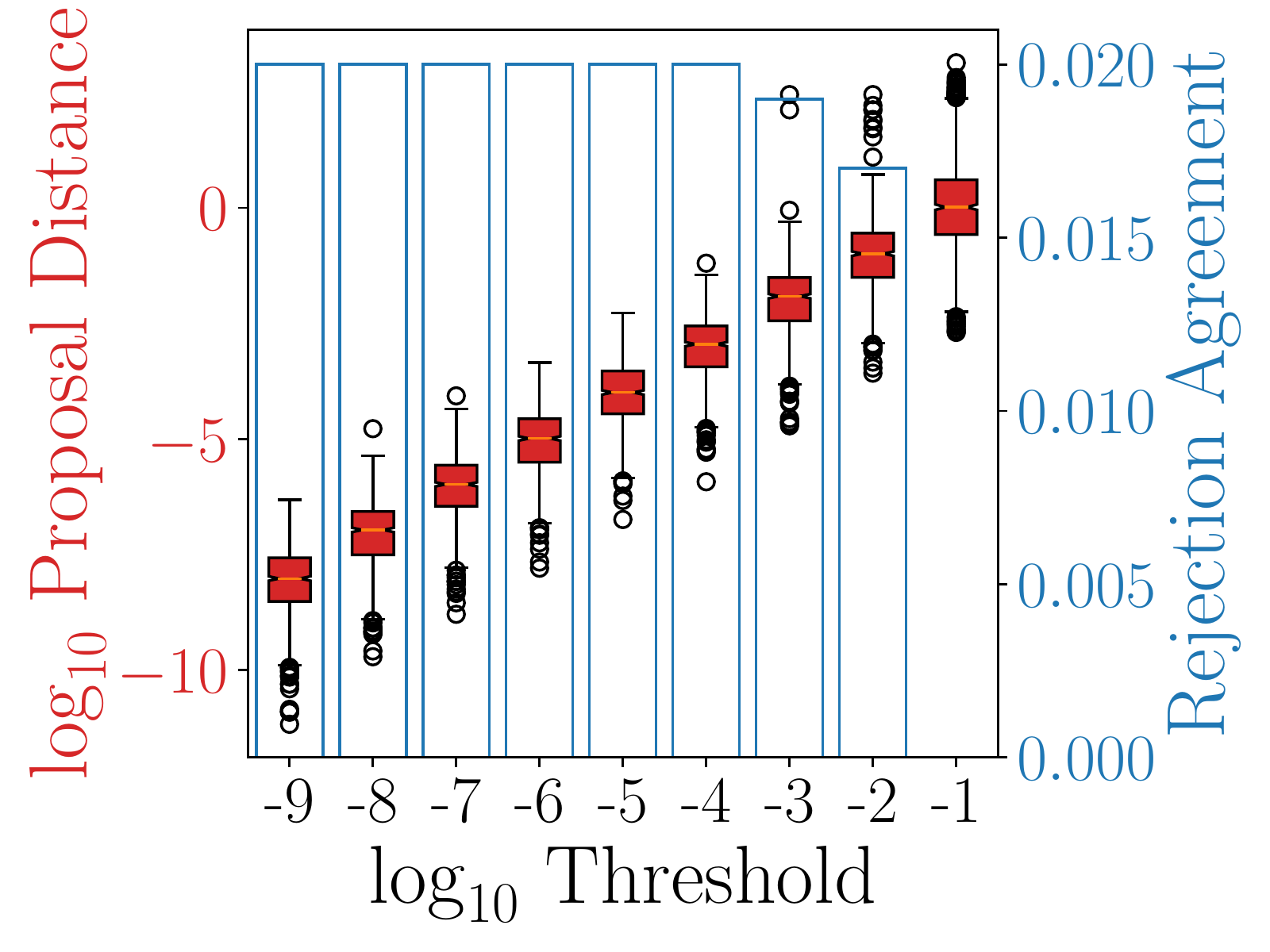}
    \caption{Difference in Transition Kernel}
    \label{subfig:t-transition-difference}
  \end{subfigure}
  \caption{Visualization of the error in reversibility (see \cref{subfig:t-reversibility}), error in volume-preservation (see \cref{subfig:t-jacobian-determinant}), and the number of decimal digits of similarity in transition kernels (see \cref{subfig:t-transition-difference}) for variable thresholds in a multivariate Student-$t$ distribution with $\sigma^2 = 10,000$.}
\end{figure}

\begin{figure}[t!]
  \centering
  \begin{subfigure}[t]{\textwidth}
  \includegraphics[width=\textwidth]{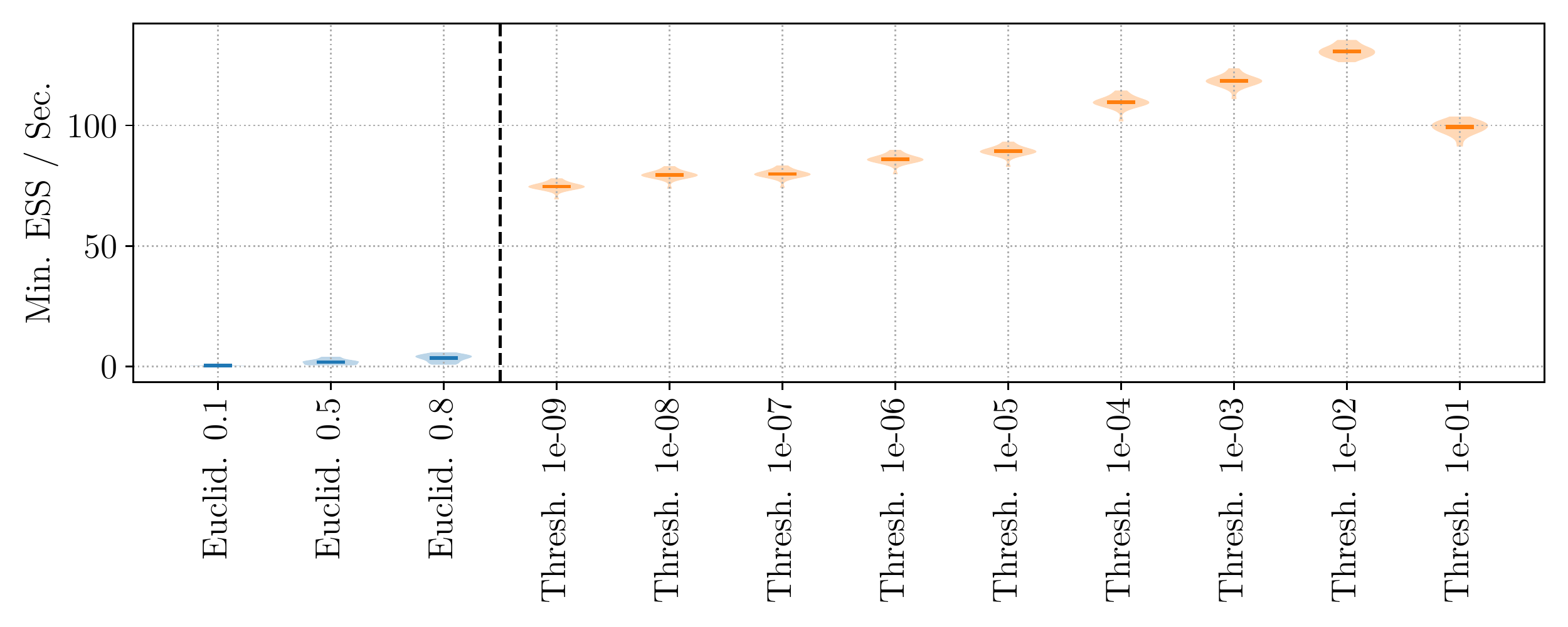}
  \caption{Minimum ESS per second in a multivariate Student-$t$ distribution.}
  \label{subfig:t-ess-per-second}
  \end{subfigure}
  
  \begin{subfigure}[t]{\textwidth}
  \includegraphics[width=\textwidth]{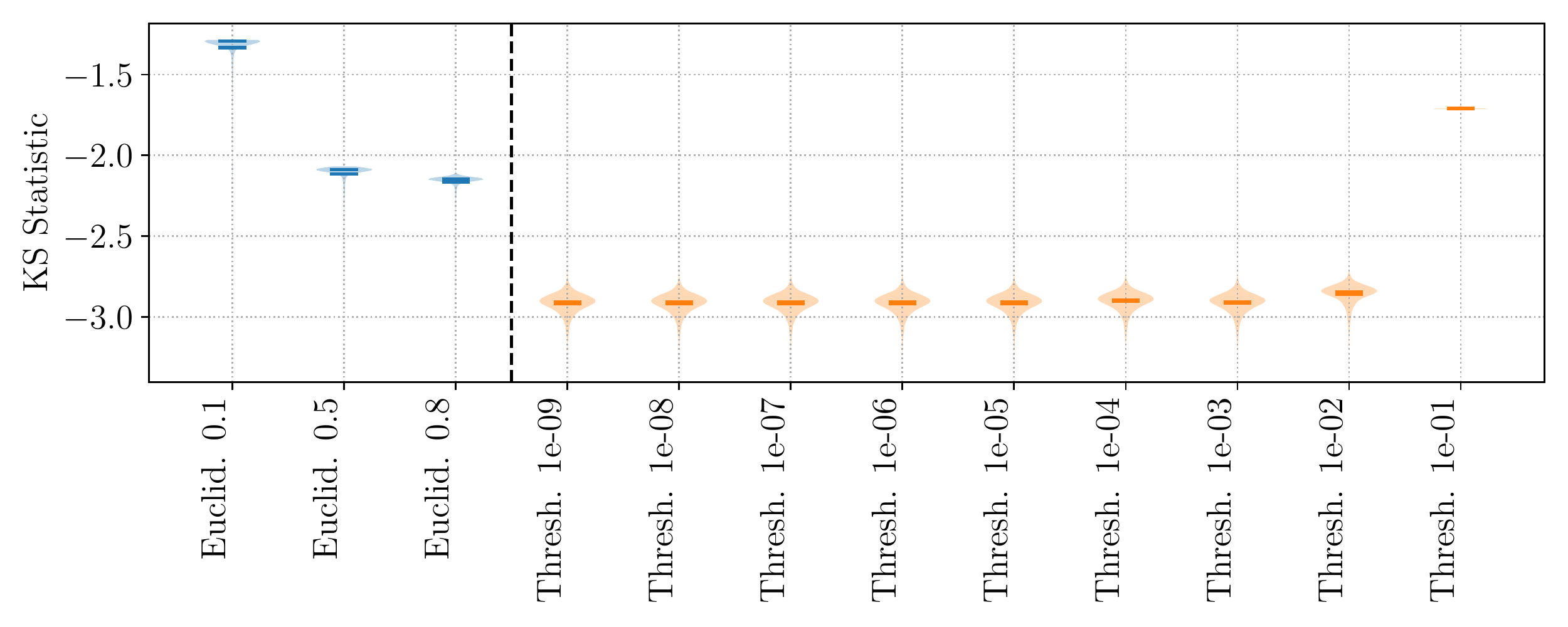}
  \caption{Sample ergodicity in a multivariate Student-$t$ distribution.}
  \label{subfig:t-ergodicity}
  \end{subfigure}
  
  \caption{Visualization of the sample quality of the variables in a multivariate Student-$t$ distribution as measured by the ESS per second (\cref{subfig:t-ess-per-second}) and the distribution of Kolmogorov-Smirnov (KS) statistics over random one-dimensional subspaces (\cref{subfig:t-ergodicity}). Distributions over ESS are computed by splitting a Markov chain of length 1,000,000 into twenty contiguous sequences of length 50,000.}
  \label{fig:t-ess}
\end{figure}

\begin{figure}[t!]
  \begin{subfigure}[t]{0.32\textwidth}
    \centering
    \includegraphics[width=\textwidth]{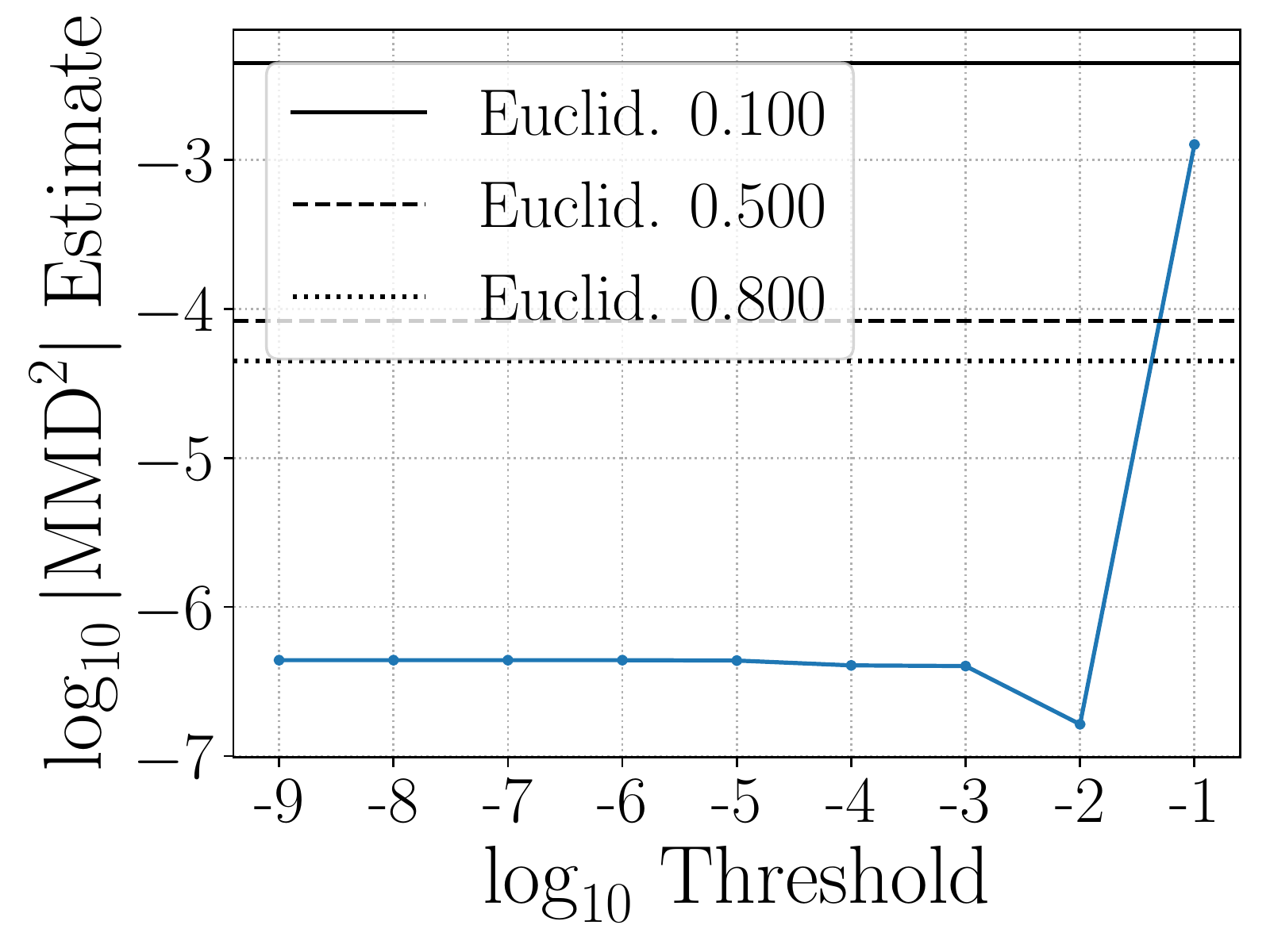}
    \caption{$\mathrm{MMD}_u^2$}
  \end{subfigure}
  ~
  \begin{subfigure}[t]{0.32\textwidth}
    \centering
    \includegraphics[width=\textwidth]{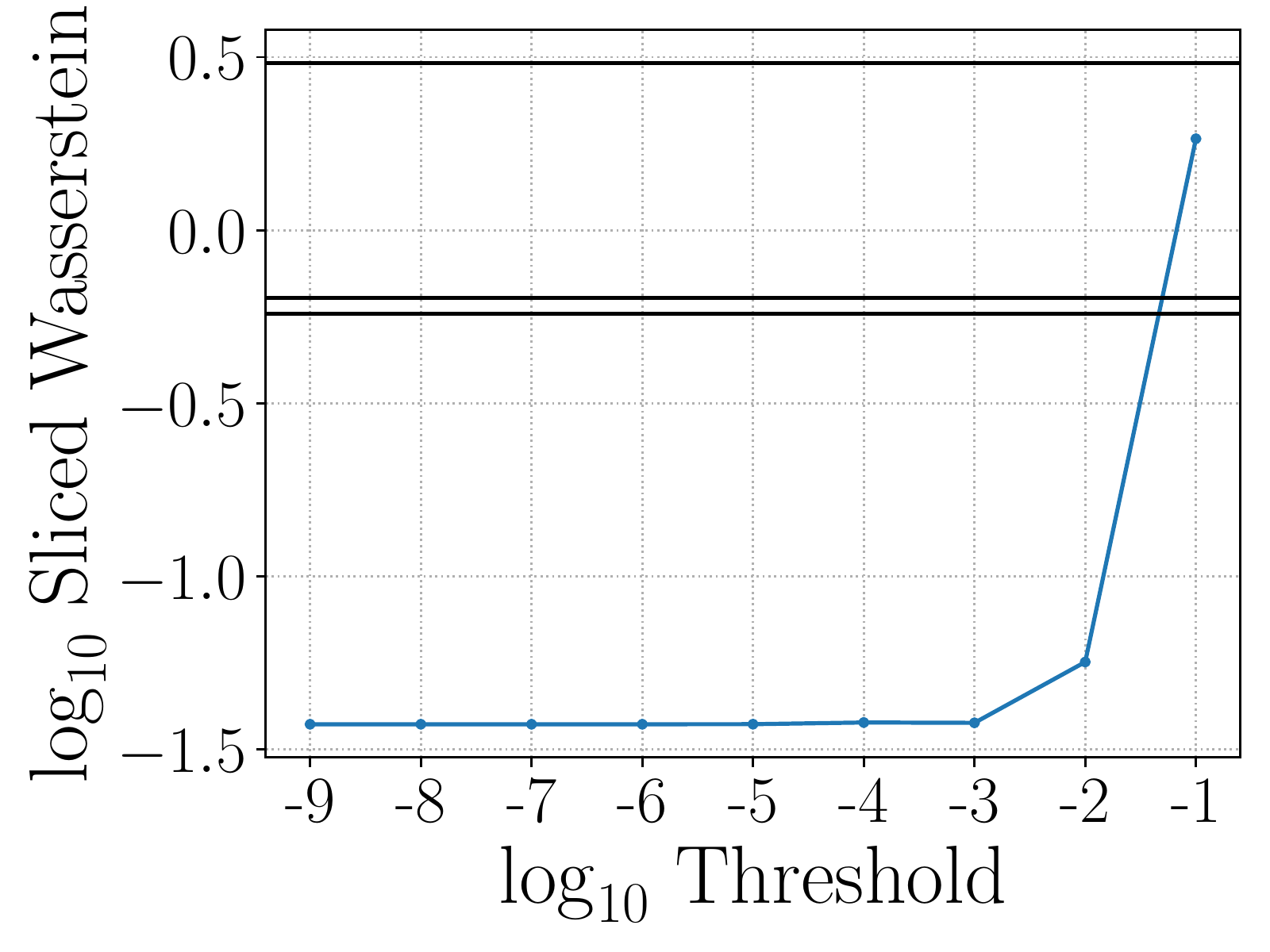}
    \caption{Sliced Wasserstein}
  \end{subfigure}
  ~
  \begin{subfigure}[t]{0.32\textwidth}
    \centering
    \includegraphics[width=\textwidth]{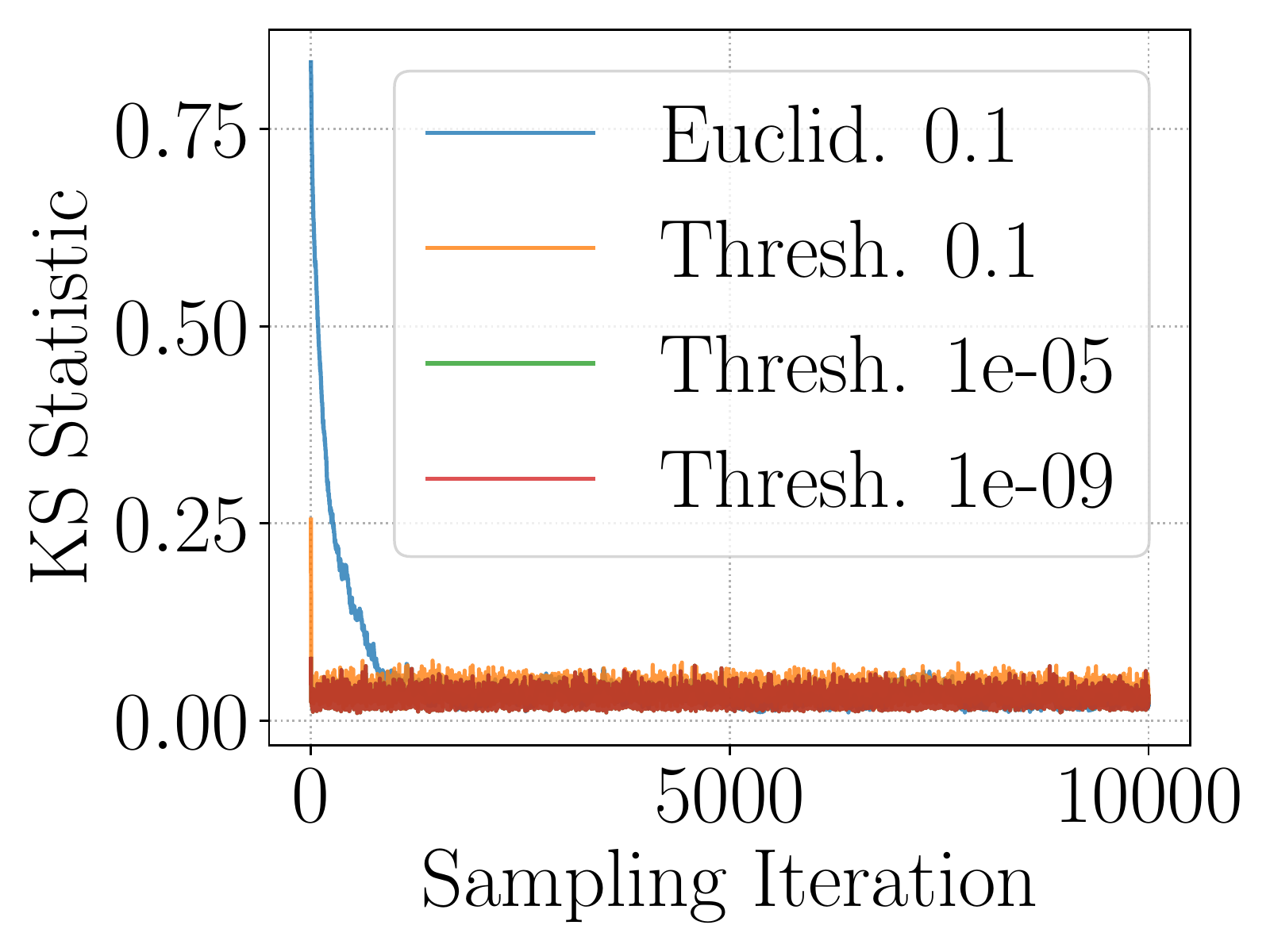}
    \caption{Ergodic Convergence}
  \end{subfigure}
  \caption{Additional measures of ergodicity in the multivariate Student-$t$ distribution. In the presence of the multiscale phenomena, the Euclidean HMC algorithms cannot achieve the same level of ergodicity as the Riemannian algorithms except in the case wherein a very large convergence tolerance is used in RMHMC. When $\sigma^2=1\times 10^4$, we also show convergence of the dimension with largest variation, whose marginal is $t(5, \sigma^2)$. The Riemannian methods exhibit faster convergence compared to HMC. For $\sigma^2=1\times 10^4$, the MMD statistic uses a bandwidth of 110.72 in the squared exponential kernel.}
  \label{fig:t-extra-ergodicity}
\end{figure}

\begin{figure}[t!]
  \begin{subfigure}[t]{0.32\textwidth}
    \centering
    \includegraphics[width=\textwidth]{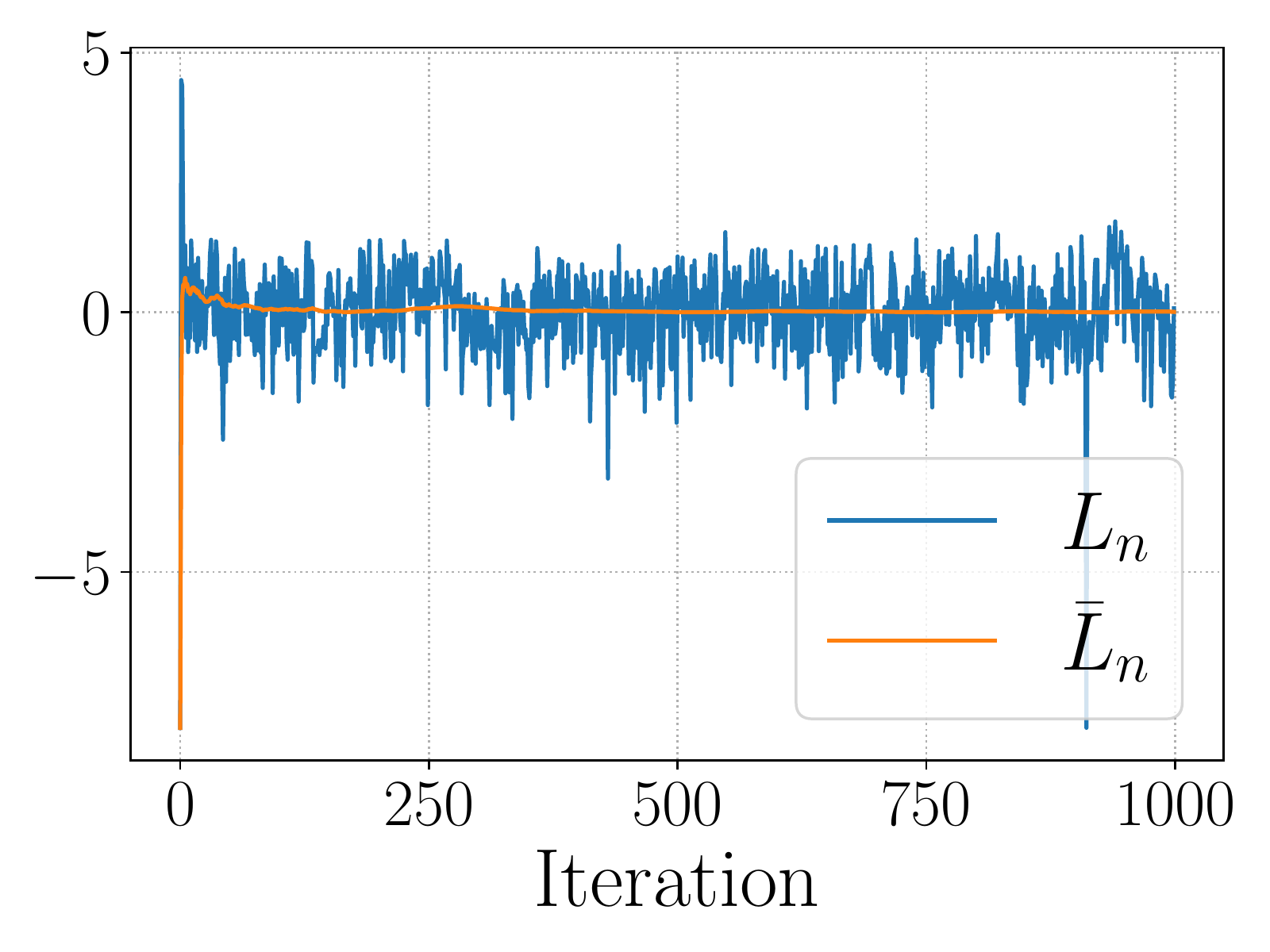}
    \caption{$L_n$ Sequence}
  \end{subfigure}
  ~
  \begin{subfigure}[t]{0.32\textwidth}
    \centering
    \includegraphics[width=\textwidth]{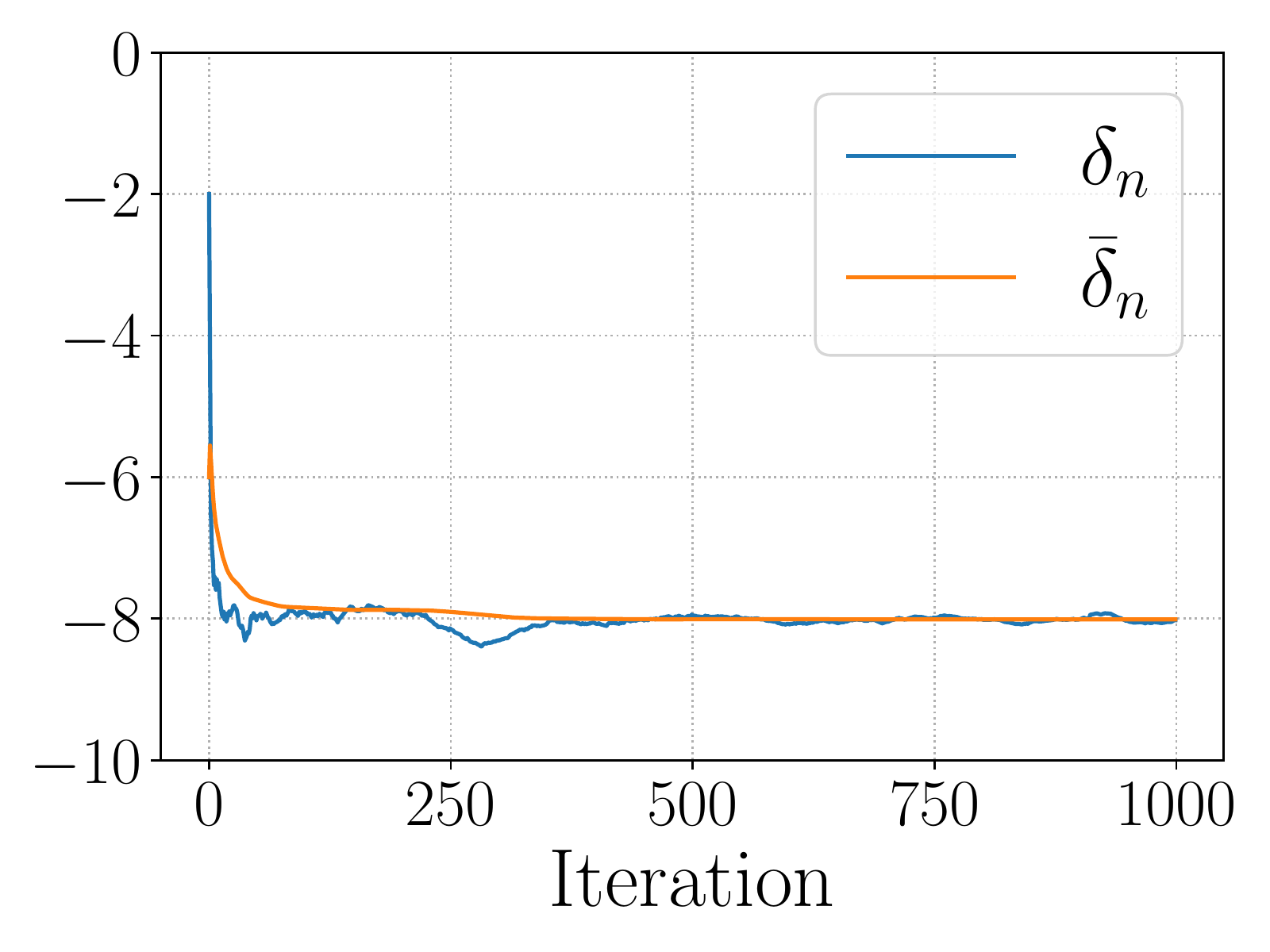}
    \caption{$\delta_n$ Sequence}
  \end{subfigure}
  ~
  \begin{subfigure}[t]{0.32\textwidth}
    \centering
    \includegraphics[width=\textwidth]{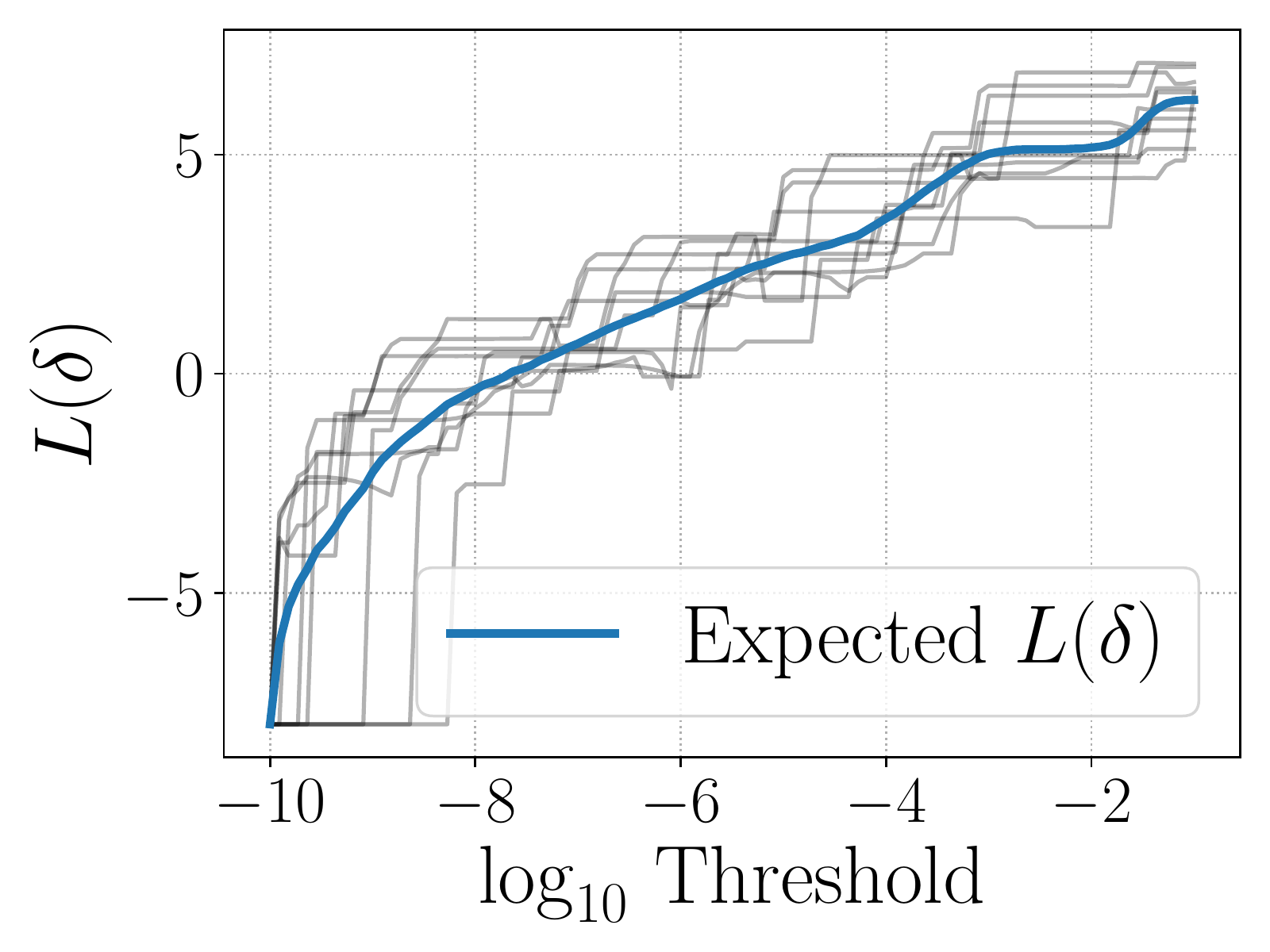}
    \caption{Monte Carlo $B(\delta)$}
  \end{subfigure}
  \caption{The use of Ruppert averaging in the multiscale Student-$t$ distribution to adaptively set the convergence threshold to achieve eight decimal digits of similarity compared to a transition kernel with a threshold of $1\times 10^{-10}$. We show a Monte Carlo approximation to $B(\delta)$, which appears smooth, monotonically increasing. The value of $\delta$ satisfying $B(\delta)=0$ is approximately $\delta = 1\times 10^{-7.83}$, which is somewhat greater than the value of $1\times 10^{-8}$ produced by Ruppert averaging. We note that the function $B(\delta)$ is approximately constant over the interval $[-3, -2]$; this could present difficulty in optimizing thresholds in this range since it violates the fourth assumption for convergence of the Robbins-Monro algorithm.}
  \label{fig:student-t-dual-averaging}
\end{figure}

Following \citet{Lan_2015}, we next consider sampling an $m$-dimensional Student-$t$ distribution with $\nu$ degrees-of-freedom whose log-density function is,
\begin{align}
    \mathcal{L}(q) = -\frac{\nu + m}{2} \log\paren{1 + \frac{1}{\nu} q^\top \Sigma^{-1} q},
\end{align}
where $\Sigma\in\R^{m\times m}$ is a positive definite matrix. The Hessian of the log-density of the multivariate Student-$t$ is,
\begin{align}
    \mathbf{H}(q) = -\frac{\nu + m}{(1 + q^\top \Sigma^{-1} q / \nu) \nu} \Sigma^{-1} + \frac{2(\nu + m)}{\nu^2 (1 + q^\top \Sigma^{-1} q / \nu)^2} (\Sigma^{-1} q) (\Sigma^{-1} q)^\top.
\end{align}
The negative Hessian of the log-density is not positive definite everywhere; however, instead of employing the SoftAbs procedure as in \cref{subsec:experiment-neal-funnel} to produce a positive definite metric, we instead consider the following metric:
\begin{align}
    \mathbf{G}(q) = \frac{\nu + m}{\nu(1 + q^\top \Sigma^{-1} q / \nu)} \Sigma^{-1}.
\end{align}
This metric can be motivated by computing the negative Hessian of the multivariate Student-$t$ log-density and keeping only the term that is guaranteed to be positive definite. The covariance of the Student-$t$ distribution is computed from the matrix $\Sigma$ and the degrees-of-freedom. In particular, when $q$ is drawn from a multivariate Student-$t$ distribution, $\mathrm{Cov}(q) = (\nu / (\nu - 2)) \Sigma$. In our experiments we set $\nu = 5$, $m=20$, and $\Sigma = \mathrm{diag}(1, \ldots, 1, \sigma^2)$; we consider values of $\sigma^2 \in \set{1\times 10^0, 1\times 10^1, 1\times 10^2, 1\times 10^3, 1\times 10^4}$. By examining variable values of $\sigma^2$, we force $\Sigma$ to incorporate increasingly severe multiscale phenomena; this can be related to the condition number of the matrix $\Sigma$, which is the ratio of the largest eigenvalue to the smallest eigenvalue. We expect Euclidean variants of HMC to struggle in the presence of multiscale distributions \citep{pourzanjani2019implicit,Lan_2015}.

For the Euclidean methods we consider twenty integration steps and a step-size in $\set{0.1, 0.5, 0.8}$. For Riemannian methods we also consider twenty integration steps and a fixed step-size of $0.3$. In \cref{fig:t-multiscale-phenomena}, we visualize the effect of the condition number on the ergodicity and effective sample size per second metrics. We show results for a Riemannian method with threshold fixed at $1\times 10^{-5}$ and report results for the best performing Euclidean algorithm among the step-sizes considered. For condition numbers of 100, 1,000, and 10,000, we see that the Euclidean algorithms exhibit degraded performance. Whereas the effective sample size per second of the Riemannian methods is effectively constant as a function of the condition number, Euclidean methods are able to generate fewer effective samples per second as the covariance matrix becomes increasingly ill-conditioned. This is also reflected in our measure of ergodicity, wherein large condition numbers result in worse sampling from the target distribution.

In the sequel, we fix $\sigma^2= 1\times 10^4$. \Cref{subfig:t-ess-per-second} shows the time-normalized ESS; we observe that RMHMC, irrespective of convergence threshold, has superior sampling efficiency. In \cref{subfig:t-ergodicity} we show the sensitivity of the ergodicity of RMHMC to the choice of threshold and compare it against the Euclidean algorithms with variable step-size. In sampling from this distribution, we observe that a threshold of $1\times 10^{-2}$ is sufficient to achieve a measure of ergodicity that is comparable to a threshold of $1\times 10^{-9}$. We further note that the Riemannian algorithm with the largest step-size achieves a measure of ergodicity comparable to the {\it best} Euclidean HMC implementation. In \cref{subfig:t-reversibility} we visualize the reversibility of the proposal: we observe that a threshold of $1\times 10^{-2}$ exhibits around one decimal digit of error in reversibility. In \cref{subfig:t-ess-per-second} we repeat this analysis for the effective sample size per second. Since smaller convergence thresholds require strictly greater computational effort, we observe that, for thresholds less than or equal to $1\times 10^{-2}$, the effective sample size per second is a decreasing function of threshold.

\subsection{Alternatives to Fixed Point Iteration: Newton's Method}\label{subsec:alternatives-newtons-method}
  
\begin{figure}[t!]
  \begin{subfigure}[t]{0.3\textwidth}
    \centering
    \includegraphics[width=\textwidth]{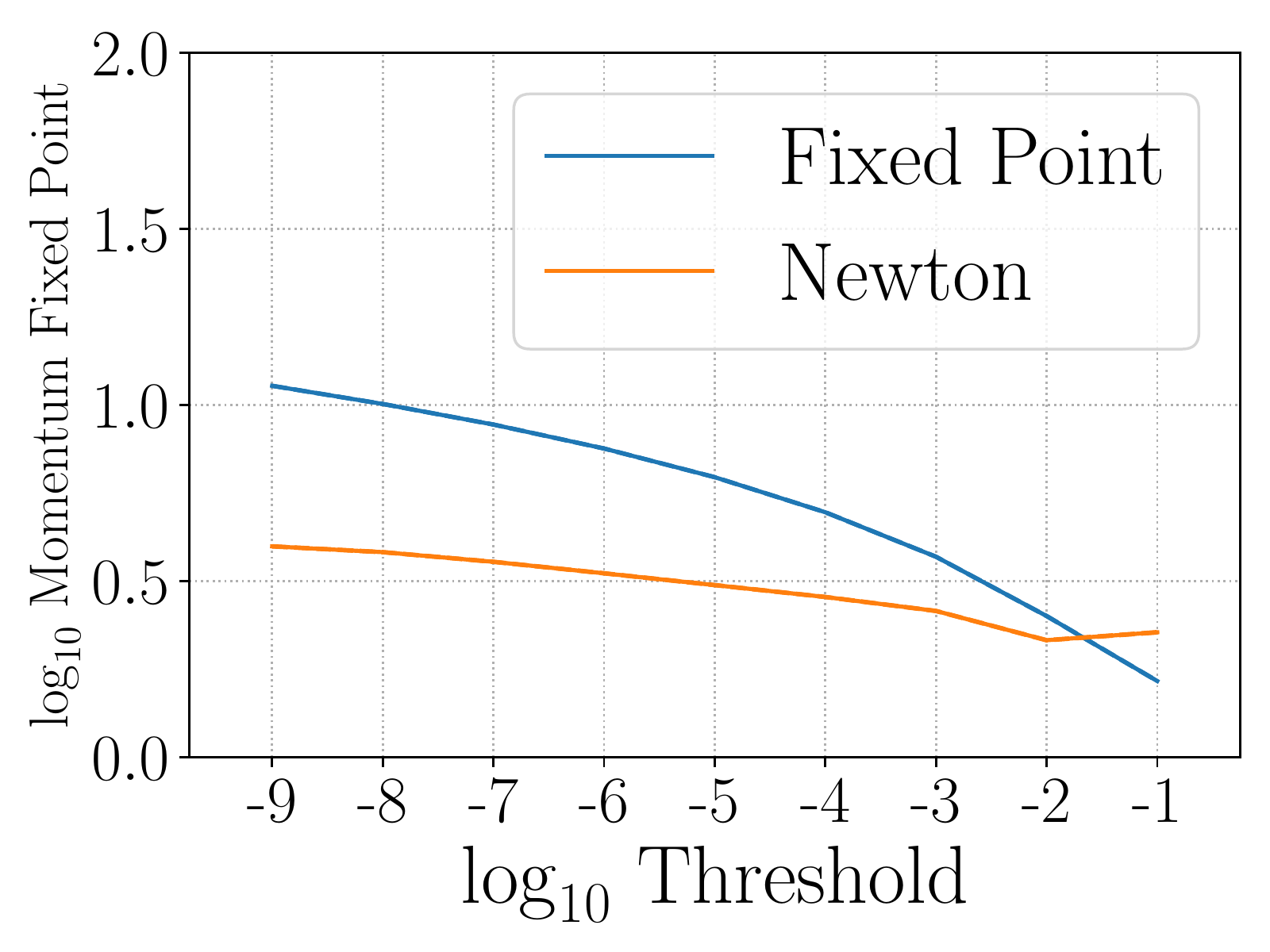}
    \caption{Banana}
    \label{subfig:fixed-point-vs-newton-banana}
  \end{subfigure}
  ~
  \begin{subfigure}[t]{0.3\textwidth}
    \centering
    \includegraphics[width=\textwidth]{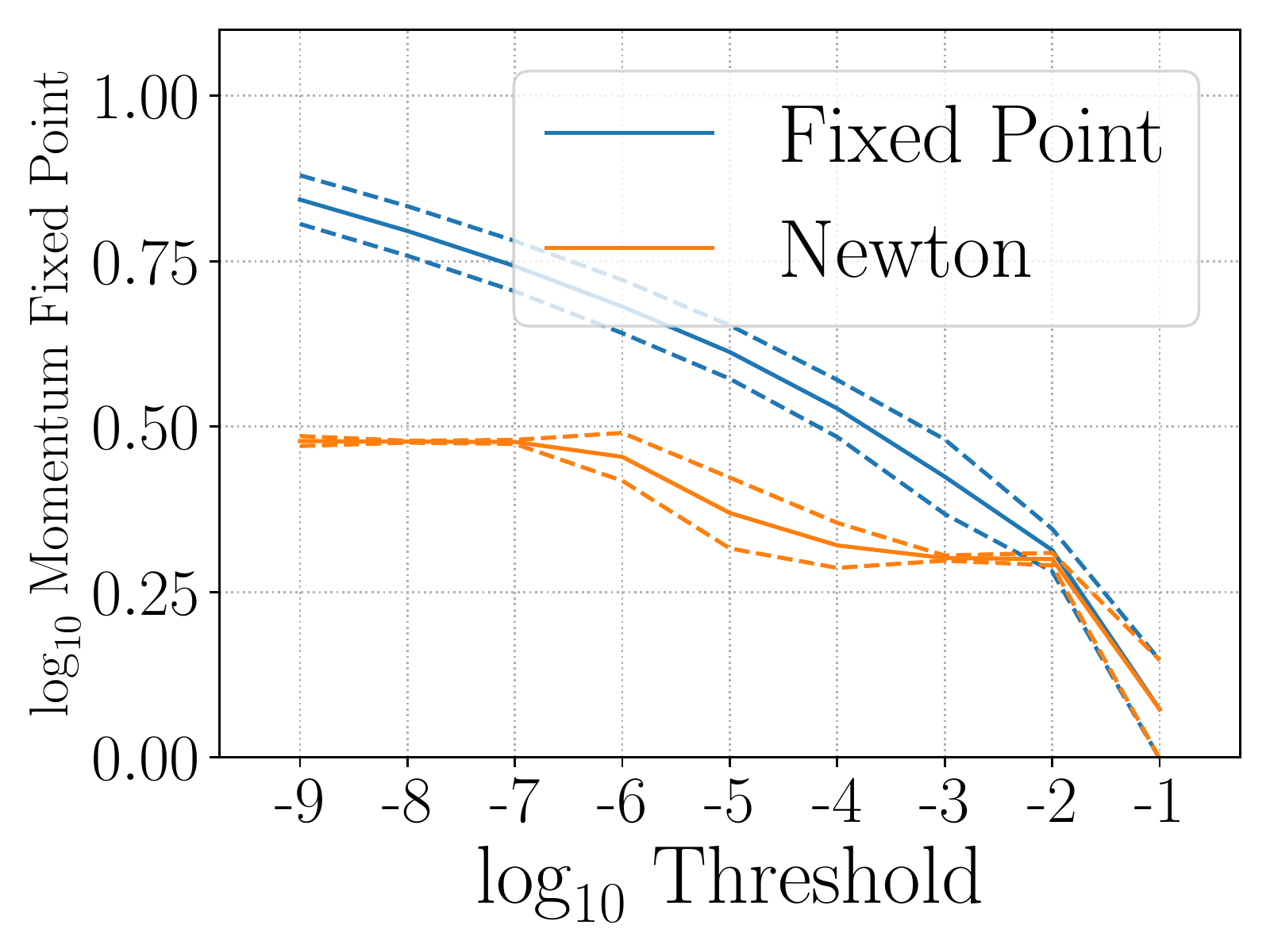}
    \caption{Logistic Regression}
    \label{subfig:fixed-point-vs-newton-logistic}
  \end{subfigure}
  ~
  \begin{subfigure}[t]{0.3\textwidth}
    \centering
    \includegraphics[width=\textwidth]{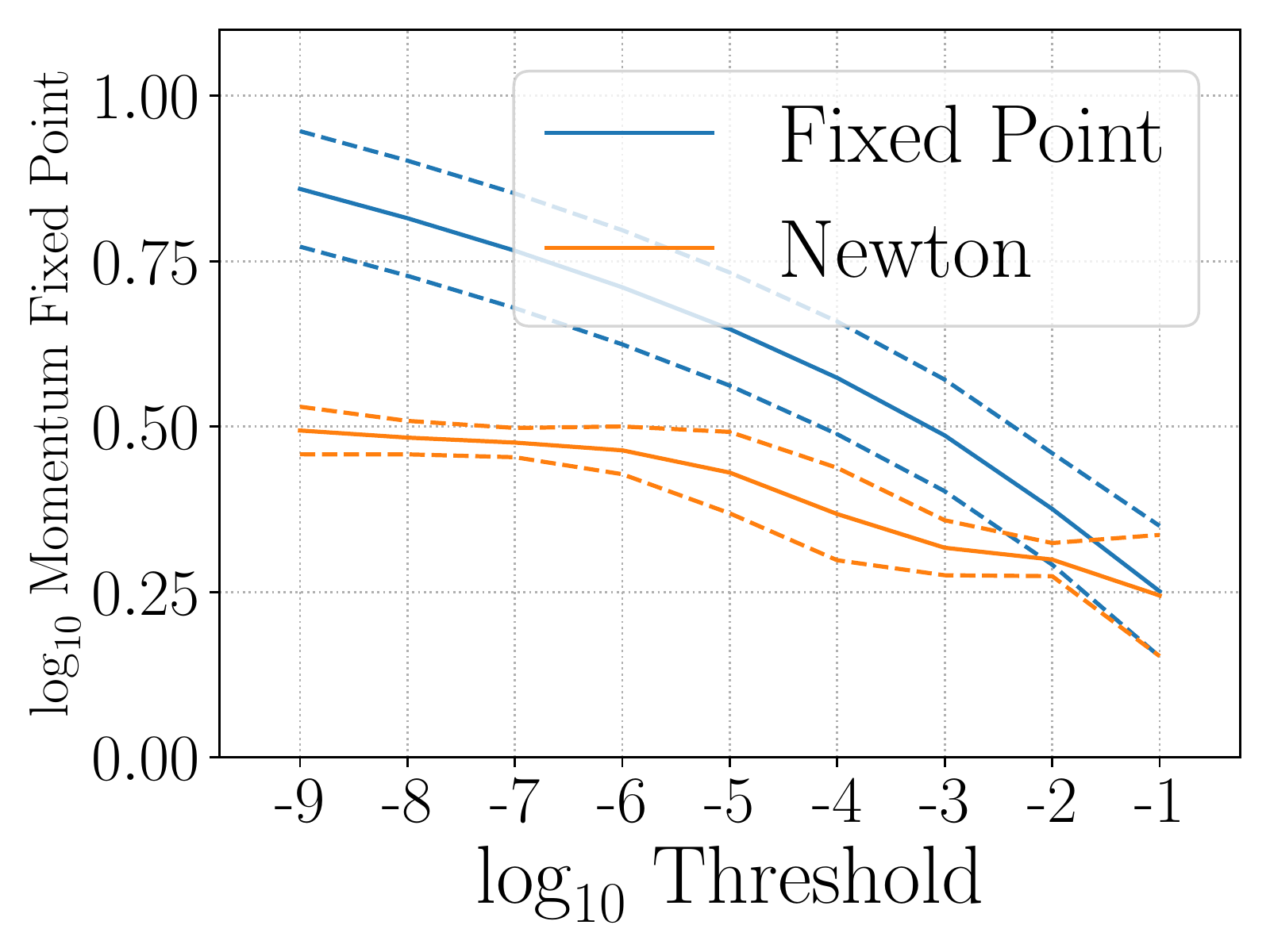}
    \caption{Stochastic Volatility}
    \label{subfig:fixed-point-vs-newton-stochastic-volatility}
  \end{subfigure}
  
  \begin{subfigure}[t]{0.3\textwidth}
    \centering
    \includegraphics[width=\textwidth]{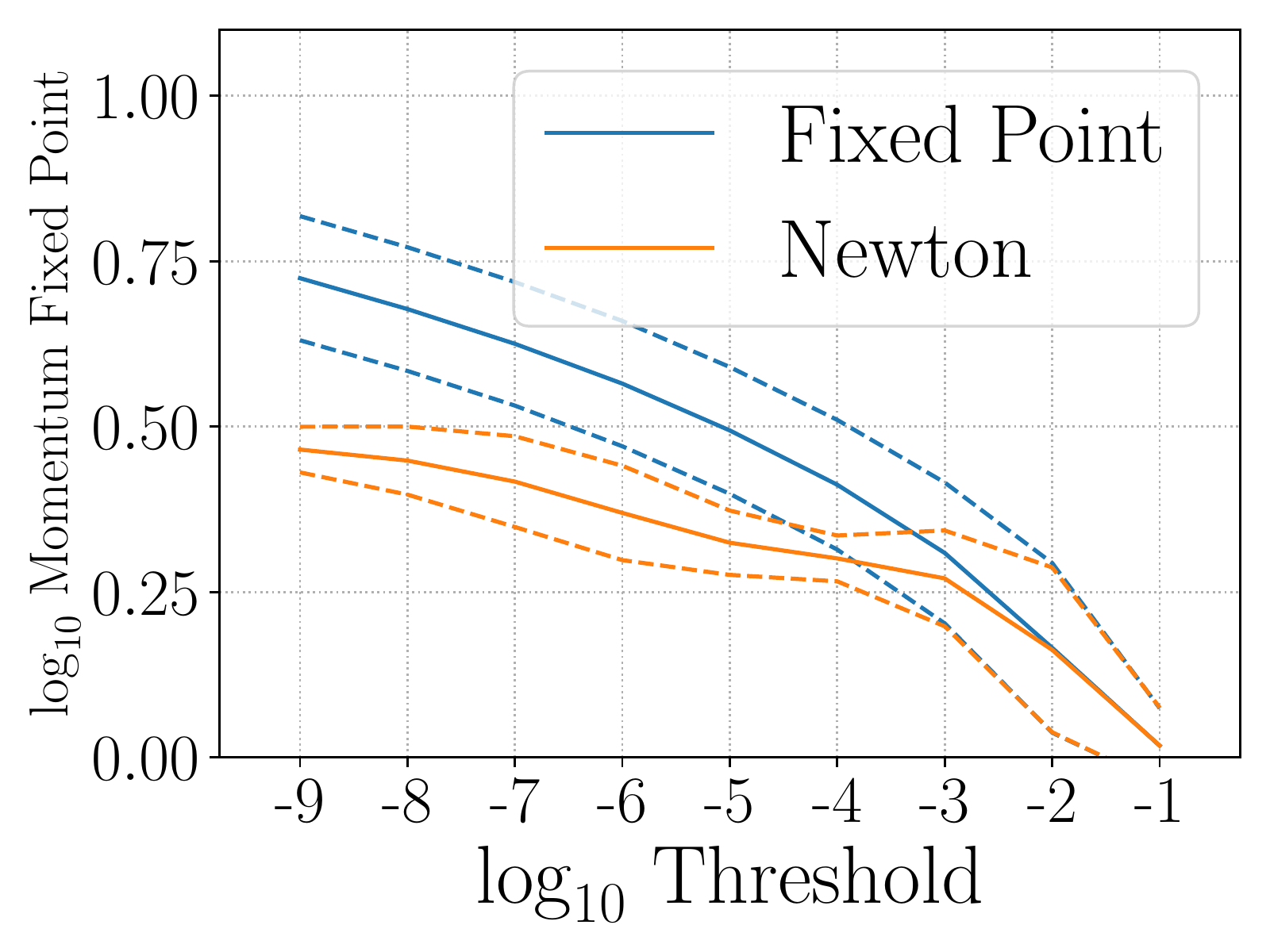}
    \caption{Cox-Poisson}
    \label{subfig:fixed-point-vs-newton-cox-poisson}
  \end{subfigure}
  ~
  \begin{subfigure}[t]{0.3\textwidth}
    \centering
    \includegraphics[width=\textwidth]{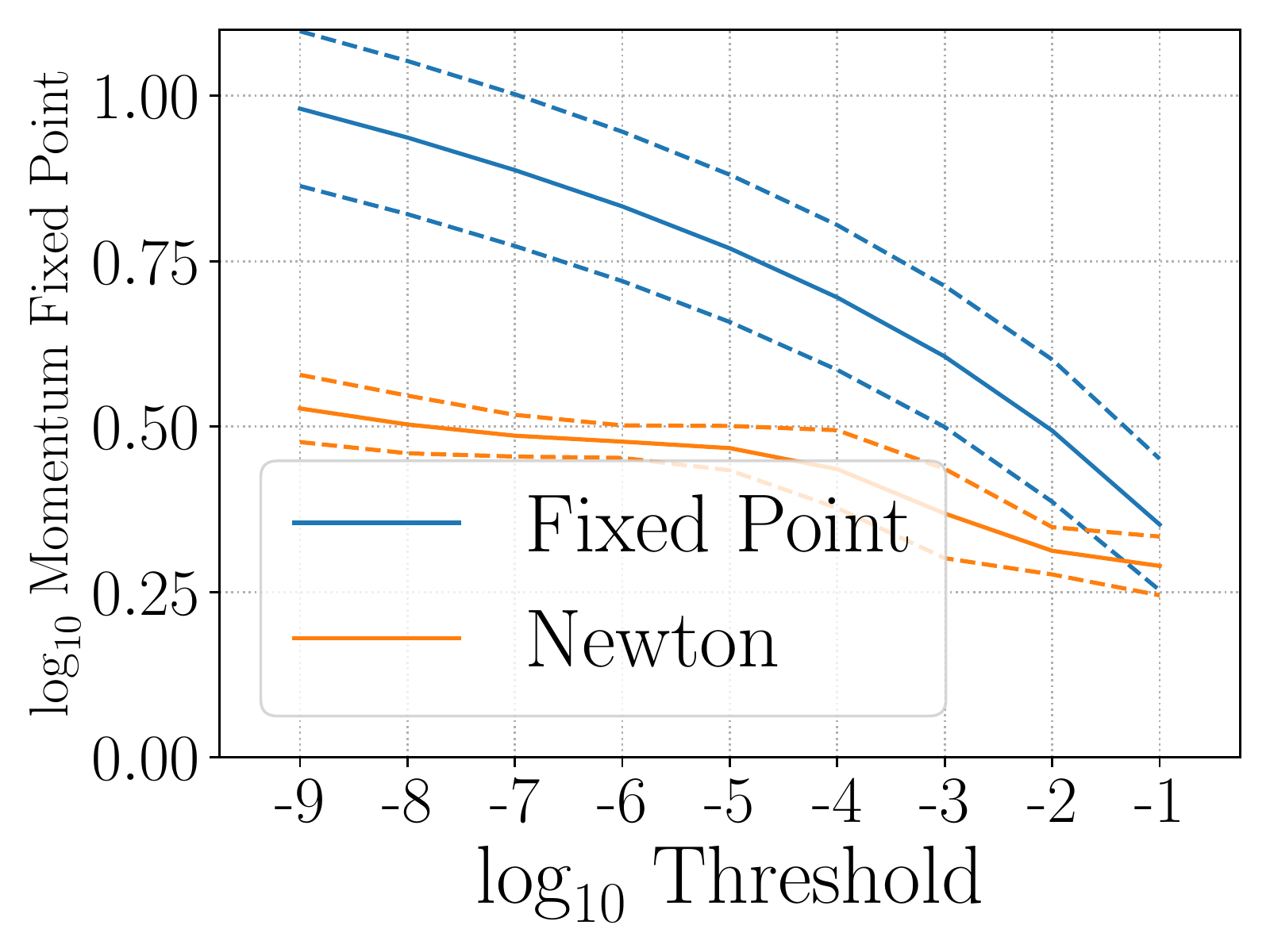}
    \caption{Fitzhugh-Nagumo}
    \label{subfig:fixed-point-vs-newton-fitzhugh-nagumo}
  \end{subfigure}
  ~
  \begin{subfigure}[t]{0.3\textwidth}
    \centering
    \includegraphics[width=\textwidth]{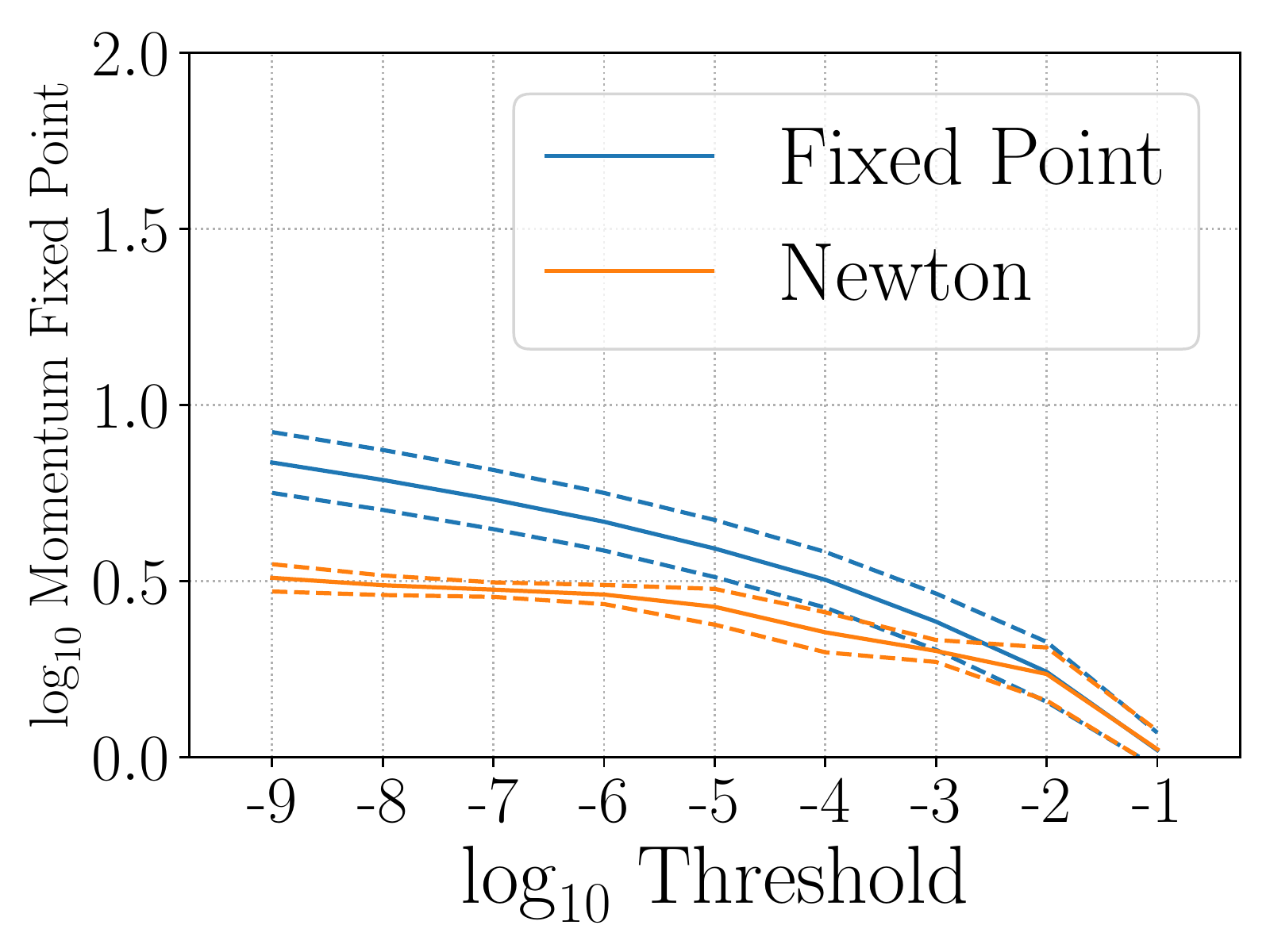}
    \caption{Student-$t$}
    \label{subfig:fixed-point-vs-newton-student-t}
  \end{subfigure}

  \caption{We compute the number of iterations required to resolve the implicit updates to the momentum variable when using either fixed point iteration or Newton's method. We report the mean number of iterations with error bars showing dispersion about the mean. In each circumstance, Newton's method converges faster than fixed point iterations, reflecting its superior order of convergence. Only around three iterations of Newton's method are required on average in each example. We observe that in many cases, Newton's method appears to be less sensitive to the convergence tolerance.}
  \label{fig:fixed-point-vs-newton-iterations}
\end{figure}

\begin{figure}[t!]
  \centering
  \begin{subfigure}[t]{\textwidth}
    \includegraphics[width=\textwidth]{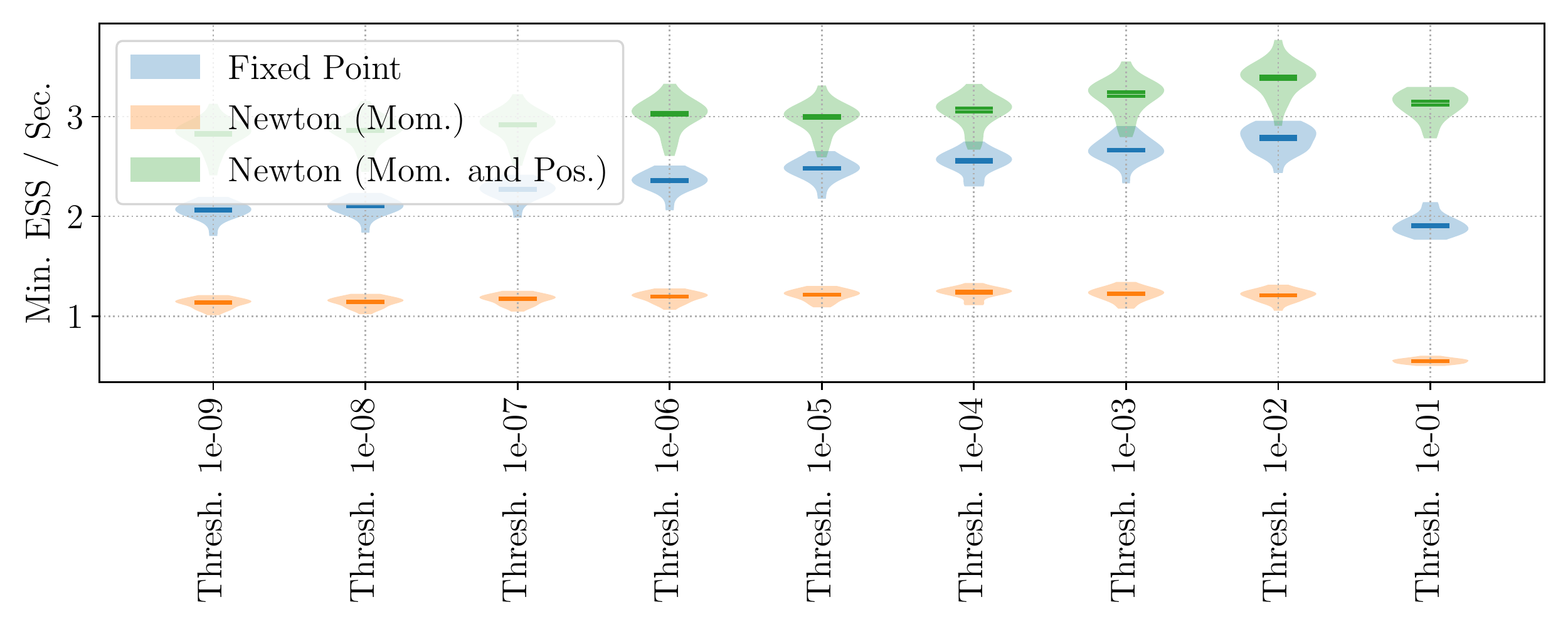}
    \caption{Banana}
    \label{subfig:fixed-point-vs-newton-ess-wose-banana}
  \end{subfigure}

  \begin{subfigure}[t]{\textwidth}
    \includegraphics[width=\textwidth]{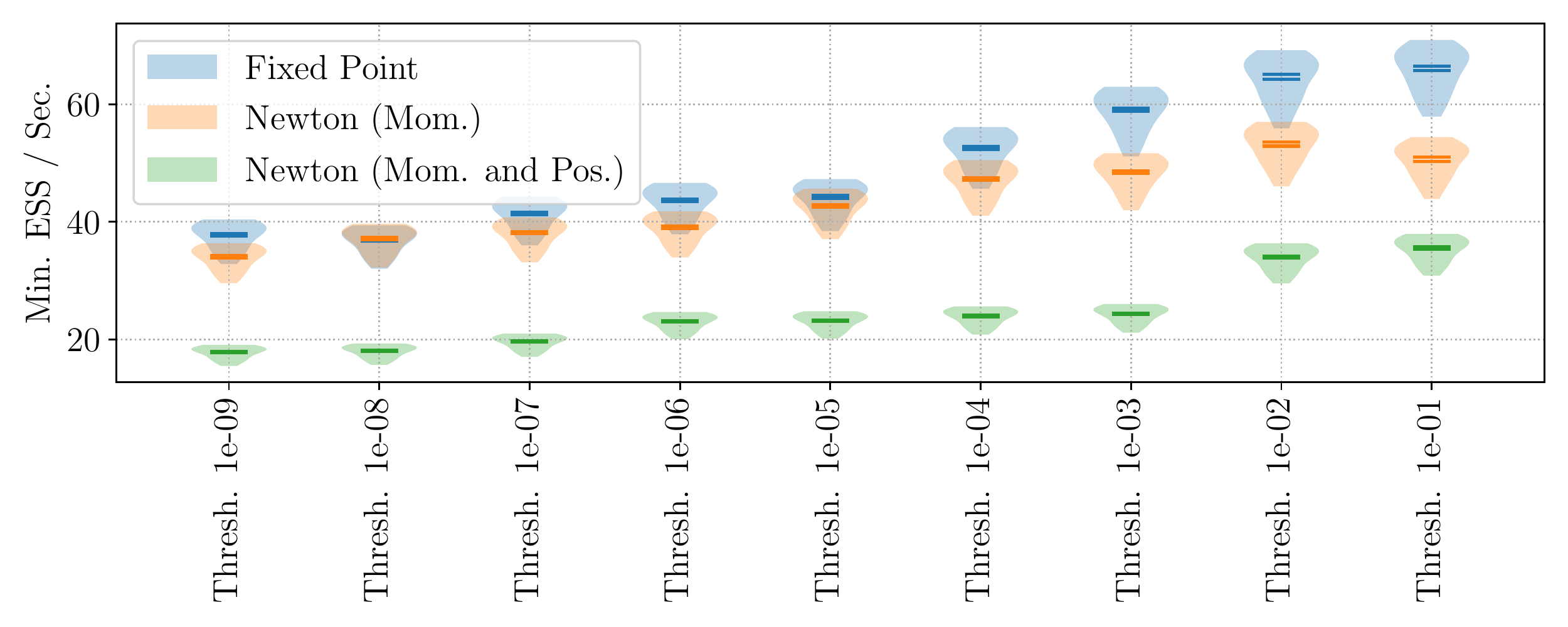}
    \caption{Logistic Regression}
  \end{subfigure}

  \begin{subfigure}[t]{\textwidth}
    \includegraphics[width=\textwidth]{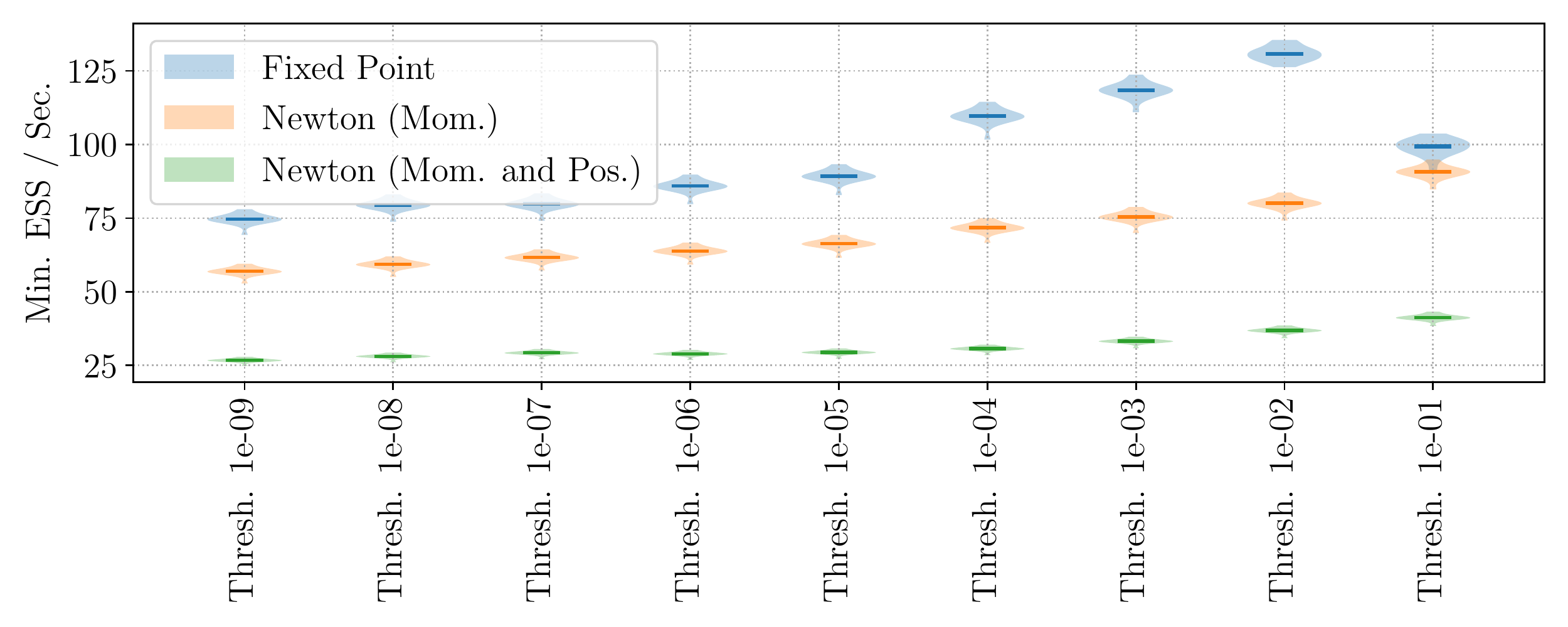}
    \caption{Student-$t$}
  \end{subfigure}
  
  \caption{Three circumstances are shown wherein the use of Newton's method to resolve the implicit update to the momentum resulted in degraded sampling efficiency due to the higher computational burden of Newton's method compared to fixed point iteration. In the banana-shaped distribution, although the posterior is only two-dimensional, significant computation is wasted while searching for solutions to the momentum update that do not exist.}
  \label{fig:fixed-point-vs-newton-ess-worse}
\end{figure}
\begin{figure}[t!]
  \centering
  \begin{subfigure}[t]{0.32\textwidth}
    \includegraphics[width=\textwidth]{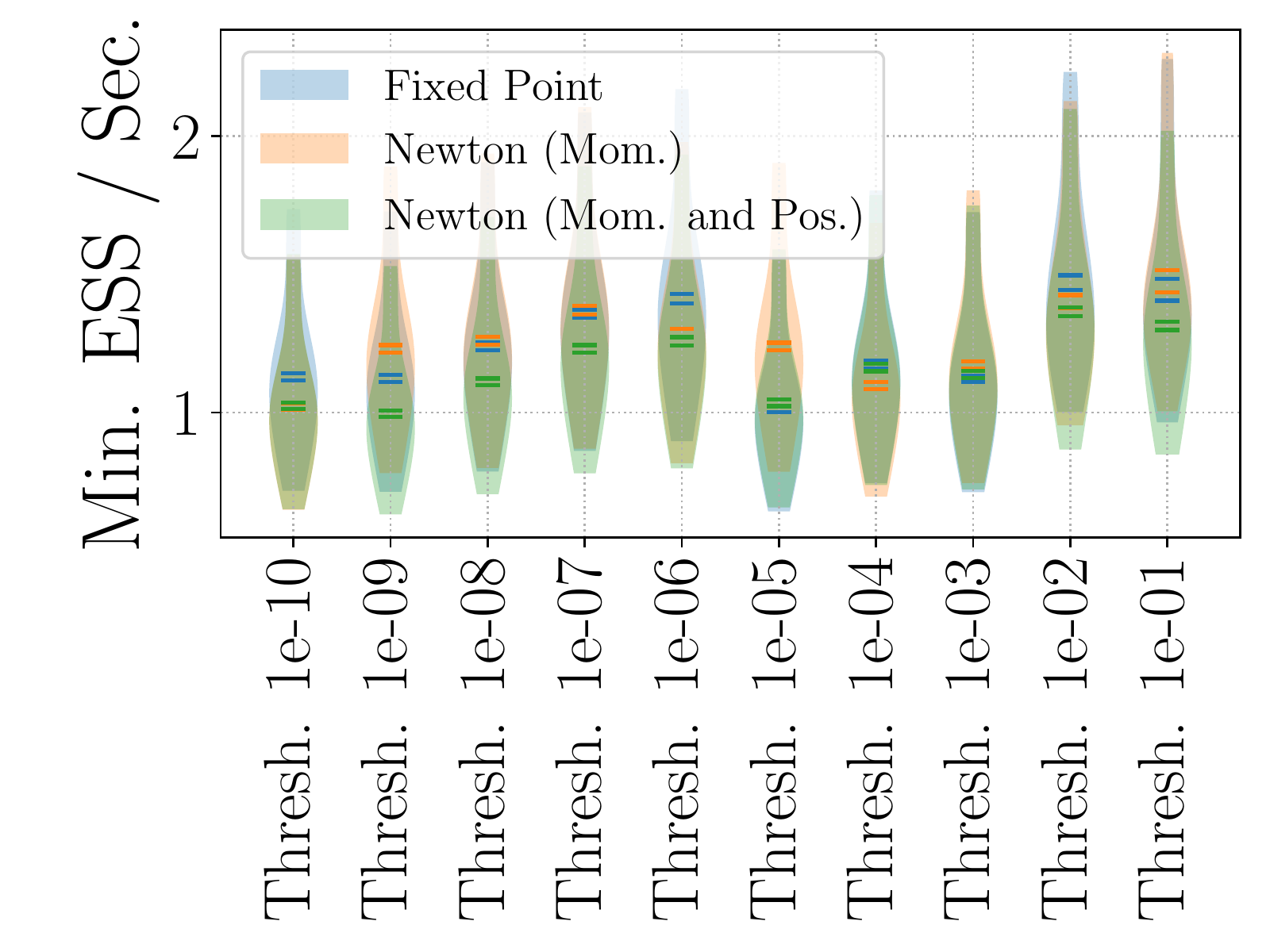}
    \caption{(SV) $\phi$}
  \end{subfigure}
  ~
  \begin{subfigure}[t]{0.32\textwidth}
    \includegraphics[width=\textwidth]{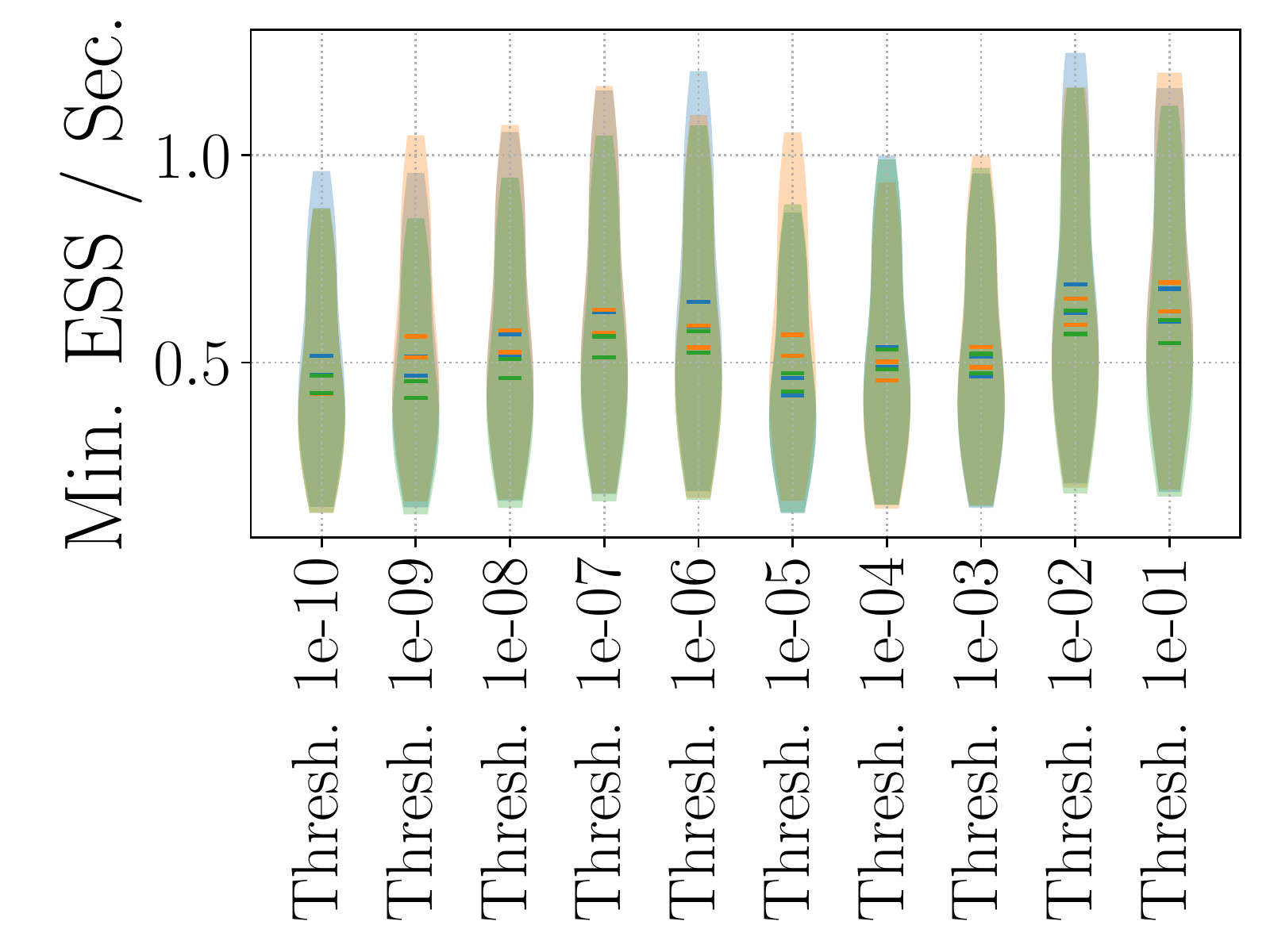}
    \caption{(SV) $\beta$}
  \end{subfigure}
  ~
  \begin{subfigure}[t]{0.32\textwidth}
    \includegraphics[width=\textwidth]{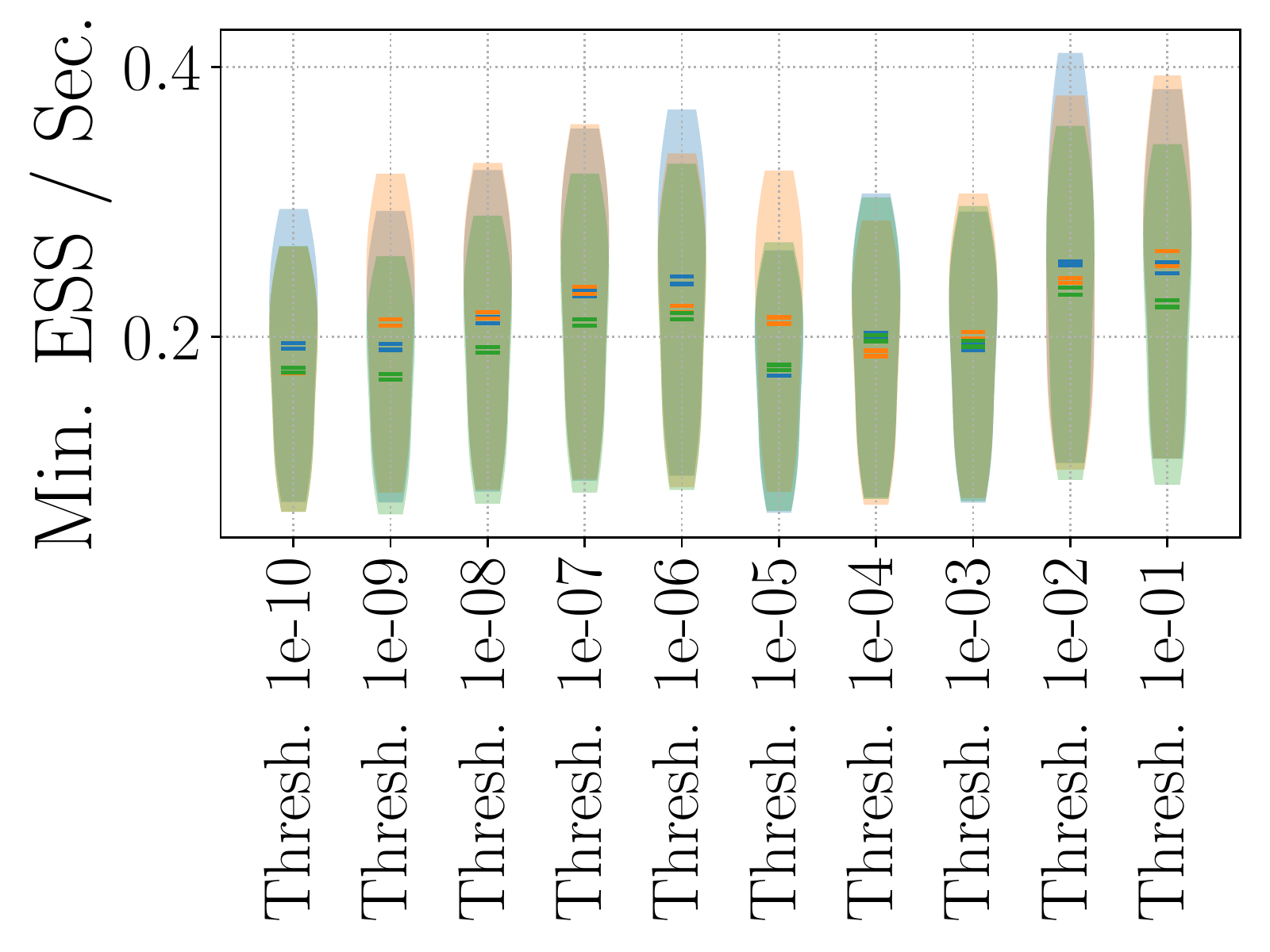}
    \caption{(SV) $\sigma^2$}
  \end{subfigure}
  
  \begin{subfigure}[t]{0.49\textwidth}
    \includegraphics[width=\textwidth]{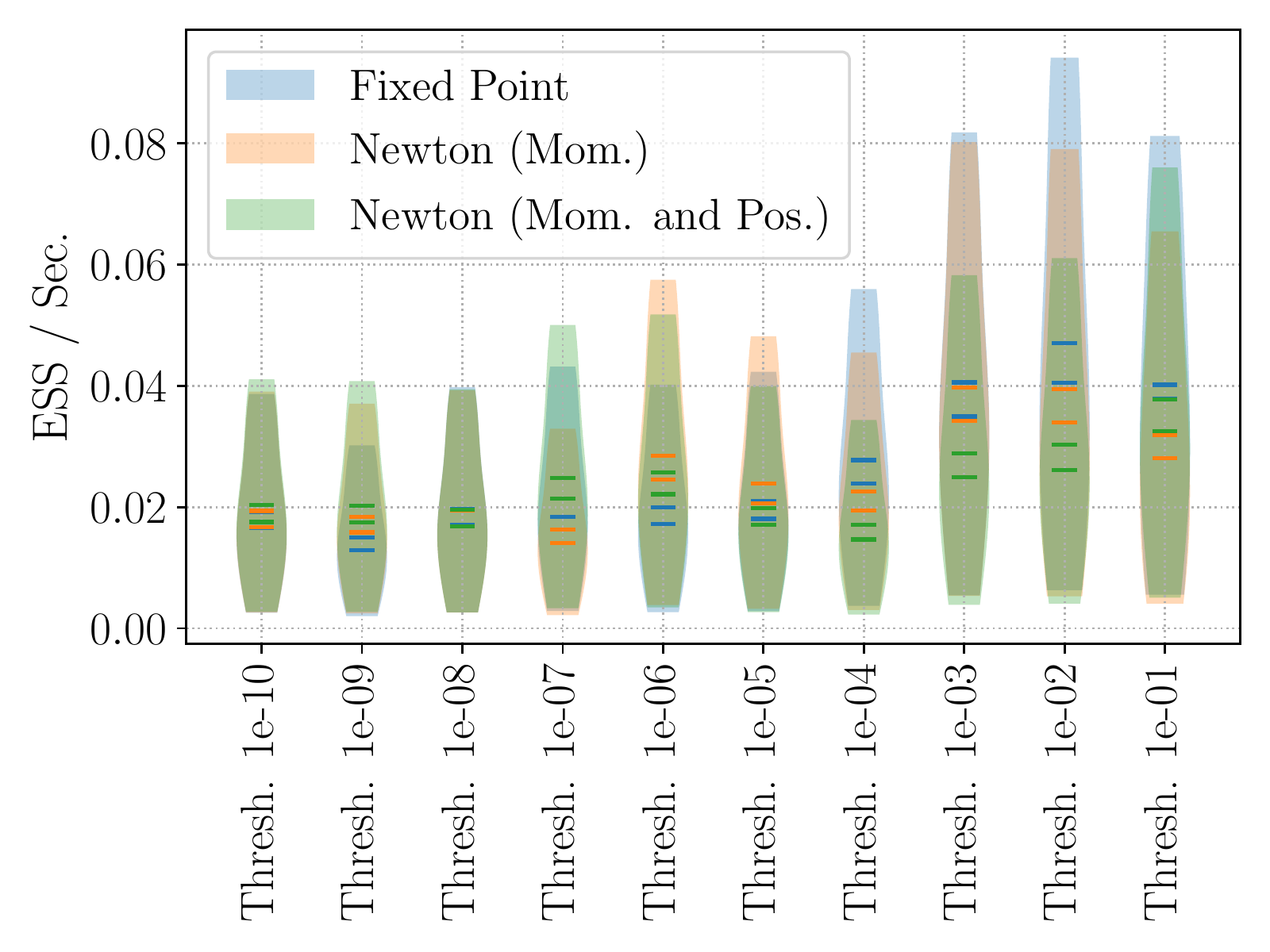}
    \caption{(CP) $\beta$}
  \end{subfigure}
  ~
  \begin{subfigure}[t]{0.49\textwidth}
    \includegraphics[width=\textwidth]{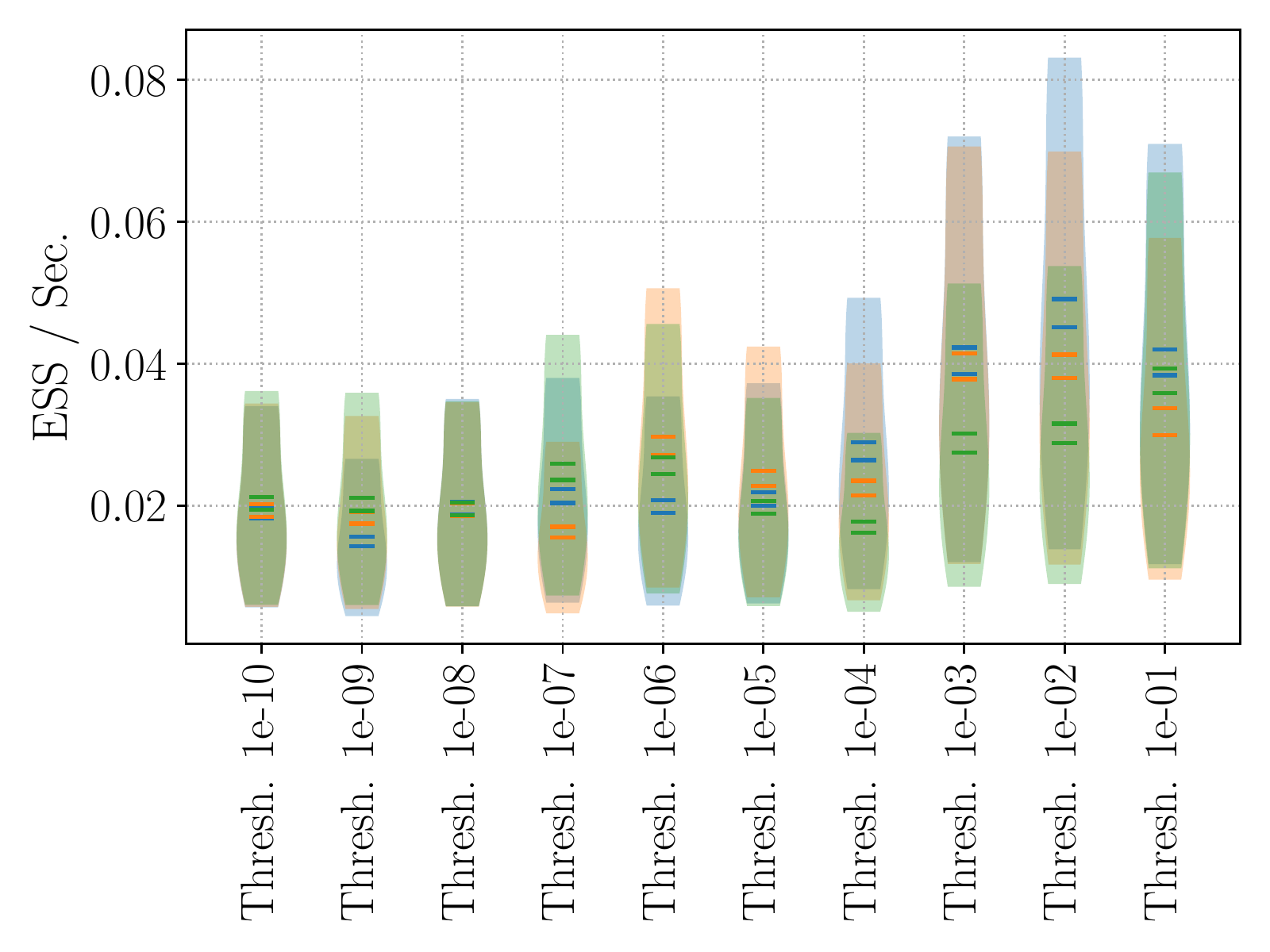}
    \caption{(CP) $\sigma$}
  \end{subfigure}
  
  \begin{subfigure}[t]{0.32\textwidth}
    \includegraphics[width=\textwidth]{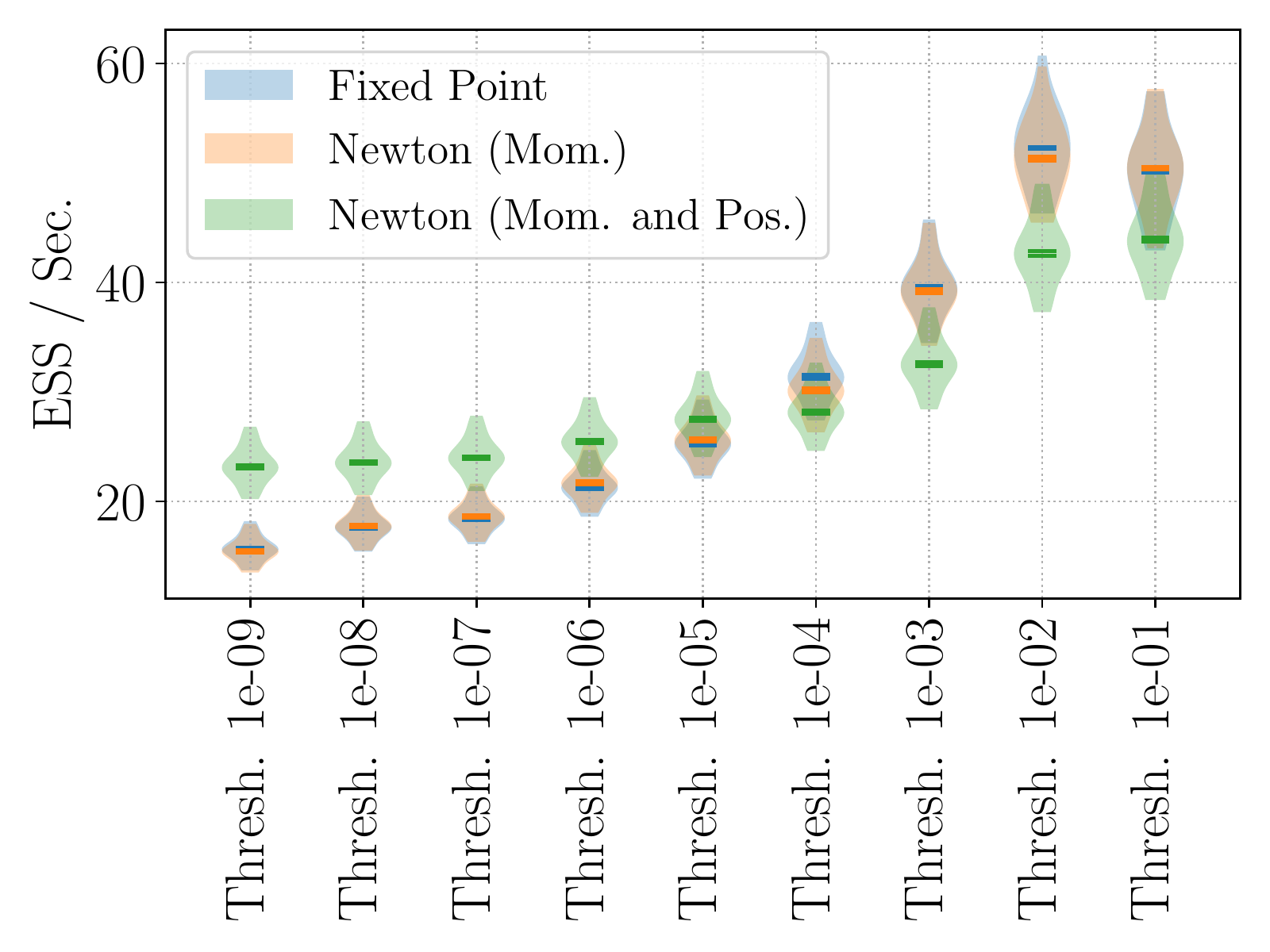}
    \caption{(FN) $a$}
  \end{subfigure}
  ~
  \begin{subfigure}[t]{0.32\textwidth}
    \includegraphics[width=\textwidth]{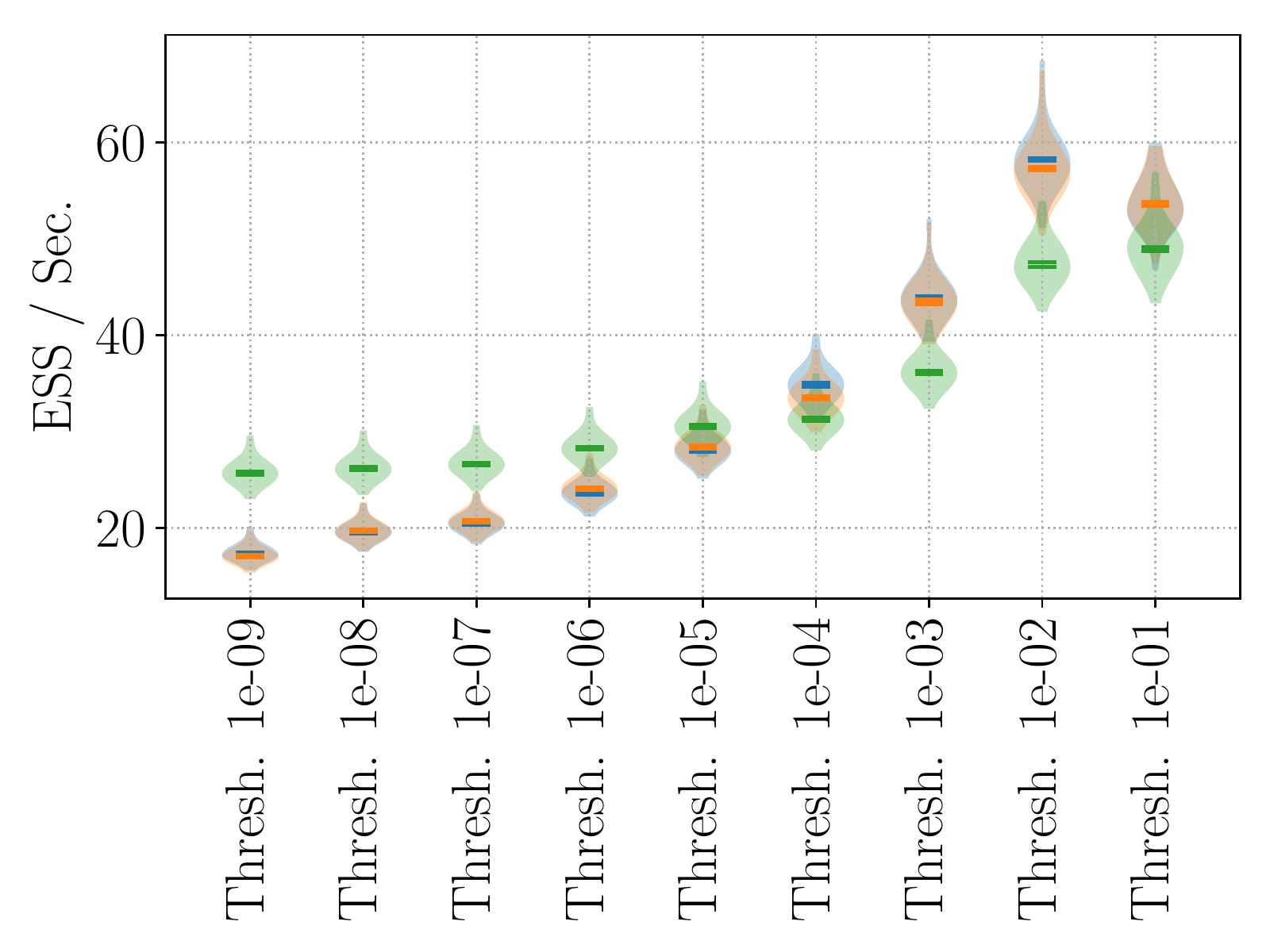}
    \caption{(FN) $b$}
  \end{subfigure}
  ~
  \begin{subfigure}[t]{0.32\textwidth}
    \includegraphics[width=\textwidth]{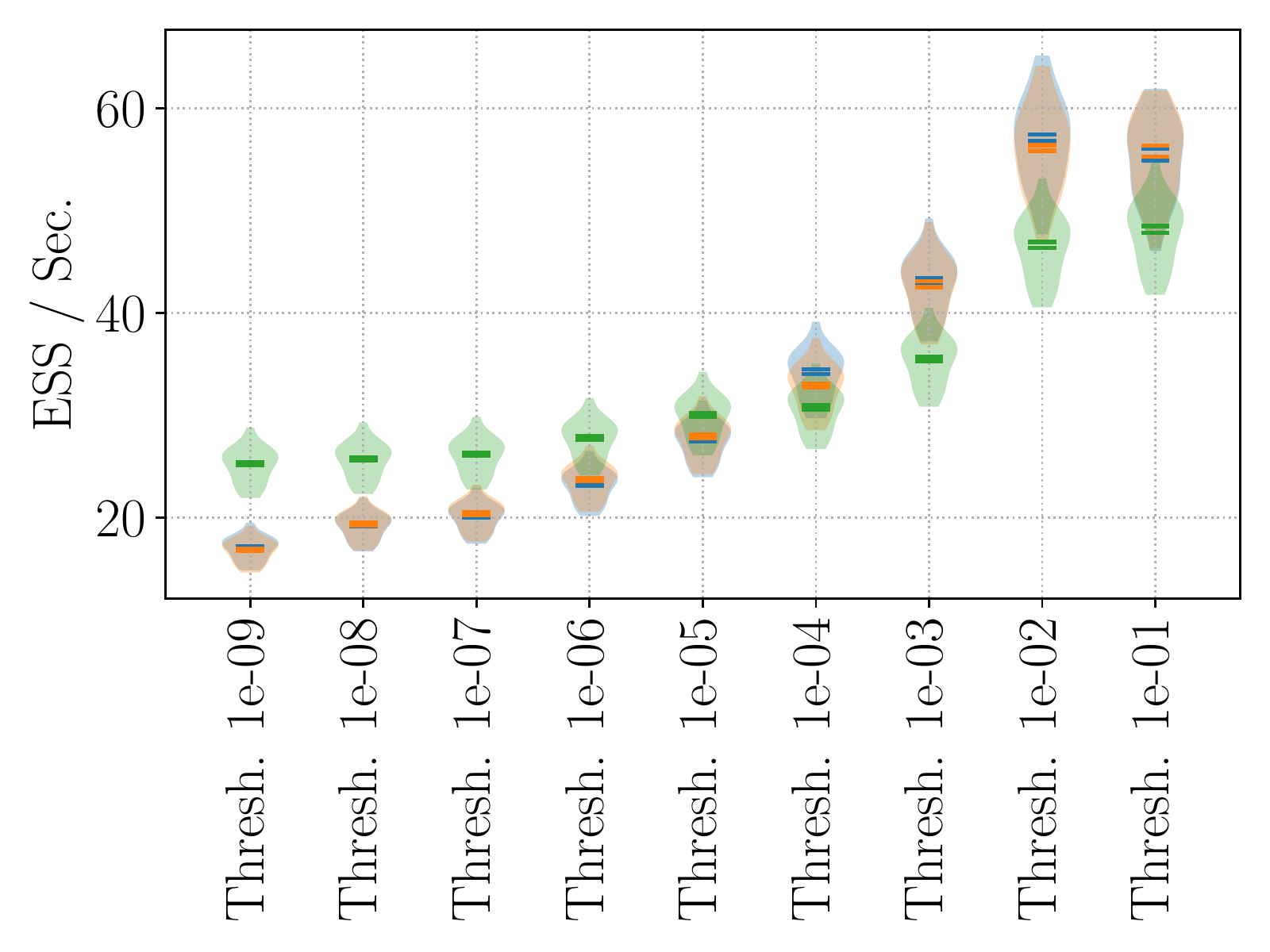}
    \caption{(FN) $c$}
  \end{subfigure}  
  
  \caption{Three circumstances in which Newton's method produced performance comparable to fixed point iteration when resolving the implicit update to the momentum. In these cases, Newton's method may actually be the preferred method, since it is somewhat less sensitive to the convergence tolerance. In the first row we show results for the stochastic volatility model (SV); the second row shows the log-Gauss Cox-Poisson model (CP); the third row shows results for the Fitzhugh-Nagumo differential equation model (FN).}
  \label{fig:fixed-point-vs-newton-ess-same}
\end{figure}
\begin{figure}[t!]
  \begin{subfigure}[t]{0.3\textwidth}
    \centering
    \includegraphics[width=\textwidth]{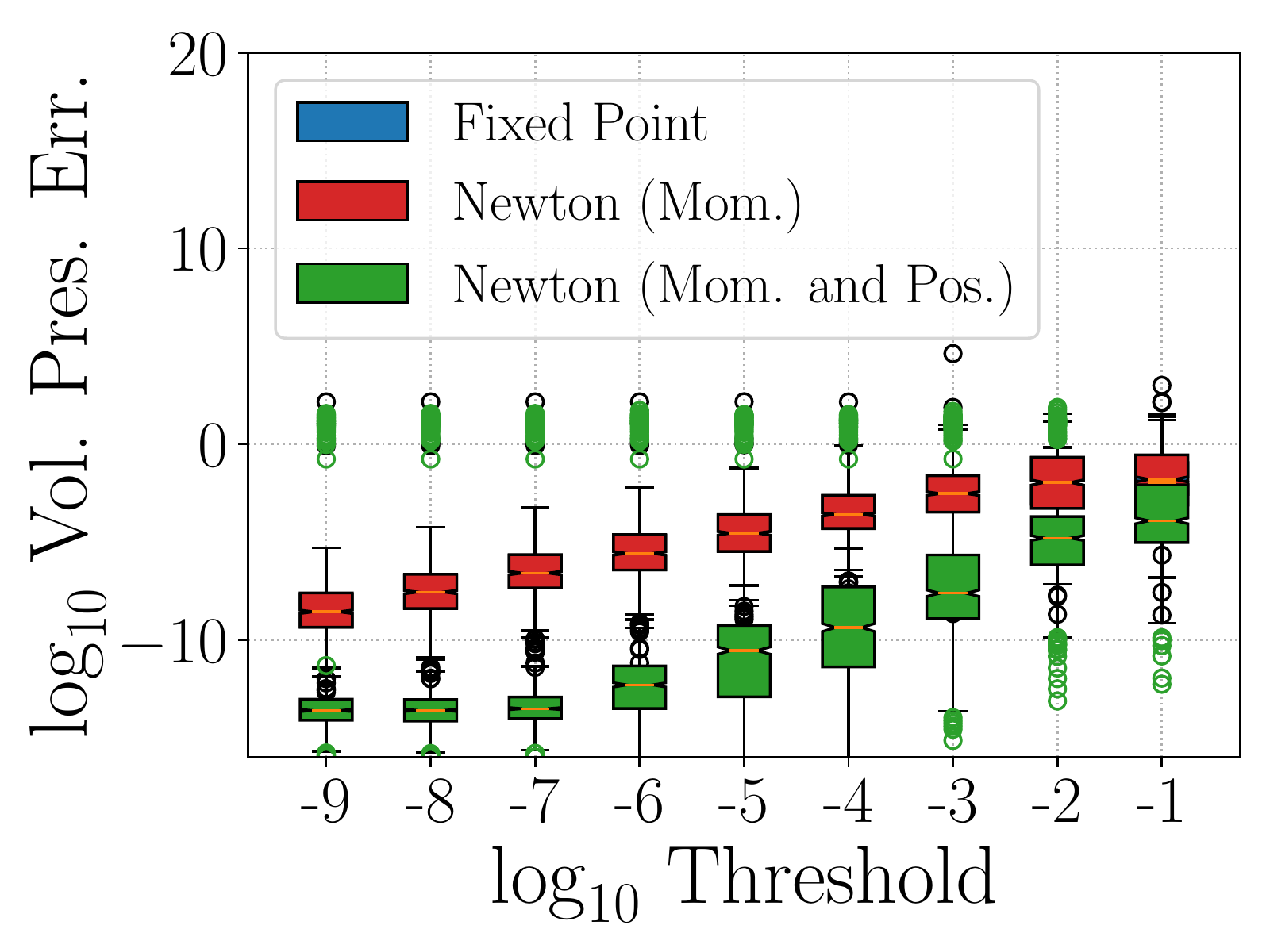}
    \caption{Banana}
  \end{subfigure}
  ~
  \begin{subfigure}[t]{0.3\textwidth}
    \centering
    \includegraphics[width=\textwidth]{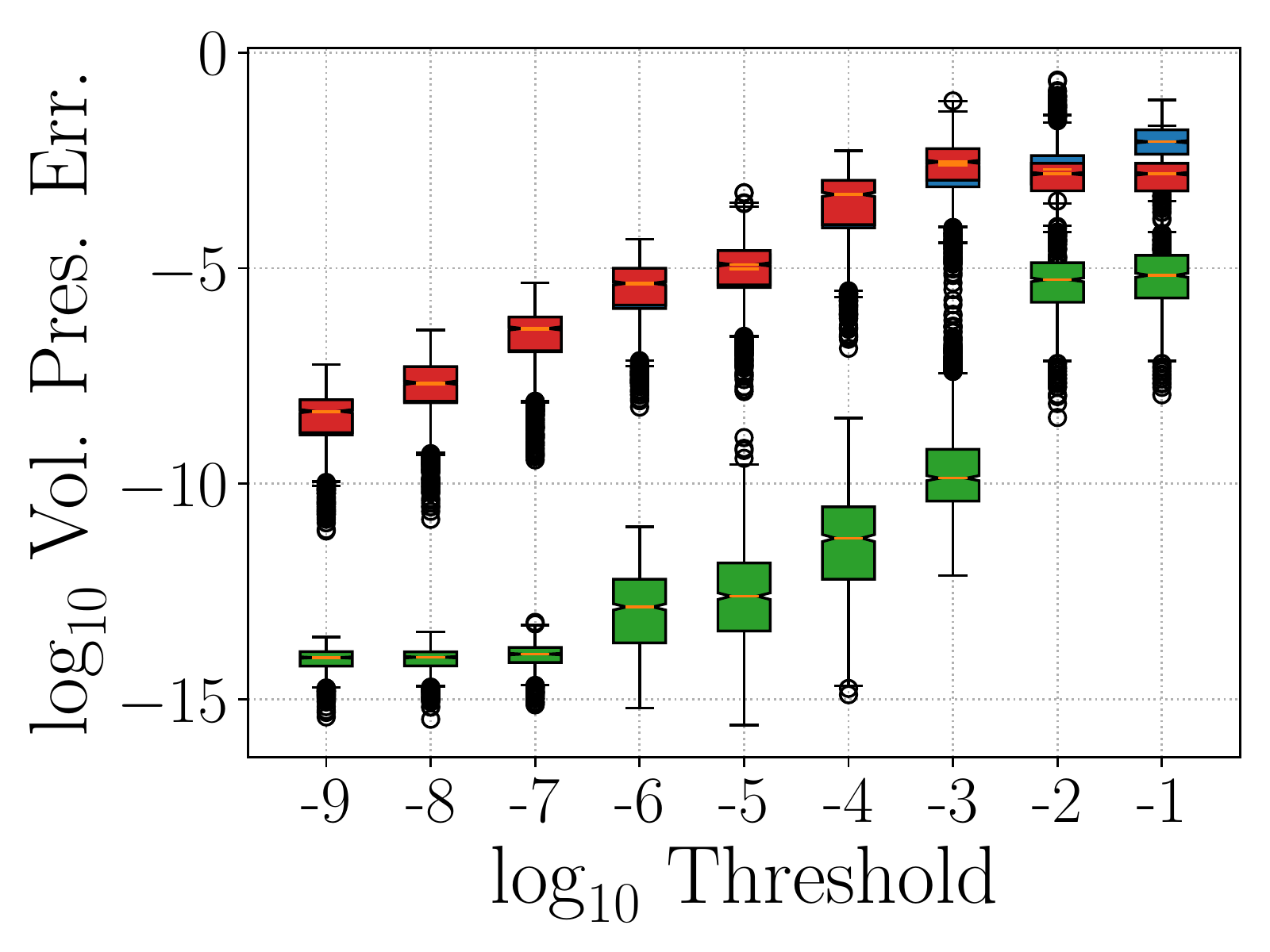}
    \caption{Logistic Regression}
  \end{subfigure}
  ~
  \begin{subfigure}[t]{0.3\textwidth}
    \centering
    \includegraphics[width=\textwidth]{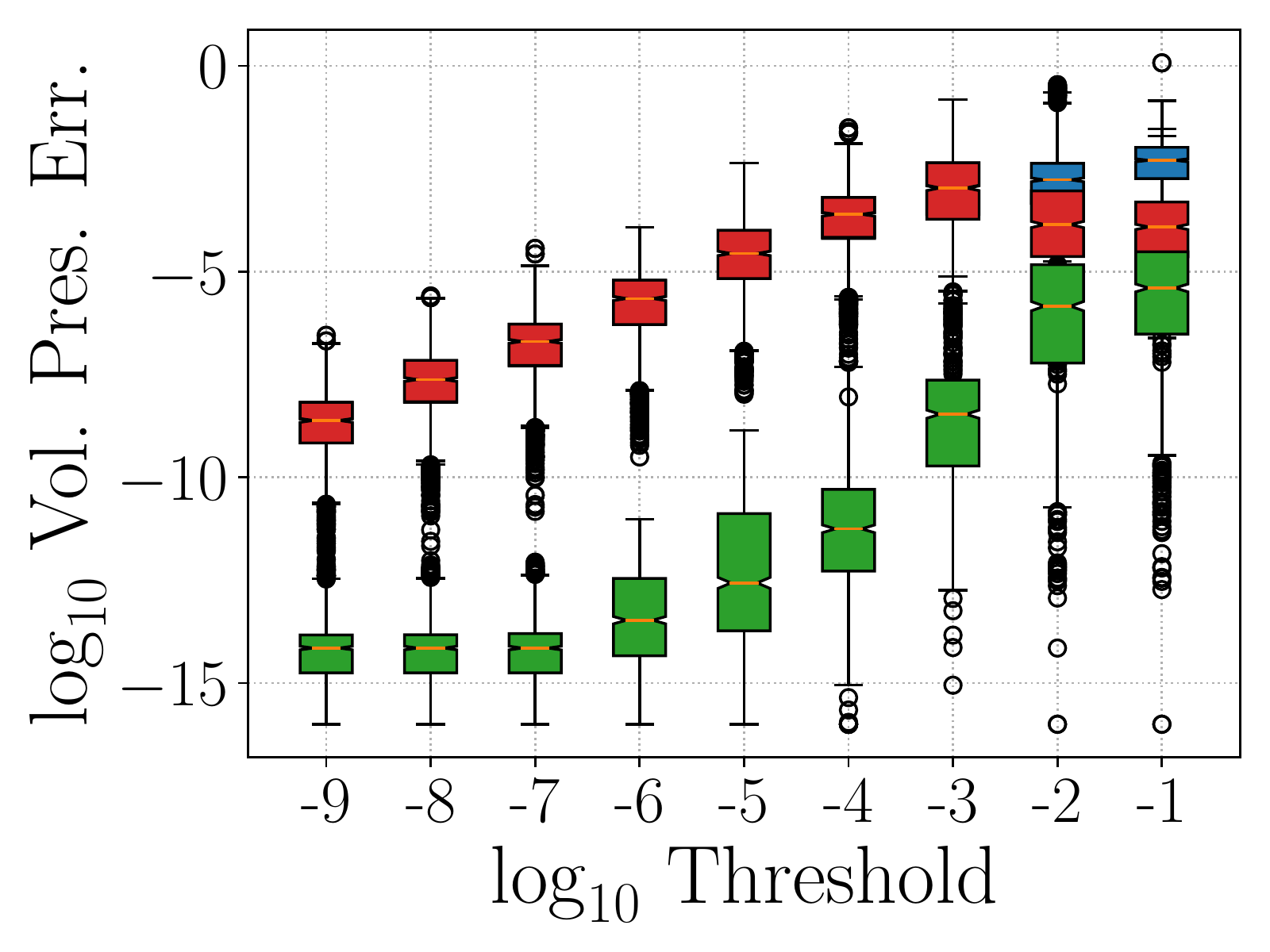}
    \caption{Stochastic Volatility}
  \end{subfigure}
  
  \begin{subfigure}[t]{0.3\textwidth}
    \centering
    \includegraphics[width=\textwidth]{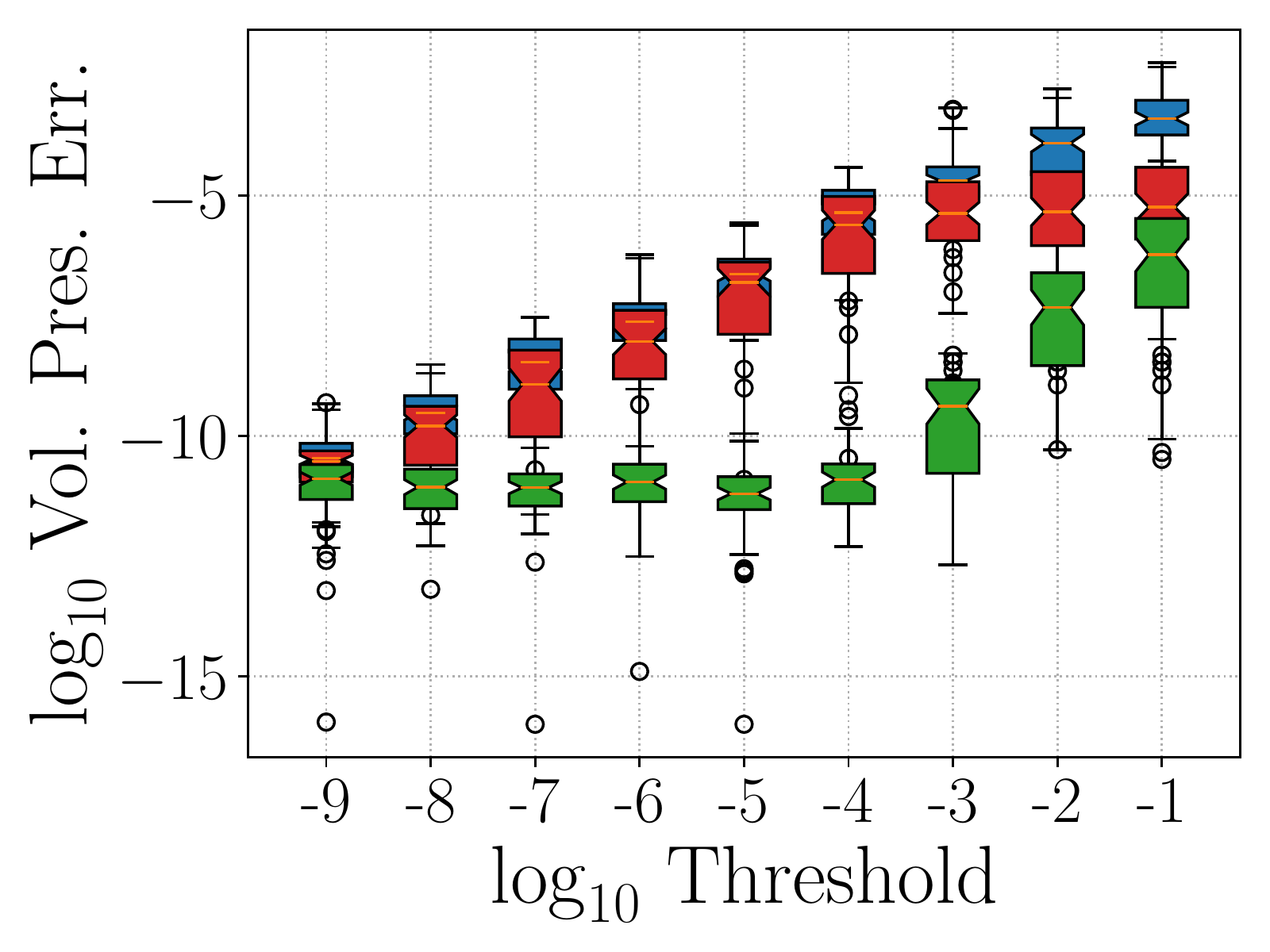}
    \caption{Cox-Poisson}
  \end{subfigure}
  ~
  \begin{subfigure}[t]{0.3\textwidth}
    \centering
    \includegraphics[width=\textwidth]{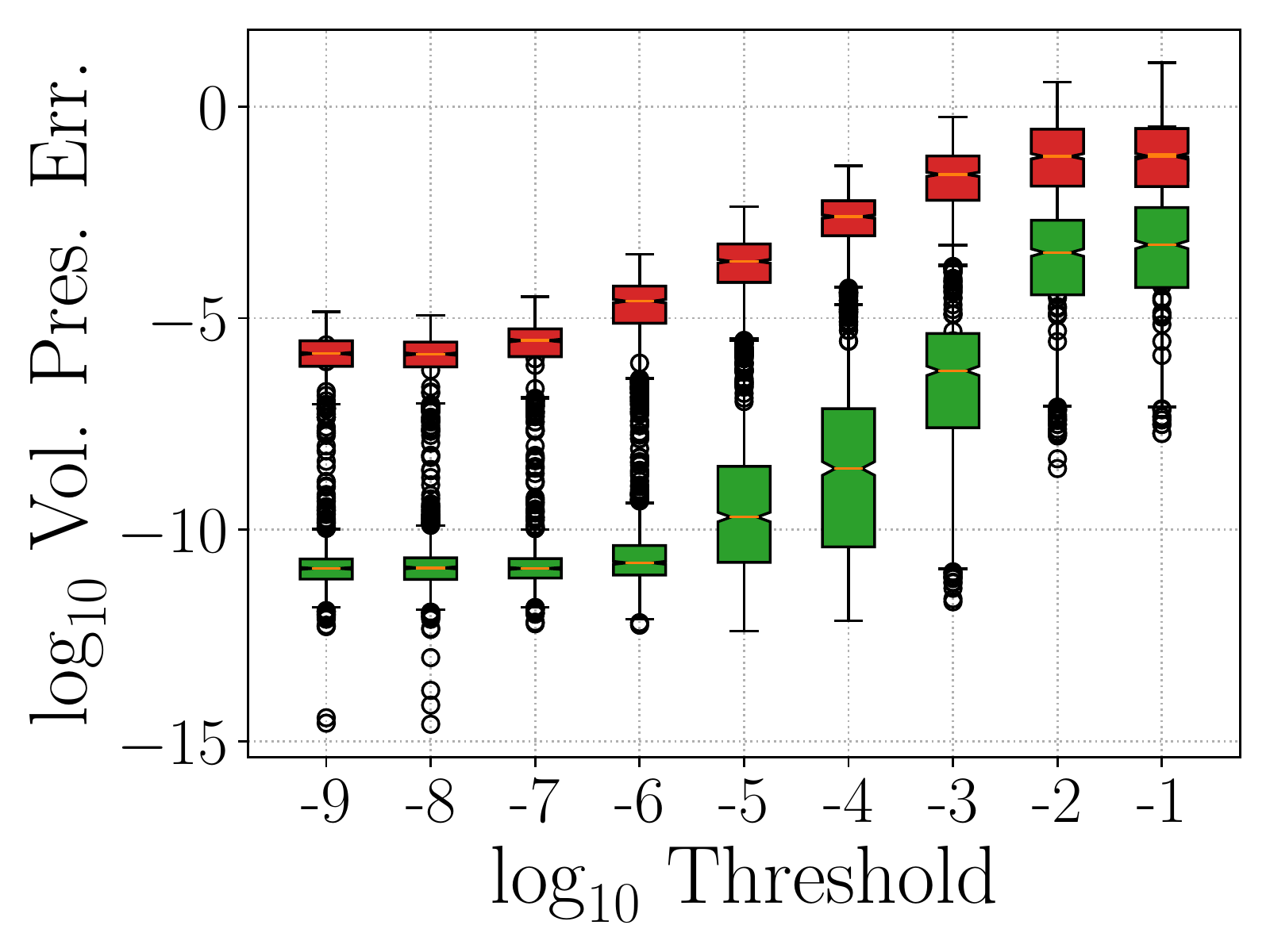}
    \caption{Fitzhugh-Nagumo}
  \end{subfigure}
  ~
  \begin{subfigure}[t]{0.3\textwidth}
    \centering
    \includegraphics[width=\textwidth]{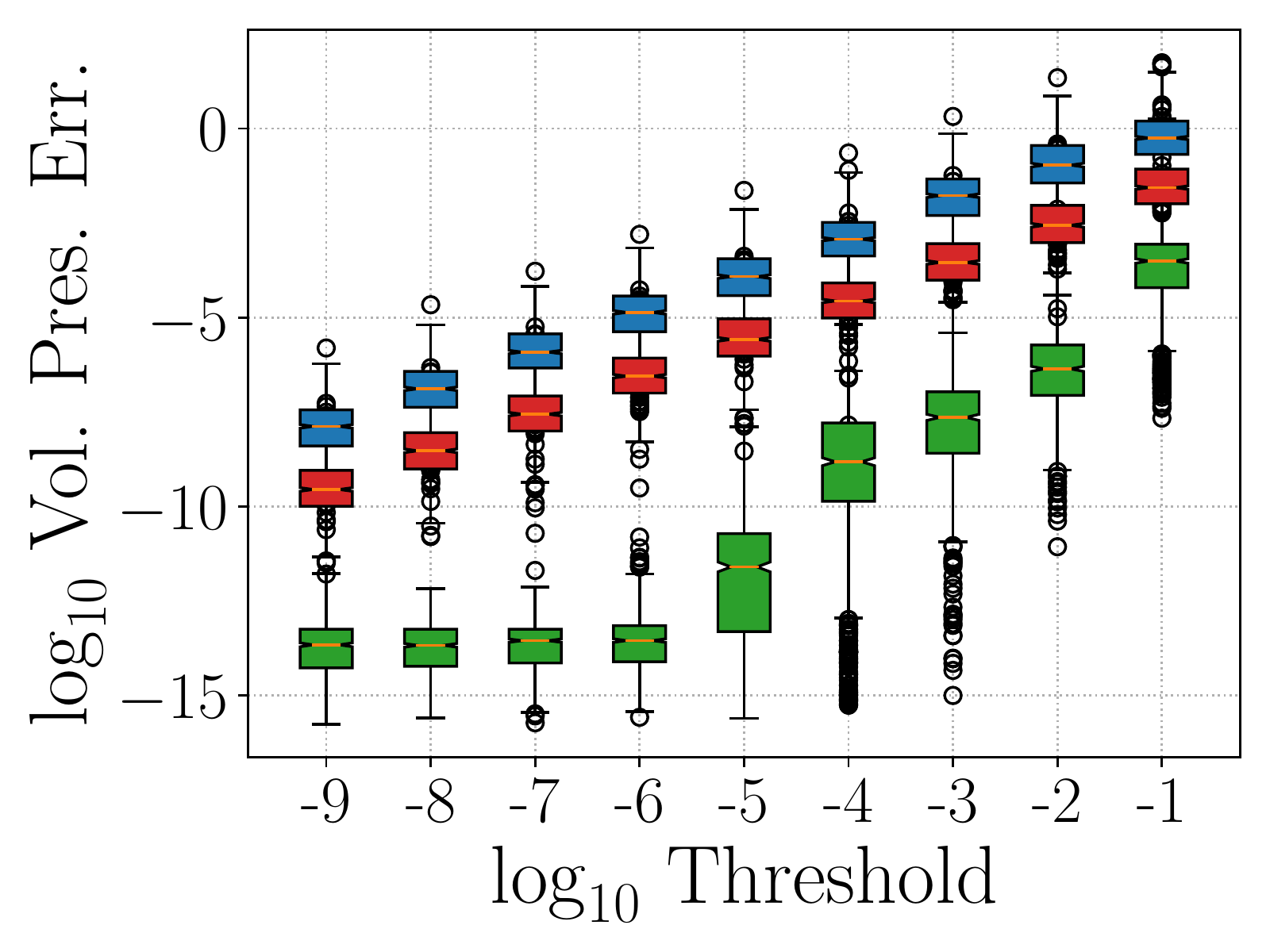}
    \caption{Student-$t$}
  \end{subfigure}

  \caption{We illustrate the dependency of the reversibility metric on the procedure used to resolve the implicit updates to momentum and position. In most examples, employing Newton's method to update both the position and the momentum variables in the generalized leapfrog integrator produced faster convergence to, and better respect for, reversibility of the proposal up to numerical precision.}
  \label{fig:fixed-point-vs-newton-reversibility}
\end{figure}

\begin{figure}[t!]
  \begin{subfigure}[t]{0.3\textwidth}
    \centering
    \includegraphics[width=\textwidth]{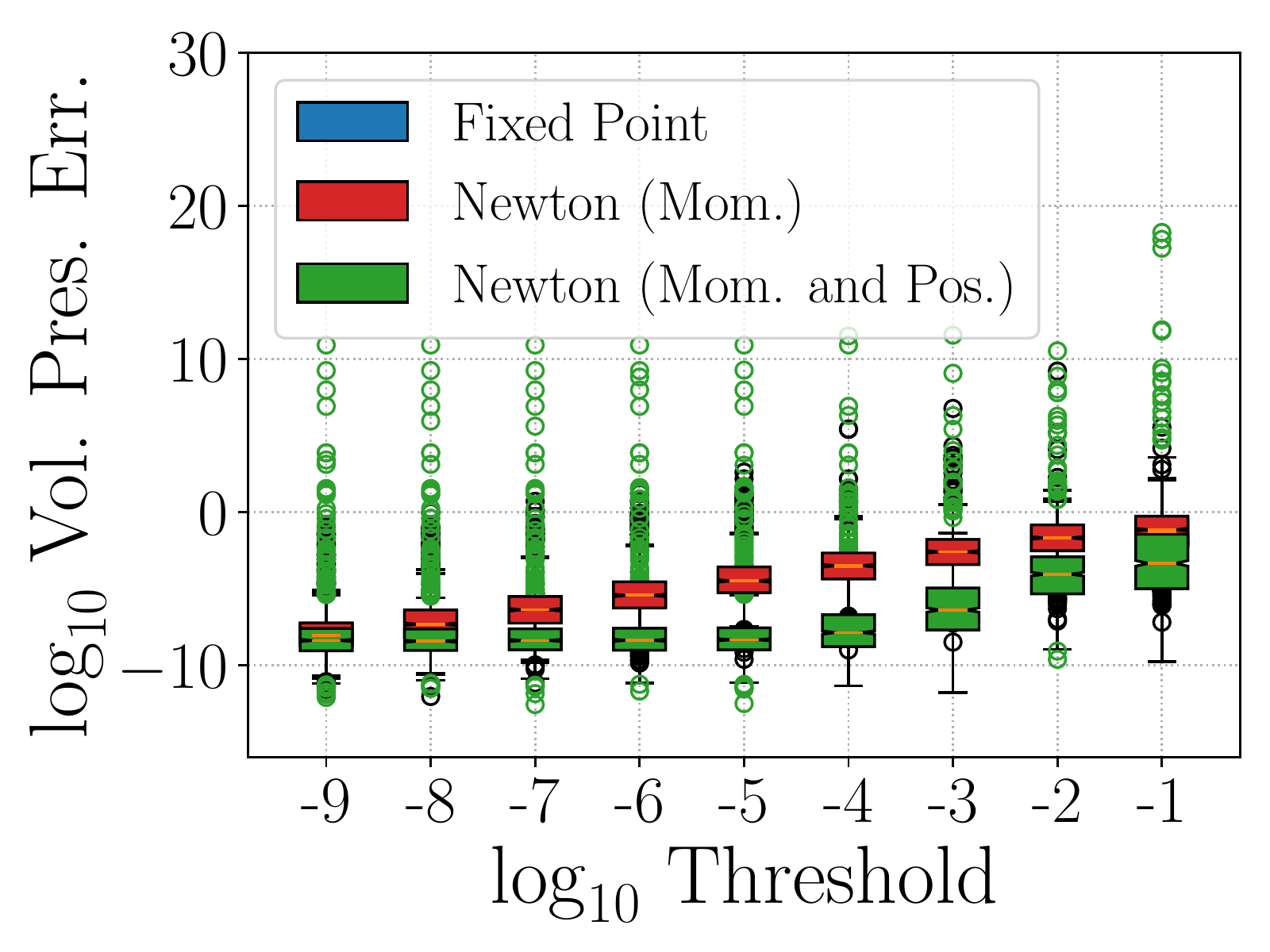}
    \caption{Banana}
  \end{subfigure}
  ~
  \begin{subfigure}[t]{0.3\textwidth}
    \centering
    \includegraphics[width=\textwidth]{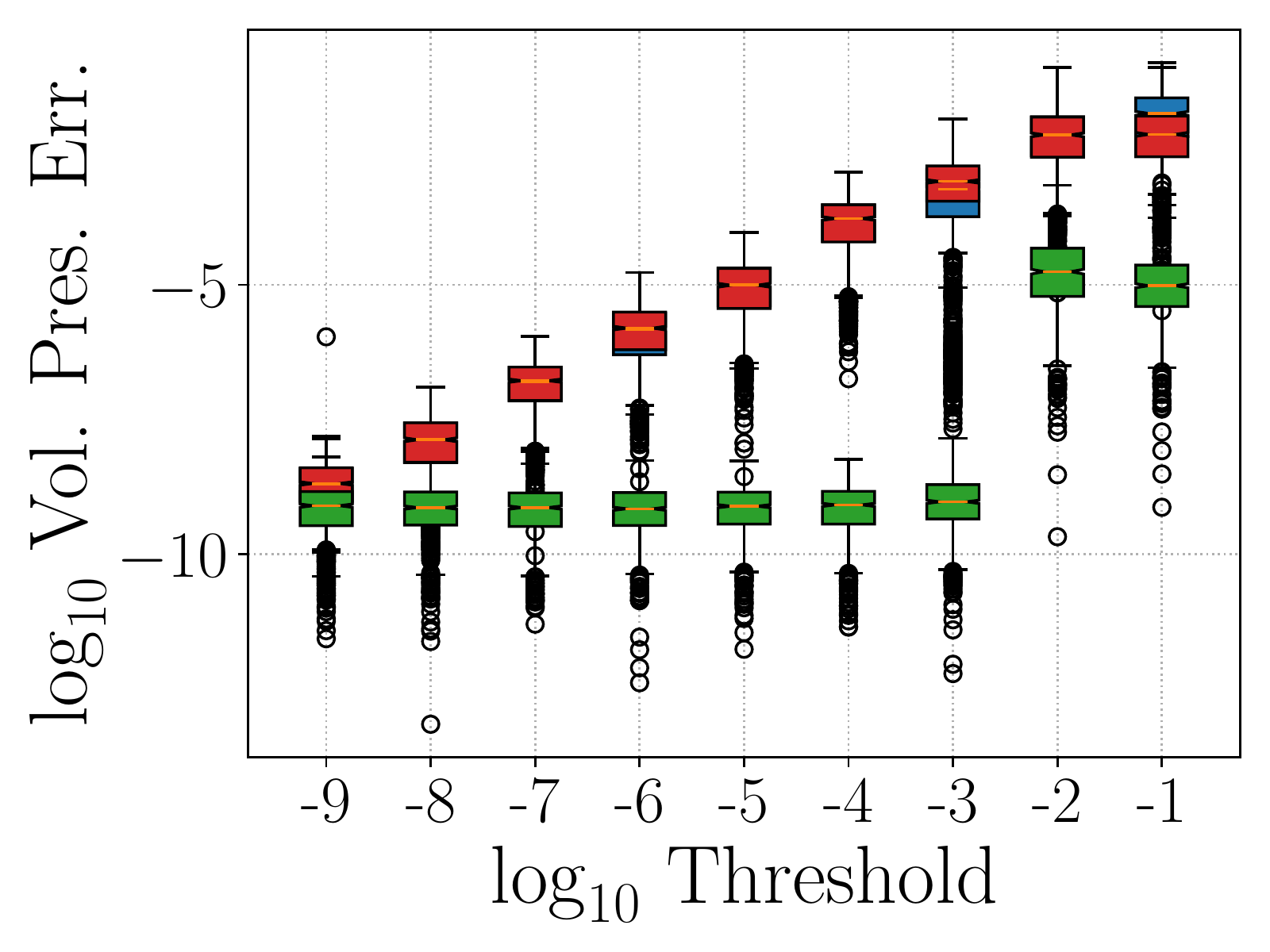}
    \caption{Logistic Regression}
  \end{subfigure}
  ~
  \begin{subfigure}[t]{0.3\textwidth}
    \centering
    \includegraphics[width=\textwidth]{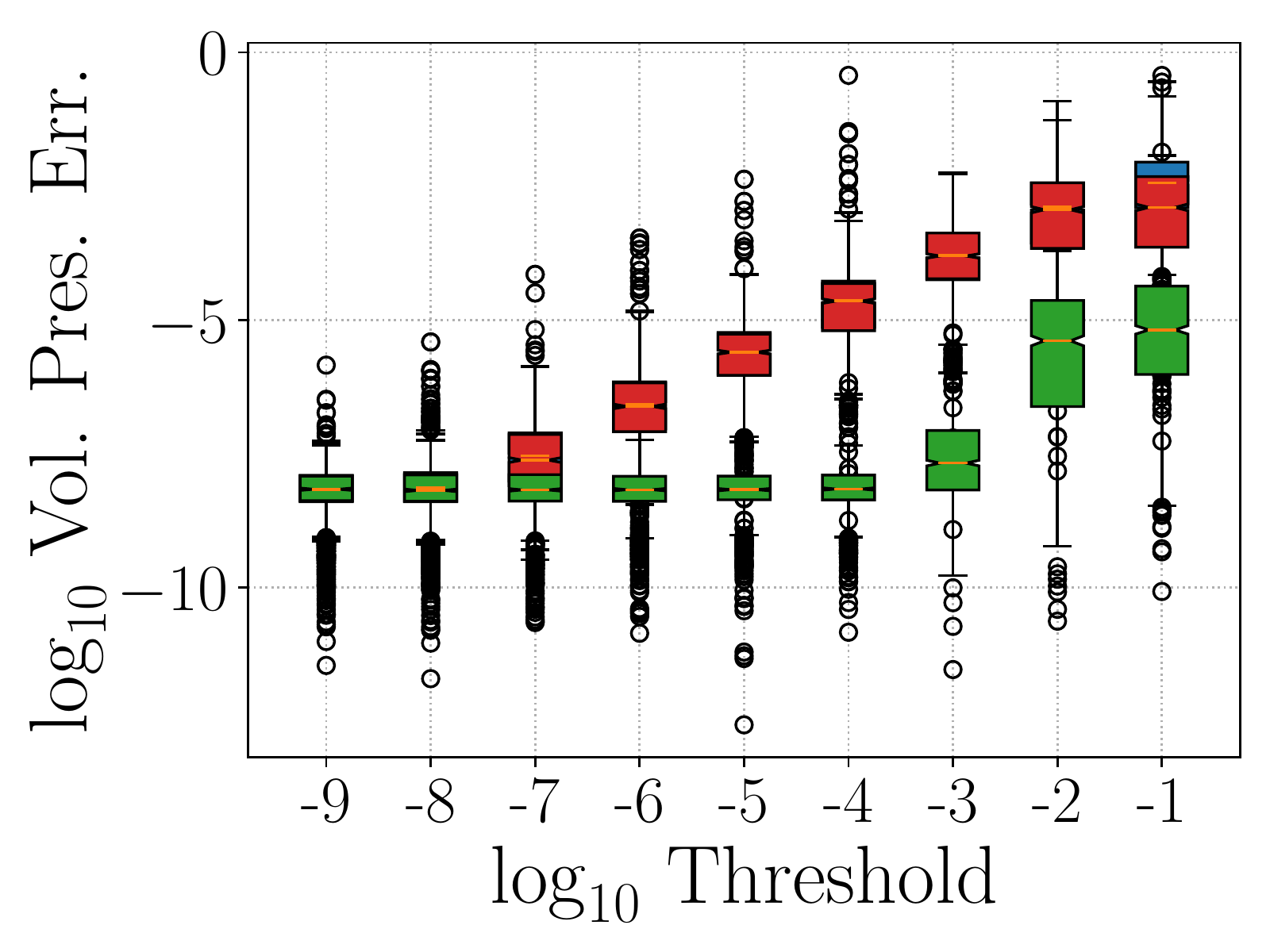}
    \caption{Stochastic Volatility}
  \end{subfigure}
  
  \begin{subfigure}[t]{0.3\textwidth}
    \centering
    \includegraphics[width=\textwidth]{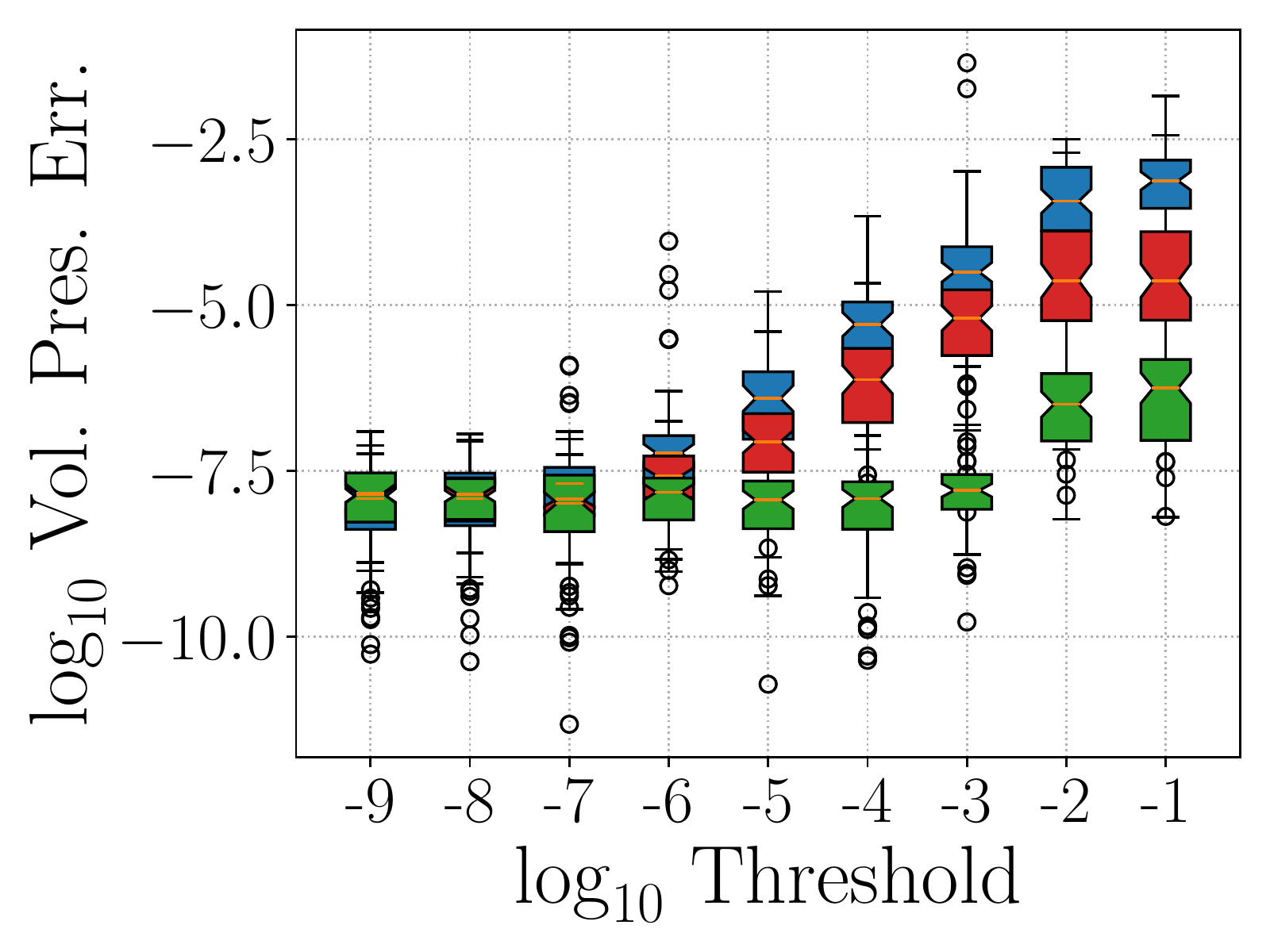}
    \caption{Cox-Poisson}
  \end{subfigure}
  ~
  \begin{subfigure}[t]{0.3\textwidth}
    \centering
    \includegraphics[width=\textwidth]{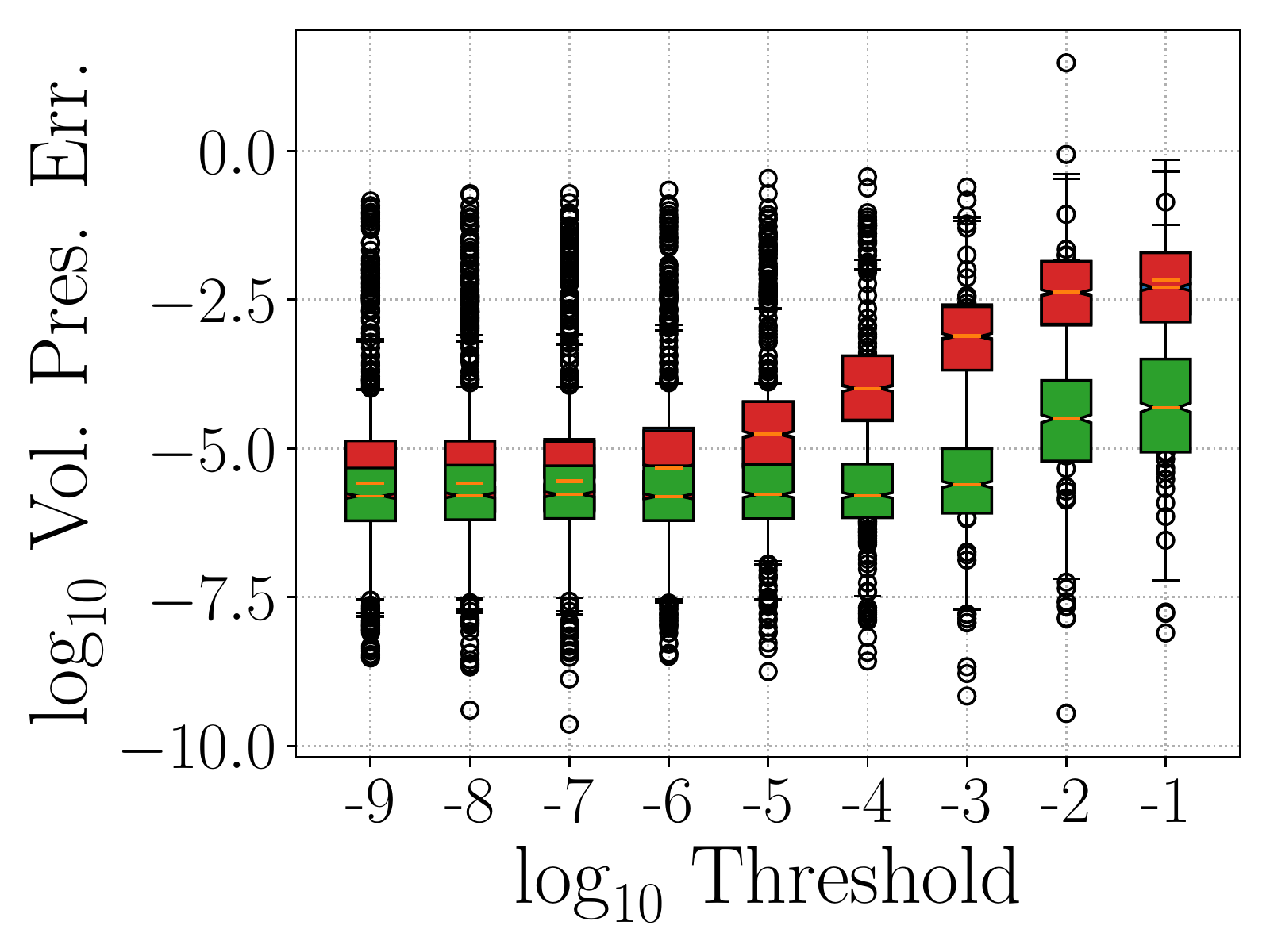}
    \caption{Fitzhugh-Nagumo}
  \end{subfigure}
  ~
  \begin{subfigure}[t]{0.3\textwidth}
    \centering
    \includegraphics[width=\textwidth]{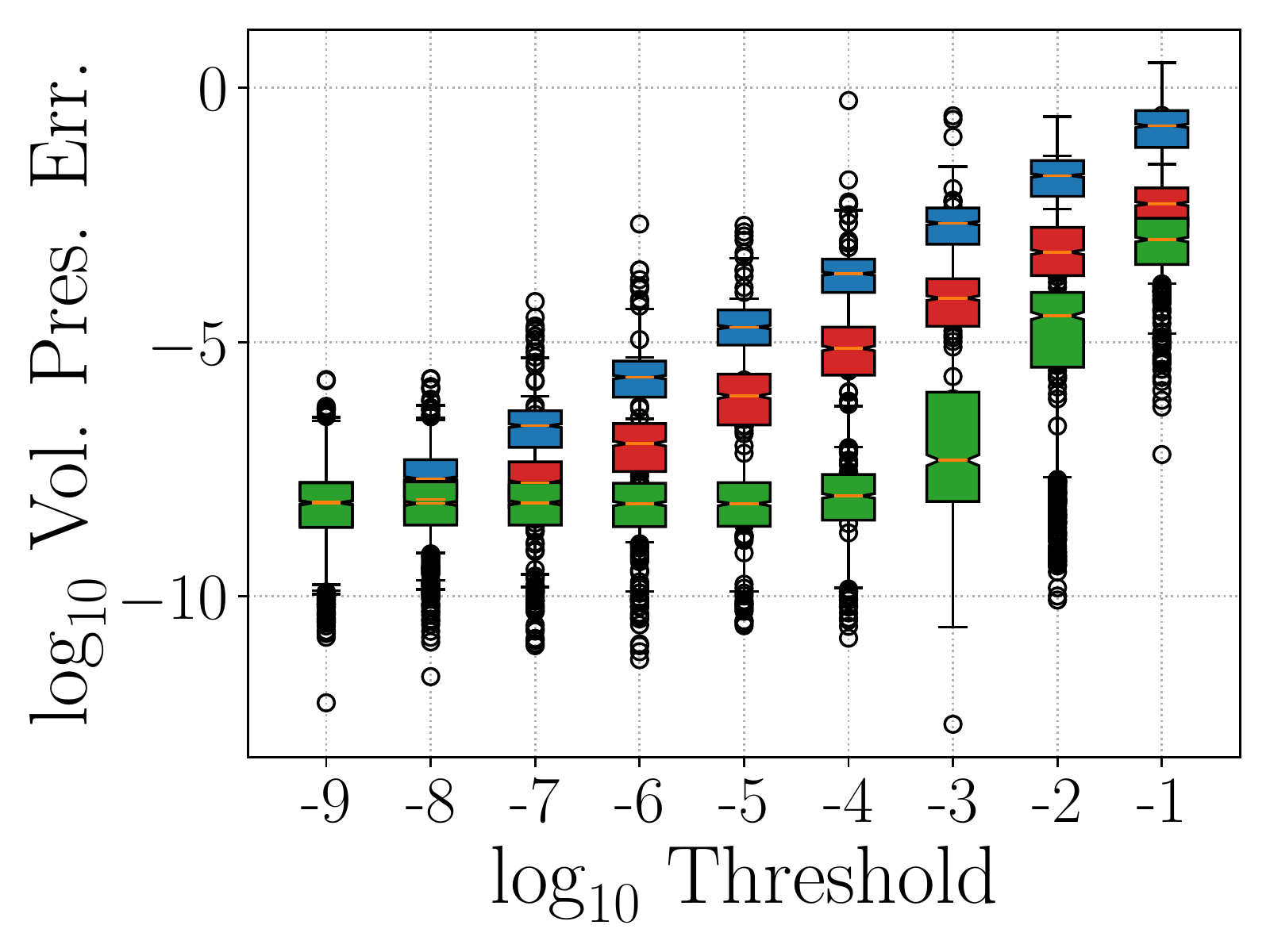}
    \caption{Student-$t$}
  \end{subfigure}

  \caption{We illustrate the dependency of the Jacobian determinant metric on the procedure used to resolve the implicit updates to momentum and position. Employing Newton's method to update both the position and the momentum variables in the generalized leapfrog integrator yields faster convergence toward a volume-preserving proposal.}
  \label{fig:fixed-point-vs-newton-jacobian-determinant}
\end{figure}

\begin{figure}[t!]
  \begin{subfigure}[t]{0.3\textwidth}
    \centering
    \includegraphics[width=\textwidth]{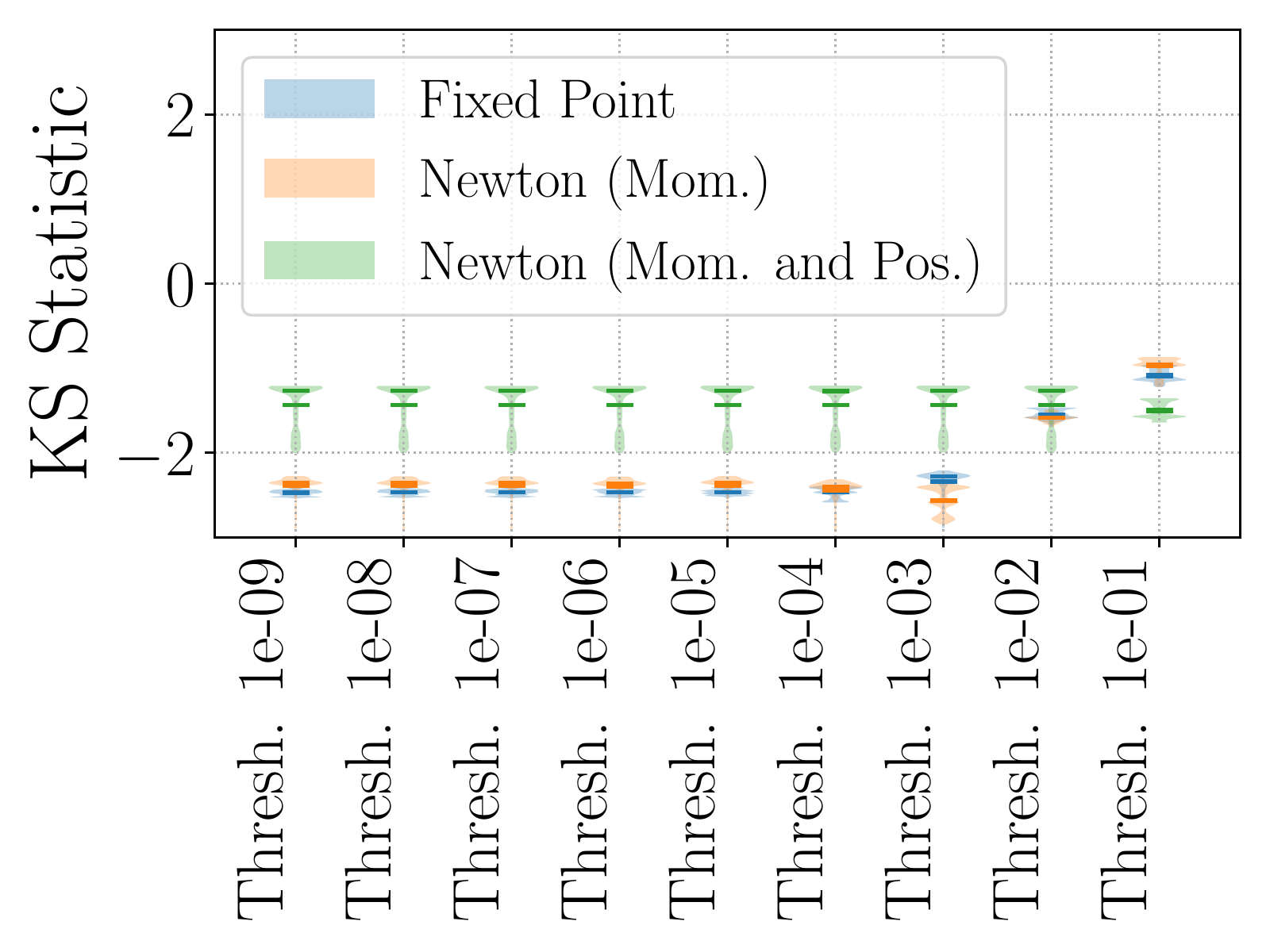}
    \caption{Banana}
  \end{subfigure}
  ~
  \begin{subfigure}[t]{0.3\textwidth}
    \centering
    \includegraphics[width=\textwidth]{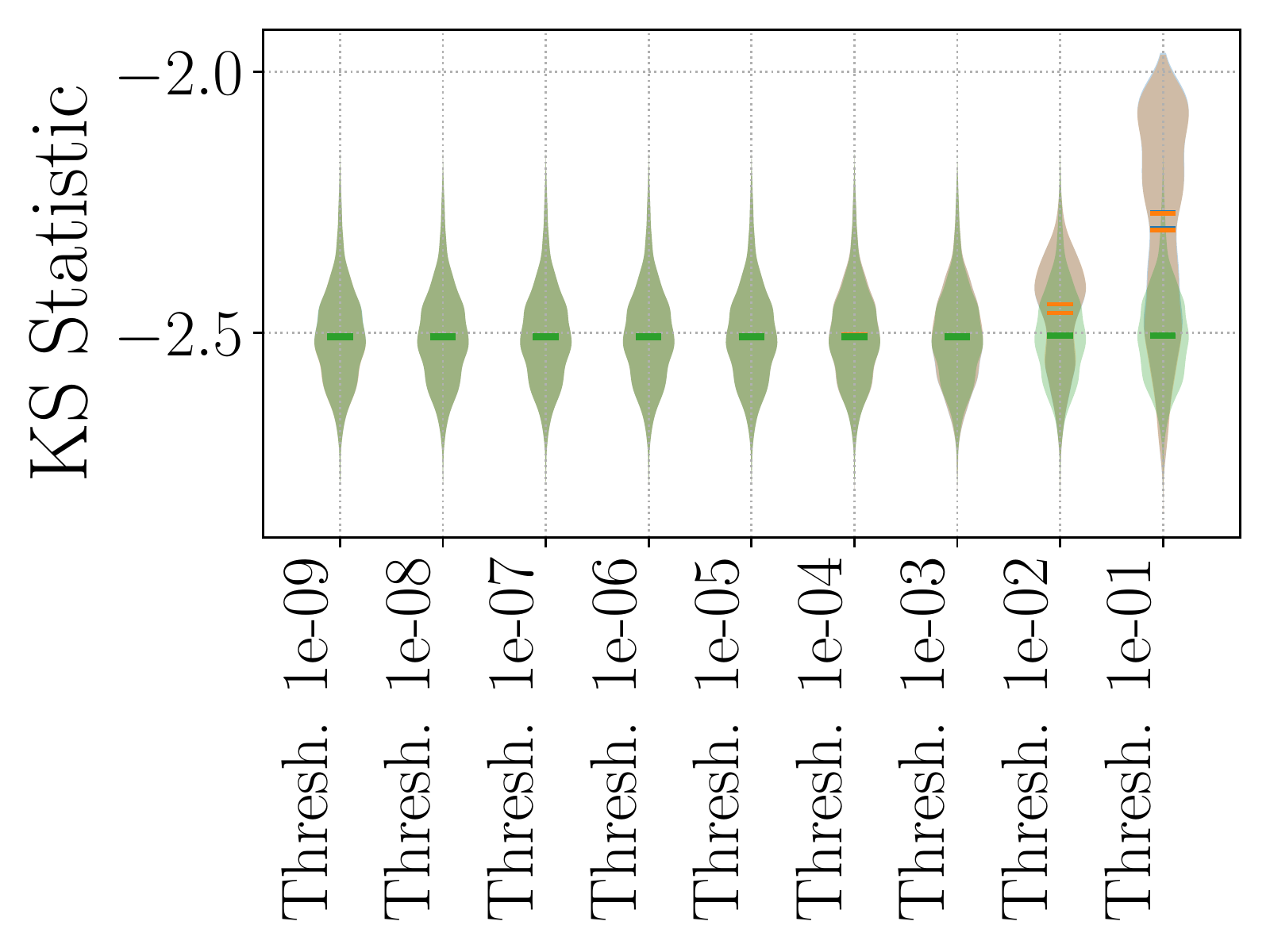}
    \caption{Fitzhugh-Nagumo}
  \end{subfigure}
  ~
  \begin{subfigure}[t]{0.3\textwidth}
    \centering
    \includegraphics[width=\textwidth]{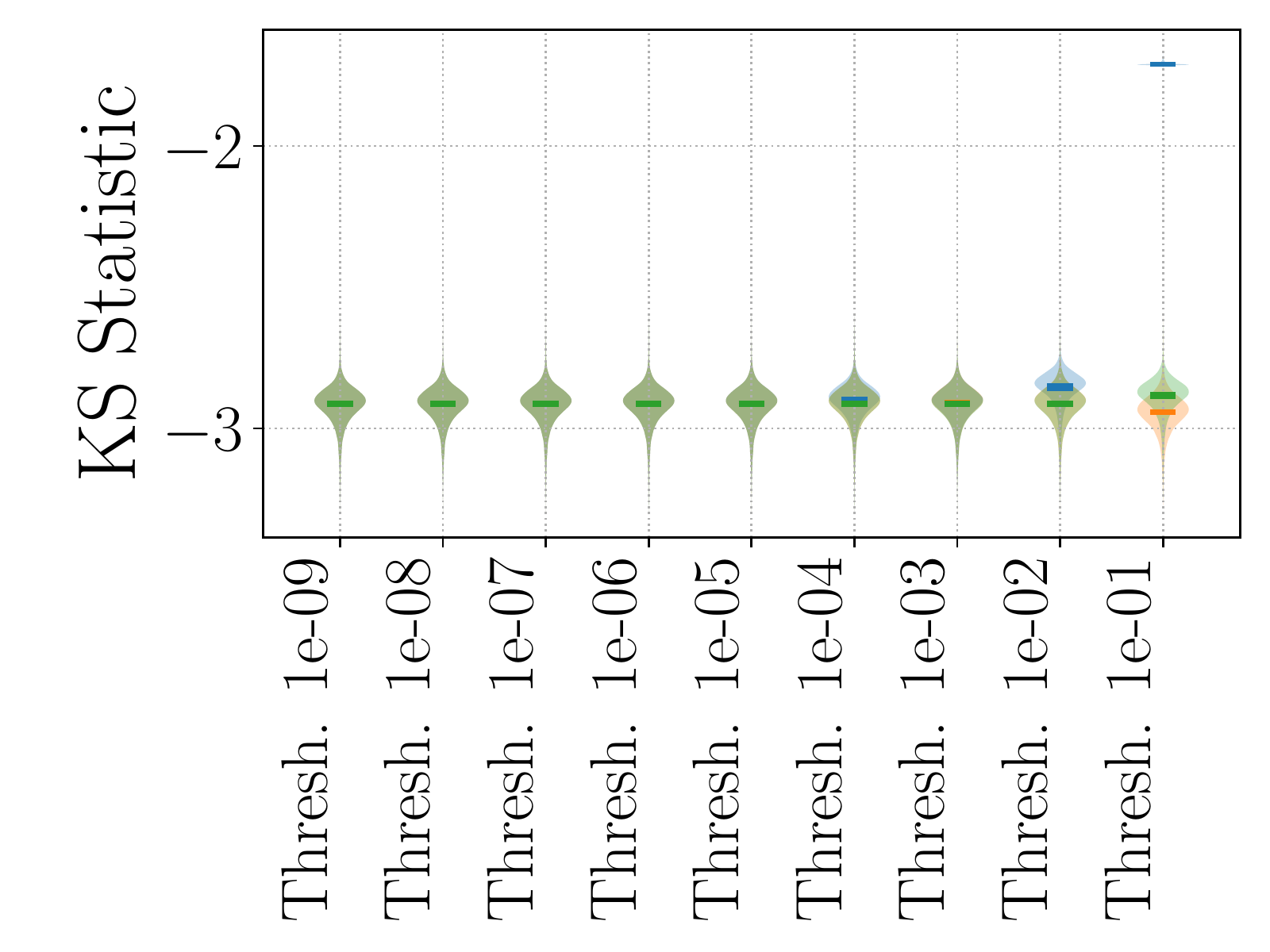}
    \caption{Student-$t$}
  \end{subfigure}

  \caption{We visualize the ergodicity properties for the three implementations of RMHMC. In the casae where Newton's method is employed for resolving the implicit updates to momentum and position, the ergodicity of the method is essentially constant with respect to the convergence tolerance. However, this is not desirable in the case of the banana-shaped posterior distribution, wherein ergodicity failures result from severe violations of reversibility and volume preservation in some cases.}
  \label{fig:fixed-point-vs-newton-ergodicity}
\end{figure}

As discussed in \cref{subsec:measuring-computational-effort}, the implicit updates for position differs fundamentally from the implicit equation to update momentum in the generalized leapfrog integrator: the update to momentum does not require recomputing quantities of the posterior. In contrast, the implicit update to the position must recompute the Riemannian metric at each iteration. Therefore, to the extent that posterior quantities are the most expensive computations in RMHMC, the momentum fixed point iteration can be resolved relatively cheaply. Therefore, if the Jacobian of the fixed point equation can be found, then it is conceivable to solve these implicit relationships via Newton's method and enjoy a faster order of convergence than fixed point iteration. Recall that the fixed equation for the momentum variable in the generalized leapfrog is
\begin{align}
  \bar{p} = p_n - \frac{\epsilon}{2} \nabla_q H(q_n, \bar{p})
\end{align}
which we can represent as a root-finding problem (in the variable $\bar{p}$) by rearranging as,
\begin{align}
    g(\bar{p}) = \bar{p} - p_n + \frac{\epsilon}{2} \nabla_q H(q_n, \bar{p}) = 0.
\end{align}
The Jacobian is therefore $\nabla g(\bar{p}) = \mathrm{Id} + \frac{\epsilon}{2} \nabla_p\nabla_q H(q_n, \bar{p})$. For Hamiltonians of the form \cref{eq:riemannian-hamiltonian}, we have
\begin{align}
  \nabla_p\nabla_q H(q, p) = \mathbf{G}^{-1}(q) \nabla_q \mathbf{G}(q) \mathbf{G}^{-1}(q) \bar{p},
\end{align}
which immediately gives a computable Jacobian to employ in Newton's method. Although Newton's method has a superior order of convergence than fixed point iteration, we observe that the two procedures do not have the same computational complexity. In particular, Newton's method requires that we solve the linear system, $\nabla g(\bar{p}) s = g(\bar{p})$ at each step; in general, the matrix $\nabla g(\bar{p})$ does not have any special structure, so solving this linear system has cubic complexity. We expect, therefore, that resolving the implicit update to momentum via Newton's method must balance the faster order of convergence with the greater computational burden.
  
We may also consider Newton iterations to resolve the implicit update to the position variable in RMHMC with the generalized leapfrog integrator. The implicit update to position in \cref{eq:generalized-leapfrog-position} and for Hamiltonians of the form \cref{eq:riemannian-hamiltonian} is,
\begin{align}
  q_{n+1} = q_n + \frac{\epsilon}{2}\paren{\mathbf{G}^{-1}(q_{n+1})\bar{p} + \mathbf{G}^{-1}(q_n) \bar{p}}.
\end{align}
This can be represented as a root-finding problem by defining the function,
\begin{align}
  g(q_{n+1}) = q_{n+1} - q_n - \frac{\epsilon}{2} \mathbf{G}^{-1}(q_{n+1})\bar{p} - \frac{\epsilon}{2} \mathbf{G}^{-1}(q_n) \bar{p},
\end{align}
whose Jacobian is,
\begin{align}
  \nabla g(q_{n+1}) = \mathrm{Id} - \frac{\epsilon}{2} \mathbf{G}^{-1}(q_{n+1}) \nabla \mathbf{G}(q_{n+1}) \mathbf{G}^{-1}(q_{n+1}) \bar{p}.
\end{align}
We therefore consider three variations of RMHMC. The first uses fixed point iterations in order to resolve both of the implicit updates to the momentum variable. The second method updates the momentum variable using Newton iterations and employs fixed point iterations for the implicit update to position. The third variation uses Newton's method for both the implicit updates.
  
In \cref{fig:fixed-point-vs-newton-iterations} we show the the average number of iterations required to resolve the implicit update to momentum when utilizing either fixed point iterations or Newton's method. We observe that Newton's method requires fewer iterations and produces a somewhat flatter curve as the convergence threshold is varied. This can be a desirable property since it eliminates some of the sensitivity of the RMHMC algorithm to variations in the convergence threshold parameter. In \cref{fig:fixed-point-vs-newton-ess-worse} we show three examples wherein the use of Newton's method degrades performance as measured by the effective sample size per second. On the other hand, in \cref{fig:fixed-point-vs-newton-ess-same}, we show three case studies where fixed point iterations and Newton's method produced compared effective sample sizes per second.
  
Examining in detail \cref{subfig:fixed-point-vs-newton-ess-wose-banana}, it may appear that applying Newton's method to both implicit updates produced a {\it benefit} in the ESS per second in the banana-shaped posterior. However, as our subsequent analysis shows, this is a mistaken conclusion. We must have confidence that the substitution of the fixed point iteration solver by Newton's method has not degraded the detailed balance condition. Therefore, we contextualize these measures of sampling efficiency by also measuring ergodicity (via the Kolmogorov-Smirnov procedure along random one-dimensional subspaces as described in \cref{subsubsec:measuring-ergodicity}) and satisfaction of the detailed balance condition. In \cref{fig:fixed-point-vs-newton-reversibility} we show the reversibility as a function of threshold for the three variations of the RMHMC procedure we consider. In many cases, Newton's method produces integrators that much more closely respect reversibility of the proposal operator; this is attributable to the second-order convergence properties of Newton's method. Indeed, in terms of reversibility, employing Newton's method on the momentum update only is virtually indistinguishable from an implementation using on fixed point iterations. However, a notable exception to the superiority of Newton's method for both position and momentum is the banana-shaped distribution, for which we see there is a notable proportion of transitions that utterly fail to respect reversibility. Similarly, in \cref{fig:fixed-point-vs-newton-jacobian-determinant}, we show how the three variations of RMHMC respect volume preservation. When employing Newton's method to resolve the implicit updates to momentum and position, the voume preservation is insensitive to the convergence threshold over many orders of magnitude. In the case of the banana-shaped distribution, we see again that there are a substantial proportion of transitions whose estimated Jacobian determinant is vastly different from unity. We note that RMHMC using only fixed point iterations or with Newton's method only applied to the momentum update are visually indistinguishable in many cases.
  
As noted previously, the principle interest in MCMC is not increasingly close satisfaction of the conditions implying detailed balance. Instead, we are interested in the closeness of the Markov chain samples to the target distribution. To investigate this property in the context of the three variations of RMHMC, we show in \cref{fig:fixed-point-vs-newton-ergodicity} the Kolmogorov-Smirnov statistic along randomly sampled one-dimensional sub-spaces for the banana-shaped distribution, the Fitzhugh-Nagumo posterior, and the multiscale Student-$t$ distribution. Our analysis of reversibility and volume preservation gives cause to believe that applying Newton's method to both momentum and position yields a sampler fraught with peril in the case of the banana-shaped distribution: indeed, this concern proves to be true, as this variation of RMHMC reveals severely degraded ergodicity relative to the competing implementations. In the case of the Fitzhugh-Nagumo posterior or the multiscale Student-$t$ distribution, the case for Newton's method is more favorable. In both cases, the ergodicity of RMHMC is essentially constant as a function of the threshold; this can be constrasted with an implementation based on fixed point iterations, for which the convergence threshold will affect ergodicity. In cases wherein Newton's method produces an integrator that satisfies the conditions of reversibility and volume preservation, its use can eliminate much of the sensitivity of the RMHMC algorithm to the selection of the convergence tolerance.

\section{Discussion and Conclusion}

This work has examined the role of convergence thresholds in Riemannian manifold Hamiltonian Monte Carlo. The integrators used in implementations of RMHMC depend on non-zero convergence tolerances, which will affect the volume preservations and reversibility properties of the integrator. Because reversibility and volume preservation are critical ingredients in the proofs of detailed balance of HMC, it is necessary to understand how these properties degrade in the presence of non-zero convergence tolerances. We would like to emphasize the following empirical observations from our case studies:
\begin{enumerate}
    \item The RMHMC algorithms tend to produce samples with better ergodicity and a larger effective sample sizes than Eulidean methods, though the $\mathcal{O}(m^3)$ computational effort represents a heavy burden. We find that the question of whether or not RMHMC has a better effective sample size {\it per second} depends on the convergence threshold being used. Very small thresholds such as $1\times 10^{-9}$ tend to be associated to the best reversibility and volume preservation, but also the largest number of fixed point iterations at each step of the generalized leapfrog integrator.
    \item We observe that the empirical sample quality, as measured by Kolmogorov-Smirnov statistics along random one-dimensional subspaces and related measures, is sensitive to the choice of convergence threshold only to a point. This implies that there exist diminishing returns to using a convergence threshold beyond some critical value. While divining the value of this critical threshold would seem challenging without the kind of analysis conducted herein, its existence suggests that proper adaptation of the convergence threshold can produce important improvements in computational efficiency.
    \item The correctness of derivatives assumes a special significance in RMHMC. In Euclidean HMC, an incorrectly specified derivative can be expected to slow mixing but not to invalidate the correctness of the MCMC procedure itself. This no longer holds in the case of RMHMC, where detailed balance depends on volume preservation and volume preservation depends on the symmetry of partial derivatives. While the use of certain automatic differentiation (AD) tools may alleviate some of these concerns, we make the additional observations that (i) metrics such as the Fisher information (which is the {\it expectation} of the negative Hessian of the log-likelihood) may not be readily obtained from AD, and (ii) implementation errors in the AD rules of even the most popular libraries are corrected on a regular basis. At the present time, we suggest that at least some minimal effort be made to ensure that volume preservation holds in a given application.
    \item We have proposed a method based on Ruppert averaging by which one may select a convergence tolerance based on specifying a desired number of decimal digits of similarity with a numerical integrator with a stringent convergence threshold. If a practitioner can hypothesize a particular minimal scale of the posterior distribution, then this quantity can guide the selection of the convergence tolerance such that differences (on average) between the numerical integrators are beneath the scale of the posterior. By tuning the convergence tolerance of the method, one can obtain improvements in computational expediency without significant detectable differences in either the ergodicity or effective sample size properties of the RMHMC procedure.
    \item We have investigated the suitability of employing Newton's method as an alternative to fixed point iteration of RMHMC, with particular attention given to the benefits this alternative may have for robustness to threshold selection. We find that Newton's method, consistent with its stronger position as a second-order solver, enjoys more rapid convergence toward numerical reversibility and volume preservation under our measures. We additionally find that, within the scope of convergence thresholds and experimental designs here, the ergodic properties of RMHMC implemented with Newton's method are less sensitive to the choice of threshold. However, we also illuminated a circumstance in the banana-shaped posterior where Newton's method actually produced degraded sampling behavior.
\end{enumerate}

This work has examined the role of the convergence threshold in RMHMC and we have suggested mechanisms by which to monitor the detailed balance conditions, to tune the convergence parameter to a given problem, and we have suggested a revised implementation of RMHMC based on Newton's method which can exhibit greater robustness to the convergence tolerance. We hope that this work raises interest and awareness around this aspect of the computation in RMHMC.

\bibliographystyle{abbrvnat}
\bibliography{thebib}

\begin{appendix}
  \section{Proofs of Reversibility and Volume Preservation of Hamiltonian Mechanics}
  \begin{proposition}
    Let $H :\R^m\times\R^m\to\R$ be a smooth Hamiltonian for which $-\nabla_p
    H(q, -p) = \nabla_pH(q, p)$ and $\nabla_q H(q, p) = \nabla_qH(q, -p)$. Let
    $(q_t, p_t)$ be solutions to Hamilton's equations of motion given in
    \cref{eq:hamiltonian-position-evolution,eq:hamiltonian-momentum-evolution}.
    Then, for fixed $\tau\in\R$, the trajectories $\tilde{q}_s = q_{\tau-s}$ and
    $\tilde{p}_s = -p_{\tau - s}$ are also solutions to Hamilton's equations of
    motion.
  \end{proposition}
  \begin{proof}
    By direct calculation,
    \begin{align}
      \frac{\mathrm{d}}{\mathrm{d}s} \tilde{q}_s &= -\dot{q}_{\tau-s} \\
      &= -\nabla_p H(q_{\tau-s}, p_{\tau-s}) \\
      &= -\nabla_p H(\tilde{q}_s, -\tilde{p}_{s}) \\
      &= \nabla_p H(\tilde{q}_s, \tilde{p}_s)
    \end{align}
    and
    \begin{align}
      \frac{\mathrm{d}}{\mathrm{d}s} \tilde{p}_s &= \dot{p}_{\tau-s} \\
      &= -\nabla_q H(q_{\tau-s}, p_{\tau-s}) \\
      &= -\nabla_q H(\tilde{q}_s, \tilde{p}_s).
    \end{align}
  \end{proof}

  \section{Proofs of Reversibility and Volume Preservation of the Generalized Leapfrog Integrator}\label{app:reversibility-volume-preservation-generalized-leapfrog}
  \begin{proposition}
    Let $H :\R^m\times\R^m\to\R$ be a smooth Hamiltonian for which $-\nabla_p
    H(q, -p) = \nabla_pH(q, p)$ and $\nabla_q H(q, p) = \nabla_qH(q, -p)$. The generalized leapfrog integrator is reversible under negation of the
    momentum variable.
  \end{proposition}
  \begin{proof}
  Given an initial position in phase-space $(q_n, p_n)$, the generalized leapfrog integrator computes the following series of updates:
  \begin{align}
      \label{eq:generalized-leapfrog-first-momentum} p_{n+1/2} &= p_n - \frac{\epsilon}{2} \nabla_q H(q_n, p_{n+1/2}) \\
      \label{eq:generalized-leapfrog-position} q_{n+1} &= q_n + \frac{\epsilon}{2}\paren{\nabla_p H(q_n, p_{n+1/2}) + \nabla_p H(q_{n+1}, p_{n+1/2})} \\
      \label{eq:generalized-leapfrog-final-momentum} p_{n+1} &= p_{n+1/2} - \frac{\epsilon}{2} \nabla_qH(q_{n+1}, p_{n+1/2})).
  \end{align}
  Now we apply the integrator a second time from the initial position $(q_{n+1}, -p_{n+1})$. This produces the following update to the momentum:
  \begin{align}
      p_{n+1+1/2} &= -p_{n+1} - \frac{\epsilon}{2} \nabla_q H(q_{n+1}, p_{n+1+1/2})
  \end{align}
  But by rearranging \cref{eq:generalized-leapfrog-final-momentum} and using the condition $\nabla_q H(q, p) = \nabla_qH(q, -p)$ we have,
  \begin{align}
      -p_{n+1/2} = -p_{n+1} -\frac{\epsilon}{2} \nabla_qH(q_{n+1}, -p_{n+1/2}),
  \end{align}
  so we obtain $p_{n+1+1/2} = -p_{n+1/2}$. Now we compute the update to the momentum variable and use the fact that $-\nabla_p
    H(q, -p) = \nabla_pH(q, p)$ to obtain,
    \begin{align}
        q_{n+2} &= q_{n+1} + \frac{\epsilon}{2} \paren{\nabla_p H(q_{n+1}, -p_{n+1/2}) + \nabla_p H(q_{n+2}, -p_{n+1/2})} \\
        &= q_{n+1} - \frac{\epsilon}{2} \paren{\nabla_p H(q_{n+1}, p_{n+1/2}) + \nabla_p H(q_{n+2}, p_{n+1/2})}.
    \end{align}
    By rearranging \cref{eq:generalized-leapfrog-position}, we see that $q_{n+2} = q_{n}$ solves this implicit relation. Finally we compute the explicit update to the momentum variable,
    \begin{align}
        p_{n+2} &= -p_{n+1/2} - \frac{\epsilon}{2} \nabla_qH(q_n, -p_{n+1/2}) \\
        &= -p_{n+1/2} - \frac{\epsilon}{2} \nabla_qH(q_n, p_{n+1/2}) \\
        &=-p_n
    \end{align}
    by rearranging \cref{eq:generalized-leapfrog-first-momentum}.
  \end{proof}
  \begin{proposition}
    The generalized leapfrog integrator is volume preserving.
  \end{proposition}
  \begin{proof}
    One can prove that the generalized leapfrog integrator is volume preserving
    by employing an analysis of differential forms. We refer the interested
    reader to \citet{leimkuhler_reich_2005} for an introduction to this
    technique. The first step is to compute the coordinate 1-forms of the
    integrator.
    \begin{align}
      \mathrm{d}p_{n+1/2} &= \mathrm{d}p_n -\frac{\epsilon}{2} \nabla_q\nabla_q H(q_n, p_{n+1/2}) ~\mathrm{d}q_n - \frac{\epsilon}{2}\nabla_p\nabla_q H(q_n, p_{n+1/2}) ~\mathrm{d}p_{n+1/2} \\
      \mathrm{d}q_{n+1} &= \mathrm{d}q_n + \frac{\epsilon}{2}\paren{\nabla_q\nabla_p H(q_n, p_{n+1/2}) ~\mathrm{d}q_n + \nabla_q\nabla_p H(q_{n+1}, p_{n+1/2}) ~\mathrm{d}q_{n+1}} + (**) ~\mathrm{d}p_{n+1/2} \\
      \mathrm{d}p_{n+1} &= \mathrm{d}p_{n+1/2} - \frac{\epsilon}{2} \nabla_q\nabla_q H(q_{n+1}, p_{n+1/2})~\mathrm{d}q_{n+1} - \frac{\epsilon}{2} \nabla_p\nabla_q H(q_{n+1}, p_{n+1/2}) ~\mathrm{d}p_{n+1/2}
    \end{align}
    Next we use the coordinate 1-forms to compute wedge products:
    \begin{align}
      \mathrm{d}q_n \wedge\mathrm{d}p_{n+1/2} &= \mathrm{d}q_n\wedge\mathrm{d}p_n - \frac{\epsilon}{2} \mathrm{d}q_n\wedge\nabla_p\nabla_q H(q_n, p_{n+1/2}) ~\mathrm{d}p_{n+1/2} \\
      \begin{split}
        \mathrm{d}q_{n+1} \wedge\mathrm{d}p_{n+1/2} &= \mathrm{d}q_n \wedge\mathrm{d}p_{n+1/2} + \frac{\epsilon}{2} \nabla_q\nabla_p H(q_n, p_{n+1/2}) ~\mathrm{d}q_n \wedge\mathrm{d}p_{n+1/2} \\
        &\qquad + \frac{\epsilon}{2} \nabla_q\nabla_p H(q_{n+1}, p_{n+1/2}) ~\mathrm{d}q_{n+1} \wedge\mathrm{d}p_{n+1/2} \\
      \end{split} \\
      \mathrm{d}q_{n+1} \wedge\mathrm{d}p_{n+1/2} &= \mathrm{d}q_n\wedge\mathrm{d}p_n + \frac{\epsilon}{2} \nabla_q\nabla_p H(q_{n+1}, p_{n+1/2})~\mathrm{d}q_{n+1}\wedge\mathrm{d}p_{n+1/2}\\
      \mathrm{d}q_{n+1} \wedge\mathrm{d}p_{n+1} &= \mathrm{d}q_{n+1} \wedge \mathrm{d}p_{n+1/2} - \mathrm{d}q_{n+1} \wedge \frac{\epsilon}{2}\nabla_p\nabla_q H(q_{n+1}, p_{n+1/2}) ~\mathrm{d}p_{n+1/2} \\
      \label{eq:generalized-leapfrog-volume-preservation} \mathrm{d}q_{n+1} \wedge\mathrm{d}p_{n+1} &= \mathrm{d}q_n\wedge\mathrm{d}p_n,
    \end{align}
    which establishes volume preservation. (Actually, this proof establishes an
    even stronger property known as {\it symplecticness}, which is important in
    the theory of Hamiltonian mechanics, but which we have eschewed in our
    discussion.) The fact that $\mathrm{d}q_{n+1} \wedge \nabla_p\nabla_q
    H(q_{n+1}, p_{n+1/2}) ~\mathrm{d}p_{n+1/2} = \nabla_q\nabla_p H(q_{n+1},
    p_{n+1/2})~\mathrm{d}q_{n+1}\wedge\mathrm{d}p_{n+1/2}$ is used in
    \cref{eq:generalized-leapfrog-volume-preservation}, which is predicated on
    the symmetry of partial derivatives. One therefore sees that one has
    required the symmetry of partial derivatives in order to establish the
    volume-preservation property. One concludes, therefore, that one cannot
    apply the generalized leapfrog method to {\it arbitrary} vector fields in
    position and momentum and expect to obtain a volume-preserving numerical
    integrator.
  \end{proof}

  \section{The Stationary Distribution of a Reversible Proposal with Non-Unit Jacobian Determinant}\label{app:stationary-distribution-non-volume-preserving}
  \begin{proposition}
    Let $\Phi : \R^m\to\R^m$ be a smooth function that is self-inverse; that is, $\Phi = \Phi^{-1}$. Denote the Jacobian determinant of $\Phi$ at $z\in \R^m$ by $J_\Phi(z)$. Then
    \begin{align}
      J_{\Phi}(\Phi(z)) = \frac{1}{J_{\Phi}(z)}.
    \end{align}
  \end{proposition}
  \begin{proof}
    By the inverse function theorem
    \begin{align}
      \mathrm{Id} = \nabla \Phi^{-1}(\Phi(z)) \cdot \nabla \Phi(z).
    \end{align}
    Taking the determinant on both sides gives,
    \begin{align}
      & 1 = J_{\Phi^{-1}}(\Phi(z)) J_\Phi(z) \\
      \implies& J_{\Phi^{-1}}(\Phi(z)) = \frac{1}{J_\Phi(z)}.
    \end{align}
    The fact that $\Phi=\Phi^{-1}$ then yields the desired conclusion.
  \end{proof}
  \begin{proposition}
    Let $\pi : \R^m\to\R_+$ be a density and let $\Phi : \R^m\to\R^m$ be a smooth, self-inverse function. Consider a Markov chain at state $z\in\R^m$ and transitions to state $z' = \Phi(z)$ with probability $\min\set{1, \pi(\Phi(z)) / \pi(z)}$. Denote this Markov chain transition operator by
    \begin{align}
      \Psi(z) = \begin{cases} \Phi(z) & ~\mathrm{w.p.}~\frac{\pi(\Phi(z))}{\pi(z)} \\ z &~\mathrm{otherwise.} \end{cases}
    \end{align}
    Then the stationary distribution of the chain is
    \begin{align}
      \bar{\pi}(z) \propto \pi(z) \cdot \sqrt{\abs{J_\Phi(z)}}.
    \end{align}
  \end{proposition}
  \begin{proof}
    The detailed balance condition states that for all regions $A, B\subset \R^m$ we have
    \begin{align}
      \underset{z\sim \bar{\pi}}{\mathrm{Pr}}\left[z\in A ~\mathrm{and}~ \Phi(z)\in B\right] = \underset{z\sim \bar{\pi}}{\mathrm{Pr}}\left[z\in B ~\mathrm{and}~ \Phi(z)\in A\right].
    \end{align}
    Let $Z \defeq \int_{\R^m} \pi(z) \cdot \sqrt{\abs{J_\Phi(z)}}~\mathrm{d}z$ be the normalizing constant of $\bar{\pi}(z)$. We have,
    \begin{align}
      &\int_{\R^m} \mathbf{1}\set{z\in A} \cdot \mathbf{1}\set{\Phi(z) \in B}\cdot \bar{\pi}(z)\cdot \min\set{1, \frac{\pi(\Phi(z))}{\pi(z)}} ~\mathrm{d}z \\
      =& \frac{1}{Z}\int_{\R^m} \mathbf{1}\set{z\in A} \cdot \mathbf{1}\set{\Phi(z) \in B}\cdot \pi(z) \cdot \sqrt{\abs{J_\Phi(z)}} \cdot \min\set{1, \frac{\pi(\Phi(z))}{\pi(z)}} ~\mathrm{d}z \\
      =& \frac{1}{Z} \int_{\R^m} \mathbf{1}\set{\Phi(z')\in A}\cdot \mathbf{1}\set{z'\in B}\cdot  \pi(\Phi(z')) \cdot \sqrt{\abs{J_\Phi(\Phi(z'))}} \cdot \min\set{1,\frac{\pi(z')}{\pi(\Phi(z'))}} \cdot \abs{J_\Phi(z')} ~\mathrm{d}z' \\
      =& \frac{1}{Z} \int_{\R^m} \mathbf{1}\set{\Phi(z')\in A}\cdot \mathbf{1}\set{z'\in B}\cdot  \pi(\Phi(z')) \cdot \sqrt{\frac{1}{\abs{J_\Phi(z')}}}\cdot  \min\set{1,\frac{\pi(z')}{\pi(\Phi(z'))}} \cdot \abs{J_\Phi(z')} ~\mathrm{d}z' \\
      =& \frac{1}{Z} \int_{\R^m} \mathbf{1}\set{\Phi(z')\in A}\cdot \mathbf{1}\set{z'\in B} \cdot \pi(z') \cdot \sqrt{\abs{J_\Phi(z')}}\cdot  \min\set{1,\frac{\pi(\Phi(z'))}{\pi(z')}} ~\mathrm{d}z' \\
      =& \int_{\R^m} \mathbf{1}\set{z\in B} \cdot \mathbf{1}\set{\Phi(z) \in A}\cdot \bar{\pi}(z)\cdot \min\set{1, \frac{\pi(\Phi(z))}{\pi(z)}} ~\mathrm{d}z 
    \end{align}
    and similarly,
    \begin{align}
    \begin{split}
        & \int_{\R^m} \mathbf{1}\set{z\in A} \cdot \mathbf{1}\set{z \in B}\cdot \bar{\pi}(z)\cdot \paren{1 - \min\set{1, \frac{\pi(\Phi(z))}{\pi(z)}}} ~\mathrm{d}z = \\ &\qquad \int_{\R^m} \mathbf{1}\set{z\in B} \cdot \mathbf{1}\set{z \in A}\cdot \bar{\pi}(z)\cdot \paren{1 - \min\set{1, \frac{\pi(\Phi(z))}{\pi(z)}}} ~\mathrm{d}z.
    \end{split}
    \end{align}
    These two equations verify the symmetry of $A$ and $B$ in the Metropolis-Hastings accept-reject decision. Therefore, 
    \begin{align}
    \begin{split}
        \underset{z\sim \bar{\pi}}{\mathrm{Pr}}\left[z\in A ~\mathrm{and}~ \Phi(z)\in B\right] &= \int_{\R^m} \mathbf{1}\set{z\in A} \cdot \mathbf{1}\set{\Phi(z) \in B}\cdot \bar{\pi}(z)\cdot \min\set{1, \frac{\pi(\Phi(z))}{\pi(z)}} ~\mathrm{d}z \\ 
        &\qquad + \int_{\R^m} \mathbf{1}\set{z\in A} \cdot \mathbf{1}\set{z \in B}\cdot \bar{\pi}(z)\cdot \paren{1 - \min\set{1, \frac{\pi(\Phi(z))}{\pi(z)}}} ~\mathrm{d}z
    \end{split} \\
    \begin{split}
        &= \int_{\R^m} \mathbf{1}\set{z\in B} \cdot \mathbf{1}\set{\Phi(z) \in A}\cdot \bar{\pi}(z)\cdot \min\set{1, \frac{\pi(\Phi(z))}{\pi(z)}} ~\mathrm{d}z \\ 
        &\qquad + \int_{\R^m} \mathbf{1}\set{z\in B} \cdot \mathbf{1}\set{z \in A}\cdot \bar{\pi}(z)\cdot \paren{1 - \min\set{1, \frac{\pi(\Phi(z))}{\pi(z)}}} ~\mathrm{d}z
    \end{split} \\
    &= \underset{z\sim \bar{\pi}}{\mathrm{Pr}}\left[z\in B ~\mathrm{and}~ \Phi(z)\in A\right]
    \end{align}
    Since $\bar{\pi}$ satisfies detailed balance with respect to the transition operator, it follows that $\bar{\pi}$ is the stationary distribution of the Markov chain with transition operator $\Psi$.
  \end{proof}
  
  This result demonstrates what can go wrong when the Metropolis-Hastings accept-reject rule does not account for the change-of-volume of the function $\Phi$. We can interpret $\bar{\pi}(z)$ as proportional to a perturbation of $\pi(z)$ according to the square-root of the Jacobian determinant of the transformation $\Phi$ at $z$. More formally,
  \begin{proposition}
    \begin{align}
      \mathrm{KL}(\pi\Vert \bar{\pi}) &= \log \underset{z\sim\pi}{\mathbb{E}} \sqrt{\abs{J_\Phi(z)}} - \underset{z\sim \pi}{\mathbb{E}} \log \sqrt{\abs{J_\Phi(z)}} \\
      &\geq 0.
    \end{align}
  \end{proposition}
  \begin{proof}
    \begin{align}
      \mathrm{KL}(\pi\Vert \bar{\pi}) &= \underset{z\sim\pi}{\mathbb{E}} \log \frac{\pi(z)}{\bar{\pi}(z)} \\
      &= \underset{z\sim\pi}{\mathbb{E}}\left[\log\pi(z) -\log\pi(z) - \log\sqrt{J_\Phi(z)} +\log Z\right] \\
      &= \log Z - \underset{z\sim\pi}{\mathbb{E}}\left[\log\sqrt{J_\Phi(z)}\right]
    \end{align}
    Now recall
    \begin{align}
      Z &= \int_{\R^m} \pi(z) \cdot \sqrt{\abs{J_\Phi(z)}}~\mathrm{d}z \\
      &= \underset{z\sim\pi}{\mathbb{E}} \sqrt{\abs{J_\Phi(z)}},
    \end{align}
    which proves equality. Non-negativity then follows from Jensen's inequality.
  \end{proof}
  
  \section{Sensitivity of Jacobian Determinant Estimates to Perturbation Size}\label{app:jacobian-determinant-perturbation}
  \begin{figure}[t!]
  \begin{subfigure}[t]{0.3\textwidth}
    \centering
    \includegraphics[width=\textwidth]{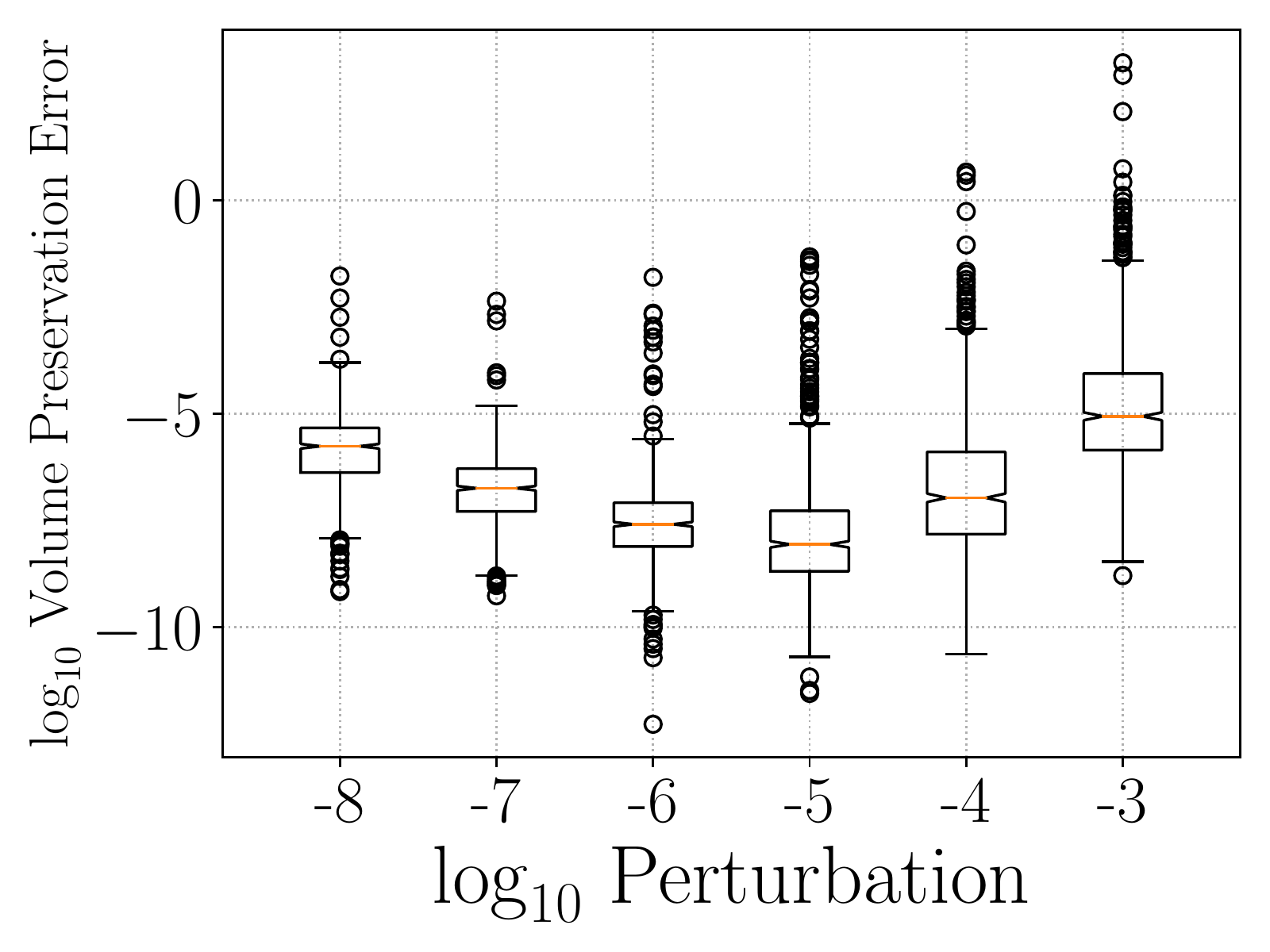}
    \caption{Banana}
    \label{subfig:jacobian-determinant-banana}
  \end{subfigure}
  ~
  \begin{subfigure}[t]{0.3\textwidth}
    \centering
    \includegraphics[width=\textwidth]{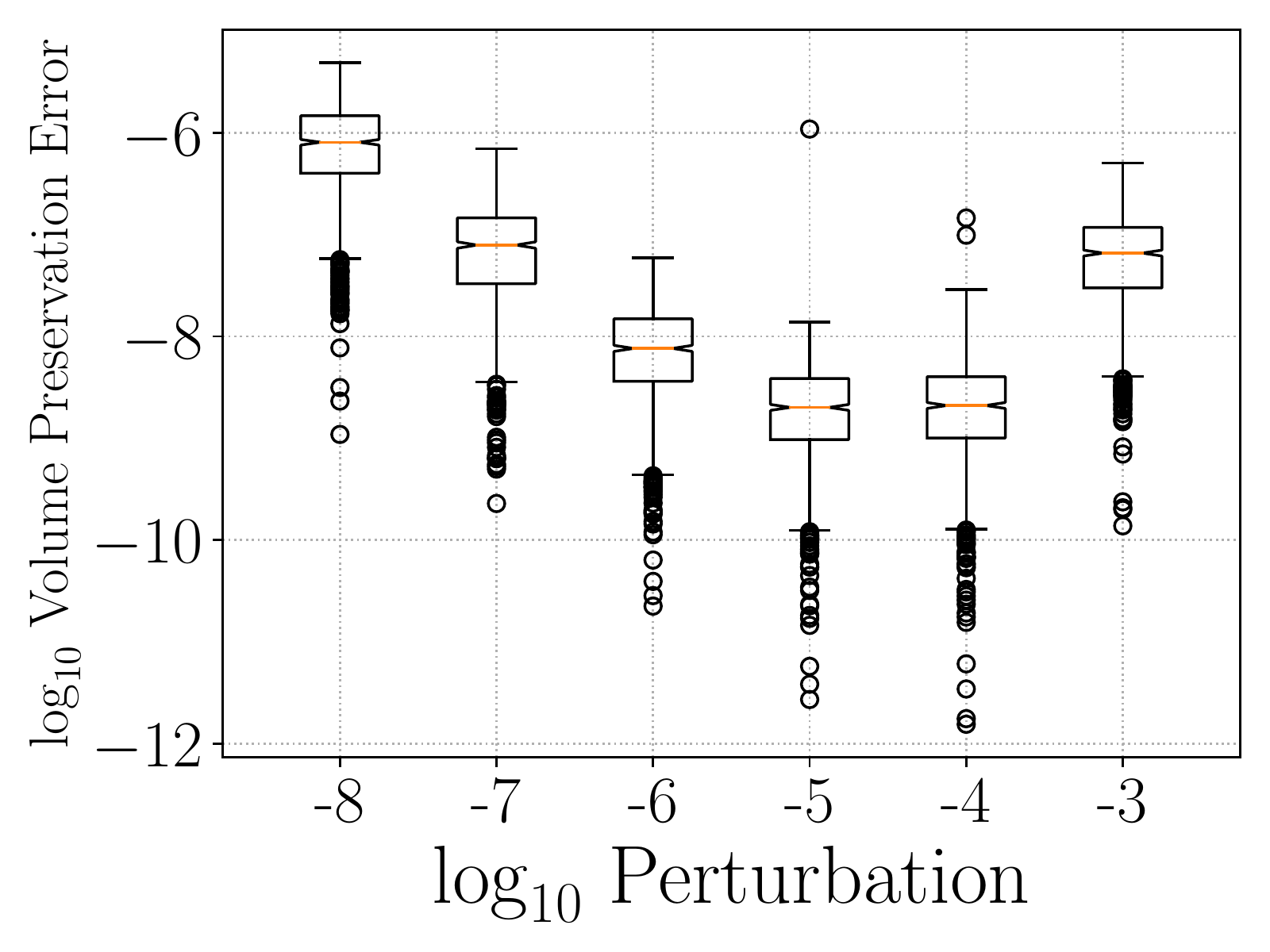}
    \caption{Logistic Regression}
    \label{subfig:jacobian-determinant-logistic}
  \end{subfigure}
  ~
  \begin{subfigure}[t]{0.3\textwidth}
    \centering
    \includegraphics[width=\textwidth]{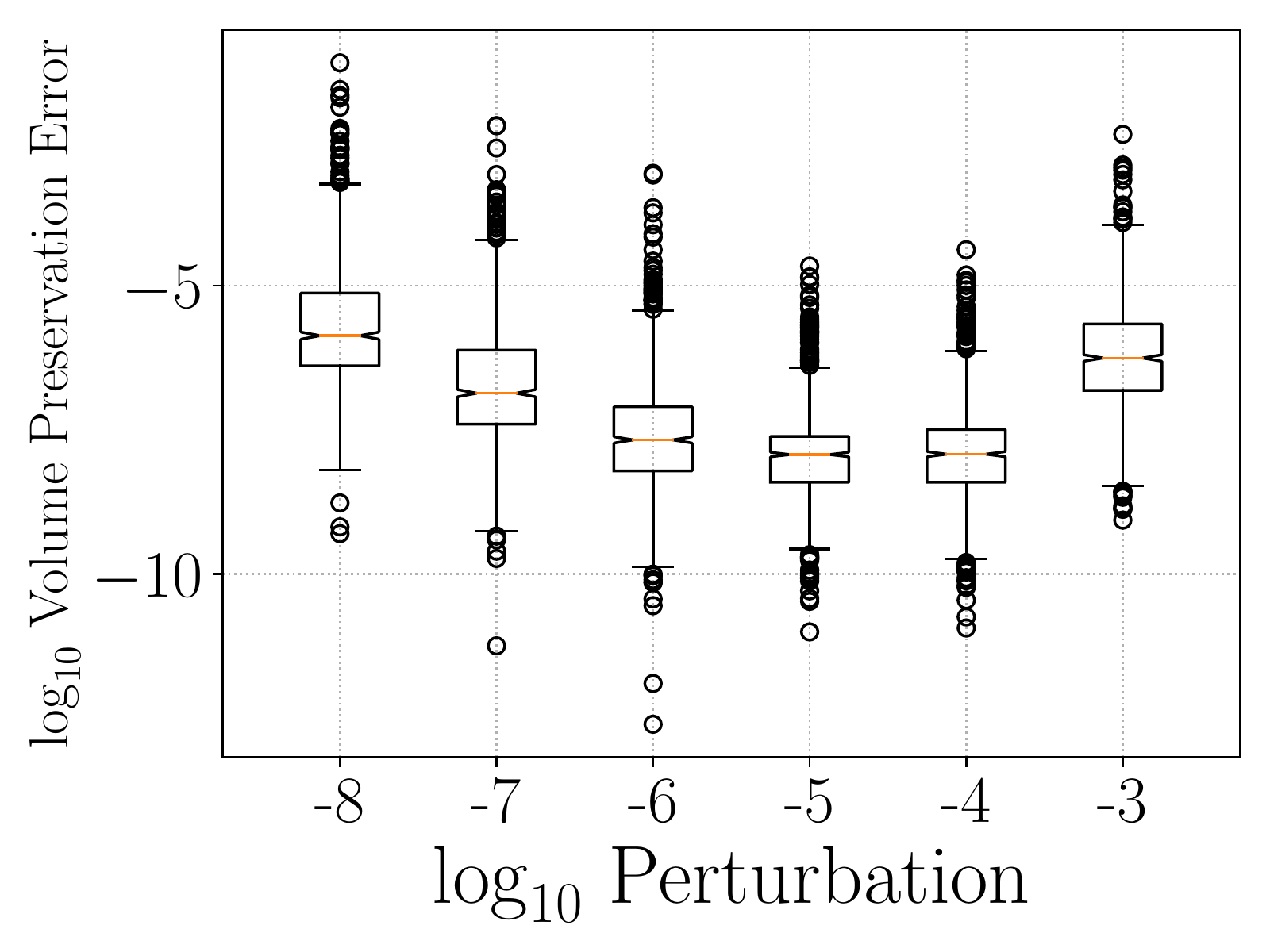}
    \caption{Neal Funnel}
    \label{subfig:jacobian-determinant-neal-funnel}
  \end{subfigure}

  \begin{subfigure}[t]{0.3\textwidth}
    \centering
    \includegraphics[width=\textwidth]{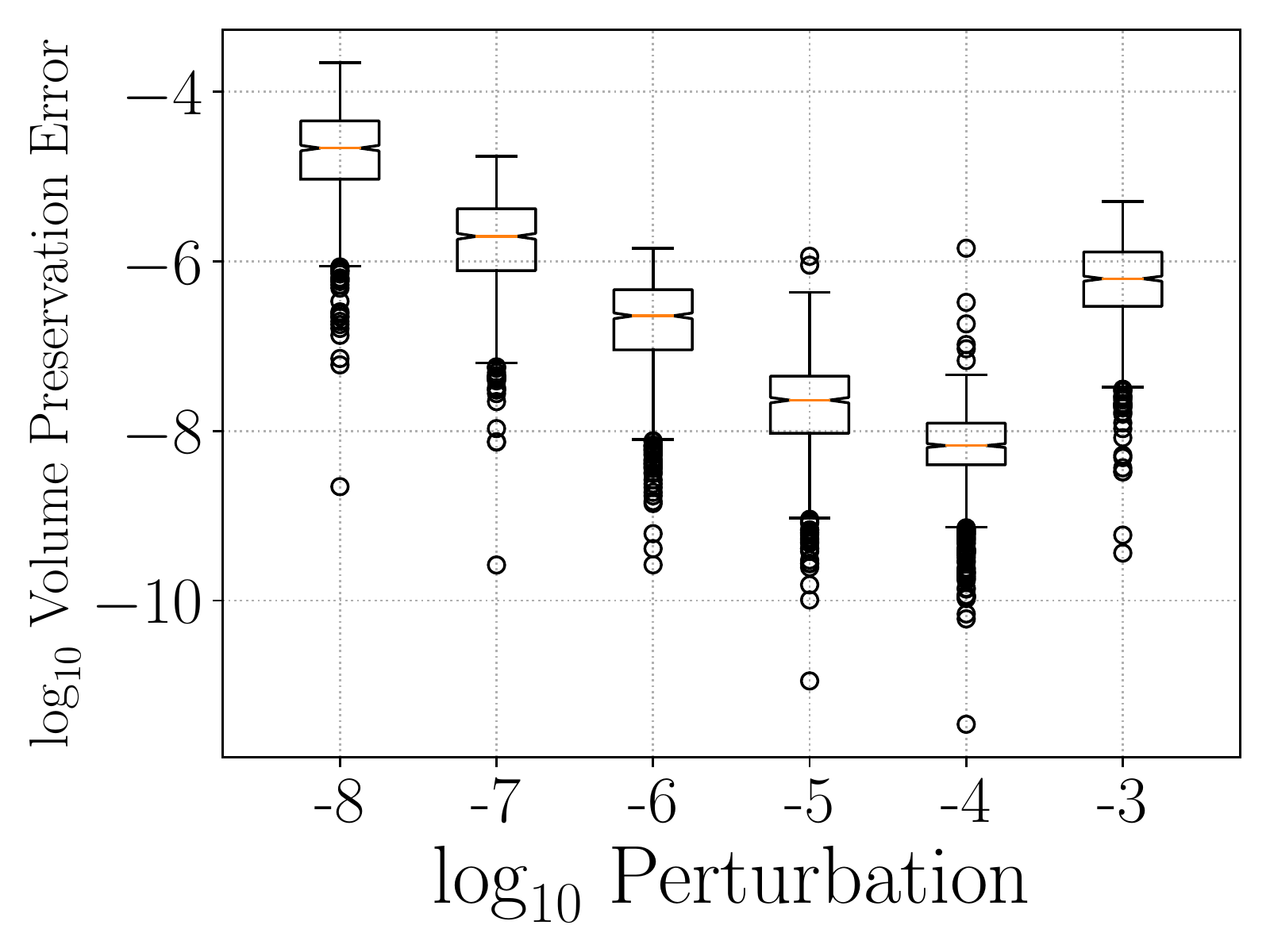}
    \caption{Stochastic Volatility}
    \label{subfig:jacobian-determinant-stochastic-volatility}
  \end{subfigure}
  ~
  \begin{subfigure}[t]{0.3\textwidth}
    \centering
    \includegraphics[width=\textwidth]{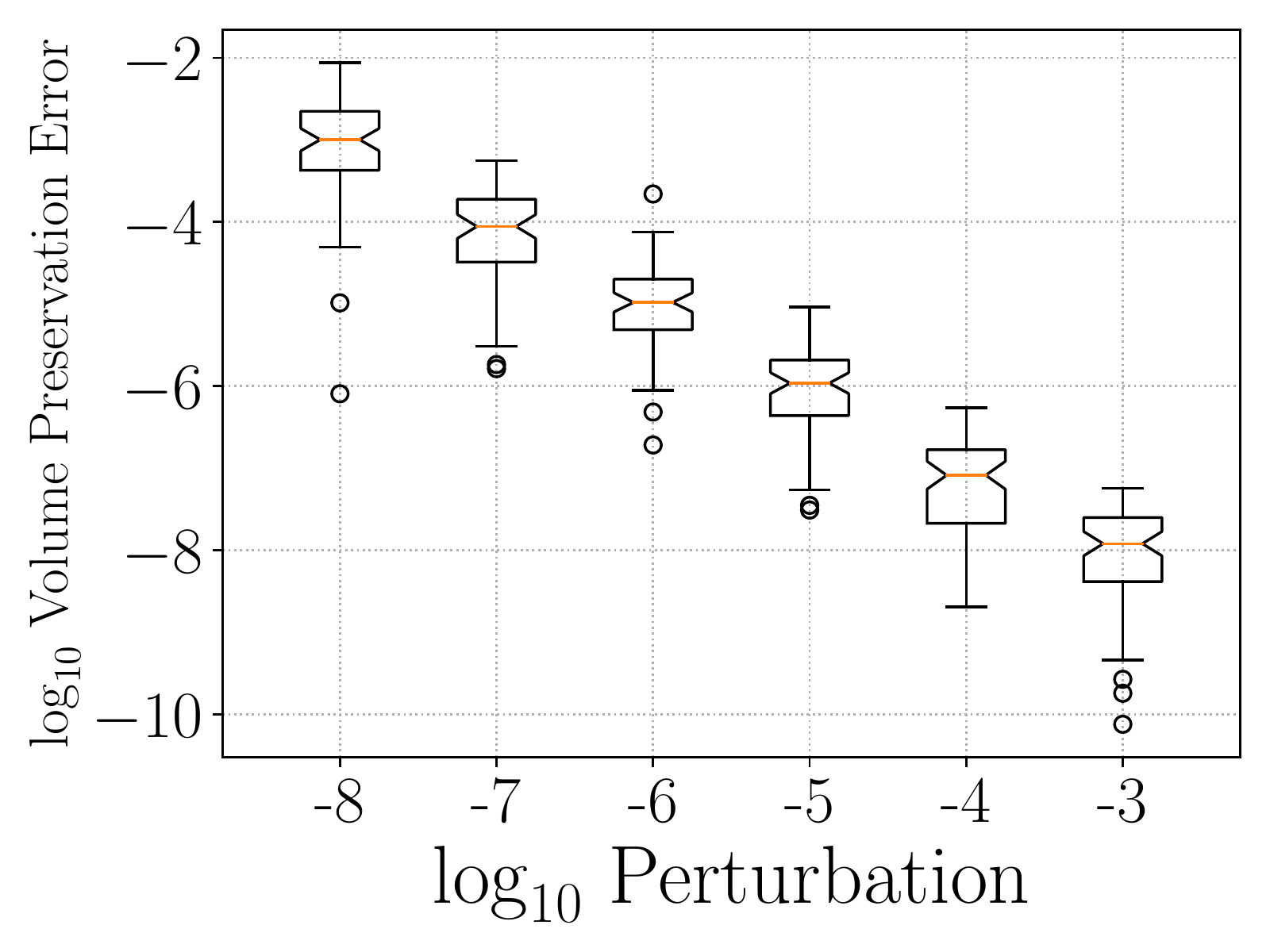}
    \caption{Cox-Poisson}
    \label{subfig:jacobian-determinant-cox-poisson}
  \end{subfigure}
  ~
  \begin{subfigure}[t]{0.3\textwidth}
    \centering
    \includegraphics[width=\textwidth]{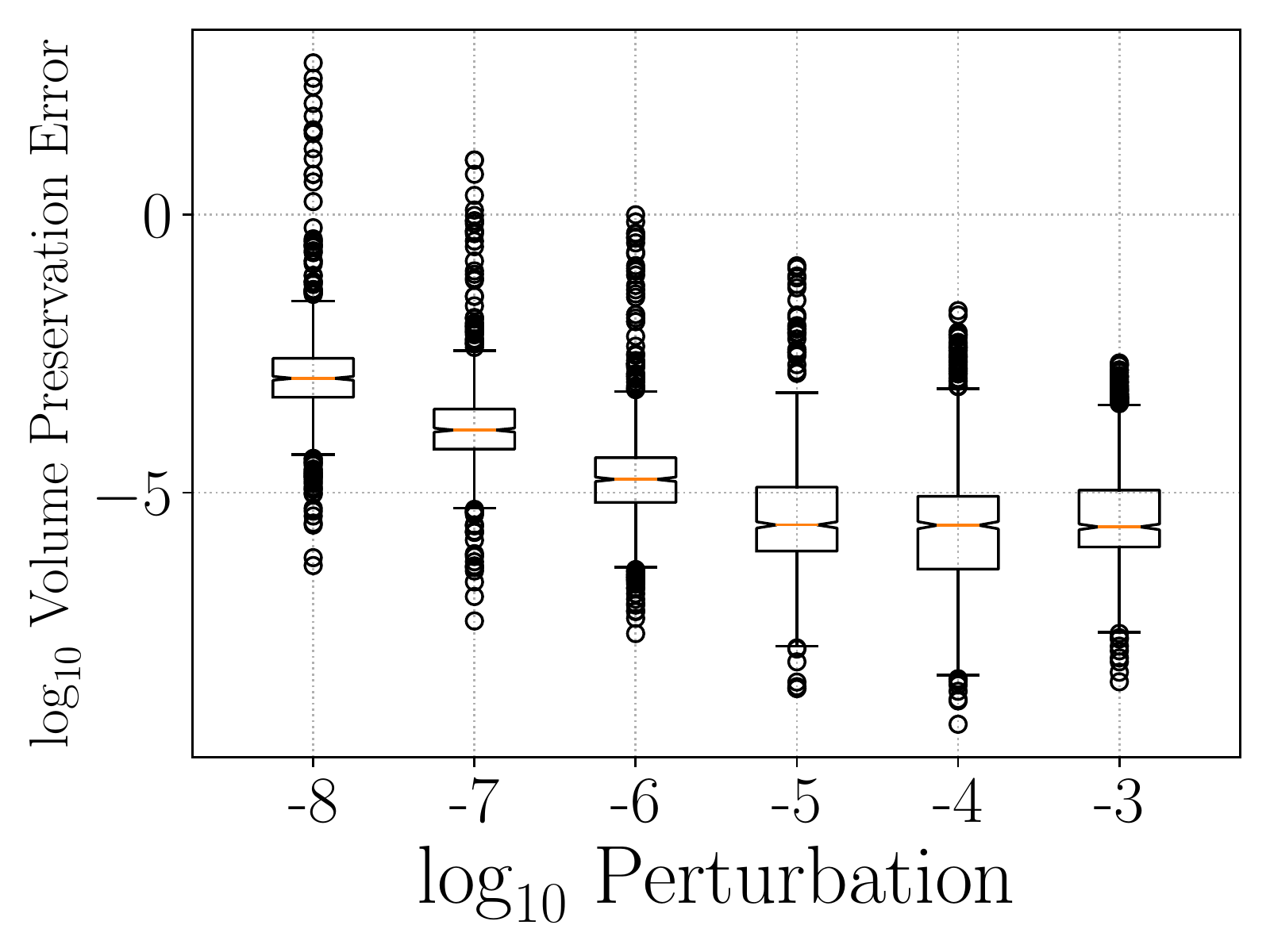}
    \caption{Fitzhugh-Nagumo}
    \label{subfig:jacobian-determinant-fitzhugh-nagumo}
  \end{subfigure}

  \begin{subfigure}[t]{0.3\textwidth}
    \centering
    \includegraphics[width=\textwidth]{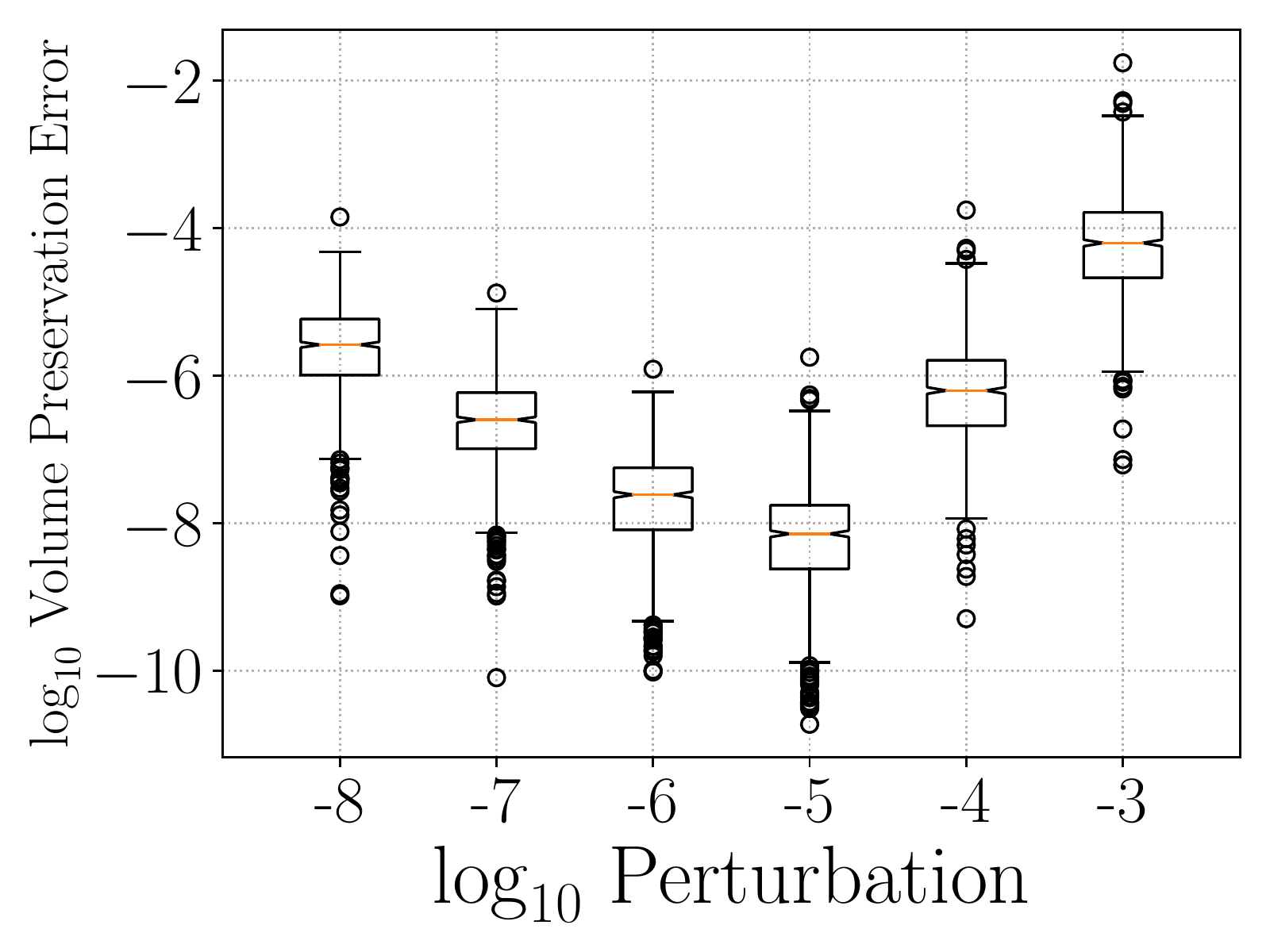}
    \caption{Student-$t$}
    \label{subfig:jacobian-determinant-student-t}
  \end{subfigure}

  \caption{We show how the estimated volume preservation of the generalized leapfrog integrator is affected by the choice of finite difference perturbation size. We report the degree of volume-preservation for the generalized leapfrog integrator with a convergence threshold of $\delta= 1\times 10^{-9}$.}
  \label{fig:jacobian-determinant-perturbation}
\end{figure}

  We evaluate the sensitivity of our estimates of volume preservation to the finite difference perturbation size. In \cref{fig:jacobian-determinant-perturbation} we show the distribution of Jacobian determinants for all of our experiments for perturbations in the set $\set{1\times 10^{-8}, 1\times 10^{-7}, 1\times 10^{-6}, 1\times 10^{-5}, 1\times 10^{-4}, 1\times 10^{-3}}$. Our criterion to select a perturbation size is to identify which perturbation produces an estimated {\it median} Jacobian determinant that is closest to zero when using a convergence threshold of $1\times 10^{-9}$. For the banana-shaped distribution in \cref{subfig:jacobian-determinant-banana}, the selected perturbation is $1\times 10^{-5}$. For the Bayesian logistic regression model in \cref{subfig:jacobian-determinant-logistic}, the selected perturbation is $1\times 10^{-5}$. For Neal's funnel distribution in \cref{subfig:jacobian-determinant-neal-funnel}, the selected perturbation is $1\times 10^{-5}$. For the stochastic volatility model in \cref{subfig:jacobian-determinant-stochastic-volatility}, the selected perturbation is $1\times 10^{-4}$. For the Cox-Poission model in \cref{subfig:jacobian-determinant-cox-poisson}, the selected perturbation is $1\times 10^{-3}$. For the Fitzhugh-Nagumo posterior in \cref{subfig:jacobian-determinant-fitzhugh-nagumo}, the selected perturbation is $1\times 10^{-5}$. For the multiscale Student-$t$ distribution with $\sigma^2 = 1\times 10^4$ in \cref{subfig:jacobian-determinant-student-t}, the selected perturbation is $1\times 10^{-5}$.

  \section{Relative Error of Reversibility}\label{app:relative-reversibility-error}
  
  \begin{figure}[t!]
  \begin{subfigure}[t]{0.3\textwidth}
    \centering
    \includegraphics[width=\textwidth]{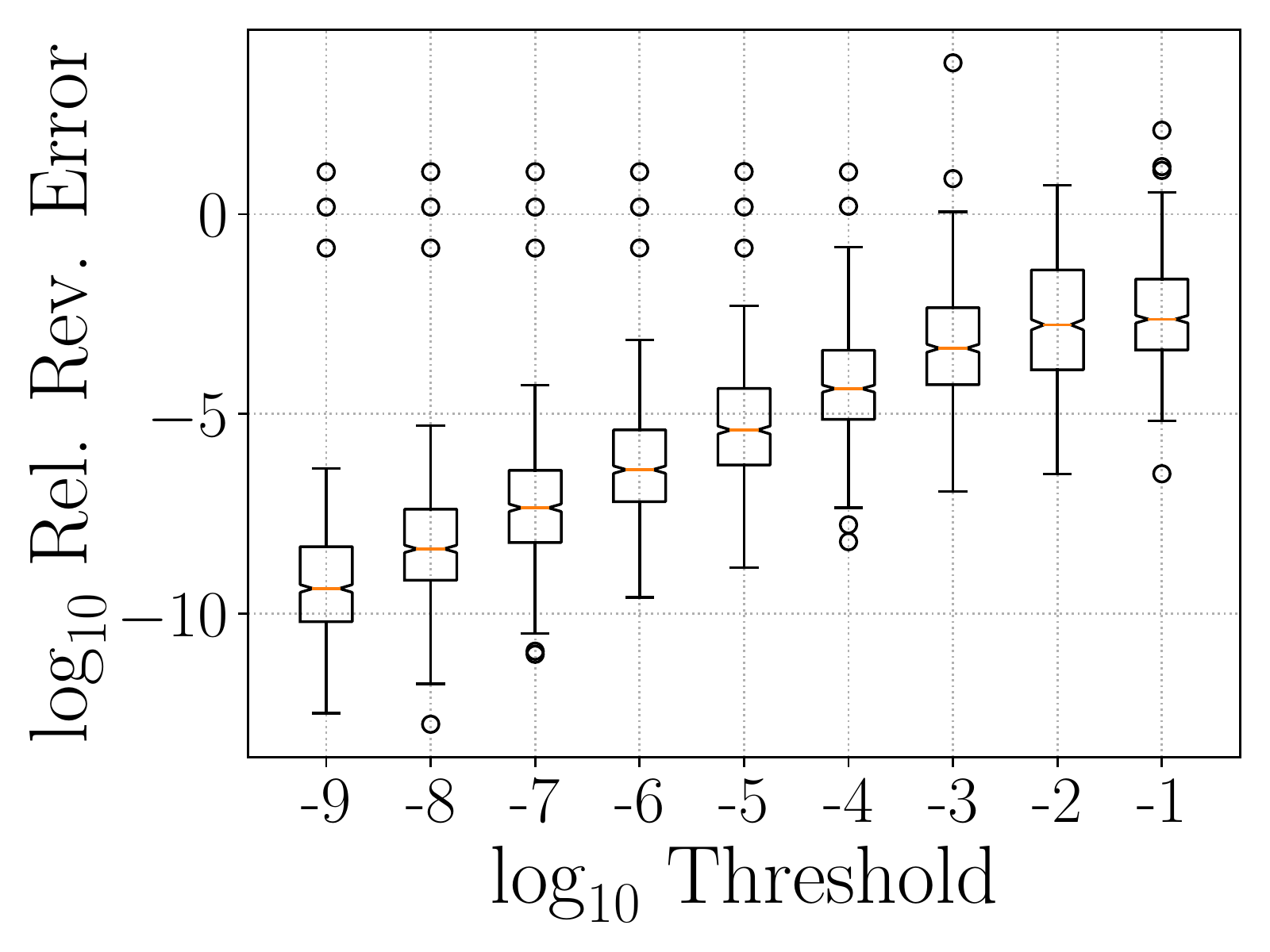}
    \caption{Banana}
    \label{subfig:relative-error-banana}
  \end{subfigure}
  ~
  \begin{subfigure}[t]{0.3\textwidth}
    \centering
    \includegraphics[width=\textwidth]{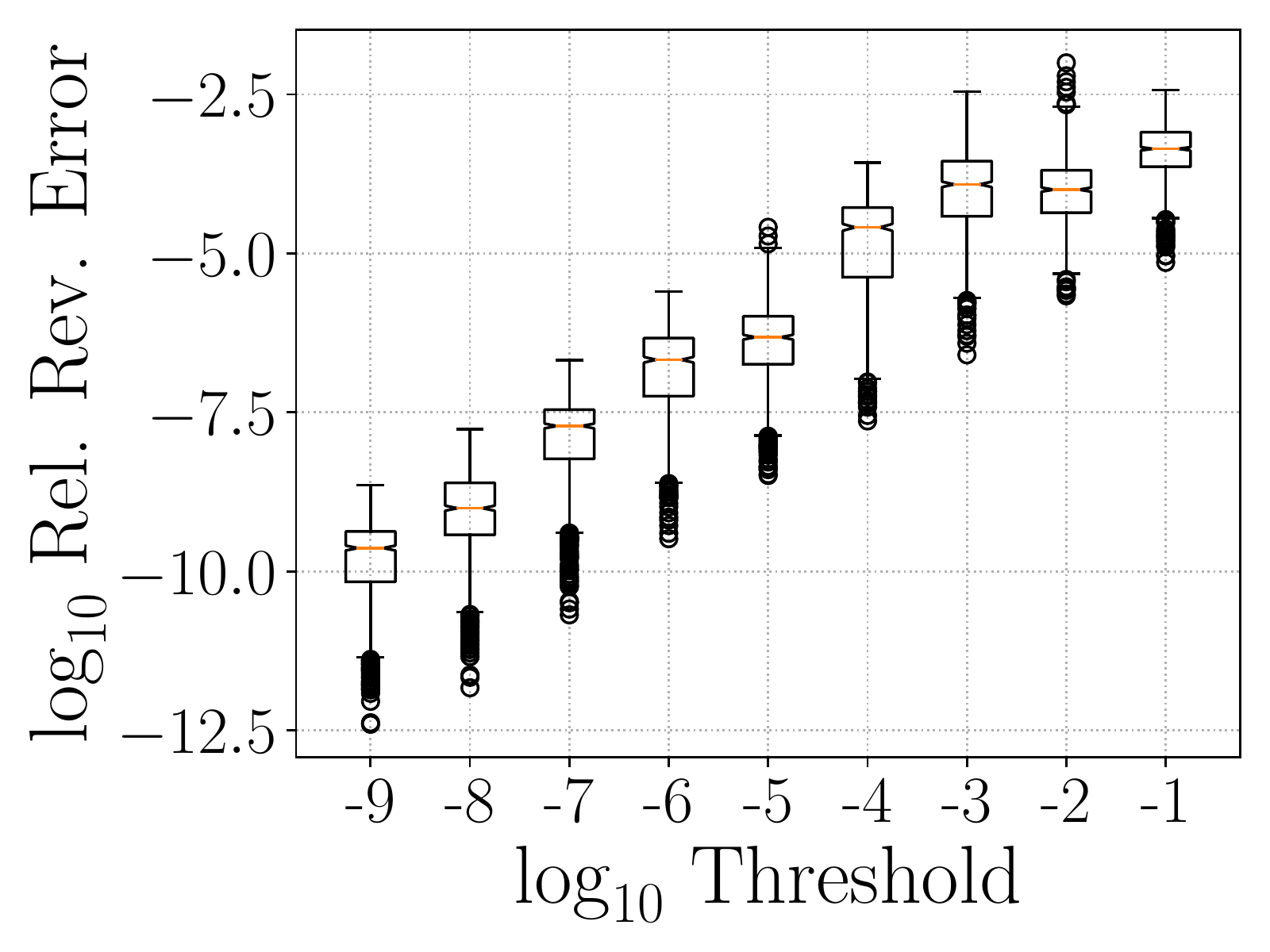}
    \caption{Logistic Regression}
    \label{subfig:relative-error-logistic}
  \end{subfigure}
  ~
  \begin{subfigure}[t]{0.3\textwidth}
    \centering
    \includegraphics[width=\textwidth]{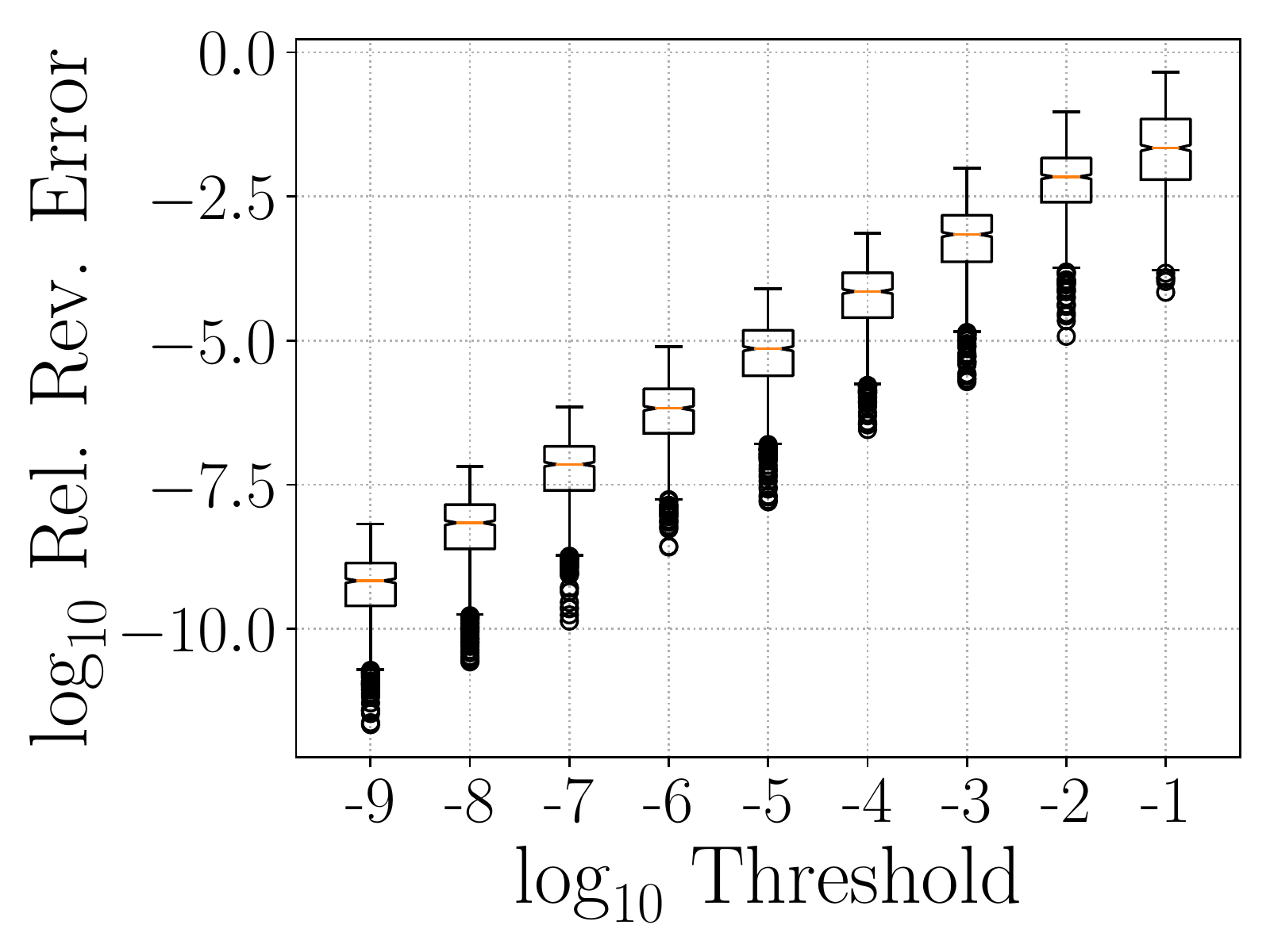}
    \caption{Neal Funnel}
    \label{subfig:relative-error-neal-funnel}
  \end{subfigure}

  \begin{subfigure}[t]{0.3\textwidth}
    \centering
    \includegraphics[width=\textwidth]{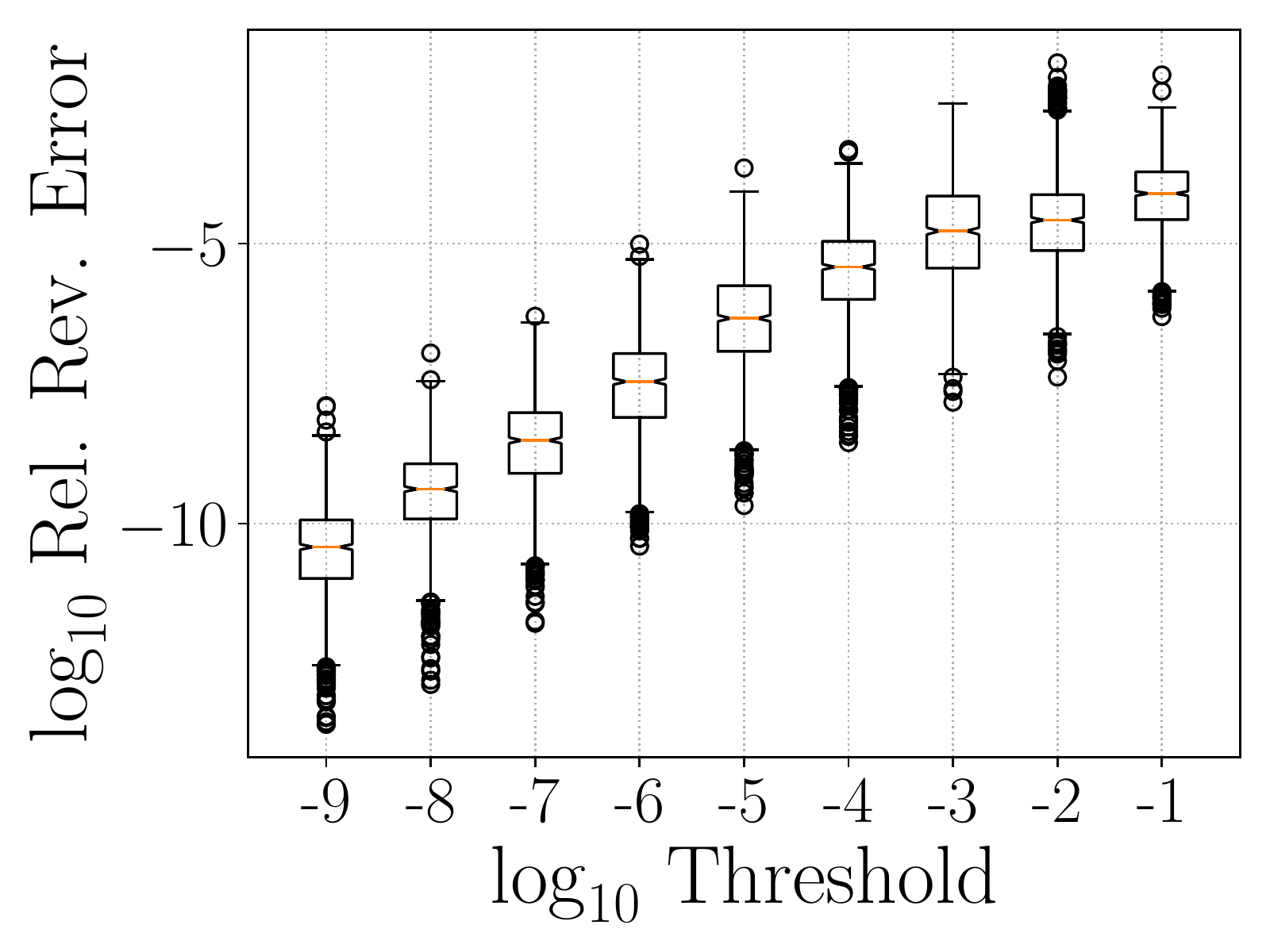}
    \caption{Stochastic Volatility}
    \label{subfig:relative-error-stochastic-volatility}
  \end{subfigure}
  ~
  \begin{subfigure}[t]{0.3\textwidth}
    \centering
    \includegraphics[width=\textwidth]{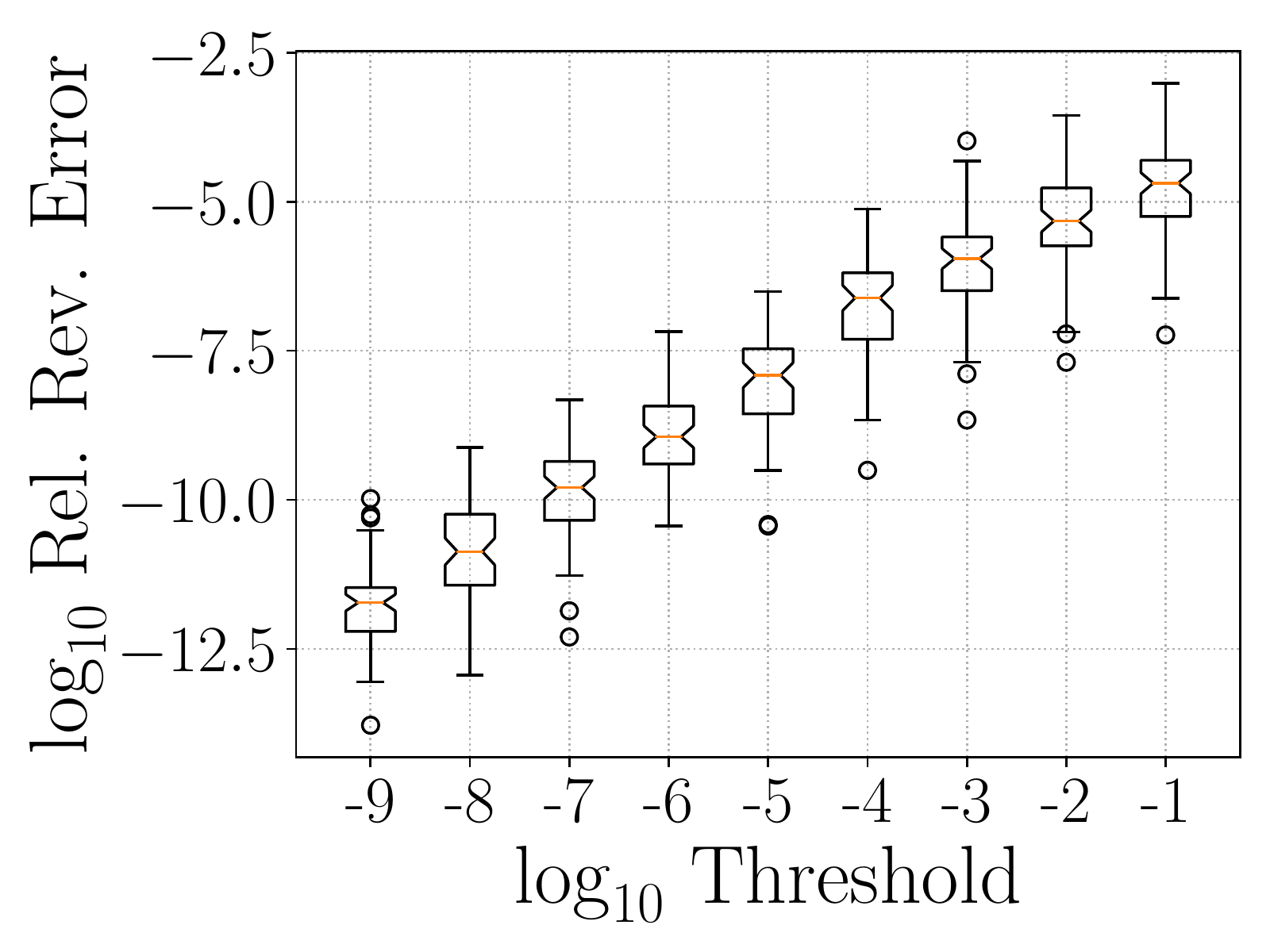}
    \caption{Cox-Poisson}
    \label{subfig:relative-error-cox-poisson}
  \end{subfigure}
  ~
  \begin{subfigure}[t]{0.3\textwidth}
    \centering
    \includegraphics[width=\textwidth]{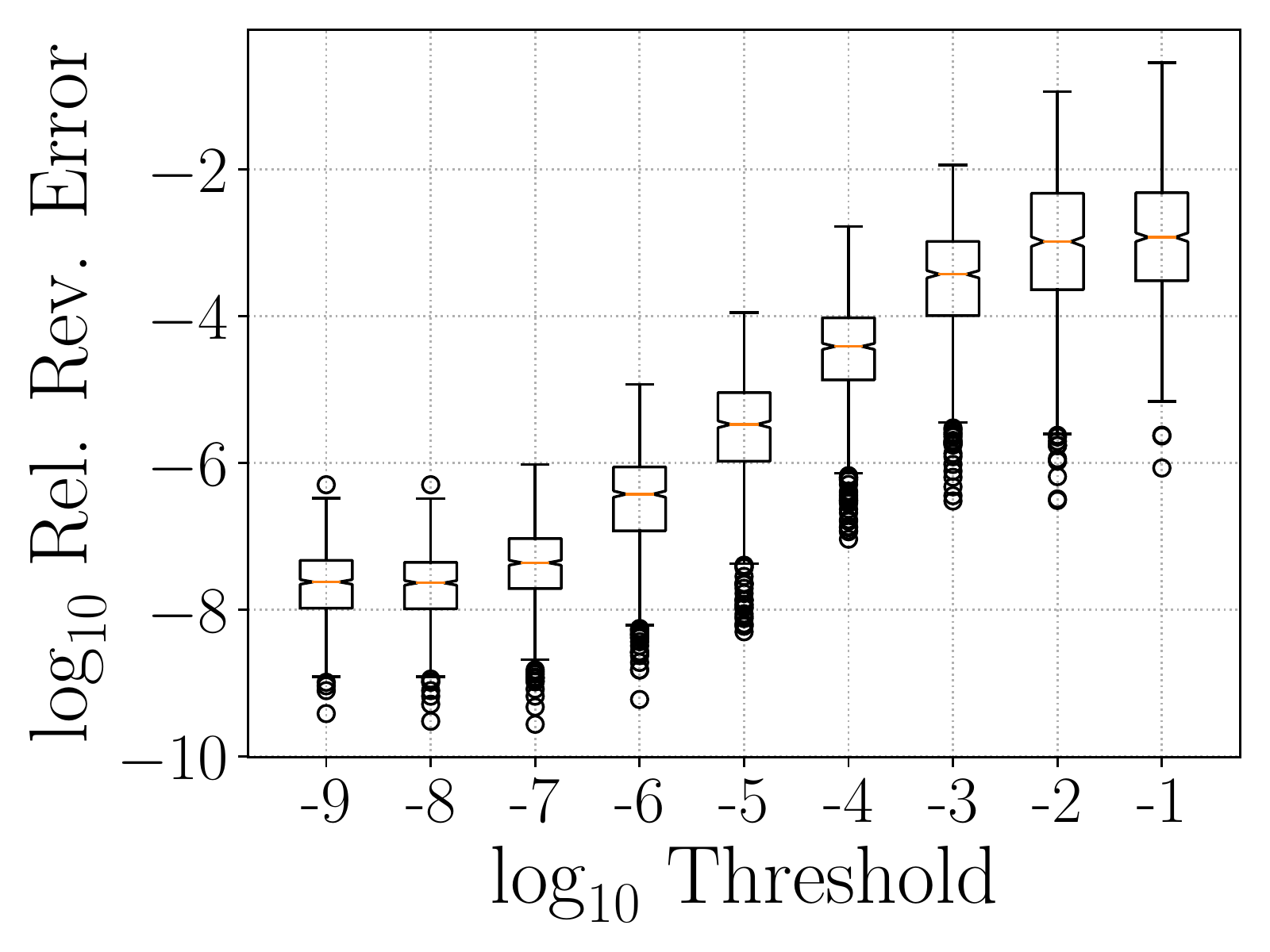}
    \caption{Fitzhugh-Nagumo}
    \label{subfig:relative-error-fitzhugh-nagumo}
  \end{subfigure}

  \begin{subfigure}[t]{0.3\textwidth}
    \centering
    \includegraphics[width=\textwidth]{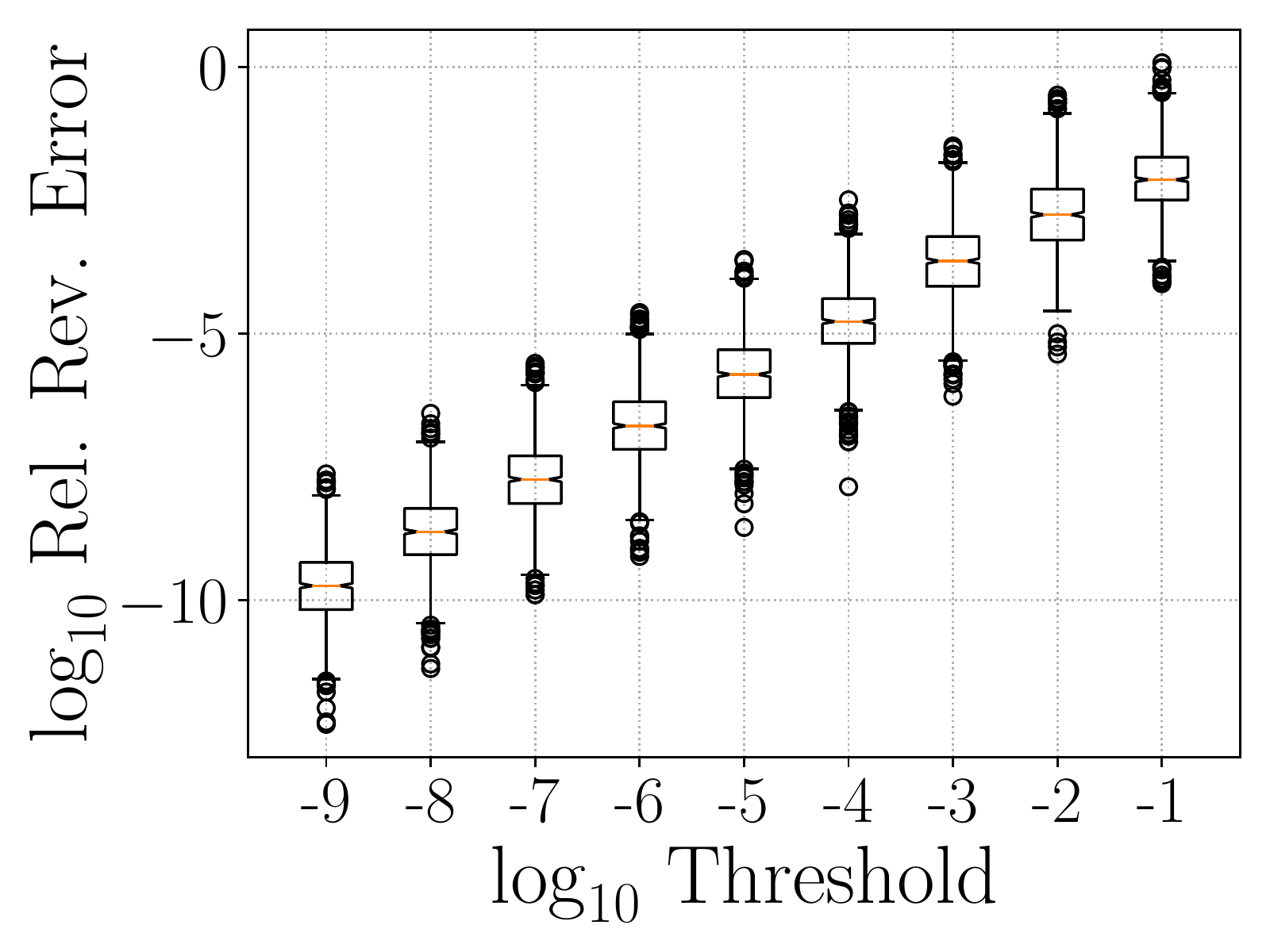}
    \caption{Student-$t$}
    \label{subfig:relative-error-student-t}
  \end{subfigure}

  \caption{We measure the relative error in reversibility by normalizing the absolute reversibility error by the norm of initial condition. This measure contextualizes the reversibility error in terms of the scale of the distribution.}
  \label{fig:reversibility-relative-error}
\end{figure}

  We show in \cref{fig:reversibility-relative-error} the relative error in reversibility for each of the posterior distributions considered in the main text. Qualitatively the relative error is similar to the absolute error in that, as expected, it is a decreasing function of the convergence threshold.
\end{appendix}

\begin{acks}[Acknowledgments]
The authors would like to thank Marcus A. Brubaker for helpful discussions. We thank Alex H. Barnett for the suggestion to investigate higher order methods for resolving the implicit updates in the generalized leapfrog method.
\end{acks}

\begin{funding}
This material is based upon work supported by the National Science Foundation Graduate Research Fellowship under Grant No. 1752134. Any opinion, findings, and conclusions or recommendations expressed in this material are those of the authors(s) and do not necessarily reflect the views of the National Science Foundation.
The works is also supported in part by NIH/NIGMS 1R01GM136780-01 and AFSOR FA9550-21-1-0317.
\end{funding}

\end{document}